\documentclass[reqno,12pt,a4paper]{article}

\usepackage[utf8]{inputenc}
\usepackage{dsfont}
\usepackage{mathtools}
\usepackage{amsthm}
\usepackage{amsmath}
\usepackage{amssymb}
\usepackage{bm}
\usepackage{mathrsfs}
\usepackage{lmodern}
\usepackage{ae}
\usepackage{arydshln}
\usepackage{enumitem}
\usepackage{comment}
\usepackage{stmaryrd}
\usepackage[draft]{fixme}
\usepackage{bbm}

\usepackage{hyperref}
\hypersetup{colorlinks=true}
\newcommand\fnsep{\textsuperscript{\, ,}}

\newcommand{\NN}{\mathbb N}

\newcommand{\RR}{\mathbb R}
\newcommand{\CC}{\mathbb C}

\newcommand{\one}{\mathbbm{1}}

\newcommand{\psigma}{\partial\sigma}

\newcommand{\af}{\mathbf{b}}
\newcommand{\cf}{\mathbf{b^*}}
\newcommand{\ab}{\mathbf{a}}
\newcommand{\cb}{\mathbf{a^*}}
\newcommand{\boson}{\mathrm{b}}
\newcommand{\fermion}{\mathrm{f}}
\newcommand{\wf}{\omega^{(\mathrm{b})}}
\newcommand{\wb}{\omega^{(\mathrm{a})}}

\newcommand{\scrJ}{\mathcal J}
\newcommand{\scrB}{\mathcal B}
\newcommand{\scrA}{\mathcal A}
\newcommand{\scrD}{\mathcal D}
\newcommand{\scrE}{\mathcal E}

\newcommand{\scrF}{\mathcal F}

\newcommand{\scrS}{\mathcal S}
\newcommand{\scrT}{\mathcal T}

\newcommand{\scrM}{\mathcal M}
\newcommand{\scrH}{\mathcal H}
\newcommand{\scrL}{\mathcal L}
\newcommand{\scrP}{\mathcal P}

\newcommand{\Fl}{F_\mathrm{l}}
\newcommand{\Fr}{F_\mathrm{r}}
\newcommand{\tF}{\widetilde{F}}
\newcommand{\tFl}{\widetilde{F}_\mathrm{l}}
\newcommand{\tFr}{\widetilde{F}_\mathrm{r}}
\newcommand{\tA}{\widetilde{A}}
\newcommand{\tB}{\widetilde{B}}
\newcommand{\tC}{\widetilde{C}}

\newcommand{\tK}{\widetilde{K}}
\newcommand{\tR}{\widetilde{R}}

\newcommand{\ti}{{\tilde{i}}}
\newcommand{\tj}{{\tilde{j}}}
\newcommand{\tell}{\tilde{\ell}}
\newcommand{\tchi}{\widetilde{\chi}}

\newcommand{\bB}{\mathbf{B}}

\newcommand{\gothh}{\mathfrak{h}}

\newcommand{\Jab}{J_{\ab}}
\newcommand{\Jcb}{J_{\cb}}
\newcommand{\Jaf}{J_{\af}}
\newcommand{\Jcf}{J_{\cf}}

\newcommand{\set}[2]{\{#1\, |\, #2\}}
\newcommand{\bigset}[2]{\bigl\{#1\, \big|\, #2\bigr\}}
\newcommand{\Bigset}[2]{\Bigl\{#1\, \Big|\, #2\Bigr\}}

\newcommand{\comp}{{\mathrm{c}}}
\newcommand{\mf}{m_{\mathrm{f}}}
\newcommand{\mb}{m_{\mathrm{b}}}
\newcommand{\kdelta}{\bm{\delta}}
\newcommand{\myTheta}{\one}

\newcommand{\re}{\mathrm{Re}}
\newcommand{\im}{\mathrm{Im}}

\newcommand{\id}{\mathrm{id}}

\newcommand{\fin}{\mathrm{fin}}
\newcommand{\Hfin}{\scrH_\fin}

\newcommand{\bare}{\mathrm{bare}}

\newcommand{\Left}{{\leftarrow}}
\newcommand{\Right}{{\rightarrow}}
\newcommand{\LR}{{\leftrightarrow}}

\newcommand{\leftT}{\overset{\leftarrow}{T}}
\newcommand{\rightT}{\overset{\rightarrow}{T}}
\newcommand{\lrT}{\overset{\leftrightarrow}{T}}

\newcommand{\ulj}{\underline{j}}

\newcommand{\uF}{\underline{F}}

\newcommand{\us}{{\underline{s}}}

\newcommand{\uut}{{\underline{\mathbf{t}}}}
\newcommand{\uutau}{{\underline{\mathbf{\tau}}}}
\newcommand{\tus}{{\tilde{\underline{s}}}}

\newcommand{\tuut}{{\tilde{\underline{\mathbf{t}}}}}
\newcommand{\usigma}{{\underline{\sigma}}}

\newcommand{\Card}{\mathrm{Card}}
\newcommand{\Split}{\mathrm{Split}}

\numberwithin{equation}{section}

\theoremstyle{plain}
\newtheorem{Th}{Theorem}[section]
\newtheorem{lem}[Th]{Lemma}
\newtheorem{Prop}[Th]{Proposition}
\newtheorem{Cor}[Th]{Corollary}

\theoremstyle{definition}
\newtheorem{Def}[Th]{Definition} 
\newtheorem{Ex}[Th]{Example}
\newtheorem{Hypothesis}[Th]{Hypothesis}

\theoremstyle{remark}
\newtheorem{rk}[Th]{Remark}

\title{Ultraviolet Renormalisation of a Quantum Field Toy Model II }
\author{Benjamin Alvarez\footnote{Corresponding author}\fnsep\footnote{Email: \texttt{benjamin.alvarez@univ-tln.fr}} \\ Aix Marseille Univ, Univ Toulon, CNRS, CPT, Marseille,\\  France \\
\\
Jacob Schach M\o ller\footnote{Email: \texttt{jacob@math.au.dk}} \\ Department of Mathematics, Aarhus University, \\
 Denmark}

\begin{document}

\maketitle

\tableofcontents

\begin{center}
\textbf{Abstract}
\end{center}
We consider a class of toy models describing a fermion field coupled with a boson field. The model can be viewed as a Yukawa model but with scalar fermions. As in our first paper, the interaction kernels are assumed bounded in the fermionic momentum variable and decaying like $|q|^{-p}$ for large boson momenta $q$. With no restrictions on the coupling strength, we prove norm resolvent convergence to an ultraviolet renormalized Hamiltonian, when the ultraviolet cutoff is removed. We do this by subtracting a sufficiently large, but finite, number of recursively defined self-energy counter-terms, which may be interpreted as arising from a perturbation expansion of the ground state energy. The renormalization procedure requires a spatial cutoff and works in three dimensions provided $p>\frac12$, which is as close as one may expect to the physically natural exponent $p = \frac12$.

\section{Introduction and Main Result} 

Quantum field theory is a successful framework in which three of the four fundamental interactions of Nature can be studied: the electromagnetic, the weak and the strong interaction. However, the computation of physical quantities often leads to divergent expressions that one has to renormalize to obtain a physically reasonable interpretation. 

A popular and effective method is Feymann's diagramatic scheme for renormalizing pertubation expansions, cf. \cite{Peskin:1995ev}. However, Feymann's method does not produce an underlying renormalized model, from which the renormalized perturbation expansions arise. See also \cite{GS14} for the Epstein-Glaser approach to renormalization of pertubation expansions. To this end, the most powerful tools goes through Euclidean field theories and Feynmann-Kac-Nelson formulas. See also the recent papers
\cite{DDJ25,DHYZ25} employing a stochastic PDE approach. However, in these approaches, a renormalized Hamiltonian only appears on the scene as a generator of time translations in a representation of the Poincar\'e group and one does not gain any insights into the structure of the renormalized Hamiltonian that would enable one to analyze its properties.

The goal of this paper is to further develop a technique to construct ultraviolet renormalized models directly in the Hamiltonian picture, which in principle should make it possible to do spectral analysis and scattering theory. See, e.g.,  \cite{DeGe00}. However, to be precise, we are working at the level of resolvents of Hamiltonians, so what we in fact produce are renormalized resolvents. However, this still directly allows for further study, using methods centered on resolvents, such as Birman-Schwinger \cite{Seir23}, Feshbach-Schur type methods \cite{BCFS} and local commutator techniques \cite{ABG}.

The present paper is a continuation of \cite{AlvaMoll2021}, where we only considered the leading order self-energy counter-term for our toy model that one may think of as a Yukawa model with scalar fermions. For a discussion of the structure of the model, we refer the reader to \cite{AlvaMoll2021}. In the present paper, we simply define the model with cutoffs without any motivating discussion.

In \cite{AlvaMoll2021}, we extended a method from  \cite{Eck1970,AW2017} for reordering the Neumann expansion of the interacting resolvent to obtain a renormalized resolvent expansion that permits removal of the ultraviolet cutoff. The method goes back to \cite{Hepp1969}. The models renormalized with this reordering idea have all had conserved fermion number. Our toy model -- like the Yukawa model from which it derives -- does not have any particle number conservation. As in the Yukawa model, we deal with the issues arising from the lack of particle number conservation, by introducing a spatial cutoff into our model. In fact, if -- in the toy model -- one drops the two interaction terms that break total particle number conservation, then our toy model would not be ultraviolet singular at all \cite{Al19}.

Most Hamiltonian ultraviolet renormalization procedures involve only leading order (in perturbation theory) self-energy corrections and - in the case of Yukawa - a mass shift. We refer to \cite{De03_01,FalHin25} for the solvable Van hove Hamiltonian, \cite{ASDJ2025,Lampart2025} for generalized spin boson model, \cite{DeGe00,GlJa77_01} for the Yukawa model and the $P(\phi)_2$ theory, \cite{Lampart2020,  Lampart2025} for  interior
boundary conditions (IBC) type methods, \cite{FalHinMar25,Fr1974,GriWun2018,HinLam24,HinMat24,MM2018,M2006,Nelson1964,AW2017} (and references therein) for Nelson type models with both non-relativistic and relativistic electrons. 

Recently, ultraviolet self-energy renormalization has been performed for the linearly coupled bose polaron model \cite{Lampart2022}, taking into account higher order self-energy contributions. This enables the construction of a renormalized model with the ultraviolet cutoff removed. However, the method employed  in \cite{Lampart2022} only yields strong resolvent convergence to a limiting renormalized Hamiltonian as the ultraviolet cutoff is removed, whereas most of the existing schemes -- at leading order -- gives norm resolvent convergence. In fact, we conjecture that implementing the method of this paper for the bose polaron considered in \cite{Lampart2022}, would yield norm resolvent convergence. We moreover believe that the present strategy could be applied to other models -- at most quadratic in the boson field -- studied in the literature. The calculus for what we call \emph{ordered operators} from Section~\ref{Sec-OrdOp} is designed to handle the fermion field in our model and is not needed if no fermion field is present, e.g. in models like the bose polaron. It should be noted that, indeed, one may fairly easily implement a simpler version of our approach for the solvable Van Hove model \cite{De03_01}. 

The current article proposes a systematic method to recursively take into account self-energy corrections -- of arbitrary order -- in a resolvent resummation scheme, yielding norm resolvent convergence when the ultraviolet cutoff is removed. Another central improvement with respect to \cite{AlvaMoll2021,Eck1970, Lampart2022, AW2017} is our ability to handle a second fermion field without any conserved particle number, neither of the two individual particle species nor of a total particle number. This constitutes a step towards extending the construction of ultraviolet renormalized Hamiltonians requiring an infinite mass shift, like Yukawa, to more singular interactions, including -- possibly -- higher dimension.   

Having a norm-convergent series representation of a renormalized resolvent, one may readily analyse the structure of the domain of the renormalized Hamiltonian, as in \cite{AW2017}. This would be a natural next step that we do not include here, since the paper is already quite long and the toy model studied here is mostly meant as a proof of concept, towards a better understanding of the Yukawa model.

\subsection{The model}

The toy model studied in this article is the same as the one in \cite{AlvaMoll2021}. We refer to \cite{AlvaMoll2021} for a discussion and motivation of the form of the model. 

 The Hilbert space we work in is a tensor product of a bosonic Fock space, $\scrF_{\boson}(\gothh)$, and a fermionic Fock space, $\scrF_{\fermion}(\gothh)$, where $\gothh=L^2(\RR^d)$ and $d\in\NN$ is the spatial dimension of the model. We therefore have: 
\begin{equation}
\scrH = \scrF_{\boson}(\gothh)\otimes \scrF_{\fermion}(\gothh).
\end{equation}
We write $\Omega_\boson\in\scrF_\boson(\gothh)$ for the bosonic vacuum and $\Omega_\fermion\in\scrF_\fermion(\gothh)$ for the fermionic vacuum. The vacuum in $\scrH$ is then the tensor product of the two vacua
\[
\Omega = \Omega_\boson \otimes \Omega_\fermion.
\]
We introduce the so-called creation and annihilation operators $\ab(q)$ and $\cb(q)$ for the bosons and $\af(k)$ and $\cf(k)$ for the fermions, fulfilling the canonical commutation and anti-commutation relations 
\[
\begin{aligned}
[\ab(q_1),\cb(q_2)] & =  \kdelta(q_1-q_2), &\qquad  \{\af(k_1),\cf(k_2)\} & =  \kdelta(k_1-k_2),\\
[\ab(q_1),\ab(q_2)] & =  0, &\qquad  \{\af(k_1),\af(k_2)\} & =  0,\\
[\cb(q_1),\cb(q_2)] & =  0, &\qquad  \{\cf(k_1),\cf(k_2)\} & =  0.
\end{aligned}
\]
Moreover, 
\[
\ab(q) \Omega_\boson = 0 \quad \textup{and} \quad \af(k)\Omega_\fermion = 0.
\]
Here $q_1,q_2,k_1,k_2\in \mathbb R^d$, $[A,B] = AB-BA$ denotes the commutator and $\{A,B\} = AB+BA$ denotes the anti-commutator. We use the letter $q$ for boson momenta and $k$ for fermion momenta.
In addition, due to acting in separate tensor components,
\[
\begin{aligned}
[\ab(q), \af(k)] & =  0, & \qquad [\ab(q), \cf(k)] & =  0, \\
[\cb(q), \af(k)] & =  0, & \qquad [\cb(q), \cf(k)] & =  0. 
\end{aligned}
\]
We recall the relativistic dispersion relations 
\[
\wb(q)=\sqrt{q^2 + \mb^2} \qquad \textup{and} \qquad \wf(k)=\sqrt{k^2+\mf^2},
\]
where $\mb$, respectively  $\mf$, labels the mass of the boson field, respectively the fermion field. We will moreover assume that 
\[
\mb > 0  \qquad \textup{and} \qquad \mf >0.
\]
The free Hamiltonian for the two independent fields is
\[
H_0= \int \wb(q) \cb(q)\ab(q) dq + \int \wf(k) \cf(k) \af(k) dk.
\]

We now turn to the interaction between the two fields.
 The interaction kernels with ultraviolet and spatial cutoffs are:
\begin{equation}
\begin{aligned}
\label{condi01}
G_{\Lambda}^{(1)}(k,q) & =  \frac{h^{(1)}(k,q)}{\wb(q)^p}  g(k-q) \chi\Bigl(\frac{k}{\Lambda}\Bigr)\chi\Bigl(\frac{q}{\Lambda}\Bigr),\\
G_{\Lambda}^{(2)}(k,q) & =  \frac{h^{(2)}(k,q)}{\wb(q)^p}  g(k+q)\chi\Bigl(\frac{k}{\Lambda}\Bigr)\chi\Bigl(\frac{q}{\Lambda}\Bigr),
\end{aligned}
\end{equation}
where the exponent $p$ is a real number that physically should be $p=1/2$. The functions $h^{(1)}$ and $h^{(2)}$ should satisfy the following hypothesis.

\begin{Hypothesis}
\label{Hypothesis-h}
For $j=1,2$,  $h^{(j)}\in L^\infty(\RR^d\times\RR^d)$.
\end{Hypothesis}

Moreover, we impose the following hypotheses on the functions implementing the ultraviolet cutoff, $\chi$, which should approximate the constant function $1$, and the (Fourier transform of a) spatial cutoff, $g$, which should approximate a delta function. 

\begin{Hypothesis}[UV cutoff]
\label{Hypothesis-chi}
The function $\chi\in L^\infty(\RR^d)$ is real-valued with $0\leq \chi\leq 1$ and has compact support $\mathrm{supp}(\chi)$. We furthermore assume that $\chi$ is continuous at $0$ with $\chi(0)=1$. For $\Lambda>0$, we set
$\chi_\Lambda(k) = \chi(k/\Lambda)$.
\end{Hypothesis}

\begin{Hypothesis}[Spatial cutoff]
\label{MainHypothesis}
The spatial cutoff $g\in L^\infty(\RR^d)$ is compactly supported with $\mathrm{supp}(g)$ contained in the unit ball $\set{z\in\CC}{|z|<1}$.
\end{Hypothesis}

\begin{rk}
    Note that the assumptions on $\chi$ and $g$ may be easily relaxed. For example, one may drop the assumption that $\chi$ is real-valued with $0\leq \chi\leq 1$ and one may also relax the assumption that $\chi$ and $g$ has compact support. These requirements are convenient, but not really necessary to establish our main result, Theorem \ref{MainTh}. 
\end{rk}

The regularised Hamiltonian is defined as follow
\begin{equation}
\label{RegularisedHamiltonian}
H_{\Lambda}= H\bigl(G^{(1)}_{\Lambda},G^{(2)}_\Lambda\bigr) = H_0 +  H_\mathrm{I}\bigl(G^{(1)}_{\Lambda},G^{(2)}_\Lambda\bigr),
\end{equation}
where
\begin{equation}\label{CutoffInteraction}
H_\mathrm{I}\bigl(G^{(1)}_{\Lambda},G^{(2)}_\Lambda\bigr)=H^{ \ab \cf }\bigl(G^{(1)}_{\Lambda}\bigr) +
H^{\cb \af}\bigl(\overline{G^{(1)}_{\Lambda}}\bigr)+H^{\cb \cf}\bigl(G^{(2)}_{\Lambda}\bigr)+
H^{\ab \af}\bigl(\overline{G^{(2)}_{\Lambda}}\bigr)
\end{equation}
and, for $F\in L^2(\RR^d\times\RR^d)$,
\begin{equation}\label{interactionterms}
\begin{aligned}
H^{\cb \af}(F) & =  \int F(k,q) \af(k) \cb(q) dk dq, & 
H^{ \ab \cf }(F) &=  \int F(k,q) \cf(k) \ab(q) dk dq,\\
H^{\ab \af}(F) &=  \int F(k,q) \af(k) \ab(q) dk dq, & 
H^{\cb \cf}(F) &=  \int F(k,q) \cf(k) \cb(q) dk dq.
\end{aligned}
\end{equation}

\subsection{Ultraviolet renormalization, the main result}

The following basic theorem has been proved in \cite{AlvaMoll2021}

\begin{Th}[The Hamiltonian with Cutoff]
\label{ThSA}
Let $G^{(1)}_{\Lambda}$ and $G^{(2)}_\Lambda$ be defined as in \eqref{condi01} and assume that the Hypotheses~\ref{Hypothesis-h}, \ref{Hypothesis-chi} and \ref{MainHypothesis} are satisfied. Then the operator defined in \eqref{RegularisedHamiltonian} is self-adjoint with domain $\scrD\left( H\right) = \scrD( H_0)$. Furthermore
$H_\Lambda \geq  - C_\Lambda$, where
\[
C_\Lambda =   1+
\int \Bigl(1+\frac1{\wb(q)}\Bigr) \Bigl(\bigl|G^{(1)}_{\Lambda}(k,q)\bigr|^2 +  \bigl|G^{(2)}_{\Lambda}(k,q)\bigr|^2\Bigr)dk dq.
\]
\end{Th}

In this paper, our goal is to prove the following main theorem.

\begin{Th}[UV Renormalized Hamiltonian]
\label{MainTh}
Let $H_\Lambda$ be as in \eqref{RegularisedHamiltonian} with interaction kernels given by \eqref{condi01},  fulfilling Hypothesis~\ref{Hypothesis-h}, \ref{Hypothesis-chi} and \ref{MainHypothesis}.  Let $E_{\Lambda} = \inf\sigma(H_{\Lambda})$. Suppose finally that $p>\frac{d}{2}-1$. Then there exists a self-adjoint operator $H$, which is bounded from below, such that $H_{\Lambda}-E_{\Lambda}$ converges in norm resolvent sense to  $H$. Moreover, the renormalized operator $H$ does not depend on the choice of the cutoff function $\chi$.
\end{Th}

This result strengthens the one in \cite{AlvaMoll2021}, where $p>\frac{d}{2}-\frac{3}{4}$ had to be assumed. 
The basic idea of the present paper is to add self-energy counter-terms, order by order, to the Hamiltonian $H_\Lambda$ and estimate recursively renormalized resolvent expansions to improve the requirement on the exponent $p$. 

In fact, to zero'th order, when there are no counter-terms, one may simply use the resolvent expansion
\begin{align}\label{IntroNeumann}
\nonumber& (H_\Lambda -z)^{-1} =  (H_0-z)^{-\frac12}\\
&\qquad \biggl\{\sum_{n=0}^\infty \Biggl( -\bigl(H_0-z)^{-\frac12} H_\mathrm{I}\bigl(G^{(1)}_{\Lambda},G^{(2)}_\Lambda\bigr) (H_0-z)^{-\frac12}\Biggr)^n
\biggr\}(H_0-z)^{-\frac12}.
\end{align}
One may use \cite[Lemma~C.1]{AlvaMoll2021} (with $\beta=\frac12$) to argue that if $p > \frac{d}2 -\frac12$, there exists $Z>0$, such that the series is absolutely convergent in norm, uniformly in $\Lambda$,  for $\re(z) \leq -Z$. Furthermore, \cite[Theorem~D.1]{AlvaMoll2021} 
then yields a renormalized Hamiltonian. This would do the trick in $1$ dimension. However, already in $d=2$, this breaks down if $p=\frac12$. However, \cite{AlvaMoll2021} covers the case $d=2$ and $p=\frac12$.

The idea of \cite{AlvaMoll2021} is to take into account the subtraction of the leading order self-energy counter-term
\begin{equation}\label{E2Lambda}
E^{(2)}_{\Lambda} = -\int \frac{\bigl|G^{(2)}_{\Lambda}(k,q)\bigr|^2}{\wf(k)+\wf(q)} dk dq,
\end{equation}
which cancels a singularity in the resolvent expansion above, coming from vacuum expectation values of products of the form
\[
H^{\ab\af}_\Lambda\bigl(\overline{G^{(2)}_\Lambda}\bigr) R_0(z)  H^{\cb\cf}_\Lambda\bigl(G^{(2)}_\Lambda\bigr).
\]
We urge the reader to get an overview of \cite{AlvaMoll2021}, in order to fix some of the underlying ideas in a much simpler setting, before going into the weeds of the present paper. 

 As we progress order by order we push the exponent $p$ higher and higher towards the limiting exponent $\frac{d}2 -1$, where we expect the self-energy renormalization to break down, as it does for the far simpler Van Hove model, cf.~\cite{De03_01}. That is, in dimension $d=3$, we can get as close as we want to $p=\frac12$, but not actually choose $p=\frac12$.

 We note that, as in \cite{AW2017}, one may exploit the reordered Neumann series expression for the resolvent of the renormalized Hamiltonian, in order to study properties of its domain. Since this paper is already quite long, we have not included such considerations here.

\subsection{Outline of the paper}

 In this subsection, we give an overview of the rest of the paper.
 If one inserts the form of the interaction \eqref{CutoffInteraction} into the Neumann expansion, cf.~\eqref{IntroNeumann}, and multiplies everything out using the distributive law, one gets a sum of all possible expressions of the form
 \begin{equation}\label{NeumannTerms}
  R_0(z) H^{s_1}(G_{s_1,\Lambda})R_0(z) \cdots R_0(z) H^{s_k}(G_{s_k,\Lambda}) R_0(z),
 \end{equation}
 where $k = 0,1,2,3,\dotsc$, the $s_i$'s labels the $4$ possible interactions terms $s_i\in\{\ab\af,\ab\cf,\cb\af,\cb\cf\}$ and $G_{s_i,\Lambda}$ is the matching interaction kernel.

 In Section~\ref{Sec-Signs}, we develop a calculus of strings of signatures, such as $s_1,\dotsc,s_k$ above, which is used at two levels. First of all, to succinctly label and handle  products of operators, such as in \eqref{NeumannTerms}, that are UV-singular, in terms of so-called handed signature strings. Secondly, to decompose arbitrary products appearing in the Neumann expansion into products of singular blocks of a given maximal length.
 
In Section~\ref{Sec-ResExp}, we will introduce the renormalized products of operators, indexed by handed signature strings, that we need to estimate in order to handle the reordered Neumann expansion. The resulting operators we refer to as handed blocks of operators. At the end of Section~\ref{Sec-ResExp}, we will introduce higher-order self-energy counter-terms as well as formulate and establish the reordered Neumann expansion of the resolvent with self-energy subtracted up to a given order. See Theorem~\ref{thm-reordering}. The summands of the reordered Neumann expansion are products of handed blocks of operators.

In Section~\ref{Sec-RegOp}, we take the first step towards estimating the handed blocks of operators from Section~\ref{Sec-ResExp}. We begin by normal ordering the handed blocks, which generates operators called regular Wick monomials that are introduced in Definitions~\ref{WickMonomial} and~\ref{RegularOperator}. We establish a calculus for regular Wick monomials that enables us to form products of such expressions, building in counter-terms from the renormalization procedure, while keeping track of an improving ultraviolet behaviour. See Lemma~\ref{HNFCT} and Lemma~\ref{HBFCTANFCO}. At the end of Section~\ref{Sec-RegOp}, we express the handed blocks of operators from the reordered Neumann series as finite sums of regular Wick monomials with their self-energy subtracted, cf.~Lemma~\ref{InductionLemma}.

 In Section~\ref{Sec-OrdOp}, we address an obstacle towards establishing estimates on the regular Wick monomials from Section~\ref{Sec-RegOp}. The main difficulty to overcome is to use boundedness of smeared fermionic annihilation and creation operators, which leads us to rewrite the regular Wick monomials as a sum of what we call ordered Wick monomials that are better suited to exploit smeared fermionic operators. Ordered Wick monomials are defined in Definition~\ref{OrderOperator}. 
 Any handed block of operators, which can be written as a finite sum of regular Wick monomials, is then a finite sum of ordered Wick monomials according to Lemma~\ref{ReorderingFermionLemma}.
 
 In Section~\ref{Sec-GSetsEst}, we establish norm estimates of regular Wick monomials, by passing first to ordered Wick monomials and then estimating these operators. This is done in Propositions~\ref{RNFCT} and~\ref{RFCT}. Using these results, any handed block of operators, may then be estimated to obtain the crucial Proposition~\ref{RegularityofthegeneralisedGsets}.

 Finally, the main result, Theorem~\ref{MainTh}, is proved in Section~\ref{Sec-MainProof}. After using Theorem~\ref{thm-reordering} to reorder the Neumann series, Proposition~\ref{FinalEstimate1}  is used to prove that the reordered Neumann series is absolutely convergent, independently of the ultraviolet cutoff parameter $\Lambda$. We finally take $\Lambda$ to infinity to get norm convergence to the resolvent of a self-adjoint operator $H$, the renormalized Hamiltonian.

\subsection{Miscellaneous notation}

 Throughout the paper, we will use the following notation.
 
 We denote by $\NN = \{1,2,3,\dotsc\}$ the set of natural numbers excluding $\{0\}$. We write $\NN_0 = \NN\cup\{0\}$ and $\NN_{0,\infty} = \NN_0\cup\{\infty\}$.  

 For $a,b\in\NN$ with $a\leq b$, we will use the notation $\llbracket a,b\rrbracket$ for the set of natural numbers ranging from $a$ to $b$, i.e., $\llbracket a,b\rrbracket = \{a,a+1,\dotsc,b-1,b\}\subset\NN$.

  We use the notation $\langle \cdot \vert \cdot \rangle$ for inner products, adopting the convention that $\langle \cdot \vert \cdot \rangle$ is conjugate linear in the first variable and linear in the second variable.

  We will often need to take fractional powers $\xi^\alpha$, $\alpha\in\RR$, of $\xi\in \CC\setminus (-\infty,0]$. Here we implicitly use a complex logarithm $\ln$ defined on $\CC\setminus (-\infty,0]$ to define $\xi^\alpha = e^{\alpha\ln(\xi)}$. Moreover, for an invertible, densely defined, normal operator $T$ with spectrum $\sigma(T)\subset \CC\setminus (-\infty,0]$, the spectral theorem ensures that $T^\alpha$ is well-defined.

Let $N_\boson$ denote the bosonic number operator acting on $\scrF_{\boson}(\gothh)$ and
let $N_\fermion$ denote the fermionic number operator acting on $\scrF_{\fermion}(\gothh)$. We abbreviate $N = N_\boson\otimes\one_{\scrF_\fermion(\gothh)} + \one_{\scrF_\boson(\gothh)}\otimes N_\fermion$ for the total number operator. We will typically drop the extra identity in the tensor product above, when $N_\boson$ and $N_\fermion$ are acting in $\scrH$, e.g.; $N_\boson$ will be used in place of $N_\boson\otimes\one_{\scrF_\fermion(\gothh)}$ and $N_\fermion$ will be used in place of $\one_{\scrF_\boson(\gothh)}\otimes N_\fermion$.

 We will be using smeared fermion annihilation and creation operators, defined for $f\in L^2(\RR^d)$ by
 \begin{equation}\label{SmearedFermions}
    \cf(f) = \int f(k) \cf(k)dk \quad \textup{and} \quad \af(f) = \int \overline{f(k)} \af(k)dk.
 \end{equation}
 Note that with this convention, we have $\af(f)^* = \cf(f)$. In addition, we recall that $\af(f)$ and $\cf(f)$ are bounded operators with $\|\af(f)\|\leq \|f\|$ and $\|\cf(f)\|\leq \|f\|$. The corresponding smeared boson operators, which we will not make explicit use of, are unbounded.
 
\section{Signature Strings and their Calculus}\label{Sec-Signs}

 What we are aiming for is to renormalize the Neumann expansion of the interacting resolvent. Each term in the Neumann expansion is a product of  the four interaction terms sandwiched by free resolvents. We will be grouping such arbitrary products into blocks that we can renormalize and estimate. How we group the factors in a long product will only depend on which of the four interaction terms are sitting between each of the resolvent pairs. In order to handle the renormalization procedure, we first build a calculus for strings of signatures $(s_1,s_2,\dotsc,s_k)$, where the $s_k$'s denotes one of the four possible interaction term. In a first reading, the reader may safely skip the proofs throughout Section~\ref{Sec-Signs}.

 \subsection{Handed signature strings}

\begin{Def}\label{def-signature}
By a \emph{signature}, we understand a choice of one of four labels $s\in \{\ab\af,\ab\cf,\cb\af,\cb\cf\}$. For $k\geq 1$, we write $\scrS^{(k)}_0 = \{\ab\af,\ab\cf,\cb\af,\cb\cf\}^k$ for the $k$-fold Cartesian product of the set of $4$ basic signatures. An element $\us$ of $\scrS^{(k)}_0$ is called a \emph{signature string} of length $k$. 
\end{Def}

We introduce functions that read of the annihilation and creation operator content in a single signature, and functions that count the difference between the numbers of annihilation and creation operators of a given type in a signature string.

\begin{Def}\label{def-countingfunctions}
Define two functions $n_\ab\colon \{\ab\af,\ab\cf,\cb\af,\cb\cf\} \to \{-1,+1\}$ and
$n_\af\colon \{\ab\af,\ab\cf,\cb\af,\cb\cf\} \to \{-1,+1\}$ by setting
\[
n_\ab(s) = \begin{cases} +1, & s\in \{\ab\af,\ab\cf\} \\
-1, & s\in \{\cb\af,\cb\cf\} \end{cases} \quad \textup{and} \quad 
n_\af(s) = \begin{cases} +1, & s\in \{\ab\af,\cb\af\} \\
-1, & s\in \{\ab\cf,\cb\cf\}. \end{cases}
\]
For $j,j',k\in\NN$ with $1\leq j\leq j'\leq k$ and
$\us=(s_1,\dotsc,s_k)\in\scrS^{(k)}_0$, we define
\[
n_\ab(j,j';\us)  = \sum_{i=j}^{j'} n_\ab(s_i) \qquad \textup{and} \qquad n_\af(j,j';\us) = \sum_{i=j}^{j'} n_\af(s_i).
\]
\end{Def}

 We will also be needing an involution of signatures and signature strings that mirrors the action of taking the adjoint of an operator.

\begin{Def}\label{def-stringinvolution} For a single signature $s$, we define and involution $s\mapsto s^*$ by setting:
\[
(\ab\af)^* = \cb\cf, \quad (\ab\cf)^* = \cb\af, \quad (\cb\af)^* = \ab\cf \quad \textup{and} \quad (\cb\cf)^* = \ab\af.
\]
For a string of signatures $\us\in \scrS^{(k)}_0$ of length $k$, we define
\[
\us^* = (s_k^*,\dotsc,s_2^*,s_1^*). 
\]
\end{Def}

\begin{rk}
Note that we have the basic rules
\begin{equation}\label{sign-adjointcount}
\begin{aligned}
n_\ab(s^*) &= -n_\ab(s), \quad  & n_\ab(j,j';\us^*) &= -n_\ab(k-j'+1,k-j+1;\us), \\
n_\af(s^*) &= - n_\af(s), \quad & n_\af(j,j';\us^*) &= -n_\af(k-j'+1,k-j+1;\us).
\end{aligned}
\end{equation}
\end{rk}

We are now ready to formulate the two main definitions of this subsection. They identify the types of signature string that will correspond to renormalized blocks of operator products. Renormalised blocks are built recursively. Some blocks can be estimated using only resolvents from the right while others using only resolvents from the left. When two neighbouring blocks compete for the same resolvent in the middle, they will be concatenated and, in some cases, renormalized. See~Definition~\ref{def-remblocks}. The order in which annihilation and creation operators may appear in the resulting (renormalized) operator products can therefore be described recursively, mirroring the recursive construction of the renormalized blocks of operators in Definition~\ref{def-remblocks}. In Remark~\ref{rem-signs-rec} below, we will briefly explain how to construct \emph{handed signature strings} through a recursive construction. However, a recursive definition of \emph{handed signature strings} is difficult to work with. We use as a definition a more convenient direct characterization of the handed signature strings, instead of the recursive construction discussed in Remark~\ref{rem-signs-rec}. This is what is done in the following two definitions.

\begin{Def}[Handed signature strings]\label{def-handedsignatures}
    Let $k\in \NN$ and $\us \in\scrS^{(k)}_0$. We say that $\us$ is \emph{handed} if the following two conditions are satisfied:
    \begin{enumerate}[label = \textup{(\arabic*)}]
        \item if $k\geq 2$, we have
          \begin{equation}\label{def-handedbasic}
    \begin{aligned}
        &\forall i\in \llbracket 1,k-1\rrbracket:  & \begin{cases} n_\ab(1,i;\us)\geq 0 \\
        n_\ab(1,i;\us)=0\Rightarrow n_\af(1,i;\us) \geq 0,
        \end{cases}\\
        &\forall i\in \llbracket 2,k\rrbracket:  & \begin{cases} n_\ab(i,k;\us)\leq 0 \\ n_\ab(i,k;\us)=0\Rightarrow n_\af(i,k;\us) \leq 0.
        \end{cases}
    \end{aligned}
    \end{equation}
    \item if $k\geq 3$ and $n_\ab(1,k;\us)=n_\af(1,k;\us)=0$, then we have the stronger conclusion.
 \begin{equation}\label{def-ambi-right}
 \forall i\in\llbracket 2,k-1\rrbracket:\qquad  n_\ab(1,i;\us)=0\Rightarrow n_\af(1,i;\us)>0.
 \end{equation} 
 \end{enumerate}
 We write $\scrS^{(k)}$ for the subset of $\scrS^{(k)}_0$ consisting of handed signature strings. 
\end{Def}

\begin{rk}\label{rem-handedsigns} Let $k\in\NN$.
    \begin{enumerate}[label = \textup{(\roman*)}]
        \item\label{rem-signslength1} If $k=1$, and any $\us= (s)\in\scrS^{(1)}_0$, all conditions in Definition~\ref{def-handedsignatures} are trivially satisfied. Hence, any of the four strings of length $1$ are handed.
        \item\label{rem-signslength2} If, $k=2$ and $\us=(s_1,s_2)\in\scrS^{(2)}_0$, the requirement \eqref{def-handedbasic} simply imposes that $n_\ab(s_1) = 1$ and $n_\ab(s_2) 
 =-1$. Hence, the handed strings $\scrS^{(2)}$ of length $2$ are 
        \[
        (\ab\af,\cb\cf), \quad (\ab\cf,\cb\af),\quad (\ab\cf,\cb\cf),\quad  (\ab\af,\cb\af).
        \]
        \item\label{item-handed-endpoints} If $k\geq 2$ and $\us\in \scrS^{(k)}$, then $n_\ab(s_1) = 1$ and $n_\ab(s_k) = -1$. This follows from \eqref{def-handedbasic}.
          \item\label{item-handedtotalsum} If $\us\in\scrS^{(k)}$, then $n_\ab(1,k;\us)\in\{-1,0,+1\}$. This is obvious for $k=1$ and for $k\geq 2$ it follows from \ref{item-handed-endpoints}, \eqref{def-handedbasic} and the estimates
          \begin{align*}
            -1 &= n_\ab(s_k) \leq n_\ab(1,k-1;\us) + n_\ab(s_k)= n_\ab(1,k;\us)\\
            & = n_\ab(s_1) +n_\ab(2,k;\us) \leq n_\ab(s_1) = +1.
          \end{align*}
          \item From \ref{item-handedtotalsum} it immediately follows that if $k$ is even, then $n_\ab(1,k;\us) = 0$, and if $k$ is odd, then $n_\ab(1,k;\us) \in\{-1,+1\}$.
          \item\label{item-handedsymmetric} Suppose $k\geq 3$ and $\us\in\scrS^{(k)}$ with $n_\ab(1,k;\us) =n_\af(1,k;\us)= 0$. Then the following property dual to \eqref{def-ambi-right} holds true
          \begin{equation}\label{def-ambi-left}
\forall i\in\llbracket 2,k-1\rrbracket:\qquad  n_\ab(i,k;\us)=0\Rightarrow n_\af(i,k;\us)<0.
          \end{equation}
           To see this, let $i\in \llbracket 2,k-1\rrbracket$ with $n_\ab(i,k;\us)=0$. Then $n_\ab(1,i-1;\us) = n_\ab(1,i-1;\us) + n_\ab(i,k;\us) = n_\ab(1,k;\us)=0$. Note that by \ref{item-handed-endpoints}, we must therefore have $i-1\geq 2$. But then it follows from \ref{def-ambi-right}  that $n_\af(1,i-1;\us)>0$.  Consequently, $n_\af(i,k;\us) = n_\af(1,k;\us) - n_\af(1,i-1;\us)  = - n_\af(1,i-1;\us)<0$.
          \item Let $\us\in\scrS_0^{(k)}$. It is a consequence of \ref{item-handedsymmetric} that $\us$ is a handed signature string if and only if $\us^*$ is a handed signature string.
    \end{enumerate}
\end{rk}

 We divide the handed string into three distinct types. The name differentiates between signature strings with a surplus of annihilation or creation operators with a third class being signature strings corresponding to operator products that preserve both boson and fermion particle numbers.  

\begin{Def}\label{def-lra}
   Let $k\in\NN$ and let $\us\in\scrS^{(k)}$ be a handed signature string of length $k$.  We say that $\us$ is:
    \begin{enumerate}[label = \textup{(\arabic*)}]
        \item\label{item-def-rh} \emph{right-handed} if in addition 
           $n_\ab(1,k;\us)\in \{0,1\}$ and $n_\ab(1,k;\us) = 0\Rightarrow n_\af(1,k;\us) >0$.
        \item\label{item-def-lh} \emph{left-handed} if in addition 
           $n_\ab(1,k;\us)\in \{-1,0\}$ and $n_\ab(1,k;\us) = 0\Rightarrow n_\af(1,k;\us) <0$.
        \item\label{item-def-ambi} \emph{ambidextrous} if       
              $n_\ab(1,k;\us) = n_\af(1,k;\us) = 0$.
 \end{enumerate}  
        We write $\scrS^{(k)}_\Right$ for the set of right-handed strings, $\scrS^{(k)}_\Left$ for the left-handed strings, and finally, $\scrS^{(k)}_\LR$ for the ambidextrous signature strings.  
\end{Def}

\begin{rk}\label{rem-signs} We make the following simple observations.
\begin{enumerate}[label = \textup{(\roman*)}]
\item\label{rem-signs-1} $\scrS^{(1)}_\Right = \{(\ab\af),(\ab\cf)\}$, $ \scrS^{(1)}_\Left = \{(\cb\af),(\cb\cf)\}$ and $\scrS^{(1)}_\LR = \emptyset$.
\item $\scrS^{(2)}_\Right = \{(\ab\af,\cb\af)\}$, $\scrS^{(2)}_\Left = \{(\ab\cf,\cb\cf)\}$ and $\scrS^{(2)}_\LR = \{(\ab\af,\cb\cf),\allowbreak(\ab\cf,\cb\af)\}$.
 \item\label{rem-signsdisjoint} For $k\in\NN$, the sets $\scrS^{(k)}_\Right$, $\scrS^{(k)}_\Left$ and $\scrS^{(k)}_\LR$ are mutually disjoint.
\item\label{rem-signleftandright} A string of signatures $\us$ is right-handed if and only if its adjoint $\us^*$ is left-handed.
    \item\label{rem-signambi} A string of signatures $\us$ is ambidexstrous if and only if its adjoint $\us^*$ is ambidexstrous.
    \item\label{rem-sign-evenodd} Suppose $k$ is an odd integer. Then $\scrS^{(k)}_\LR =\emptyset$ and if $\us\in \scrS^{(k)}_\Right$, we have $n_\ab(1,k;\us)=1$, and if $\us\in\scrS^{(k)}_\Left$, we have $n_\ab(1,k;\us)=-1$. 
    \item\label{rem-signnovacuumR} For $\us\in\scrS^{(k)}_\Right$, the property \eqref{def-ambi-right} is satisfied. Indeed, let $i\in \llbracket 2,k-1\rrbracket$ and assume that $n_\ab(1,i;\us)=0$. Then we have $n_\ab(1,k;\us) = n_\ab(i+1,k) \leq 0$, hence we must have $n_\ab(1,k;\us)=n_\ab(i+1;\us)= 0$ and therefore $n_\af(1,k;\us)>0$, (since $\us$ is right-handed) and $n_\af(i+1,k;\us)\leq 0$ by \eqref{def-handedbasic}. Since $n_\af(1,i;\us) = n_\af(1,k;\us)- n_\ab(i+1,k;\us)>0$, we may now conclude the claim. 
    \item\label{rem-signnovacuumL} For $\us\in\scrS^{(k)}_\Left$,the property \eqref{def-ambi-left} holds true. Indeed, this follows from \ref{rem-signleftandright} and~\ref{rem-signnovacuumR} above by taking adjoints.
\end{enumerate}
\end{rk}

\subsection{Signature calculus}

\begin{Def}[Composition of signature strings]
  Let $\us = (s_1,s_2,\dotsc,s_k)\in\scrS^{(k)}_0$ and $\us' = (s'_1,s'_2,\dotsc,s'_{k'})\in\scrS^{(k')}_0$ be two strings of signatures. We define a new string of signatures $\us'' = \us\circ\us'$ of length $k''= k+k'$ by setting
  $\us''= (s_1,s_2,\dotsc,s_k,s'_1,s'_2\dotsc,s'_{k'})$. We call $\us''$ the composition of $\us$ and $\us'$.
\end{Def}

\begin{rk}\label{rem-signComp} Let $k,k'\in\NN$, $\us\in\scrS^{(k)}_0$ and $\us'\in \scrS^{(k')}_0$.
Note that $(\us\circ\us')^* = \us'^{*}\circ \us^*$.
\end{rk}

\begin{rk}\label{rem-signs-rec} Before continuing, we pause to explain how to construct handed signature strings recursively. The starting point is to set
$\scrS^{(1)}_\Right = \{(\ab\af),(\ab\cf)\}$, $ \scrS^{(1)}_\Left = \{(\cb\af),(\cb\cf)\}$ and $\scrS^{(1)}_\LR = \emptyset$, cf. Remark~\ref{rem-signs}~\ref{rem-signs-1}.
Secondly, fix $k>1$ and assume the sets $\scrS^{(\ell)}_\Right, \scrS^{(\ell)}_\Left$ and $\scrS^{(\ell)}_\LR$ have been defined for $\ell\in\NN$ with $1\leq \ell<k$. 
Then we set
\begin{align*}
\scrS^{(k)} = \Bigset{\us_1\circ \us_2}{\exists 1\leq \ell< k:\ \us_1\in \scrS^{(\ell)}_\Right, \us_2 \in \scrS^{(k-\ell)}_\Left\cup\scrS^{(k-\ell)}_\LR }\\
\cup \, \Bigset{\us_1\circ \us_2}{\exists 1\leq \ell< k:\ \us_1\in \scrS^{(\ell)}_\Right\cup\scrS^{(\ell)}_\LR, \us_2 \in \scrS^{(k-\ell)}_\Left }
\end{align*}
and define
 \begin{align*}
\scrS^{(k)}_\Right &= \Bigset{\us\in \scrS^{(k)}}{n_\ab(\us) > 0 \textup{ or } n_\ab(\us) = 0\textup{ and }n_\af(\us)>0},\\
\scrS^{(k)}_\Left &= \Bigset{\us\in \scrS^{(k)}}{n_\ab(\us) < 0 \textup{ or } n_\ab(\us) = 0\textup{ and }n_\af(\us)<0},\\
\scrS^{(k)}_\LR &= \Bigset{\us\in \scrS^{(k)}}{n_\ab(\us) = 0\textup{ and }n_\af(\us)=0}.
 \end{align*}
That the recursively defined sets of handed signature strings coincide with those from Definitions~\ref{def-handedsignatures} and~\ref{def-lra}, follows from Propositions~\ref{prop-signcomp} and~\ref{propo-signdecomp} below.
\end{rk}

When we build renormalized operator blocks in Subsect. \ref{subsec-RenBlocks}, we compose operators corresponding to handed signature strings. The operator product will correspond to a concatenation of handed signature strings. The following proposition shows that handedness is preserved under the relevant concatenations. 

\begin{Prop}[Composition of handed signature strings] \label{prop-signcomp}
    Let $k_\Right,k_\Left,\allowbreak k_\LR\in\NN$, and let $\us_\Right\in \scrS^{(k_\Right)}_\Right$, $\us_\Left\in\scrS^{(k_\Left)}_\Left$ and $\us_\LR\in\scrS^{(k_\LR)}$ be handed signature strings. The following holds true:
    \begin{enumerate}[label = \textup{(\arabic*)}]
        \item\label{item-SignComp1} $\us_\Right\circ \us_\LR\in\scrS^{(k_\Right+k_\LR)}_\Right$ and $\us_\LR\circ \us_\Left \in \scrS^{(k_\LR+k_\Left)}_\Left$.
        \item\label{item-SignComp2}  $\us = \us_\Right\circ\us_\Left\in \scrS^{(k)}$ with $k= k_\Right+k_\Left$. More precisely, $n_\ab(1,k;\us) \in\{-1,0,1\}$ and
        \begin{enumerate}[label = \textup{(\arabic{enumi}\alph*)}]
            \item if $n_\ab(1,k;\us) = 1$, or $n_\ab(1,k;\us) = 0$ and $n_\af(1,k;\us) >0$,  then $\us \in\scrS^{(k)}_\Right$,
            \item if $n_\ab(1,k;\us) = -1$,  or $n_\ab(1,k;\us) = 0$ and $n_\af(1,k;\us) <0$,then $\us \in\scrS^{(k)}_\Left$,
            \item if $n_\ab(1,k;\us) = n_\af(1,k;\us) = 0$, then $\us\in\scrS^{(k)}_\LR$.
        \end{enumerate}
    \end{enumerate}
\end{Prop}

\begin{proof} We divide the proof into three steps.

 \emph{Step I:} As a first step, let $\us'\in \scrS^{(k')}_\Right\cup\scrS^{(k')}_\LR$ and $\us''\in \scrS^{(k'')}_\Left\cup\scrS^{(k'')}_\LR$  be two handed signature strings of lengths $k'$ and $k''$. Let $\us = \us'\circ \us''$ and $k=k'+k''$. We aim to show that $\us$ satisfies \eqref{def-handedbasic}.
  
   First, for any $i \in \llbracket 1, k' \rrbracket$, we have $n_\ab(1,i;\us) = n_\ab(1,i;\us') \geq 0$ and $n_\ab(1,i;\us)=0$ implies $n_\ab(1,i;\us')=0$ and therefore $n_\af(1,i;\us)=n_\af(1,i;\us')\geq 0$. Here we used that $\us'$ satisfies \eqref{def-handedbasic} if $i<k'$ and that $\us'$ is either right-handed or ambidexstrous if $i=k'$, cf. Definition~\ref{def-lra}.    

   If now  $i \in \llbracket  k' +1, k-1\rrbracket $, then $n_\ab(1,i;\us) = n_\ab(1,k'; \us')+n_\ab(1,i-k';\us'')\geq 0$. Moreover, if $n_\ab(1,i;\us)=0$ then $n_\ab(1,k';\us') = n_\ab(1,i-k';\us'') = 0$ and therefore $n_\af(1,k'; \us'),n_\af(1,i-k';\us'') \geq 0$. Consequently, $n_\af(1,i;\us)\geq 0$.

  Summing up, we have established the first line in \eqref{def-handedbasic}. That the rest of \eqref{def-handedbasic} also holds now follows by passing to the adjoint $\us^* = (\us'\circ\us'')^* = (\us'')^*\circ (\us')^*$, which is again a composition of two handed signature strings as considered above. See Remark~\ref{rem-signs}~\ref{rem-signleftandright} and~\ref{rem-signambi}, as well as Remark~\ref{rem-signComp}. Invoking the part of \eqref{def-handedbasic} that was just proved for $\us^*$ and going back with the involution completes the argument. See also \eqref{sign-adjointcount}.

 \emph{Step II:} We now establish \ref{item-SignComp1}. It suffices to consider the case $\us = \us_\Right\circ \us_\LR$ and $k = k_\Right + k_\LR$. The other case again follows by passing to the adjoint.
    
 Recalling that $n_\ab(1,k_\LR;\us_\LR)= n_\af(1,k_\LR;\us_\LR)=0$, we have
        \[
        n_\ab(1,k;\us) = n_\ab(1,k_\Right;\us_\Right) \qquad \textup{and}\qquad n_\af(1,k;\us) = n_\af(1,k_\Right;\us_\Right).
        \]
        Since $\us_\Right\in\scrS^{(k_\Right)}_\Right$, we conclude that
        $n_\ab(1,k;\us)\in\{0,1\}$ and if $n_\ab(1,k;\us) = 0$, we have $n_\af(1,k;\us)>0$. This proves that $\us\in \scrS^{(k_\Right+k_\LR)}_\Right$.

\emph{Step III:}        We finally turn to \ref{item-SignComp2}. Let $\us_\Right\in\scrS^{(k_\Right)}_\Right$, $\us_\Left\in\scrS^{(k_\Left)}_\Left$ and abbreviate  $k=k_\Right+k_\Left$ and $\us = \us_\Right \circ \us_\Left$. 

        We begin by ensuring that $\us$ is a handed signature string. For this we still need to establish \eqref{def-ambi-right} and, hence, we may
        suppose $n_\ab(1,k;\us) = n_\af(1,k;\us) = 0$,  $k\geq 3$ and that we have an $i\in \llbracket 2,k-1\rrbracket$ with $n_\ab(1,i;\us)=0$. 
        
        If $i\leq k_\Right$. Then $n_\ab(1,i;\us_\Right) =0$. Since $\us_\Right$ is right handed, we conclude that $n_\af(1,i;\us)= n_\af(1,i;\us_\Right)>0$. See Remark~\ref{rem-signs}~\ref{rem-signnovacuumR} for $i< k_\Right$ and Definition~\ref{def-lra}~\ref{item-def-rh} if $i = k_\Right$.

         If $i>k_\Right$, then $n_\ab(i+1-k_\Right,k_\Left;\us_\Left) = n_\ab(i+1,k;\us)=0$, we may conclude that $n_\af(i+1-k_\Right,k_\Left;\us_\Left)<0$. Here we used Remark~\ref{rem-signs}~\ref{rem-signnovacuumL}.
        Hence, $n_\af(1,i;\us) = - n_\af(i+1,k;\us) = - n_\af(i+1-k_\Right,k_\Left;\us_\Left)>0$, which completes the proof of \eqref{def-ambi-right}.

        Having established that $\us$ is a handed signature string, it must fall into one of the three possible categories, left-handed, right-handed or ambidexstrous.
\end{proof}

 Note that the compositions considered in Proposition~\ref{prop-signcomp} are the only possible. To make this precise, we have the following.

\begin{lem}\label{lem-badcomp} Let $k',k''\in\NN$ and $\us'\in\scrS^{(k')}_0$ and $\us''\in\scrS^{(k'')}_0$. Set $\us= \us'\circ\us''$ and $k= k'+k''$. We have
\begin{enumerate}[label = \textup{(\arabic*)}]
    \item\label{item-badcomp-right} Suppose $\us''\in\scrS^{(k'')}_\Right$. Then $\us\not\in\scrS^{(k)}$. 
\item\label{item-badcomp-left} Suppose $\us'\in\scrS^{(k')}_\Left$. Then
$\us\not\in\scrS^{(k)}$. 
\item\label{item-badcomp-ambiR} Suppose $\us'\not\in\scrS^{(k')}_\Right$ and $\us''\in\scrS^{(k'')}_\LR$. Then
$\us\not\in\scrS^{(k)}$.
\item\label{item-badcomp-ambiL} Finally, suppose $\us'\in\scrS^{(k')}_\LR$ and $\us''\not\in\scrS^{(k'')}_\Left$. Then
$\us\not\in\scrS^{(k)}$.
\end{enumerate}
\end{lem}

\begin{proof} We begin with \ref{item-badcomp-right}.
Suppose towards a contradiction that $\us\in\scrS^{(k)}$. 
Since $\us$ satisfies \eqref{def-handedbasic}, we have $n_\ab(1,k'';\us'') = n_\ab(k'+1,k;\us)\leq 0$ and since $\us''$ is right-handed, we must have $n_\ab(1,k'';\us'')=0$ and $n_\af(1,k'';\us'')>0$. But since we may now also observe that $n_\ab(k'+1,k;\us) = 0$, we may use again that $\us$ is assumed handed to conclude that $n_\af(1,k'';\us'')=n_\af(k'+1,k;\us)\leq 0$. This establishes a contradiction. 

The claim \ref{item-badcomp-left} follows from \ref{item-badcomp-right} by passing to adjoints, cf.~Remark~\ref{rem-signs}~\ref{rem-signleftandright} and Remark~\ref{rem-signComp}.

To see \ref{item-badcomp-ambiR},  assume towards a contradiction that $\us\in\scrS^{(k)}$. From the observation that $n_\ab(k'+1,k;\us)= n_\ab(1,k'';\us'')=0$ and $n_\af(k'+1,k;\us)=n_\af(1,k'';\us'')=0$, we conclude that $\us'$ must 
satisfy \eqref{def-handedbasic}. Indeed, for $j\in\llbracket 1,k'-1\rrbracket$, we have $n_\ab(1,j;\us')=n_\ab(1,j;\us)\geq 0$ and
$n_\ab(1,j;\us')=0\Rightarrow n_\ab(1,j;\us)=0 \Rightarrow n_\af(1,j;\us) = n_\af(1,j;\us')\geq 0$. Similarly, for $j\in \llbracket 2,k'\rrbracket$, we have $n_\ab(j,k';\us') = n_\ab(j,k;\us) \leq 0$ and
$n_\ab(j,k';\us')=0 \Rightarrow n_\ab(j,k;\us)=0\Rightarrow n_\af(j,k';\us')= n_\af(j,k;\us)\leq 0$. 

Since $n_\ab(k'+1,k;\us)= n_\af(k'+1,k;\us)=0$, we see that $k'+1< k$, cf. Remark~\ref{rem-handedsigns}~\ref{item-handed-endpoints}, and therefore that $\us$ does not have the property \eqref{def-ambi-left}. Hence it follows that $\us$ must be right-handed. See Remark~\ref{rem-handedsigns}~\ref{item-handedsymmetric} and Remark~\ref{rem-signs}~\ref{rem-signnovacuumL}. But this implies that $n_\ab(1,k';\us')=n_\ab(1,k';\us)\geq 0$ and if $n_\ab(1,k';\us')=0$, then we have $n_\af(1,k';\us')= n_\ab(1,k';\us)>0$. Here we used  Remark~\ref{rem-signs}~\ref{rem-signnovacuumR}. But then $\us'$ must be a right-handed signature string, which is a contradiction.

 Finally, \ref{item-badcomp-ambiL} follows from \ref{item-badcomp-ambiR} by passing to adjoints, where we again make use of Remark~\ref{rem-signs}~\ref{rem-signleftandright} and~\ref{rem-signambi} and Remark~\ref{rem-signComp}.
\end{proof}
 
When we recursively construct renormalized operator blocks affiliated with handed signature strings, cf. Definition~\ref{def-remblocks}, we need to decompose a handed signature string into a concatenation of two shorter handed signature strings. We then use the renormalized operator blocks pertaining to the two shorter strings to recursively build an operator affiliated with the original string. In Remark~\ref{rem-SplitIndependence}, we show that this construction does not depend on the choice of decomposition of the handed signature string. For this purpose, we need to classify the ways in which we may decompose a handed signature string into two shorter handed signature strings. This is the purpose of the following proposition.

\begin{Prop}[Decomposition of handed signature strings]\label{propo-signdecomp} Let $k\in\NN$ with $k\geq 2$ and $\us\in\scrS^{(k)}$. Let $\Split(\us)$ denote the set of $j\in\llbracket 1,k-1\rrbracket$, such that $(s_1,\dotsc,s_{j})\in\scrS^{(j)}_\Right\cup\scrS^{(j)}_\LR$ and
$(s_{j+1},\dotsc,s_k)\in\scrS^{(k-j)}_\LR\cup \scrS^{(k-j)}_\Left$. Then $\Split(\us)\neq \emptyset$ and we furthermore have
\begin{enumerate}[label = \textup{(\arabic*)}]
\item\label{item-sd-right} Suppose $\us\in\scrS^{(k)}_\Right$ and $j\in \Split(\us)$, then $(s_1,\dotsc,s_{j})\in\scrS^{(j)}_\Right$. 
\item\label{item-sd-left} Suppose $\us\in\scrS^{(k)}_\Left$ and $j\in\Split(\us)$, then
$(s_{j+1},\dotsc,s_k)\in\scrS^{(k-j)}_\Left$.
\item\label{item-sd-ambi} Suppose $\us\in\scrS^{(k)}_\LR$ and $j\in \Split(\us)$, then we have both $(s_1,\dotsc,s_{j})\in\scrS^{(j)}_\Right$ and $(s_{j+1},\dotsc, s_k)\in\scrS^{(k-j)}_\Left$.
\item\label{item-sd-deg} Write $\Split(\us) = \{j_1,\dots,j_u\}$ with $j_1<j_2<\dotsc,<j_u$. If $u\geq 2$, we furthermore have:
\begin{enumerate}[label = \textup{(\arabic{enumi}\alph*)}]
    \item\label{item-sd-dega} If $\us\in \scrS^{(k)}_\Right$, then $(s_{j_u+1},\dotsc,s_k)\in\scrS^{(k-j_u)}_\Left$.
    \item\label{item-sd-degb} If $\us\in \scrS^{(k)}_\Left$, then $(s_1,\dotsc,s_{j_1})\in\scrS^{(k-j_1+1)}_\Right$.
    \item\label{item-sd-degc} For any $i\in\llbracket 1,u-1\rrbracket$, we have
$(s_{j_i+1},\dotsc,s_{j_i})\in\scrS^{(j_{i+1}-j_i)}_\LR$.
\end{enumerate}
\end{enumerate}
\end{Prop}

\begin{proof} Let $k\in\NN$ with $k\geq 2$ and $\us\in\scrS^{(k)}$. 

 The key task is to prove that $\Split(\us)\neq\emptyset$. 
In the process we establish the properties \ref{item-sd-right}--\ref{item-sd-ambi} of the two factors, in the case where $\Split(\us)$ is a singleton.

    Consider $\us \in \scrS^{(k)} $. We may assume that $\us \in \scrS^{(k)}_\Right\cup \scrS^{(k)}_\LR$, the case $\us \in \scrS^{(k)}_\Left $ being the adjoint case of $\us \in \scrS^{(k)}_\Right $ can be deduced from the case treated. See Remark~\ref{rem-signs}~\ref{rem-signleftandright} and~\ref{rem-signambi} as well as Remark~\ref{rem-signComp}.
    
    For $N=0,1$, abbreviate
    \[
     \Split_N(\us) = \bigset{i\in\llbracket 1,k-1\rrbracket }{n_\ab(1,i;\us)=N}.
    \]
    Since $n_\ab(1,1;\us)=1$, we observe that $\Split_1(\us)\neq\emptyset$. 
    Fix $N=0$ if $\Split_0(\us)\neq\emptyset$, otherwise set $N=1$.

    We separate the argument into several steps.

\emph{Step I:}   
   Fix $i_0$ to be the \emph{largest} $i_0\in \Split_N(\us)$, such that we have 
    \begin{equation}\label{choice-of-i0}
    n_\af(1,i_0;\us) =\min\bigset{n_\af(1,i;\us)}{i\in \Split_N(\us)}.
    \end{equation}
  Note that 
  \begin{equation}\label{decomp-strictlypos}
   N=0 \quad \Rightarrow  \quad n_\af(1,i_0;\us) >0. 
  \end{equation}
  This holds since \eqref{def-ambi-right} is applicable both for right-handed and ambidexstrous signature strings, cf. Remark~\ref{rem-signs}~\ref{rem-signnovacuumR}. (Observe that if $N=0$ then $k\geq 3$, since -- by Remark~\ref{rem-signs}~\ref{item-handed-endpoints} -- we have  $n_\ab(1,1;\us)=1$).
   We proceed to argue that $i_0\in \Split(\us)$.  
  Put $\us' = (s_1,\dotsc,s_{i_0})$ and $\us'' = (s_{i_0+1},\dotsc, s_k)$.
  
   \emph{Step II:}   In this step, we argue that $\us'\in\scrS^{(i_0)}_\Right$. In fact, due to \eqref{decomp-strictlypos}, it only remains to prove that
     \begin{equation}\label{decompsign-need1}
\forall j\in \llbracket 2,i_0\rrbracket:\quad \begin{cases}  n_\ab(j,i_0;\us') \leq 0\\
 n_\ab(j,i_0;\us')=0 \Rightarrow n_\af(j,i_0;\us')\leq 0.
\end{cases}
     \end{equation}
 First note that there is nothing to prove if $i_0=1$, so we may assume that $i_0\geq 2$. Then for $j\in\llbracket 2,i_0\rrbracket$, compute
 $N = n_\ab(1,i_0;\us) = n_\ab(1,j-1;\us') + n_\ab(j,i_0;\us') \geq N + n_\ab(j,i_0;\us')$. Here we used that $N$ was chosen to be the smallest $N$ with $\Split_N(\us)\neq\emptyset$.  This implies that $n_\ab(j,i_0;\us)\leq 0$ and we conclude the first part of \eqref{decompsign-need1}. 

    Now suppose that $n_\ab(j,i_0;\us') =0$. Then, from the computation, $n_\ab(1,j-1;\us)=  n_\ab(1,j-1;\us') +n_\ab(j,i_0;\us') = n_\ab(1,i_0;\us) = N$, we conclude that $j-1\in \Split_N(\us)$ as well. Compute $n_\af(1,i_0;\us) = n_\af(1,j-1;\us)+n_\af(j,i_0;\us')$. Due to the choice \eqref{choice-of-i0} of $i_0$, we conclude that  $n_\af(j,i_0;\us')\leq 0$. This completes the verification of \eqref{decompsign-need1}.

\emph{Step III:} We now turn to proving that $\us''\in\scrS^{(k-i_0)}$. We first prove that
   \begin{equation}\label{decompsign-need2}
\forall j\in \llbracket 1,k-i_0-1\rrbracket:\quad \begin{cases} n_\ab(1,j;\us'') \geq 0 \\
  n_\ab(1,j;\us'')=0 \Rightarrow n_\af(1,j;\us'') > 0.
\end{cases}
     \end{equation}
 Note that there is nothing to prove here if $i_0 = k-1$. So for the verification of \eqref{decompsign-need2}, we assume that $i_0\leq k-2$.

 Let $j\in \llbracket 1, k-i_0-1\rrbracket$. From the computation \begin{equation}\label{decomp-stepIIhelp}
     n_\ab(1,i_0+j;\us) =n_\ab(1,i_0;\us) + n_\ab(i_0+1,i_0+j;\us) = N + n_\ab(1,j;\us''),
     \end{equation}
    and the inequality $n_\ab(1,i_0+j;\us)\geq N$, coming from the choice of $N$, the first part of \eqref{decompsign-need2} follows.

 Now suppose that $n_\ab(1,j;\us'') = 0$. Then the equation \eqref{decomp-stepIIhelp}
  implies that $i_0+j\in \Split_N(\us)$. Note that we have the fermionic analogue of the computation \eqref{decomp-stepIIhelp}:
  \[
  n_\af(1,i_0+j;\us) = n_\af(1,i_0;\us)+ n_\af(i_0+1,i_0+j;\us) = n_\af(1,i_0;\us)+ n_\af(1,j;\us'').
  \]
  This equation and the choice of $i_0$ as the largest element in $\Split_N(\us)$ satisfying \eqref{choice-of-i0}, leads us to conclude that $n_\af(1,j;\us'') > 0$. This completes the proof of \eqref{decompsign-need2}.

  Note that \eqref{decompsign-need2} also covers \eqref{def-ambi-right}, such that we have completed the proof of $\us''\in\scrS^{(k'')}$.

 \emph{Step IV:} In Step III, we proved that $\us''$ is a handed signature string. In order to conclude that $\us''\in\scrS^{(k'')}_\Left\cup \scrS^{(k'')}_\LR$, it suffices to argue that $\us''\not\in \scrS^{(k'')}_\Right$. See Remark~\ref{rem-signs}~\ref{rem-signsdisjoint}.

 But this follows from Lemma~\ref{lem-badcomp}~\ref{item-badcomp-right}, since $\us = \us'\circ\us''$ is a handed signature string.

 It follows from the composition rules in Proposition~\ref{prop-signcomp}~\ref{item-SignComp1} that if $\us$ is ambidexstrous and $\us'$ is right-handed, then $\us''$ is forced to be left-handed. 

 This completes the proof that $\Split(\us)\neq \emptyset$ and that the properties in \ref{item-sd-right}--\ref{item-sd-ambi} holds, if $\Split(\us)$ is a singleton.

\emph{Step V:}
 If $\Split(\us)$  is a singleton, the properties of the factors in \ref{item-sd-right}, \ref{item-sd-left} and \ref{item-sd-ambi} were established above. We therefore turn our attention to the case when $\Split(\us)$ has at least two elements. 

  Assume that $\us$ is right-handed or ambidexstrous. The left-handed case will as usual follow by passing to the adjoint $\us^*$.

Let $j,j'\in \Split(\us)$ with $j<j'$. Since $(s_1,\dotsc,s_{j'})$ and $(s_{j+1},\dotsc,s_k)$ both satisfy \eqref{def-handedbasic}, we conclude that $n_\ab(j+1,j';\us)$ is both non-negative and non-positive, hence we must have $n_\ab(j+1,j';\us)=0$. Appealing to \eqref{def-handedbasic} yet again, we may now conclude that $n_\af(j+1,j';\us)=0$ as well. Hence, we have established that
  \begin{equation}\label{sd-preamb}
 \forall j,j'\in \Split(\us), j<j':\quad    \begin{cases} (s_{j+1},\dotsc,s_{j'}) \quad \textup{satisfies \eqref{def-handedbasic}},\\
 n_\ab(j+1,j';\us) = n_\af(j+1,j';\us)=0.
 \end{cases}
 \end{equation}

\emph{Step VI:} We are now in a position to establish \ref{item-sd-right}, \ref{item-sd-ambi} and \ref{item-sd-dega}. Note again that \ref{item-sd-left} and \ref{item-sd-degb} follow from \ref{item-sd-right} and \ref{item-sd-dega} by taking adjoints. 

 Let $j\in \Split(\us)$ and abbreviate $\us'=(s_1,\dotsc,s_j)$ and $\us''= (s_{j+1},\dotsc,s_k)$. By Remark~\ref{lem-badcomp}, we know that $\us'$ cannot be left-handed and $\us''$ cannot be right-handed. Hence, it suffices to show that neither $\us'$ not $\us''$ can be ambidexstrous. Note that we have either $n_\ab(1,j;\us)>0$ or $n_\ab(1,j;\us)=0$ and $n_\af(1,j;\us)>0$, since $\us$ is assumed either right-handed or ambidexstrous. This implies that $\us'$ cannot be ambidexstrous and therefore is right-handed as claimed.

 Assume towards a contradiction that $\us''$ is ambidexstrous. Let $j'\in \Split(\us)$ with $j'\neq j$. 
 
 If $j'>j$, we would by \eqref{sd-preamb} have that $n_\ab(1,j'-j;\us'')= n_\ab(j+1,j';\us)=0$ and $n_\af(1,j'-j;\us'') = n_\af(j+1,j';\us)=0$, which contradicts \eqref{def-ambi-right}, which $\us''$ should satisfy. 
 
 Hence $j'<j$. But then, by a similar argument, the signature string $\us'''=(s_{j'},\dotsc,s_k)$ cannot be ambidexstrous, and therefore must be left-handed. But this is contradicted by the computation $n_\ab(1,k-j';\us''') = n_\ab(j'+1,k;\us) = n_\ab(j'+1,j;\us) + n_\ab(j+1,k;\us)= 0$, which uses \eqref{sd-preamb} and the assumption that $\us''$ is ambidexstrous. Similarly, we have $n_\af(1,k-j';\us''')=0$. But this is impossible, since $\us'''$ is left-handed.

This completes the proof of \ref{item-sd-right}--\ref{item-sd-ambi} as well as \ref{item-sd-dega} and \ref{item-sd-degb}.

 \emph{Step VII:} It remains to establish \ref{item-sd-degc}.
 Write $\Split(\us) = \{j_1,\dotsc,j_u\}$ with $j_1<j_2<\dots<j_u$ and $u\geq 2$. Set $\us_i = (s_{j_i+1},\dotsc, s_{j_{i+1}})$, for $i\in\llbracket 1,u-1\rrbracket$. Keeping \eqref{sd-preamb} in mind, we observe that in order to conclude that $\us_i$ is ambidexstrous, it suffices to verify the property \eqref{def-ambi-right}. In particular, we observe that $j_{i+1}-j_i \geq 2$ for all $i\in\llbracket 1,u-1\rrbracket$.

 Fix an $i\in\llbracket 1,u-1\rrbracket$ and $\ell\in \llbracket 1, j_{i+1}-j_i -1\rrbracket$.
 We must argue that $n_\ab(1,\ell;\us_i) = 0$ implies $n_\af(1,\ell;\us_i) >0$.  
Assume towards a contradiction that this is false and $n_\ab(1,\ell;\us_i) = n_\af(1,\ell;\us_i)=0$. The strategy is to prove that this assumption will imply that $j_i+\ell\in\Split(\us)$, which would be a contradiction.

Observe first that
\begin{align*}
n_\ab(j_i+\ell+1,j_{i+1};\us) & =  n_\ab(1,\ell;\us_i) + n_\ab(j_i+\ell+1,j_{i+1};\us) \\
 & =  n_\ab(j_i+1,j_i+\ell;\us) + n_\ab(j_i+\ell+1,j_{i+1};\us) \\
 & = n_\ab(j_i+1,j_{i+1};\us) = 0.
\end{align*}
Similarly, we may compute
\[
n_\af(j_i+\ell+1,j_{i+1};\us)  = 0.
\]
Set $\usigma' = (s_1,\dotsc,s_{j_i+\ell})$ and $\usigma'' = (s_{j_i+\ell+1},\dotsc, s_k)$. Given the two computations above, together with \eqref{sd-preamb}, 
it is now straightforward to check that both $\usigma'$ and $\usigma''$ satisfy \eqref{def-handedbasic}. Furthermore, since $\us'= (s_1,\dotsc,s_{j_{i}})$ is right-handed and $\us'' = (s_{j_{i+1}+1},\cdots,s_k)$ is left-handed, we use the two identities above to verify that $\usigma'$ is right-handed and $\usigma''$ is left-handed. Note that the property \eqref{def-ambi-right} is not present in this case.
This shows that $j_i+\ell\in \Split(\us)$ and establishes a contradiction, since there are no elements from $\Split(\us)$ between $j_{i}$ and $j_{i+1}$.
\end{proof}

\subsection{Tuples of handed signature strings}

In this section, we will be discussing minimal splits of arbitrary signature strings into handed signature strings of a given maximal length. 
Recall that in a Neumann expansion of the resolvent we get terms of the form \eqref{NeumannTerms} with all possible signature strings appearing as summands. In order to estimate such summands we will resum the Neumann expansion order by order arriving at new summands that are products of longer and longer blocks of renormalized operators, corresponding to a decomposition of an arbitrary signature string into successive handed signature strings of longer and longer length. 
We will be using the notation
\[
\ulj = (j_1,\dotsc,j_\ell) \quad \textup{and}\quad |\ulj| = j_1+\cdots +j_\ell
\]
for multiindices $\ulj$ with integer entries $j_i\in\NN$ and length $|\ulj|$. The number $\ell\in\NN$ counting the number of entries in $\ulj$ is suppressed from the notation.

\begin{Def}\label{def-tuples} Let $k,n \in \NN$. We introduce the set  $\scrT^{(n,k)}$  of tuples of handed signature strings of total length $k$ and maximal block length $n$ as follows:
\begin{align*}
\scrT^{(n,k)}  & = \Bigset{ \uut = (\us_1, \dots, \us_\ell) \in  \scrS^{(j_1)} \times \dots \times  \scrS^{(j_\ell)} }{  \ell \in \NN, \ulj \in \llbracket 1,n\rrbracket^\ell, \,|\ulj| = k, \\
 &  \qquad \forall i \in \llbracket 1, \ell-1\rrbracket: \, j_i+j_{i+1}\leq n \Rightarrow \us_i\circ \us_{i+1} \notin  \scrS^{(j_i+j_{i+1})} }.
\end{align*}
\end{Def} 

\begin{rk}\label{rem-tuples} 
If $n = 1$ and $k\in\NN$, the set of tuples $\scrT^{(1,k)}$ is the same as the set of all strings of length $k$: $\scrT^{(1,k)}\simeq \scrS^{(k)}_0$ through the trivial identification
\[
\bigl((s_1),(s_2),\dotsc,(s_k)\bigr) \leftrightarrow (s_1,s_2,\dotsc,s_k).
\]

In the above example with $n=1$ there is only one tuple corresponding to the same string of signatures. This may not be the case if $n>1$. If $n=3$ and $k=4$, the following two tuples from $\scrT^{(3,4)}$ illustrates this:
\[
\bigl((\ab\af,\ab\af,\cb\cf), (\cb\cf)\bigr) \quad \textup{and} \quad 
\bigl((\ab\af),(\ab\af,\cb\cf, \cb\cf)\bigr).
\]
The underlying signature string $(\ab\af,\ab\af,\cb\cf,\cb\cf)$ is the same for both tuples.
\end{rk}

The tuples from $\scrT^{(k,n)}$ will be used to index the summands in the renormalized Neumann expansion of the resolvent. The redundancy in Remark~\ref{rem-tuples} above gives rise to an over-counting issue. One may resolve this by adding more constraints in Definition~\ref{def-tuples} in order to remove the redundancy. We proceed differently, and lift the redundancy using an equivalence relation on $\scrT^{(k,n)}$.

\begin{Def}
Let $k,n \in \NN$. We introduce the following equivalence relation $\sim$ defined on $\scrT^{(n,k)}$ as follows
\[
\forall \uut, \uut' \in \scrT^{(n,k)}:\qquad  \uut \sim \uut' \quad \Leftrightarrow \quad \us_1\circ \dots \circ \us_\ell = \us'_1\circ \dots \circ \us'_{\ell'}.
\]
\end{Def}

\begin{Prop}\label{prop-tuple-bijection} Let $k,n\in\NN$ the following map from $\scrT^{(n,k)}/\!\sim$ to $\scrS^{(k)}_0$ is a bijection:
\[
\scrT^{(n,k)}/\!\sim\, \ni \, [\uut] \,\mapsto\, \us_1\circ\us_2\cdots\circ\us_\ell.
\]
Write $\varphi^{(n,k)}\colon \scrS^{(k)}_0\to \scrT^{(n,k)}/\!\sim$ for the inverse bijection.
\end{Prop}

\begin{proof}
 It is clear from the choice of equivalence relation that the map is well-defined and injective. It remains to show that the map is surjective. 
 We do this by induction after $n$. For $n=1$ it follows from Remark~\ref{rem-tuples}.

 Let $n\in\NN$ and assume that the proposition has been proved for $\scrT^{(n,k)}$. Let $\us = (s_1,s_2,\dotsc,s_k)\in\scrS^{(k)}_0$. By the induction hypothesis, there exists a $\uut = (\us_1,\dotsc,\us_\ell)\in \scrT^{(n,k)}$ with
 $\us_1\circ\cdots\circ \us_\ell  = \us$. 

 We introduce an auxiliary set of tuples
 \begin{align}\label{auxtuples}
\nonumber \scrT^{(n+1,k)}_\mathrm{A}  & = \Bigset{(\us_1, \dots, \us_{\ell}) \in  \scrS^{(j_1)} \times \dots \times  \scrS^{(j_{\ell})} }{  \ell \in \NN, \ulj \in \llbracket 1,n+1\rrbracket^{\ell}, \,|\ulj| = k, \\
 &  \qquad \forall i \in \llbracket 1, \ell-1\rrbracket: \, j_i+j_{i+1}\leq n \Rightarrow \us_i\circ \us_{i+1} \notin  \scrS^{(j_i+j_{i+1})} }.
\end{align}
Clearly $\scrT^{(n,k)}\subset \scrT^{(n+1,k)}_\mathrm{A}$.

Let us return to $\uut\in \scrT^{(n,k)}$ and let $I(\uut) = \set{i\in\llbracket 1,\ell-1\rrbracket}{j_i+j_{i+1} = n+1, \us_i\circ\us_{i+i}\in\scrS^{(n+1)}}$. If $I(\uut)=\emptyset$, then
$\uut\in\scrT^{(n+1,k)}$ and we are done. We may therefore assume that $I(\uut)\neq \emptyset$. Let $i\in I(\uut)$. Define a new tuple $\uut' = (\us'_1,\dotsc,\us'_{\ell-1})$ by setting
\[
\forall k\in \llbracket 1,\ell-1\rrbracket:\quad \us'_k = \begin{cases} 
\us_k, & k< i\\
\us_i\circ\us_{i+1}, & k=i\\
\us_{k+1}, & k> i.
\end{cases}
\]
Note that $\uut'\in \scrT^{(n+1,k)}_\mathrm{A}$ and the cardinality of $I(\uut')$ is strictly less than that of $I(\uut)$. By repeating the procedure, removing one more $i\in I(\uut')$, one eventually reaches a $\uut''=(\us''_1,\dotsc,\us''_{\ell''})\in\scrT^{(n+1,k)}_\mathrm{A}$ with $I(\uut'') = \emptyset$. But we then have $\uut''\in\scrT^{(n+1,k)}$. Since the underlying signature does not change in the recursive procedure, we find that
$\us''_1\circ \cdots \circ\us''_{\ell''} = \us$ and we are done.
\end{proof}

For convenience, we introduce functions that read off the starting index $b$ and ending index $e$ of each $\us_i$ inside a tuple $\uut$.

\begin{Def} Let $n,k\in\NN$ and $\uut = (\us_1,\dotsc,\us_\ell)\in\scrT^{(n,k)}$. We define
\[
\forall i\in\llbracket 1,\ell\rrbracket:\quad \begin{cases}
b(i;\uut) = k-(j_1+\cdots+j_i)+1\\
e(i;\uut) = j_1+\cdots+ j_i.
\end{cases}
\]
Here the $j_i$'s are the lengths of the $\us_i$'s.
We furthermore set: 
\begin{align*}
\scrB(\uut) & =  \bigset{b(i;\uut)}{i\in\llbracket 1,\ell\rrbracket, \us_i\in \scrS^{(j_i)}_\Right},\\
\scrE(\uut) & =  \bigset{e(i;\uut)}{i\in\llbracket 1,\ell\rrbracket, \us_i\in \scrS^{(j_i)}_\Left},\\
\scrA(\uut) &= \bigset{\bigl(b(i,\uut),e(i,\uut)\bigr)}{i\in\llbracket 1,\ell\rrbracket, \us_i\in \scrS^{(j_i)}_\LR}.
\end{align*}
\end{Def}

\begin{rk} The terms \eqref{NeumannTerms} in the Neumann expansion of the interacting resolvent are naturally indexed by elements of $\scrS_0^{(k)}$, whereas terms in the resummed -- to order $n$ -- Neumann expansion will be indexed by elements of $\scrT^{(n,k)}/\!\sim$. See Proposition~\ref{prop-tuple-bijection} above and Theorem~\ref{thm-reordering} below. The summands in the resummed Neumann expansion are constructed in Section~\ref{subsec-RenSummands}, a priori for each tuple $\uut\in\scrT^{(n,k)}$. A key observation is made in  Proposition~\ref{prop-EquivSummands}, where we establish that if $\uut\sim\uut'$, then the associated renormalized summands are the same. The arguments we make go via induction after $n$, the order of resummation. For the induction arguments to work, we need to understand what features of a tuple are the same for all equivalent tuples. This is the content of Proposition~\ref{prop-equivalent-tuples} below.

However, before we get to Proposition~\ref{prop-equivalent-tuples}, we pause to address an
issue with Definition~\ref{def-tuples} that we need to handle first. The issue is exemplified by the following tuple $\uut\in\scrT^{(4,5)}$:
\[
\uut = \bigl((\ab\af),(\ab\af,\cb\cf,\ab\cf,\cb\cf)\bigr).
\]
The second handed signature string $(\ab\af,\cb\cf,\ab\cf,\cb\cf)$ of maximal length $4$ naturally splits into a composition of two shorter handed signature strings
$(\ab\af,\cb\cf)$ and $(\ab\cf,\cb\cf)$. However, if we replace the handed string of length $4$ with the two shorter handed signature strings of length $2$, we get
\[
\uut' = \bigl((\ab\af),(\ab\af,\cb\cf), (\ab\cf,\cb\cf)\bigr).
\]
But this is not an element of $\scrT^{(3,5)}$, because the composition of the first two handed signature strings  $(\ab\af)\circ(\ab\af,\cb\cf) = (\ab\af,\ab\af,\cb\cf)$ is an element of $\scrS^{(3)}$. This will be an issue in Section~\ref{subsec-RenSummands}, when we wish to argue by induction that the renormalized summands attached to a tuple $\uut$ in a Neumann expansion is independent of the representative $\uut$. 

The tool we use to overcome the obstruction outlined in this remark is Lemma~\ref{lem-equivalent-tuples} just below, which ensures that given a tuple $\uut$,  we can always find a representative $\uut'$ from the same equivalence class as $\uut$ that does not have the splitting issue that $\uut$ may have. In the example above, we can choose
\[
\uut' = \bigl((\ab\af,\ab\af,\cb\cf),(\ab\cf,\cb\cf)\bigr)\in\scrT^{(4,5)}.
\]
In this case, there are no handed signature strings of maximal length $4$ in $\uut'$. 
\end{rk}

\begin{lem}\label{lem-equivalent-tuples}  Let $n,k\in\NN$ and $\uut = (\us_1,\dotsc,\us_\ell)\in\scrT^{(n+1,k)}$. There exists a $\uut' = (\us'_1,\dotsc,\us'_{\ell'})\in\scrT^{(n+1,k)}$ with $\uut\sim\uut'$, such that
\[
\scrB(\uut) = \scrB(\uut'),\quad \scrE(\uut) = \scrE(\uut'), \quad \scrA(\uut) = \scrA(\uut')
\]
and such that for any $i\in\llbracket 1,\ell'\rrbracket$ with $\us'_i\in\scrS^{(n+1)}$, there exists a split $\us'_i = \us'_{i;1}\circ\us'_{i;2}$ with $\us'_{i;u}\in \scrS^{(j_{i;u})}$, $u=1,2$, and $j_{i;1} + j_{i;2} = j_i$ satisfying
\begin{align}
\label{tuplehelp1}
&\textup{If } i\geq 2, & & j_{i-1}+j_{i;1}\leq n & & \textup{and } \us'_{i;1}\in\scrS^{(j_{i;1})}_\LR, & &\textup{then } \us'_{i-1}\not\in \scrS^{(j_{i-1})}_\Right,\\
\label{tuplehelp2}
&\textup{If } i\leq \ell'-1, & &j_{i;2}+j_{i+1}\leq n & & \textup{and } \us'_{i;2}\in\scrS^{(j_{i;2})}_\LR, & &\textup{then } \us'_{i+1}\not\in \scrS^{(j_{i+1})}_\Left.
\end{align}
\end{lem}

\begin{proof} We divide the proof into two steps. 

\emph{Step I:} Setting the stage. Let us introduce
    \[
    I=\bigset{ i \in\llbracket 1, \ell \rrbracket}{ \us_i \in \scrS^{(n+1)} }.
    \]
  Observe that from Proposition~\ref{propo-signdecomp}, we know that for any $i \in I$,  there exists $j_{i;1},j_{i;2}\in\NN$ with $j_{i;1}+j_{i;2} = n+1$ as well as $\us_{i;1} \in  \scrS^{(j_{i;1})}$ and $\us_{i;2} \in  \scrS^{(j_{i;2})}$, such that
  \[
   \us_i = \us_{i;1}\circ\us_{i;2}
  \]
  and:
  \begin{itemize}[left=0pt .. \parindent]
      \item if $\us_i \in \scrS^{(n+1)}_\Right$, then we may choose  $\us_{i;1} \in  \scrS^{(j_{i;1})}_\Right$ and $\us_{i;2} \in  \scrS^{(j_{i;2})}_\Left\cup \scrS^{(j_{i;2})}_\LR$,
       \item if $\us_i \in \scrS^{(n+1)}_\Left$, then we may choose $\us_{i;1} \in  \scrS^{(j_{i;1})}_\Right\cup \scrS^{(j_{i;1})}_\LR$ and $\us_{i;2} \in  \scrS^{(j_{i;2})}_\Left$,
       \item if $\us_i \in \scrS^{(n+1)}_\LR$, then we may choose $\us_{i,1} \in  \scrS^{(j_{i;1})}_\Right$ and $\us_{i;2} \in  \scrS^{(j_{i;2})}_\Left$.
  \end{itemize}

  We now define two mutually disjoint subsets of $I$
  \begin{align*}
   I_\Left &= \bigset{i\in I}{i\geq 2,\, \us_{i;1}\in\scrS^{(j_{i;1})}_\LR \textup{ and } \us_{i-1}\in\scrS^{(j_{i-1})}_\Right \textup{ with } j_{i-1}+j_{i;1}\leq n},\\
   I_\Right &= \bigset{i\in I}{i\leq \ell-1,\, \us_{i;2}\in\scrS^{(j_{i;2})}_\LR \textup{ and } \us_{i+1}\in\scrS^{(j_{i-1})}_\Left \textup{ with } j_{i;2}+j_{i+1}\leq n}.
  \end{align*}
  Note that if $I_\Left= I_\Right = \emptyset$, we may simply pick $\uut' = \uut$ and be done with the proof. Note also that if $i\in I_\Left$, then $\us_i$ and $\us_{i;2}$ are left-handed and if $i\in I_\Right$, then $\us_i$ and $\us_{i;1}$ are right-handed. See Lemma ~\ref{lem-badcomp}.
  
  Observe that for $i\in I_\Left$, we must have $i-1\not\in I$, and for $i\in I_\Right$, we must have $i+1\not\in I$. Furthermore, for any $i\in\llbracket 1,\ell\rrbracket$, we cannot have both $i-1\in I_\Right$ and $i+1\in I_\Left$.

\emph{Step II:} We are now in a position to construct $\uut'$.  We define a new element of $\scrT^{(n+1,k)}$,
$\uut' = (\us'_1,\dotsc,\us'_\ell)$, with the same number of elements $\ell$ as $\uut$, by setting
\[
\forall i\in\llbracket 1,\ell\rrbracket:\quad \us'_i = \begin{cases} 
\us_{i;2}, & i \in I_\Left\\
\us_{i}\circ\us_{i+1;1}, & i+1\in I_\Left\\
\us_{i;1}, & i \in I_\Right\\
\us_{i-1;2}\circ\us_{i}, & i-1\in I_\Right\\
\us_i, & \textup{otherwise}.
\end{cases}
\]
Clearly, $\uut\sim\uut'$. Observe that for any $i$, the handedness of $\us_i$ and $\us'_i$ is the same. Furthermore, if $\us_i$ and $\us'_i$ are right-handed, then $b(i;\uut) = b(i,\uut')$, if $\us_i$ and $\us'_i$ are left-handed, then $e(u;\uut) = e(i,\uut')$, and finally if $\us_i$ and $\us'_i$ are ambidexstrous, then nothing has been altered and we have  $\us_i = \us_i'$ and hence,  $(b(i;\uut),e(i;\uut)) = (b(i;\uut'),e(i;\uut'))$. 

This completes the proof, since by construction, $\uut'$ satisfies both \eqref{tuplehelp1} and \eqref{tuplehelp2}.
\end{proof}

\begin{Prop}\label{prop-equivalent-tuples} Let $k,n,n'\in\NN$,
     $\uut\in\scrT^{(n,k)}$ and $\uut' \in \scrT^{(n',k)}$. Suppose that $\uut \sim \uut'$.\footnote{The equivalence relation is meaningful also if $n\neq n'$.} We have the following
    \begin{enumerate}[label = \textup{(\arabic*)}]
        \item\label{item-equivtuples1}  If $n=n'$, then
        \[
        \scrB(\uut) = \scrB(\uut'), \quad \scrE(\uut) = \scrE(\uut'), \quad
        \scrA(\uut) = \scrA(\uut').
        \]
        \item\label{item-equivtuples2} If $n'<n$, then
         \[
        \scrB(\uut) \subset \scrB(\uut') \quad \textup{and} \quad \scrE(\uut) \subset \scrE(\uut').
        \]
    \end{enumerate}
\end{Prop}

\begin{proof} 
We begin with \ref{item-equivtuples1}. In fact \ref{item-equivtuples2} will then follow from the proof of \ref{item-equivtuples1} for $n'<n$.

The proof goes by induction after $n$. If $n=1$, the statement is obvious. Hence we assume that \ref{item-equivtuples1} is true for some $n\geq 1$ and we must then verify  \ref{item-equivtuples1} with $n$ replaced by $n+1$.

Let $\uut,\uut'\in \scrT^{(n+1,k)}$ with $\uut\sim\uut'$. 
By Lemma~\ref{lem-equivalent-tuples}, we may assume that both $\uut$ and $\uut'$ satisfy the properties \eqref{tuplehelp1} and \eqref{tuplehelp2}.

\emph{Step I:} In this step, we produce $\tuut = (\tus_1,\dotsc,\tus_{\tell}),\tuut'= (\tus'_1,\dotsc,\tus'_{\tell'})\in \scrT^{(n,k)}$ with $\tuut\sim \uut, \tuut'\sim\uut'$ and hence $\tuut\sim\tuut'$.

It suffices to discuss how to generate $\tuut$, since the same construction will apply to obtain $\tuut'$.

We will be working in the auxiliary set of tuples \eqref{auxtuples}, denoted by $\scrT^{(n+1,k)}_\mathrm{A}$, which appeared in the proof of Proposition~\ref{prop-tuple-bijection}. Note that
\begin{align*}
\scrT^{(n+1,k)} & \subset \scrT_\mathrm{A}^{(n+1,k)} \\
\scrT^{(n,k)} & = \Bigset{\uut\in\scrT^{(n+1,k)}_\mathrm{A}}{\forall i\in\llbracket 1,\ell\rrbracket:\, j_i\leq n}.
\end{align*}

Set
\[
I(\uut)  = \bigset{i\in\llbracket 1,\ell\rrbracket}{j_i = n+1}.
\]
For $i\in I(\uut)$, we use the splitting of $\us_i = \us_{i;1}\circ \us_{i;2}$ coming from Lemma~\ref{lem-equivalent-tuples}.

Our strategy is to replace $\uut$ by another tuple $\uutau\in\scrT^{(n+1,k)}_\mathrm{A}$ satisfying  
 \begin{equation}\label{recursiontotau}
    \uut\sim\uutau,\quad     \scrB(\uut) \subset \scrB(\uutau) \quad \textup{and} \quad \scrE(\uut) \subset \scrE(\uutau),
        \end{equation}
        such that $I(\uutau)\subsetneq I(\uut)$. 
Indeed, we define $\uutau= (\usigma_1,\dotsc,\usigma_{\ell+1})$ by the reverse procedure as in the proof of Proposition~\ref{prop-tuple-bijection}.
Fix an $i_0\in I(\uut)$ and define:
\[
\forall i\in\llbracket 1,\ell+1\rrbracket:\quad \usigma_i = \begin{cases} 
\us_{i_0;1}, & i = i_0\\
\us_{i_0:2}, & i = i_0+1\\
\us_{i}, & i <i_0\\
\us_{i-1}, & i> i_0+1.
\end{cases}
\]
Then indeed, $\uutau\in\scrT^{(n+1,k)}_\mathrm{A}$ by Lemma~\ref{lem-equivalent-tuples}, the cardinality of $I(\uutau)$ is one less than that of $I(\uut)$ and \eqref{recursiontotau} holds true.

This sets up a recursive procedure, we can use to reach, in finitely many steps, a $\tuut\in\scrT^{(n+1,k)}_\mathrm{A}$ with $I(\tuut) = \emptyset$, and hence a $\tuut\in\scrT^{(n,k)}$ satisfying
 \begin{equation}\label{recursiontotau2}
    \uut\sim \tuut,\quad    \scrB(\uut) \subset \scrB(\tuut) \quad \textup{and} \quad \scrE(\uut) \subset \scrE(\tuut).
        \end{equation}        
The same procedure can be run on $\uut'$ to obtain a $\tuut'$ with the same properties.

\emph{Step II:} Invoking the induction hypothesis. Let $b\in \scrB(\uut)\subset \scrB(\tuut)$. 
We proceed -- via $\tuut$ and $\tuut'$ -- to argue that $b\in\scrB(\uut')$ as well.

Since $b\in\scrB(\uut)$, there exists $i\in\llbracket 1,\ell\rrbracket$ such that
$\us_i$ is right-handed with $b = b(i;\uut)$. Similarly, since $b\in\scrB(\tuut)$, there exists $\ti\in\llbracket 1,\tell\rrbracket$ such that $\tus_{\ti}$ is right-handed with $b = b(\ti;\tuut)$.

By the induction hypothesis $b\in \scrB(\tuut')$ and there exists $\ti'\in \llbracket 1,\tell'\rrbracket$, such that $\tus'_{\ti'}$ is right-handed with $b = b(\ti';\tuut')$.
If $b\in \scrB(\uut')$ there is nothing to prove, so we assume towards a contradiction that $b\not\in \scrB(\uut')$. But then, since $b\in\scrB(\tuut')$, the index $b$ must be the start of a left-handed or ambidexstrous string of length $n+1$ that is split, when passing from $\uut'$ to $\tuut'$ during the construction of $\tuut$ and $\tuut'$ in Step I. 
That is, there must exist an $i'\in\llbracket 1,\ell'\rrbracket$, such that $\us'_{i'}\in \scrS^{(n+1)}_{\Left}\cup \scrS^{(n+1)}_\LR$ and some splitting $\us'_{i'} = \us'_{i';1}\circ \us'_{i';2}$ with $\us'_{i';1} = \tus'_{\ti'}$ right-handed and $\us'_{\ti';2} = \tus'_{\ti'+1}$ left-handed. Here we used Proposition~\ref{propo-signdecomp}~\ref{item-sd-left} and~\ref{item-sd-ambi} to argue that $\us'_{\ti',2}$ must be left-handed. 


Since $\tus'_{\ti'+1}$ is left-handed, we have $b+n = e(\ti'+1;\tuut')\in\scrE(\tuut')$. Again, invoking the induction hypothesis, there exists a $\tj \in \llbracket \ti+1,\tell\rrbracket$, such that $\tus_{\tj}$ is left-handed and
$e(\tj;\tuut) = b+n$.  

We observe that $\tj = \ti+1$, i.e. the left-handed $\tus_{\tj}$ comes immediately after the right-handed $\tus_{\ti}$. Indeed, if this was not the case, then the sequence $\tus_{\ti}, \tus_{\ti+1},\dotsc,\tus_{\tj}$ would contain at least three elements. Since the sequence begins with a right-handed string and ends with a left-handed string, there must at least be one consecutive pair of string $\tus_{u},\tus_{u+1}$ that can be composed following one the composition rules in Proposition~\ref{prop-signcomp}. But then $\tus_{u}\circ\tus_{u+1}$ is a handed signature of length at most $n$, which would be a contradiction with $\tuut\in\scrT^{(n,k)}$.

In conclusion $\tus_{\ti}\circ \tus_{\ti+1} = (s_b\dotsc,s_{b+n})$, which was not a right-handed string. This means that the right-handed string $\us_{i}$ that began at $b$ does not have length $n+1$. This implies that it has not been split during the passage from $\uut$ to $\tuut$ and we therefore must have $\us_i = \tus_{\ti}$. 

But this would absurdly imply that $\us_{i+1} = \tus_{\ti+1}$ 
and $\us_i\circ\us_{i+1}\in\scrS^{(n+1)}$. We have arrived at a contradiction and can conclude that $\scrB(\uut)\subset \scrB(\uut')$. But we could have started with $\uut'$ as well, and hence $\scrB(\uut) = \scrB(\uut')$.

If $i\in \scrE(\uut)$, we can similarly argue that $i\in \scrE(\uut')$ and hence, $\scrE(\uut) = \scrE(\uut')$. In the process we have also established
\ref{item-equivtuples2}.

\emph{Step III:} It remains to prove that $\scrA(\uut) = \scrA(\uut')$. Note that this is still part of the induction step.

Let $(b,e) \in\scrA(\tuut)$. That is, there exists $i\in \llbracket 1,\ell\rrbracket$, such that $\us_i\in \scrS^{(e-b+1)}_\LR$ and
$e-b+1\leq n+1$.

Assume first that $e-b+1 \leq n$. Then $\us_i$ is not split in the passage from $\uut$ to $\tuut$ and there exists $\ti\in \llbracket 1,\tell\rrbracket$, such that $\us_i = \tus_\ti$. Hence $(b,e) \in\scrA(\tuut)$. By the induction hypothesis, we must have $(b,e)\in\scrA(\tuut')$ and there exists $\ti'\in\llbracket 1,\tell'\rrbracket$, such that $\tus'_{\ti'} = \tus_\ti$, $b(\ti';\tuut') = b$ and $e(\ti';\tuut') = e$.

If $(b,e) \in\scrA(\uut')$ there is nothing to prove, so we may assume towards a contradiction that $(b,e)\not\in\scrA(\uut')$.
This means that $\tus'_{\ti'}$ has come from $\uut'$ by splitting a handed signature of length $n+1$. That is, there exists $i'\in\llbracket 1,\ell'\rrbracket$, such that $\us'_{i'}\in \scrS^{(n+1)}$ and one may split
$\us'_{i'} = \us'_{i';1}\circ \us'_{i';2}$, using the split we know exists from Lemma~\ref{lem-equivalent-tuples}. One of the two components is our ambidexstrous $\tus'_{\ti'}$. We assume that $\tus'_{\ti'} = \us'_{i';1}$, with the other case being similar. But then $\us'_{i';2}$ must be left-handed and hence $b+n\in\scrE(\uut')$. By what has already been established, we conclude that $b+n\in \scrE(\uut)$ as well. Hence there exists $j\in\llbracket i+1,\ell\rrbracket$, such that $\us_j$ is left-handed with $e(j;\uut) = b+n$.

We have reached a contradiction, since it has to be possible to compose two consecutive signature strings between $\us_i$ and $\us_j$ to obtain a handed string of length at most $n+1$, which would be absurd.

Next assume that $e-b+1 = n+1$, such that $e = b+n$. Then the ambidexstrous signature string $\us_i$ is split in two, using Lemma~\ref{lem-equivalent-tuples}, when passing from $\uut$ to $\tuut$. This means that there exists 
$\ti\in\llbracket 1,\tell\rrbracket$, such that $\us_i = \tus_{\ti}\circ \tus_{\ti+1}$, $\tus_{\ti}$ is right-handed and $\tus_{\ti+1}$ is left-handed.

But this again means that $b = b(\ti;\tuut)\in\scrB(\tuut)$ and $e = e(\ti+1;\tuut)\in\scrE(\tuut)$.
By the induction hypothesis, we have $b\in \scrB(\tuut')$ and $e\in\scrE(\tuut')$. Let $\ti'<\tj'$ be such that $\tus'_{\ti'}$ is right-handed, $\tus'_{j'}$ is left-handed, $b(\ti';\tuut')=b$ and $e(\tj';\tuut')=e$.

As in Step III, we may argue that $\tj' = \ti'+1$, since it would otherwise be possible to compose two consecutive strings between $\tus'_{\ti'}$ and
$\tus'_{\tj'}$ and get a handed string of length at most $n$.

By what was just proved, we know that $b\not\in \scrB(\uut')$ and $b+n\not\in \scrE(\uut')$. But this implies that $\tus'_{\ti'}$ and $\tus_{\ti'+1}$ must have been produced by splitting a signature string of length $n+1$, when passing from $\uut'$ to $\tuut'$. Since $\tuut'_{\ti'}$ is right-handed, it cannot be the second component in a split and $\tuut'_{\ti'+1}$ is right-handed and cannot be the first component in a split, cf.~Lemma \ref{lem-badcomp}. 

The only remaining option is that there exits an $i'\in\llbracket 1,\ell'\rrbracket$, such that $\us'_{i'} = \tus'_{\ti'}\circ \tus'_{\ti'+1} = \us_i$ and is ambidexstrous. Hence $(b,e)\in \scrA(\uut')$ and we are done.
\end{proof}

\section{The Renormalized Resolvent Expansion}\label{Sec-ResExp}

 We write 
 \[
 \Hfin = \bigset{\psi\in\scrH}{\exists n\in\NN:\, \one_{[N\geq n]}\psi = 0}
 \]
 for the 
 vector-space of finite particle states. Finally, let
\[
\scrL_\fin = \scrL (\Hfin;\Hfin)
\]
denote the complex algebra of linear operators $T\colon \Hfin \mapsto \Hfin$. 
We use the following notation for the closed complex left half-plane, with and without zero,
\[
\CC_- =\bigset{z\in \CC}{ \re(z)\leq 0 } \qquad \textup{and}\qquad \CC_-^* = \CC_-\setminus \{0\}.
\]

We will be using the notation
\begin{equation}\label{F-tuples}
    \uF^{(k)} = (F_1,\dotsc,F_k) \in \bigl( L^2(\RR^d\times\RR^d)\bigr)^k 
\end{equation}
for $k$-tuples of square-integrable functions of fermion and boson momenta.

\begin{Def} For $k\in\NN$, we define $\scrM^{(k)}$ to be the vector-space of complex multi-linear functions $M$ of $k$ variables 
\[
\bigl( L^2(\RR^d\times\RR^d)\bigr)^k \ni \uF^{(k)}\mapsto M(\uF^{(k)})\in\scrL_\fin.
\]    
\end{Def}

\subsection{Renormalized handed blocks of operators}\label{subsec-RenBlocks}

Let $\us = (s_1,\dotsc,s_k)\in\scrS^{(k)}$, a handed signature string of length $k$. 
The purpose of this subsection is to recursively define, for each such handed signature string, renormalized versions of the operator product
\[
H^{s_1}(F_1) R_0(z) H^{s_2}(F_2) \cdots R_0(z) H^{s_k}(F_k),
\]
where $z\in\CC_-$, $F_i\in L^2(\RR^d\times\RR^d)$ and the $H^{s_i}$'s are the interaction terms from \eqref{interactionterms}, depending on which of the four signatures $s_i\in\{\ab\af,\ab\cf,\cb\af, \cb\cf\}$ we have.
Note that for any signature $s$ and $F\in L^2(\RR^d\times\RR^d)$, we have 
$H^s(F)\in\scrL_\fin$. Likewise, if $z\in\CC_-^*$, then $R_0(z)\in\scrL_\fin$ and even for $z\in\CC_-$, we have $\one_{[N\geq 1]}R_0(z)\in\scrL_\fin$.

\begin{Def}[Renormalized handed blocks of operators]\label{def-remblocks} For $s\in\{\ab\af,\ab\cf,\allowbreak\cb\af,\cb\cf\}$, $z\in\CC_-$ and $F\in L^2(\RR^d\times\RR^d)$, we set
\[
T^{(1)}_{(s)}(z;F) = -H^s(F).
\]
For $k\in\NN$ and $z\in\CC_-$, We recursively define operators $T^{(k)}_\us(z;\cdot)\in\scrM^{(k)}$ associated with handed strings $\us\in\scrS^{(k)}$ as follows. Assume 
$T^{(k')}_{\us'}(z;\cdot)$ has been constructed for all $k'\leq k$ and $\us'\in\scrS^{(k')}$. Let $\us\in\scrS^{(k+1)}$. Pick a split $j\in\Split(\us)$ and write $\us = \us'\circ\us''$ with 
$\us'\in\scrS^{(k')}_\Right\cup \scrS^{(k')}_\LR$ and $\us''\in\scrS^{(k'')}_\Left\cup \scrS^{(k'')}_\LR$, $k'+k'' = k+1$ and not both of them ambidextrous. See Proposition~\ref{propo-signdecomp}.
For $\us\in\scrS^{(k+1)}_\Right\cup\scrS^{(k+1)}_\Left$, we define
\[
T^{(k+1)}_{\us}(z;F_1,\dotsc,F_{k+1})= T^{(k')}_{\us'}(z;F_1,\dotsc,F_{k'}) R_0(z) T^{(k')}_{\us'}(z;F_{k'+1},\dotsc,F_{k+1}).
\]
For $\us\in \scrS^{(k+1)}_\LR$ ($\us$ is ambidexstrous), then we define
\begin{equation*}
T^{(k+1)}_{\us,\bare}(z;F_1,\dotsc,F_{k+1})= T^{(k')}_{\us'}(z;F_1,\dotsc,F_{k'}) R_0(z) T^{(k')}_{\us'}(z;F_{k'+1},\dotsc,F_{k+1})
\end{equation*}
and finally:
\begin{equation}\label{handedrenormblock}
T^{(k+1)}_{\us}(z;\uF^{(k+1)})= T^{(k+1)}_{\us,\bare}(z;\uF^{(k+1)}) -\bigl\langle \Omega \,\big\vert\, T^{(k+1)}_{\us,\bare}(0;\uF^{(k+1)})\Omega\bigr\rangle.
\end{equation}
\end{Def} 

There are two issues to consider regarding Definition~\ref{def-remblocks}. 
The first is to observe that the operators constructed permit one to take $z=0$. A priori, they only makes sense for $z\in \CC_-^*$.
The second issue is independence of the choice of splitting $j\in\Split(\us)$ in the recursive construction, to ensure that the objects are canonical. We address these issues in the following two remarks.

\begin{rk}\label{rem-nosing}  Observe that one may readily establish the following identities as one proceeds with the recursive construction
\begin{align*}
 N_\boson T^{(k)}_\us(z;\uF^{(k)})& = T^{(k)}_\us(z;\uF^{(k)})\bigl(N_\boson + n_\ab(1,k;\us)\one\bigr)\\
N_\fermion T^{(k)}_\us(z;\uF^{(k)}) & =  T^{(k)}_\us(z;\uF^{(k)})\bigl(N_\fermion + n_\af(1,k;\us)\one\bigr).
\end{align*}
Here the counting functions $n_\ab$ and $n_\af$ were defined in Definition~\ref{def-countingfunctions}.

Recalling Definition~\ref{def-handedsignatures}, we may conclude from the intertwining relations above that for $k\in\NN$
\[
\begin{aligned}
&\forall \us\in\scrS^{(k)}_\Left: & &  \vert \Omega\rangle\langle \Omega\vert T^{(k)}_\us(z;\uF^{(k)}) = 0,\\
&\forall \us\in\scrS^{(k)}_\Right: & &    T^{(k)}_\us(z;\uF^{(k)})\vert \Omega\rangle\langle \Omega\vert = 0.
\end{aligned}
\]
From this it now follows, as part of the recursive construction, that one may indeed take $z=0$ in $T^{(k)}_{\us}(z;\uF^{(k)})$.
\end{rk}

\begin{rk}\label{rem-SplitIndependence} That the recursive definition of the $T^{(k)}_\us$'s is independent on the choice of split $j\in\Split(\us)$, may be seen as follows. Assume that for some $k\in\NN$ and all $k'\leq k$ and $\us'\in \scrS^{(k')}$, the definition of $T^{(k')}_{\us'}$ is independent of the choice of split $j'\in\Split(\us')$. (For $k=1$ this is trivially satisfied.)

Let now $\us\in\scrS^{(k+1)}$. If $\Split(\us)$ is a singleton, then we do not have much choice and there is nothing to prove. (Always the case if $k+1=2$.) Hence we may assume that $\Split(\us) = \{j_1,\dotsc,j_u\}$ with $u\geq 2$ and $1\leq j_1<j_2\cdots <j_u\leq k$. 

For $v\in\llbracket 1, u\rrbracket$, write $\us_v' = (s_1,\dotsc,s_{j_v})\in\scrS^{(j_v)}_\Right$ and $\us_v'' = (s_{j_v+1},\dotsc,s_{k+1})\in\scrS^{(k+1-j_v)}_\Left$ (cf. Proposition~\ref{propo-signdecomp}). Furthermore, if $v\leq u-1$, write $\usigma_v = (s_{j_v+1},\dotsc,s_{j_{v+1}})\in\scrS^{(j_{v+1}-j_v)}_\LR$.
Using the induction hypothesis, we may compute for $1\leq v\leq u-1$:
\begin{align*}
& T^{(j_{v})}_{\us'_v}(z,F_1,\dotsc,F_{j_v})R_0(z)T^{(k+1-j_v)}_{\us''_v}(z;F_{j_{v}+1},\dotsc,F_{k+1})\\
 & \quad  = T^{(j_{v})}_{\us'_v}(z,F_1,\dotsc,F_{j_v})R_0(z)\\
 & \qquad \Bigl\{T^{(j_{v+1}-j_v+1)}_{\usigma_v}(z;F_{j_v+1},\dotsc,F_{j_{v+1}}) R_0(z) T^{(k+1-j_{v+1})}_{\us''_{v+1}}(z;F_{j_{v+1}+1},\dotsc,F_{k+1})\Bigr\}\\
  & \quad  = \Bigl\{ T^{(j_{v})}_{\us'_v}(z,F_1,\dotsc,F_{j_v})R_0(z)T^{(j_{v+1}-j_v+1)}_{\usigma_v}(z;F_{j_v+1},\dotsc,F_{j_{v+1}})\Bigr\} \\
  &  \qquad R_0(z) T^{(k+1-j_{v+1})}_{\us''_{v+1}}(z;F_{j_{v+1}+1},\dotsc,F_{k+1})\\
  &  \quad = T^{(j_{v+1})}_{\us'_{v+1}}(z,F_1,\dotsc,F_{j_{v+1}})R_0(z) T^{(k+1-j_{v+1})}_{\us''_{v+1}}(z;F_{j_{v+1}+1},\dotsc,F_{k+1}).
\end{align*}
From this identity, it readily follows that $T^{(k+1)}_\us$ does not depend on the choice $j_v\in\Split(\us)$ of split.  
\end{rk}

\begin{Ex} As an illustration, recalling Remark~\ref{rem-signs}~\ref{rem-signslength2}, the recursive construction for $k=2$ yields the following operators:
\[
\begin{aligned}
&(\rightarrow) &  T^{(2)}_{(\ab\af,\cb\af)}(z;F_1,F_2) & = H^{\ab \af}(F_1) R_0(z)H^{\cb \af}(F_2), \\
&(\leftarrow) &  T^{(2)}_{(\ab\cf,\cb\cf)}(z;F_1,F_2)& =  H^{\ab \cf}(F_1) R_0(z) H^{\cb \cf}(F_2),\\
&(\leftrightarrow,\bare) &  T^{(2)}_{(\ab\af,\cb\cf),\bare}(z;F_1,F_2) & = H^{\ab \af}(F_1)  R_0(z) H^{\cb \cf}(F_2), \\
&(\leftrightarrow,\bare) &  T^{(2)}_{(\ab\cf,\cb\af),\bare}(z;F_1,F_2) & = H^{\ab \cf}(F_1)R_0(z) H^{\cb \af}(F_2), \\
&(\leftrightarrow) &   T^{(2)}_{(\ab\af,\cb\cf)}(z;F_1,F_2) & =  H^{\ab \af}(F_1)  R_0(z) H^{\cb \cf}(F_2)\\
& & &\ - \bigl\langle \Omega\,\big\vert\, H^{\ab \af}(F_1)  R_0(0) H^{\cb \cf}(F_2)  \Omega\bigr\rangle\cdot\one, \\
&(\leftrightarrow) &   T^{(2)}_{(\ab\cf,\cb\af)}(z;F_1,F_2) &=  H^{\ab \cf}(F_1)R_0(z) H^{\cb \af}(F_2).
\end{aligned}
\]
The arrows on the far left above, indicate the handedness of the operator. The notation $(\LR,\bare)$, indicates an ambidexstrous operator before subtracting counter-term. The counter-term associated with the ambidexstrous signature string $(\ab\cf,\cb\af)$ is equal to zero.
Note that 
\begin{equation}
   -\bigl\langle \Omega, H^{\ab \af}(\overline{F_1})  R_0(0) H^{\cb \cf}(F_2) \Omega\bigr\rangle =  -\int \frac{\overline{F_1(k,q)} F_2(k,q)}{\wf(k) + \wb(q)} dk dq
\end{equation}
is the counter-term $E(F_1,F_2)$ from \cite{AlvaMoll2021} and
$E\bigl(G^{(2)}_\Lambda,G^{(2)}_\Lambda\bigr) = E^{(2)}_\Lambda$, the counter-term from \eqref{E2Lambda}.
\end{Ex}

\begin{lem}\label{lem-basicblockbound} Let $k\in\NN$, There exists a constant $C$ 
only depending on $k$ and the masses $\mb$ and $\mf$, such that for any $z\in \CC_-$ and $\uF^{(k)}\in (L^2(\RR^d\times\RR^d))^k$, we have the following
\begin{enumerate}[label = \textup{(\arabic*)}]
    \item\label{item-bb-k1}  If $\us\in\scrS^{(1)}_\Right$, then $T^{(1)}_\us(z;F_1)(N_\boson+1)^{-\frac12}$ and $(N_\boson+1)^{-\frac12} T^{(1)}_{\us^*}(z;F_1)$ extend from $\Hfin$ to a bounded operators on $\scrH$ with
    \[
    \bigl\| T^{(1)}_\us(z;F_1)(N_\boson+1)^{-\frac12}\bigr\|\leq C\|F_1\|, \quad
    \bigl\|(N_\boson+1)^{-\frac12} T^{(1)}_{\us^*}(z;F_1)  \bigr\|\leq C\|F_1\|
    \]
    \item\label{item-bb-kgeneral} For $k'\in\llbracket 2,k\rrbracket$ and $\us\in\scrS^{(k')}$, the operators $T^{(k')}_\us(z;\uF^{(k')})$ extend from $\Hfin$ to bounded operators on $\scrH$ and
    \[
    \bigl\|T^{(k')}_\us(z;\uF^{(k')})\bigr\|\leq C\|F_1\|\cdots\|F_{k'}\|.
    \] 
    \item\label{item-bb-bare} For $k'\in\llbracket 2,k\rrbracket$ and $\us\in\scrS^{(k')}_\LR$, the operators $T^{(k')}_{\us,\bare}(z;\uF^{(k')})$ extend from $\Hfin$ to bounded operators on $\scrH$ and
    \[
    \bigl\|T^{(k')}_{\us,\bare}(z;\uF^{(k')})\bigr\|\leq C\|F_1\|\cdots\|F_{k'}\|.
    \] 
    \end{enumerate}
\end{lem}

\begin{proof}  The claim \ref{item-bb-k1} with $\us\in\scrS^{(1)}$ follows from the following computation for $\psi,\varphi\in\Hfin$ and $F\in L^2(\RR^d,\RR^d)$:
\begin{align*}
\bigl|\bigl\langle \psi, H^{\ab\af^*}(F)\psi\bigr\rangle \bigr| & \leq 
 \int \bigl|\bigl\langle \psi, \af^*(F(\cdot,q))\ab(q)\varphi\bigr\rangle \bigr|dq\\
 & \leq  \int \|F(\cdot,q)\| \|\ab(q)\varphi\|dq \|\psi\|\\
 &\leq \|F\| \|N^{\frac12} \varphi\|\|\psi\|.
\end{align*}
Here we used that $\|\af^*(h)\|\leq \|h\|$ for any $h\in L^2(\RR^d)$.
For $H^{\ab\af}(F)$, the computation is the same, except for keeping track of complex conjugations: $\int F(k,q) \af(k)dk = \af( \overline{F(\cdot,q)})$, following standard convention for smeared annihilation operators.
The remaining two cases follows from passing to adjoints under the integral sign and thereby switching the roles of $\psi$ and $\varphi$.

 For $k\geq 2$, the claims \ref{item-bb-kgeneral} and \ref{item-bb-bare} follow easily from the recursive definition of the $T^{(k)}_\us$'s by induction after $k$. Observe that if one should need a factor of $(N_\ab+1)^{-\frac12}$, in case one of the factors in a splitting $\us = \us'\circ\us''$ has length $1$, then it can be extracted from the free resolvent $R_0(z)$ sandwiched in the middle at the cost of a factor $\mb^{-1/2}$. Recall Remark~\ref{rem-nosing}.
\end{proof} 

\begin{Cor} Let $k\in\NN$ and $\us\in\scrS^{(k)}$. Then for any $z\in\CC_-$ and $\uF^{(k)}\in (L^2(\RR^d\times\RR^d))^k$, we have
\[
\bigl(T^{(k)}_{\us}(z;F_1,\dotsc,F_k)\bigr)^*_{\vert \Hfin} = 
T^{(k)}_{\us^*}(\bar{z};\overline{F}_k,\dotsc,\overline{F}_1).
\]
\end{Cor}

\begin{proof} The corollary follows easily by induction, keeping Remark~\ref{rem-SplitIndependence} in mind. Note that for $k=1$, the adjoint operator $(T^{(1)}(z;F_1))^*$ is densely defined on a domain that contains $\Hfin$. 
\end{proof}

\subsection{Renormalized summands}\label{subsec-RenSummands}

Recall that our goal is to resum the Neumann series into a new series with renormalized summands that allow for better estimates. See Theorem~\ref{thm-reordering} below. In this section, we construct and estimate the renormalized summands.

\begin{Def}\label{def-RenSummands}
    For any $n,k\in\NN$, $\uut=(\us_1,\dotsc,\us_\ell) \in\scrT^{(n,k)}$ and $z\in\CC_-^*$ we define $S_\uut^{(n,k)}\in\scrM^{(k)}$. For $\uF^{(k)} = (F_1,\dotsc,F_k)\in (L^2(\RR^d\times\RR^d))^k$, we set
    \[
      S_\uut^{(n,k)}(z;\uF^{(k)}) = R_0(z) \prod^{\ell}_{i=1} \Bigl[T^{(j_i)}_{\us_i}(z;F_{b(i;\uut)},\dotsc,F_{e(i;\uut)}) R_0(z)\Bigr]\in\scrL_\fin.
    \]
    Here $j_i$ is the length of the handed signature string $\us_i$.
\end{Def}

The following lemma follows directly from the definition above together with Lemma~\ref{lem-basicblockbound}.

\begin{lem}\label{lem-basicsummandbound}  For any $n,k\in\NN$, there exists a constant $C$ that only depends on $n$ and the masses $\mb,\mf$, such that for any $\uut=(\us_1,\dotsc,\us_\ell)\in\scrT^{(n,k)}$, $z\in\CC_-^*$ and $\uF^{(k)}\in (L^2(\RR^d\times\RR^d))^k$, the operator $S_\uut^{(n,k)}(z;\uF^{(k)})$ extends from $\Hfin$ to a bounded operator on $\Hfin$ and
\[
\bigl\|S_\uut^{(n,k)}(z;\uF^{(k)})\bigr\| \leq \frac{C^\ell}{|\re z|^{1+\frac{\ell}{2}}} \|F_1\|\cdots\|F_k\|.
\]
\end{lem}

Our remaining task is to establish the following proposition, which shows that the renormalized summands only depend on the tuples $\uut$ through $\uut$'s equivalence class.

\begin{Prop}\label{prop-EquivSummands} Let $n,k\in\NN$. 
    For $\uut, \uut' \in \scrT^{(n,k)}$ with $\uut \sim \uut'$, we have for any $z\in\CC_-^*$ and $\uF^{(k)}\in (L^2(\RR^d\times\RR^d))^k$:
    \[
    S_\uut^{(n,k)}(z;\uF^{(k)}) = S_{\uut'}^{(n,k)}(z;\uF^{(k)}).
    \]
\end{Prop}

\begin{proof} 
    We will proceed by induction on $k$. If $k=1$ and $\uut,\uut'\in \scrT^{(n,1)}$ with $\uut\sim\uut'$, then $\uut = \uut'$ and there is nothing to prove. 
    Assume that the proposition is correct for $k'\leq k\in\NN$. 
   
    Let now $\uut,\uut'\in\scrT^{(n,k+1)}$ with $\uut\sim\uut'$ and
write
    \begin{equation*}
        \uut  = (\us_1, \dotsc, \us_\ell)\quad \textup{and}\quad     \uut'  = (\us'_1, \dotsc, \us'_{\ell'}),
    \end{equation*}  
    where $\us_i\in\scrS^{(j_i)}$, $\us'_{i'}\in\scrS^{(j'_{i'})}$,
    $j_1+\cdots+ j_\ell = k+1$ and $j'_1+\cdots +j'_{\ell'} = k+1$.

 \emph{Step I:} Reduction of the problem.
    Suppose $\us_i$ is ambidexstrous for all $i\in\llbracket 2,\ell-1\rrbracket$, $\us_1$ is not left-handed and $\us_\ell$ is not right-handed.
    Then by Proposition~\ref{prop-equivalent-tuples}, there exists $i'\in\llbracket 2,\ell'\rrbracket$ with $i'+\ell-2 \leq \ell'$, such that $b(2;\uut) = b(i';\uut')$ and
    $\us_i = \us'_{i'+i-2}$  for all $i\in\llbracket 2,\ell-1\rrbracket$.

    We now claim that $i'=2$ and $\us'_1 = \us_1$ as well. 
    If either $\us_1$ is ambidexstrous, or one of the signature strings $\us'_1,\dotsc\us'_{i'-1}$ are ambidexstrous, then we are done by Proposition~\ref{prop-equivalent-tuples}. 
    
    Assume towards a contradiction that $i'>2$. 
    By Lemma~\ref{lem-badcomp}, we know that $\us'_{i'-1}$ cannot be right-handed, since we then would have
    $(\us'_1\circ\cdots\circ \us'_{i'-2})\circ \us'_{i'-1} = \us_1$, which is a handed signature string. Similarly, $\us'_1$ cannot be left-handed. 

     Summing up $\us'_1$ must be right-handed and $\us'_{i'-1}$ must be left-handed. But then there are at least two consecutive $\us'_{j}$, $\us'_{j+1}$ with $1\leq j<j+1\leq i'-1$ that can be composed to form a handed signature string of length at most $n$. This is not allowed in a tuple and we may conclude that $i'=2$ and, hence, $\us'_1=\us_1$.

    Similarly, we may argue that $\ell  = \ell'$ and $\us_\ell = \us'_{\ell}$ and conclude that $\uut = \uut'$.

    In conclusion, we may without loss of generality assume that either there exists an $i\geq 2$, such that $\us_i$ is right-handed, or there exists an $i\leq \ell-1$ with $\us_i$ left-handed.
    
\emph{Step II:} Assume that we have an $i\geq 2$ with $\us_i\in\scrS^{(j_i)}_\Right$ being a right-handed signature string. The other case with a left-handed $\us_i$ and $i\leq \ell-1$ is completely symmetric.

     Let $ b= b(i;\uut)\in\scrB(\uut)$. 
    From Proposition~\ref{prop-equivalent-tuples}, we know that $b\in\scrB(\uut')$ as well, and hence there exists $i'\in \llbracket 1, \ell'\rrbracket$ such that $\us'_{i'}\in\scrS^{(j'_{i'})}_\Right$ and  $b(i';\uut') = b(i;\uut)$. Note that $i'\geq 2$.

    We may define four new tuples $\uutau_\mathrm{l},\uutau'_\mathrm{l}\in \scrT^{(n,b-1)}$ and $\uutau_\mathrm{r},\uutau'_\mathrm{r}\in \scrT^{(n,k-b+1)}$, by splitting $\uut$ and $\uut'$ at the index $b$.
     \[
 \begin{aligned}
  & \uutau_\mathrm{l} = (\us_1,\dotsc,\us_{i-1}), & & \uutau_\mathrm{r} = (\us_i,\dotsc,\us_{\ell}),\\
   & \uutau'_\mathrm{l} = (\us'_1,\dotsc,\us'_{i'-1}), & & \uutau'_\mathrm{r} = (\us'_{i'},\dotsc,\us'_{\ell'}).
 \end{aligned}
     \]  
 Note that $b-1<k$ and $k-b+1<k$ (since $b\geq 2$).
 Since $b = b(i;\uut) = b(i';\uut')$, we must have
 \[
 \uutau_\mathrm{l}\sim \uutau'_{\mathrm{l}} \quad \uutau_\mathrm{r} \sim \uutau'_\mathrm{r}.
 \]

  We are now in a position to use the induction hypothesis to compute
  {\allowdisplaybreaks
  \begin{align*}
& S^{(n,k)}_\uut(z;\uF^{(k)})  = R_0(z)
       \prod^{\ell}_{\nu=1} \Bigl\{T^{(j_i)}_{\us_\nu}(z;F_{b(\nu;\uut)},\dotsc,F_{e(\nu;\uut)}) R_0(z)\Bigr\}\\
       &\quad  = S^{(n,b-1)}_{\uutau_\mathrm{l}}(z;F_1,\dotsc, F_{b-1}) 
        \prod^{\ell}_{\nu= i}\Bigl\{ T^{(j_\nu)}_{\us_\nu}(z;F_{b(\nu;\uut)},\dotsc,F_{e(\nu;\uut)}) R_0(z)\Bigr\}\\
        & \quad =  S^{(n,b-1)}_{\uutau'_\mathrm{l}}(z;F_1,\dotsc, F_{b-1}) 
        \prod^{\ell}_{\nu= i} \Bigl\{T^{(j_\nu)}_{\us_\nu}(z;F_{b(\nu;\uut)},\dotsc,F_{e(\nu;\uut)}) R_0(z)\Bigr\}\\
        & \quad = \prod_{\nu' = 1}^{i'-1} \Bigl\{R_0(z) T^{(j'_{\nu'})}_{\us'_{\nu'}}(z;F_{b(\nu';\uut')},\dotsc, F_{e(\nu',\uut')}) \Bigr\} 
        S^{(n,k-b+1)}_{\uutau_\mathrm{r}}(z;F_b,\dotsc,F_k)\\
        & \quad = \prod_{\nu' = 1}^{i'-1} \Bigl\{ R_0(z) T^{(j'_{\nu'})}_{\us'_{\nu'}}(z;F_{b(\nu';\uut')},\dotsc, F_{e(\nu',\uut')})\Bigr\}
        S^{(n,k-b+1)}_{\uutau'_\mathrm{r}}(z;F_b,\dotsc,F_k)\\
        & \quad = \prod_{\nu' = 1}^{\ell'}\Bigl\{ R_0(z) T^{(j'_{\nu'})}_{\us'_{\nu'}}(z;F_1,\dotsc,F_k) \Bigr\} R_0(z)\\
        & \quad = S^{(n,k)}_{\uut'}(z;\uF^{(k)}).        
  \end{align*}
  }
  This completes the proof.
\end{proof}

\subsection{Reordering theorem}

 In this subsection, we formulate and verify our formula for the renormalized Neumann expansion. The rest of the paper will then be concerned with getting estimates on the summands that are uniform in the ultraviolet cutoff.

 For $s\in \{\ab\af,\cb\cf,\ab\cf,\cb\af\}$, we set
  \[
    G_{s,\Lambda} = \begin{cases}
      G^{(1)}_\Lambda, & \textup{if } s = \ab\cf\\
      \overline{G}_\Lambda^{(1)},  & \textup{if } s = \cb\af\\
      G^{(2)}_\Lambda, &  \textup{if } s =\cb\cf\\
      \overline{G}_\Lambda^{(2)}, & \textup{if } s = \ab\af.
    \end{cases}
  \]

  For $N\in\NN$, we may now define the self-energy counter-term at order $N$ to be

  \begin{Def}[Self-energy]\label{def-selfenergy} For $\ell\in\NN$ and $\us\in\scrS^{(2\ell)}_\LR$, we define the associated self-energy contribution to be
  \[
    E_{\us,\Lambda}^{(2\ell)} = -  \bigl\langle \Omega\,\big\vert\, T^{(2\ell)}_{\us,\bare}\bigl(0;G_{s_1,\Lambda},\dotsc, G_{s_k,\Lambda}\bigr)\Omega \bigr\rangle.
  \]
  For $N\in\NN$, the total self-energy up to order $N$ is
  \[
    E^{(N)}_\Lambda =   \sum_{\ell\in\NN, 2\ell\leq N} \sum_{\us\in\scrS^{(2\ell)}_\LR}E_{\us,\Lambda}^{(2\ell)}.
  \]
\end{Def}
 We are now in position to formulate the reordering theorem.

  \begin{Th}[Reordering theorem]\label{thm-reordering} Let $N\in\NN$ and $\Lambda>0$. There exists a $C_N(\Lambda)>0$, which depends only on $N,\Lambda$ and the masses $\mb,\mf$, such that for any $z\in\CC$ with $\re(z)\leq -C_N(\Lambda)$, we have
  \[
    \bigl(H_\Lambda - E_\Lambda^{(N)} - z \bigr)^{-1} = R_0(z) + \sum_{k=1}^\infty
    \sum_{[\uut]\in\scrT^{(N,k)}/\!\sim} S^{(N,k)}_\uut\bigl(z;G_{s_1,\Lambda},\dotsc, G_{s_k,\Lambda}\bigr),
  \]
  with the right-hand side being absolutely convergent in operator norm.
  \end{Th}

  The rest of this section is devoted to proving the reordering theorem. The objects and notation introduced in the remainder of this subsection will not be used elsewhere. We apologize to the reader for what is essentially a nested application of the distributive law, being somewhat obscured by notation.

We begin with a lemma about the location of ambidexstrous sub-strings of a signature string $\us$. 
Let us fix $k\in\NN$ and $\us\in\scrS^{(k)}_0$. Define the set of ambidexstrous sub-strings of $\us$ to be
\[
\scrA_\us^{(k)} = \Bigset{(j,j')\in \llbracket 1,k\rrbracket^2}{j<j', \ (s_j,\dotsc,s_{j'})\in \scrS^{(j'-j+1)}_\LR}.
\]
For $(j,j')\in\scrA^{(k)}_\us$, we write $\us_{(j,j')} = (s_j,\dotsc,s_{j'})$.

\begin{lem}\label{lem-disjointambi} Let $(i,i'),(j,j')\in\scrA^{(k)}_\us$ with $(i,i')\neq (j,j')$ be two distinct ambidexstrous sub-strings. Then either $\llbracket i,i'\rrbracket\cap \llbracket j,j'\rrbracket = \emptyset$,  $\llbracket i,i'\rrbracket \subsetneq \llbracket j,j'\rrbracket$ or $\llbracket j,j'\rrbracket \subsetneq \llbracket i,i'\rrbracket$.
\end{lem}

 \begin{proof} Let $(i,i'),(j,j')\in\scrA^{(k)}_\us$ with $(i,i')\neq (j,j')$. Assume towards a contradiction that
 $j \leq  i \leq j' \leq  i'$ with either $j<i$ or $j' < i'$. Suppose without loss of generality that $j<i$. (If we have $j'<i'$, we can consider $(k-j'+1,k-j+1)$ and $(k-i'+1,k-i+1)$ as ambidexstrous sub-strings of $\us^*$ and thereby reduce to the case handled.) In particular, it implies that $j'-j+1\geq 3$ and $i'-i+1\geq 3$.

 Suppose $n_\ab(j,i-1;\us) >0$. Since $0 = n_\ab(j,j';\us) = n_\ab(j,i-1;\us) + n_{\ab}(i,j';\us)$, we conclude that $n_\ab(i,j';\us)<0$, which is not possible, since $\us_{(i,i')}$ is ambidexstrous. 

 Since $n_\ab(j,i-1;\us) \geq 0$, due to $\us_{(j,j')}$ being a handed signature string, we therefore must have  
 $n_\ab(j,i-1;\us) = 0$. But then the computation above gives us that $n_{\ab}(i,j';\us)=0$. 
  Since this implies, from \eqref{def-ambi-right}, that $0 = n_\af(j,j';\us) =n_\af(j,i-1;\us) + n_\af(i,j';\us)  >0$, we arrive at a contradiction.
 \end{proof}

For $n\in\NN$, we write
\[
\scrA_\us^{(n,k)} = \Bigset{(j,j')\in \scrA^{(k)}_\us }{j'-j+1 = n},
\]
for the ambidexstrous sub-strings of length $n$. Note that $\scrA^{(n,k)}_\us =\emptyset$ if $n$ is odd or if $n>k$.

We will move recursively from a naive Neumann expansion, cf.~\eqref{naiveneumann} below, of the resolvent $(H_\Lambda - E_\Lambda^{(N)} - z )^{-1}$ to the reordered sum in Theorem~\ref{thm-reordering}. In the $n$'th recursive step ($n < N$), we will match ambidexstrous blocks of operators of length $n+1$ with matching counterterms. For this purpose we need intermediate reordering formulas, where we have only renormalized handed blocks of operators up to length $n$.
In order to label summands in the intermediate reordering expansions, we introduce the following set of all possible locations of remaining conterterms.
For $n,N\in\NN$ with $n\leq N$ we set
\begin{equation}\label{Ps-nNk}
\begin{aligned}
\scrP_\us^{(n,N,k)} & =  \Bigset{U\subset \scrA^{(k)}_\us}{\forall (j,j')\in U:\ n < j'-j+1 \leq N,\\
& \forall (j,j'),(i,i')\in U\textup{ with }(j,j')\neq (i,i'):\ \llbracket j,j'\rrbracket \cap \llbracket i,i'\rrbracket = \emptyset}.
\end{aligned}
\end{equation}
The elements of $\scrP_\us^{(n,N,k)}$ are collections of pairwise disjoint ambidexstrous sub-strings of lengths between $n$ and $N$ (with $n$ not included). 

\begin{rk} Let $n,N,k\in\NN$ with $n<N$. We remark the following:
\begin{enumerate}[label = \textup{(\roman*)}]
\item We have $\emptyset\in \scrP^{(n,N,k)}_\us$ and
 $\scrP_\us^{(N,N,k)}= \{\emptyset\}$.
\item For $n'\in\NN$ with $n<n'\leq N$, and any $U\subset\scrA^{(n',k)}_\us$, it follows from Lemma~\ref{lem-disjointambi} that we have $U\in \scrP^{(n,N,k)}_\us$.
\item For $n'\in\NN$ with $n<n'\leq N$, we have $\scrP^{(n',N,k)}_\us \subset \scrP^{(n,N,k)}_\us$.
\end{enumerate}
\end{rk}

Suppose now that $n<N$. Let $U_0\in \scrP_\us^{(n+1,N,k)}$. We say that $U\in \scrP^{(n,N,k)}_\us$ is \emph{subordinate} to $U_0$, written $U\preceq U_0$, if 
\[
U_0\subset U \quad \textup{and} \quad U\setminus U_0 \subset \scrA^{(n+1,k)}_\us.
\]
\begin{rk} Let $n,N,k\in\NN$ with $n<N$. We remark the following:
\begin{enumerate}[label = \textup{(\roman*)}]
\item Let $U_0\in \scrP_\us^{(n+1,N,k)}$. Then $U_0\in \scrP^{(n,N,k)}_\us$ as well and $U_0\preceq U_0$. 
\item\label{item-rem-Psets2} Let $U\in \scrP_\us^{(n,N,k)}$ and set $U_0 = U\setminus \scrA^{(n+1,k)}_\us$. Then $U_0\in\scrP^{(n+1,N,k)}_\us$ and $U\preceq U_0$.
\item\label{item-rem-Psets3} Let $U_0,U_0'\in \scrP_\us^{(n+1,N,k)}$ and $U,U'\in \scrP_\us^{(n,N,k)}$. If $U_0\neq U_0'$, $U\preceq U_0$ and $U'\preceq U_0'$, then $U \neq U'$. 
\item  From \ref{item-rem-Psets2} and \ref{item-rem-Psets3} it follows that we have the disjoint union
\[
\scrP^{(n,N,k)}_\us = \bigcup_{U_0\in\scrP^{(n+1,N,k)}_\us} \Bigset{U\in\scrP^{(n,N,k)}_\us}{ U\preceq U_0}.
\]
\end{enumerate}
\end{rk}

Let $n,N,k\in\NN$ with $n\leq N$ and take a $U\in\scrP_\us^{(n,N,k)}$. Let $u=\Card(U)$ and enumerate the elements
$U = \{(j_1,j_1'),\dotsc,(j_u,j_u')\}$, such that $\mu\leq \nu\Rightarrow j_\mu\leq j_\nu$. (We would get same ordering of the pairs if we had used the second coordinate instead of the first.) 

Define $U^\comp =\{(\ell_1,\ell'_1),\dotsc,(\ell_{u+1},\ell'_{u+1})) \subset \NN_0^2$ with $\Card(U^\comp)= u+1$ as follows. If $u=0$, we set $(\ell_1,\ell_1') = (1,k)$. If $u>0$, we set first $(\ell_1,\ell_1') = (1,j_1-1)$ and  $(\ell_{u+1},\ell'_{u+1}) = (j_u'+1,k)$. Finally, for $1<i<u+1$, we set $(\ell_i,\ell_i') = (j'_{i-1}+1,j_i-1)$. For the lengths $L_i$, $i=1,2,\dotsc,u+1$, of $(\ell_i,\ell'_i)$, we write $L_i = \ell'_i-\ell_i +1 \in\NN_0$.

For each $i\in \llbracket 1,u+1\rrbracket$, we associate an operator. If $L_i\geq 1$, select a tuple
\begin{equation}\label{TupleChoice}
\uut_i \in \varphi^{(n,L_i)}\bigl((s_{\ell_1},\dotsc,s_{\ell'_i})\bigr)\in
\scrT^{(n,L_i)}/\!\sim,
\end{equation}
where $\varphi^{(n,L_i)} \colon \scrS_0^{(L_i)} \to \scrT^{(n,L_i)}/\!\sim$ is the (inverse) bijection from Proposition~\ref{prop-tuple-bijection}. If $L_i=0$, we write $\uut_i=()$ for a "tuple" with zero elements.

With the above in place, we may now associate to $U\in \scrP^{(n,N,k)}_\us$ the operator
\begin{equation}\label{IntSummands}
S_{\us,U,\Lambda}^{(n,N,k)}(z) = \biggl(\prod_{(j,j')\in U} E^{(j'-j+1)}_{\us_{(j,j')},\Lambda} \biggr)\prod_{i=1}^{u+1} S_{\uut_i}^{(n,N,L_i)}(z;G_{s_{\ell_i},\Lambda},\dotsc, G_{s_{\ell'_i},\Lambda})
\end{equation}
with the convention that if $L_i=0$, then $S_{()}^{(n,N,0)}(z)= R_0(z)$.
This operator does not depend on the choice of representative \eqref{TupleChoice}.

The following lemma is the key step in the recursion. It takes a summand in the partially reordered, up to length $n+1$, expansion and multiplies out, using the distributive law, each renormalized ambidexstrous handed block \eqref{handedrenormblock} of (even) length $n+1$. 

\begin{lem}\label{lem-CritResumStep} Let $n,N,k\in\NN$ with $n<N$ and let $U_0\in \scrP^{(n+1,N,k)}_\us$. Then
\begin{equation}\label{KeyRenFormStep}
S_{\us,U_0,\Lambda}^{(n+1,N,k)}(z)=
\sum_{U\in \scrP^{(n,N,k)}_\us, U\preceq U_0} S_{\us,U,\Lambda}^{(n,N,k)}(z).
\end{equation}  
\end{lem}

\begin{proof} Recall from \eqref{Ps-nNk} and \eqref{IntSummands} the definitions of the sets $\scrP^{(n,N,k)}_\us$ and the operators $S_{\us,U,\Lambda}^{(n,N,k)}(z)$. We divide the proof into two steps.

\emph{Step I:} The case $n=N-1$. In this case, we must have $U_0 = \emptyset$. Note also that $\scrP^{(N-1,N,k)}_\us = 2^{\scrA^{(N,k)}_\us}$, the set whose elements are the subsets of $\scrA^{(N,k)}_\us$. The formula \eqref{KeyRenFormStep} therefore reduces to
\begin{equation}\label{KeyRenFormStep0}
   S_{\us,\emptyset,\Lambda}^{(N,N,k)}(z)=
\sum_{ U\subset \scrA^{(N,k)}_\us} S_{\us,U,\Lambda}^{(N-1,N,k)}(z). 
\end{equation}

Let $\uut$ be a representative for the equivalence class $\varphi^{(N,k)}(\us)\in\scrT^{(N,k)}/\!\sim$, where -- again -- the bijection $\varphi^{(N,k)}$ comes from Proposition~\ref{prop-tuple-bijection}. Then 
\begin{equation}\label{ResumTopLevel}
S^{(N,N,k)}_{\us,\emptyset,\Lambda}(z) = S^{(N,k)}_{\uut}(z;G_{s_1,\Lambda},\dotsc,G_{s_k,\Lambda}).
\end{equation}
Write $\uut = (\usigma_1,\dotsc,\usigma_\nu)$ and recall that the $\usigma_i$'s have length at most $N$.

Note that for any $(j,j')\in \scrA^{(N,k)}_\us$, there exists
$i\in\llbracket 1,\nu\rrbracket$, such that $\usigma_i = (s_j,\dotsc,s_{j'})$. To see this, it suffices -- by Proposition~\ref{prop-equivalent-tuples}~\ref{item-equivtuples1} -- to observe that we only need to produce one $\uutau$ with this property and $\uutau\sim \uut$. Indeed, we may form (at most) two tuples $\uut' =(\us'_1,\dotsc,\us'_{\nu'}) \in \varphi^{-1}((s_1,\dotsc,s_{j-1}))$ (if $j>1$), $\uut'' = (\us''_1,\dotsc,\us''_{\nu''}) \in \varphi^{-1}((s_{j'+1},\dotsc,s_\nu))$ (if $j'<\nu$). Then
$\uutau = (\us'_{1},\dotsc,\us'_{\nu'},\usigma_i,\us''_{1},\dotsc,\us''_{\nu''})\in\scrT^{(N,k)}$ and $\uutau\sim\uut$.

We proceed by induction after $c=\Card(\scrA^{(N,k)}_\us)$. We start with the case $c=0$. In this case $\scrA^{(N,k)}_\us = \emptyset$ and hence $\scrP^{(N-1,N,k)} = \{\emptyset\}$. So the sum on the right-hand side of \eqref{KeyRenFormStep0} has only one term coming from $U = U_0 = \emptyset$. Since $S^{(n,N,k)}_{\us,\emptyset,\Lambda}(z) = S^{(n,k)}_{\uut'}(z;G_{s_1,\Lambda},\dotsc,G_{s_k,\Lambda})$, where $\uut'\in\varphi^{(N-1,k)}(\us)\in\scrT^{(N-1,k)}/\!\sim$, it suffices by \eqref{ResumTopLevel} and Proposition~\ref{prop-EquivSummands} to produce one $\uut'\in\scrT^{(N-1,k)}$ with $\uut'\sim\uut$ and   $S^{(N,k)}_\uut = S^{(N-1,k)}_{\uut'}$. But such a $\uut'$ may easily be produced from $\uut=(\us_1,\dotsc,\us_\ell)$ by choosing a split in $\Split(\us_i)$ for each of the $\us_i$ that are in $\scrS^{(N)}$. Here it is crucial that none of the $\us_i$ of length $N$ are ambidexstrous, so the recursive definition of the blocks $T^{(N)}_{\us_i}$ do not introduce a counter-term. See Definition~\ref{def-remblocks}.

 We now assume that \eqref{KeyRenFormStep0} holds if $c\leq c_0\in\NN_0$. Suppose that $c = c_0+1$. Let $(j,j+N-1) \in\scrA^{(N,k)}_\us$ be such that for any $(i,i+N-1)\in\scrA^{(N,k)}_\us$, we have $i\leq j$. That is, we pick the last ambidexstrous sub-string of length $N$. 
 
 Split  $\us$ into three strings
 \[
 \us' = (s_1,\dotsc,s_{j-1})\in\scrS^{(j-1)}_0 \quad \textup{and}\quad \us'' = (s_{j+N},\dotsc,s_k)\in\scrS^{(k-N-j+1)}_0,
 \]
 as well as $\usigma = (s_j,\dotsc,s_{j+N-1})\in\scrS^{(N)}_\LR$.
Note that $\us'$ has $c_0$ ambidexstrous sub-strings of length $N$, whereas $\us''$ has none. If $j=1$, we use the convention that $\us' = ()$ is a string of length $0$, and likewise, if $j+N = k+1$, then $\us'' = ()$ is a string of length $0$. 

We may therefore use the induction hypothesis on $\us'$ and $\us''$ and conclude from \eqref{KeyRenFormStep0} that
\[
S^{(N,N,j-1)}_{\us',\emptyset,\Lambda}(z) = \sum_{U'\subset \scrA^{(N,j-1)}_{\us'}} S^{(N-1,N,j-1)}_{\us',U',\Lambda}(z)
\]
and
\[
S^{(N,N,k-N- j+1)}_{\us'',\emptyset,\Lambda}(z) = S^{(N-1,N,k-N-j+1)}_{\us'',\emptyset,\Lambda}(z).
\]
If $j=1$ and $\us'= ()$, then we must have $U'=\emptyset$ and $S^{(N,N,0)}_{\us',\emptyset} = S^{(N-1,N,0)}_{\us',\emptyset} = R_0(z)$. Likewise, if $j+N = k+1$, then
$\us'' = ()$ and 
\[
S^{(N,N,0)}_{\us'',\emptyset,\Lambda}(z) = S^{(N-1,N,0)}_{\us'',\emptyset,\Lambda}(z) = R_0(z).
\]

 Split $\usigma = \usigma_\Right\circ\usigma_\Left$, where $\usigma_\Right\in\scrS^{(N_\Right)}_\Right$, $\usigma_\Left\in\scrS^{(N_\Left)}_\Left$ and $N_\Right,N_\Left\in\NN$ with $N_\Right+N_\Left = N$.
 Here we used Proposition~\ref{propo-signdecomp}~\ref{item-sd-ambi}.
 Compute, using Definitions~\ref{def-remblocks} and~\ref{def-selfenergy},
 \[
T^{(N)}_\usigma = T^{(N)}_{\usigma,\bare} + E^{(N)}_{\usigma,\Lambda}
=T^{(N_\Right)}_{\usigma_\Right}R_0(z)T^{(N_\Left)}_{\usigma_\Left}+E^{(N)}_{\usigma,\Lambda}.
 \]

 Let $U'\subset \scrA^{(N,j-1)}_{\us'}$. Compute first
 \[
 S^{(N-1,N,j-1)}_{\us',U',\Lambda}(z) E^{(N)}_{\usigma,\Lambda}S^{(N-1,N,k-j-N+1)}_{\us'',\emptyset,\Lambda}(z) = S^{(N-1,N,k)}_{\us,U'\cup\{(j,j+N-1)\},\Lambda}(z),
 \]
 where we used that $(U'\cup \{(j,j+N-1)\})^\comp = (U')^\comp\cup \{(j+N,k)\}$ keeping in mind that the complement is computed relative to the relevant ambient interval, $\llbracket 1, j-1\rrbracket$ for $U'$ and $\llbracket 1,k\rrbracket$ for $U'\cup\{(j,j+N-1)\}$.

Next we verify the identity
 \begin{equation}\label{DiffPartOfDistri}
S^{(N-1,N,j-1)}_{\us',U',\Lambda}(z) T^{(N_\Right)}_{\usigma_\Right}R_0(z)T^{(N_\Left)}_{\usigma_\Left}S^{(N-1,N,k-j-N)}_{\us'',\emptyset,\Lambda} = S^{(N-1,N,k)}_{\us,U',\Lambda}(z).
 \end{equation}
To see this, select the last element $(\ell,\ell')\in (U')^\comp$ (computed with respect to $\llbracket 1,j-1\rrbracket$.) Then, either $(\ell,\ell') = (v,j-1)$ for some $v\leq j-1$, in the case $(j-N,j-1)\not\in U'$, and $(\ell,\ell')=(j,j-1)$, if $(j-N,j-1)\in U'$ (or $j=1$).

Observe next that $(U')^\comp$ computed with reference to $\llbracket 1,k\rrbracket$ on the right-hand side of \eqref{DiffPartOfDistri}, equals $(U'\setminus \{(\ell,\ell')\})^\comp\cup \{(\ell,k)\}$, where $(U')^\comp$ here is computed in the context of the left-hand side, namely with respect to $\llbracket 1,j-1\rrbracket$.

We may now see the identity \eqref{DiffPartOfDistri} as follows. With $\ell$ being the first component of $(\ell,\ell')$ picked above,
we construct $\uutau \in\varphi^{(N-1,k-\ell+1)}((s_{\ell},\dotsc,s_k))$ as follows. Select tuples $\uutau'=(\us'_1\dotsc,\us'_{u'}) \in\varphi^{(N-1,j-\ell)}((s_\ell,\dotsc,s_{j-1}))$ and
  $\uutau'' = (\us'_1\dotsc,\us'_{u''}) \in\varphi^{(N-1,k-j-N+1)}((s_{j+N},\dotsc,s_k))$. We may now construct $\uutau$ by setting
  \[
\uutau = (\us'_1,\dotsc,\us'_{u'},\usigma_\Right,\usigma_\Left,\us''_1,\dotsc,\us''_{u''}).
  \]
  That this is indeed a tuple in $\scrT^{(N-1,k-\ell+1)}$ follows from $\usigma_\Right$ being right-handed and $\usigma_\Left$ being left-handed, such that it is immaterial what type of strings we have as $\us'_{u'}$ and $\us''_1$. We may therefore conclude that
  \[
     S^{(N-1,k-\ell+1)}_{\uutau} = S^{(N-1,j-\ell)}_{\uut'} T^{(N_\Right)}_{\usigma_\Right} R_0(z) T^{(N_\Left)}_{\usigma_\Left}  S^{(N-1,k-j-N+1)}_{\uut''},
  \]
  which implies \eqref{DiffPartOfDistri} as desired.

\emph{Step II:} The general case. Let $n,N,k\in\NN$ with $n<N$.
Fix a $U_0\in\scrP^{(n+1,N,k)}_\us$. The argument is by induction after $c = \Card(U_0)$. We begin with $c=0$, which is almost what was covered above. In this case, using \eqref{KeyRenFormStep0} with $N = n+1$,
\begin{align*}
S^{(n+1,N,k)}_{\us,\emptyset,\Lambda}(z) & = S^{(n+1,n+1,k)}_{\us,\emptyset,\Lambda}(z)
= \sum_{U\subset \scrA^{(n+1,k)}_{\us}} S^{(n,n+1,k)}_{\us,U,\Lambda}(z)\\
&=\sum_{U\in \scrP^{(n,N,k)}_{\us}, U\preceq \emptyset} S^{(n,N,k)}_{\us,U,\Lambda}(z).
\end{align*}

Now let $c_0\in\NN_0$ and assume that the identity \eqref{KeyRenFormStep0} holds true for $U_0$ with $\Card(U_0)\leq c_0$. Assume that $c = \Card(U_0) = c_0+1$, which forces us to have $n+1<N$.  Pick $(j,j')\in U_0$, such that for any $(i,i')\in U_0$, we have
$i\leq j$. This corresponds to the last ambidexstrous sub-string  indexed by $U_0$. Note that $j'-j +1 > n+1$.

Put $U_0' = U_0 \setminus \{(j,j')\}$, which has $\Card(U_0') =c_0$. Split $\us$ into three parts, $\us' = (s_1,\dotsc, s_{j-1})$ , $\us'' = (s_{j'+1},\dotsc,s_k)$ and
$\usigma = \usigma_{(j,j')}=(s_j,\dotsc,s_{j'})$. We again allow for the case $j=1$, where $\us' = ()$ is a string of length $0$, and the case $j'+1 = k+1$, where $\us'' = ()$.

Then $U_0' \in\scrP^{(n+1,N,j-1)}_{\us'}$, so we can use the induction hypothesis to compute
\[
S^{(n+1,N,j-1)}_{\us',U_0'} = \sum_{U'\in\scrP^{(n,N,j-1)}_{\us'}, U'\preceq U_0'} S^{(n,N,j-1)}_{\us',U'}.
\]
Also, using the induction start,
\[
S^{(n+1,N,k-j')}_{\us'',\emptyset} = \sum_{U''\in \scrP^{(n,N,k-j')}_{\us''}, U''\preceq \emptyset} S^{(n,N,k-j')}_{\us'',U''}. 
\]
Note that for $U'\in \scrP^{(n,N,j-1)}_{\us'}$ and $U''\in \scrP^{(n,N,k-j')}_{\us''}$ with $U'\preceq U_0'$ and $U''\preceq \emptyset$, we have
\[
E_{\usigma,\Lambda}^{(j'-j+1)} S^{(n,N,j-1)}_{\us',U'} S^{(n,N,k-j')}_{\us'',U''} = S^{(n,N,k)}_{\us, U'\cup \{(j,j')\}\cup U''}
\]
and that $U=  U'\cup \{(j,j')\}\cup U''\in\scrP^{(n,N,k)}_\us$ with $U\preceq U_0$. Here we used the computation $U^\comp = (U')^\comp\cup (U'')^\comp$. 

Conversely, any $U\in\scrP^{(n,N,k)}_\us$ with $U\preceq U_0$ contains $(j,j')$ and can be written as above with $U' \in\scrP^{(n,N,j-1)}_{\us'}$, $U''\in\scrS^{(n,N,k-j')}_{\us''}$ and $U'\preceq U_0'$ as well as $U''\preceq \emptyset$. This completes the proof. 
\end{proof}

We are now finally in a position to end with: 

\begin{proof}[Proof of Theorem~\ref{thm-reordering}]
Let $N\in\NN$.
Using the notation from above, we do a Neumann expansion of the resolvent
\begin{align}\label{naiveneumann}
\nonumber \bigl(H_\Lambda - E^{(N)}_\Lambda-z\bigr)^{-1}
& = R_0(z)+ \sum_{n=1}^\infty R_0(z)\Bigl\{\Bigl(-H_\mathrm{I}\bigl(G^{(1)}_\Lambda,G^{(2)}_\Lambda\bigr)+E^{(N)}_\Lambda\Bigr)R_0(z)\Bigr\}^n \\
& = R_0(z) + \sum_{k=1}^\infty\sum_{\us\in\scrS^{(k)}_0} \sum_{U\in \scrP^{(1,N,k)}_\us} S^{(1,N,k)}_{\us,U,\Lambda}(z).
\end{align}
For the last equality, we -- for each $n$ -- simply multiplied out all the $n$ brackets, each containing $4$ terms from $H_\mathrm{I}$, cf~\eqref{CutoffInteraction}, and for $E^{(N)}_\Lambda$, we get a term for each $\usigma\in\scrS^{(2i)}_\LR$ with $2i\leq N$ (see Definition~\ref{def-selfenergy}). Observe that it follows from Lemma~\ref{lem-basicblockbound}~\ref{item-bb-k1} and~\ref{item-bb-bare} that there exists a $C_N(\Lambda)>0$ (as in the formulation of Theorem~\ref{thm-reordering}), such that the two infinite series above are absolutely convergent for $z$ with $\re z\leq -C_N(\Lambda)$. Note that $n$ and $k$ in the sums above do not count the same thing. Whereas $n$ counts the number of factors in the first Neumann expansion, $k$ counts the number of coupling functions, including those inside the counter-terms.

Fix $k\in\NN$ and $\us\in\scrS^{(k)}_0$. We may now use Lemma~\ref{lem-CritResumStep} recursively to conclude that
\begin{align*}
& \sum_{U\in \scrP^{(1,N,k)}_\us} S^{(1,N,k)}_{\us,U,\Lambda}(z)\\
 & =  \sum_{U_0\in \scrP^{(2,N,k)}_\us}\sum_{U\in \scrP^{(1,N,k)}_\us, U\preceq U_0} S^{(1,N,k)}_{\us,U,\Lambda}(z)\\
 & =  \sum_{U\in \scrP^{(2,N,k)}_\us}S^{(2,N,k)}_{\us,U,\Lambda}(z)\\
 &\qquad \vdots\\
 &= \sum_{U\in \scrP^{(N-1,N,k)}_\us}S^{(N-1,N,k)}_{\us,U,\Lambda}(z)\\
 &= S^{(N,N,k)}_{\us,\emptyset,\Lambda}(z)
 = S^{(N,k)}_\uut\bigl(z;G_{s_1,\Lambda},\dotsc,G_{s_k,\Lambda}\bigr),
\end{align*}
where $\uut\in\scrS^{(N,k)}$ is a representative of the equivalence class $\varphi^{(N,k)}(\us)\in\scrT^{(N,k)}/\!\sim$.
Inserting back into the Neumann expansion, we conclude that
\begin{equation*}
\bigl(H_\Lambda - E^{(N)}_\Lambda-z\bigr)^{-1}
 =R_0(z) + \sum_{k=1}^\infty
\sum_{[\uut]\in\scrT^{(N,k)}/\!\sim} S^{(N,k)}_\uut\bigl(z;G_{s_1,\Lambda},\dotsc,G_{s_k,\Lambda}\bigr),
\end{equation*}
which concludes the proof. Recall that the right-hand side is absolutely convergent for $z$ with $\re(z)\leq -C_N(\Lambda)$. One may also appeal to Lemma~\ref{lem-basicsummandbound} to get absolute convergence (for a possibly different $C_N(\Lambda)$).
\end{proof}

\section{Regular Wick Monomials}\label{Sec-RegOp}

The renormalized handed blocks of operators, $T^{(k)}_\us(z;\uF^{(k)})$ for $k\in\NN$ and $\us\in\scrS^{(k)}$, cf. Definition~\ref{def-remblocks}, are difficult to estimate directly. Recalling the main ideas of \cite{AlvaMoll2021}, the strategy begins with computing normal ordered expressions for these operators. The terms in the normal ordered expressions will be the regular Wick monomials of this section. 

\subsection{The definition of regular Wick monomials}

\begin{Def}\label{Notation}
Let $n \in \NN$. We introduce the following general notations: 
\begin{itemize}[left=0pt .. \parindent]
\item $\Jab \subset \llbracket 1, n \rrbracket$, respectively $\Jcb \subset \llbracket 1, n \rrbracket$, will label a set of indices associated to bosonic annihilation operators, respectively creation operators. We impose that $\Jab \cap \Jcb = \emptyset$.
\item $\Jaf \subset \llbracket 1, n \rrbracket$, respectively $\Jcf \subset  \llbracket 1, n \rrbracket$, will label a set of indices associated to fermionic annihilation operators, respectively creation operators.  We impose that $\Jaf \cap \Jcf = \emptyset$.
\item Abbreviate $\scrJ = J_\ab\cup J_\cb\cup J_\af\cup J_\cf$. 
\item  $I_\ab \subset \llbracket 1, n \rrbracket\setminus (\Jab\cup \Jcb) $ is the subset of indices labelling bosonic annihilation operators that have been contracted,  and $f_{\ab}\colon I_\ab\to \llbracket 1, n \rrbracket\setminus (\Jab\cup \Jcb\cup I_\ab)$ a bijection with $f_\ab(i) > i$, for all $i\in I_\ab$ encodes which bosonic creation operator that was involved in the contraction.
\item $I_\af \subset \llbracket 1, n \rrbracket\setminus (\Jaf\cup \Jcf) $ is a subset of indices labelling fermionic annihilation operators that have been contracted, and $f_{\af}\colon I_\af \to \llbracket 1, n \rrbracket\setminus (\Jaf\cup \Jcf\cup I_\af)$ a bijection with $f_\af(i) >i$, for all $i\in I_\af$ encodes which fermionic creation operator that was involved in the contraction.
\end{itemize}
\end{Def}

\begin{rk}\label{FromSetsToSignature} Let  $J_\ab,J_\cb,J_\af,J_\cf$ and $I_\ab,I_\af$ be subsets of $\llbracket 1,n\rrbracket$, together with the functions $f_\ab$ and $f_\af$, be given as in Definition~\ref{Notation}. Then there is a unique signature string $\us\in\scrS^{(n)}_0$ associated with the sets.
Noting that $J_\ab\cup J_\cb\cup I_\ab\cup f_\ab(I_\ab)=
J_\af\cup J_\cf\cup I_\af\cup f_\ab(I_\af)= \llbracket 1,n\rrbracket$, we set 
\[
s_i = \begin{cases}
    \ab\af, & \textup{if } i\in (J_\ab\cup I_\ab)\cap (J_\af\cup I_\af)\\
 \ab\cf, & \textup{if } i\in (J_\ab\cup I_\ab)\cap (J_\cb\cup f_\af(I_\af))\\
 \cb\af, & \textup{if } i\in (J_\cb\cup f_\ab(I_\ab))\cap (J_\af\cup I_\af)\\
 \cb\cf, & \textup{if } i\in (J_\cb\cup f_\cb(I_\ab))\cap (J_\cf\cup f_\af(I_\af)).
\end{cases}
\]
Conversely, given a signature string $\us\in\scrS^{(n)}_0$ one may reconstruct the sets $J_\ab\cup I_\af,J_\cb\cup f_\ab(I_\ab),J_\af\cup I_\af,J_\cf\cup f_\af(I_\af)$, but the signature string does not contain any information about contractions.
\end{rk}

We aim to normal order the renormalized handed blocks of operators, introduced in Subsect.~\ref{subsec-RenBlocks}, with a view towards obtaining improved norm estimates that are uniform in the ultraviolet cutoff $\Lambda$ (when the $F_i$'s are replaced by appropriate coupling functions). 
We introduce a class of Wick monomials that will appear when we normal order renormalized blocks. As we shall see in Lemma~\ref{HNFCT}, the product of two Wick monomials can be written as a finite sum of Wick monomials.

\begin{Def}[Wick Monomial]
\label{WickMonomial} Let $n\in\NN$. We say that $T$ is a \emph{Wick monomial} of length $n$, if $T\colon \CC_-^*\to \scrM^{(n)}$ and there exist sets $\Jab$, $\Jcb$, $\Jaf$, $\Jcf$, $I_\af$, $I_\ab$ and functions $f_\ab,f_\af$,  defined as in Definition~\ref{Notation}, together with 
$\scrA \subset \llbracket 1,n-1 \rrbracket$ ($\scrA=\emptyset$ if $n=1$) and 
a continuous function $\scrL(\{k_j\}_{j \in I_\af}, \{ q_j\}_{j \in I_\ab})$ such that for $z\in\CC_-^*$ and $\uF^{(n)} = (F_1,\dotsc,F_n)\in (L^2(\RR^d\times\RR^d))^n$:
\begin{align}
\label{GeneralFormRegularOperator}
 \nonumber T(z;\uF^{(n)}) & =    \int \prod^n_{i = 1} F_i(k_i, q_i) \scrL\bigl(\{k_j\}_{j \in I_\af}, \{ q_j\}_{j \in I_\ab}\bigr) \\
\nonumber & \quad \prod_{i \in I_\af}\kdelta(k_i-k_{f_{\af}(i)})  \prod_{i \in I_\ab}\kdelta(q_i-q_{f_{\ab}(i)})   \prod_{j \in  \Jcf} \cf(k_j)\prod_{j \in \Jcb } \cb(q_j)\\
 & \quad \prod_{i \in \scrA} R_0(z-C_{i}-R_i)  \prod_{j \in \Jab } \ab(q_j) \prod_{j \in  \Jaf} \af(k_j) \prod_{i=1}^n dq_i dk_i,
\end{align}
where: 
\begin{equation}\label{Ccomponents}
C_i =  \sum_{\underset{j \in \Jab }{j \leq i}} \wb(q_j) +  \sum_{\underset{j \in  \Jaf}{j \leq i}} \wf(k_j)  + \sum_{\underset{j \in \Jcb}{i < j}} \wb(q_j) + \sum_{\underset{j \in  \Jcf}{i < j}} \wf(k_j),
\end{equation}
and 
\begin{equation}
    \label{Rcomponents}
    R_i = \sum_{\underset{ j \leq i <  f_{\ab}(j)}{j\in I_\ab}}\wb(q_j) + \sum_{\underset{ j\leq i <  f_{\af}(j)}{j\in I_\af}}\wf(k_j).
\end{equation}
\end{Def}

\begin{Def}[Indices of a Wick monomial]\label{def-Tnumbers} Remark~\ref{FromSetsToSignature} leads us to introduce two indices that are associated with a Wick monomial $T$ as introduced in Definition~\ref{WickMonomial}. Let $\us\in\scrS^{(n)}_0$ be the signature string from Remark~\ref{FromSetsToSignature} associated with $T$. Then we set
\[
n_\ab(T) = n_\ab(1,n;\us) \quad \textup{and} \quad n_\af(T) = n_\af(1,n;\us).
\]
\end{Def}

\begin{rk}
Note that the indices $n_\ab(T)$ and $n_\af(T)$ only depend on $T$ through its signature $\us$. That is, they are independent of the contractions in the Wick monomial. 
\end{rk}

When attempting to estimate Wick monomials, some particle channels have to be studied in a specific way, for example, we will have to exploit boundedness of smeared fermionic annihilation and creation operators, or proceed differently when operators have been all contracted. Therefore, the sets $\Jab$, $\Jaf$, $\Jcb$ and $\Jcf$ are not enough to state our regularity conditions and that is why we need to introduce the concept of covers: 

\begin{Def}[Cover]
\label{DefCover}
Let the sets $\Jaf, \Jcf, \Jab$ and $\Jcb$ be as in Definition~\ref{Notation}. A \emph{cover} is a partition $(P_{\ab}, P_{\af}, P_{\cb}, P_{\cf})$ of $\scrJ= \Jaf \cup \Jcf \cup \Jab \cup \Jcb$ into pairwise disjoint sets, fulfilling the following properties:  $P_\ab\cup P_\cb \cup P_\af \cup P_\cf = \scrJ$ and
\begin{equation*}
    P_{\ab}  \subseteq \Jab, \quad 
    P_{\af}  \subseteq \Jaf, \quad
    P_{\cb}  \subseteq \Jcb \quad\textup{and} \quad 
     P_{\cf}  \subseteq \Jcf.
\end{equation*}
\end{Def}

In preparing to estimate Wick monomials, we will have to set the stage by placing the creation and annihilation operators in a particular order. More precisely, the order of the fermionic operators whose indices correspond to non-contracted bosonic operators will have to be treated with caution in order to use boundedness of smeared fermionic operators. This will be the subject of Section~\ref{Sec-OrdOp} on Ordered Wick Monomials. We introduce admissible maps that will take care of this ordering: 

\begin{Def}[Admissible maps]
\label{Admissiblemaps}
Let $\Jaf, \Jcf, \Jab, \Jcb\subset \llbracket 1, n \rrbracket$ with a cover $(P_{\cb},P_{\ab},P_{\cf},P_{\af})$, be as in Definitions~\ref{Notation} and~\ref{DefCover}. Then
\begin{itemize}[left=0pt .. \parindent]
\item A function $\sigma\colon J_\cf\cup J_\af\to \NN_{0,\infty}$ is called \emph{admissible} w.r.t. the cover if:
\begin{itemize}
\item For any  $j \in P_{\cf}\cup P_\af$, we have $\sigma(j) = j $.
\item For any $j \in \Jcf \setminus P_{\cf}$, $\sigma(j) \leq j$.
\item For any $j \in \Jaf \setminus P_{\af}$, $\sigma(j) \geq j$.
\item The range of $\sigma$ is included in $\{0\} \cup \llbracket 1, n \rrbracket \cup \{\infty\}$.
\end{itemize}
\item  
$\psigma = \set{j\in J_\cf\cup J_\af}{\sigma(j)=0 \textup{ or } \sigma(j) = +\infty} = \set{j\in J_\cf}{\sigma(j)=0}\cup \set{j\in J_\af}{\sigma(j)=+\infty}$.
\end{itemize}
\end{Def}

\begin{rk}\label{rem-id-admissible} Note that the identity map $\id\colon J_\af\cup J_\cf\to \llbracket 1,n\rrbracket$, defined by $\id(j) = j$ is always admissible with $\partial\id = \emptyset$.
\end{rk}

To estimate the uncontracted annihilation and creation operators in a Wick monomial using kinetic energy bounds, we will distribute fractional powers of the available resolvents onto the annihilation and creation operators. To do so, it is useful to introduce the following two definitions:

\begin{Def}
\label{SumsinRes}
Let us consider $\Jaf, \Jcf, \Jab, \Jcb$  defined as in Definition \ref{Notation}, a cover $(P_{\ab}, P_{\af}, P_{\cb}, P_{\cf})$ defined as in Definition~\ref{DefCover}, together with an admissible map $\sigma\colon \llbracket 1,n\rrbracket\to \NN_{0,\infty}$. For any $i \in \llbracket 1, n-1 \rrbracket$ we define
\[
C_i(\sigma) = \sum_{\underset{j\in  \Jcb }{j > i}} \wb(q_j) + \sum_{\underset{j\in  \Jcf }{\sigma(j)>i}} \wf(k_j)+\sum_{\underset{j\in \Jab }{j \leq i}} \wb(q_j) +\sum_{\underset{j\in \Jaf }{\sigma(j) \leq i}}\wf(k_j).
\]
Note that $C_i(\id) = C_i$, cf. \eqref{Ccomponents}.
\end{Def}

\begin{Def}[Admissible exponents]\label{def-admexp} Let $\Jaf, \Jcf, \Jab, \Jcb$ be as in Definition~\ref{Notation}, and a cover $(P_{\ab}, P_{\af}, P_{\cb}, P_{\cf})$ be as in Definition~\ref{DefCover}. Let $\scrA\subset \llbracket 1,n-1\rrbracket$ and suppose $\{\alpha_i\}_{i\in\scrJ}$ is given with $\alpha_i\geq 0$ for all $i\in\scrJ$. A  collection of exponents $\{\gamma_{i;j}\}_{i\in\scrA,j\in \scrJ}$ is called \emph{admissible exponents} if  for all $j\in\scrJ$ and $i\in\scrA$: $\gamma_{i,j}\geq 0$ and :
 \begin{equation}
 \label{ConstraintOnGamma}
\begin{aligned}
& \textup{if } j\in  P_{\cb}\cup \bigl(P_\cf\cap (J_\cb \cup f_\ab(I_\ab))\bigr)\textup{ and } i\geq j:  & & \gamma_{i;j}  = 0, \\
&  \textup{if } j\in P_{\cf}\cap I_\ab\textup{ and } i\geq f_\ab(j): & & \gamma_{i;j}  = 0,\\
& \textup{if }  j\in  P_{\ab}\cup \bigl(P_\af\cup (J_\ab\cup I_\af)\bigr)\textup{ and } i <j: & &  \gamma_{i;j}  = 0,\\
&   \textup{if }  j\in P_{\af}\cap f_\ab(I_\ab)\textup{ and }  i< f_\ab^{-1}(j): & & \gamma_{i;j}  = 0, 
\end{aligned}
\end{equation}
together with
\begin{equation}\label{SumOfGammas}
\forall j \in \scrJ : \quad\alpha_j = \sum_{i\in \scrA} \gamma_{i;j}, \qquad
\forall i \in \scrA:  \quad
\overline{\gamma}_i := 
  \sum_{j\in \scrJ  } \gamma_{i;j} \leq 1.
\end{equation}
\end{Def}

We are now ready, keeping the definitions above in mind, to introduce the notion of regular Wick monomial:

\begin{Def}[Regular Wick monomial]
\label{RegularOperator}
Let $T$ be a Wick monomial of length $n$. We say that $T$ is \emph{regular},
if there exists a positive constant $c_T$ such that for any cover $(P_{\ab},P_{\af}, P_{\cb},P_{\cf})$, any admissible map $\sigma\colon J_\af\cup J_\cf\to \NN_{0,\infty}$, any collection $\{\alpha_i\}_{i =1}^n$ with
\begin{equation}
\label{Caraalpha}
 \forall i \in \llbracket 1, n \rrbracket: \quad 0 \leq \alpha_i \leq 1 \qquad \textup{and} \qquad
\sum^n_{i =1} \alpha_i  =  n-1,
\end{equation}
 there exist $\{\beta_{i}\}_{i=1}^n$ with $0\leq \beta_i\leq \alpha_i$ and $\beta_{i}=0$ for $i\in\psigma$, and a collection of admissible exponents 
 $\{\gamma_{i;j}\}_{i\in\scrA,j\in \scrJ}$,
such that for any $z\in \CC_-^*$, the following estimate holds:
\begin{align}
\label{regularityproperty}
\nonumber & \Bigl|\scrL\bigl( \{ k_j\}_{j \in I_\af}, \{ q_j\}_{j \in I_\ab}\bigr)\Bigr|  \prod_{j \in P_{\ab}\cup P_{\cb}} \wb(q_j)^{-\alpha_j} \prod_{ j \in P_{\af}\cup P_{\cf}} \wf(k_j)^{-\alpha_j} \\ 
\nonumber & \quad \qquad \prod_{i\in I_\ab}\kdelta\bigl(q_{i} - q_{f_\ab(i)}\bigr) \prod_{j\in I_\af}\kdelta\bigl(k_{j} - k_{f_\af(j)}\bigr) \biggl\|  \prod_{i\in \scrA} R_0\bigl(z-C_i(\sigma) -R_i\bigr)^{1-\overline{\gamma}_i}  \biggr\|\\
&\qquad \leq c_T\frac{ \prod_{i\in I_\ab}\kdelta(q_{i} - q_{f_\ab(i)}) \prod_{j\in I_\af}\kdelta(k_{j} - k_{f_\af(j)})}{\prod^n_{i=1}[\wb(q_i)]^{\alpha_i-\beta_{i}}[\wf(k_i)]^{\beta_{i}}}.
\end{align}
\end{Def}

Estimate \eqref{regularityproperty} is an important tool to estimate renormalized handed blocks of operators. The definitions of coverings, admissible exponents and the estimate \eqref{regularityproperty} ensures that the product of Regular Wick Monomials are finite sums of Regular Wick Monomials. See Lemma~\ref{HBFCTANFCO}. 

\begin{rk}\label{rem-Wick}
Observe that, pointwise in the $q_j$'s and $k_j$'s, the right-hand side of \eqref{regularityproperty} is bounded uniformly in $z\in\CC_-^*$. 
Since, for $z\in \CC_-^*$, we have  $\|R_0(z-C_i(\sigma)-R_i)\|^2 = ((C_i(\sigma)+R_i+|\re(z)|)^2 + \im(z)^2)^{-1}$ , we conclude that for any $i\in\scrA$, the function $C_i(\sigma)+R_i$ is non-zero, i.e., there is at least one dispersion relation present in the total sum, either in $C_i(\sigma)$ or in $R_i$. Hence, choosing $\sigma=\id$, regular Wick monomials are defined for all $z\in \CC_-$ and the estimate in \eqref{regularityproperty} holds also for $z=0$.
\end{rk}

When estimating regular Wick monomials, one can observe three types of behaviors. Either the monomial can be controlled by a resolvent at its right, or by a resolvent at its left, or alternatively, we have contracted all the annihilation and creation operators and we need to add a counterterm. The precise definition of these three types of monomials are stated in the following definition. 

\begin{Def}[Handed Wick monomials]\label{def-handedWick} Let $T$ be a regular Wick monomial. We say that $T$ is:
\begin{enumerate}[label = \textup{(\arabic*)}]
    \item\label{item-RHWick}  \emph{right-handed} if  $\Jab \cup (\Jaf\setminus  \Jcb ) \neq \emptyset$. Right-handed Wick monomials will be denoted by
$\rightT$.
\item\label{item-LHWick}  \emph{left-handed} if  $\Jcb \cup(\Jcf\setminus \Jab  ) \neq \emptyset$. Left-handed Wick monomials will be denoted $\leftT$.
\item\label{item-FCWick}  \emph{fully contracted} if  $\Jcb =\Jcf =  \Jab = \Jaf  =\emptyset$. 
\end{enumerate}
\end{Def}

\begin{rk}
Note that if 
\begin{equation*}
\Jab \cup (\Jaf\backslash \Jcb )  = \emptyset\quad\textup{and}\quad
\Jcb \cup (\Jcf\backslash \Jab)  =  \emptyset,
\end{equation*}
then 
\[
\Jab=\Jcb=\Jaf=\Jcf=\emptyset.
\]
In other words, if a regular Wick monomial is neither right-handed nor left-handed, then it is fully contracted.
\end{rk}

\begin{rk}\label{rem-naToHanded} Let $T$ be a regular Wick monomial of length $n$. Recall the indices $n_\ab(T)$ and $n_\af(T)$ from Definition~\ref {def-Tnumbers}. We have the following simple observations:
\begin{enumerate}[label = \textup{(\roman*)}]
\item If we have either $n_\ab(T)>0$ or we have $n_\ab(T) = 0$ and $n_\af(T)>0$, then $T$ is right-handed.
\item If we have either $n_\ab(T)<0$ or we have $n_\ab(T) = 0$ and $n_\af(T)<0$, then $T$ is left-handed.
\item If both $n_\ab(T) = n_\af(T) = 0$, then $T$ is either fully contracted or both left- and right-handed at the same time.
\end{enumerate}
\end{rk}

We finally introduce the concept of adjoint of a Wick monomial, which will enable us to simplify many of our proofs.

\begin{Def}[Adjoint of  a Wick monomial] Let $T$ be a Wick monomial. We define the adjoint of $T$, denoted by $T^*$, as the map $T^*\colon \CC^*_- \times 
L^2(\RR^d\times \RR^d)^{n}\to \scrL_\fin$ defined by setting 
\[
T^*(z;F_1, \dotsc, F_n)= T(\overline{z};\overline{F}_n,\dotsc, \overline{F}_1)^*_{\vert\Hfin}.
\]
\end{Def}

\begin{lem}
\label{LemAdjoint}
    Let $T$ be a Wick monomial of length $n$. Then its adjoint $T^*$ is also a Wick monomial of length $n$. Furthermore, if $T$ is regular, so is $T^*$ and we may choose the bounding constant $c_{T^*} = c_T$.
\end{lem}

\begin{rk}\label{rem-TAdjoint}
 Let $\us\in\scrS^{(n)}_0$ be the signature string associated with $T$ (cf.~Remark~\ref{FromSetsToSignature}). Then $\us^*$ is the signature string associated with $T^*$. Hence, we have $n_\ab(T^*) = -n_\ab(T)$ and $n_\af(T^*) = -n_\af(T)$. See Definition~\ref{def-Tnumbers} for the indices $n_\ab(T)$ and $n_\af(T)$.
\end{rk}

\begin{proof}
    Since $T$ is a Wick monomial it has the form \eqref{GeneralFormRegularOperator}, and we have
  \begin{align}\label{FormOfAdjoint}
 \nonumber \bigl(T(z;\overline{F_1}, &\dots \overline{F_n})\bigr)^*  =   \int \prod^n_{j = 1} F_j(k_j, q_j) \ \overline{ \scrL\bigl(\{k_j\}_{j \in I_\af}, \{ q_j\}_{j \in I_\ab}\bigr)} \\
\nonumber  & \quad \prod_{j \in I_\af}\kdelta(k_j-k_{f_{\af}(j)})  \prod_{j \in I_\ab}\kdelta(q_j-q_{f_{\ab}(j)}) 
  \prod_{j \in  \Jaf} \cf(k_j)\prod_{j \in \Jab } \cb(q_j)\\
  & \quad \prod_{i \in \scrA} R_0(\overline{z}-C_{i}-R_i)  \prod_{j \in \Jcb } \ab(q_j) \prod_{j \in  \Jcf} \af(k_j) \prod_{j=1}^n dq_j dk_j.
\end{align}
Let us define a bijection $\varphi\colon\{0\}\cup \llbracket 1 , n \rrbracket\cup \{\infty\} \to\{0\}\cup \llbracket 1 , n \rrbracket\cup \{\infty\}$ by
\begin{equation}\label{varphi}
\varphi(i) = \begin{cases}
n+1-i & \textup{ if } i \in \llbracket 1 , n \rrbracket\\
0 & \textup{ if } i=\infty\\
\infty & \textup{ if } i=0. 
\end{cases}
\end{equation}
Although $\varphi^{-1} = \varphi$, we write $\varphi^{-1}$ below, to better track where objects live. 
We may now introduce the following objects
\begin{align}\label{adjoint-primeobjects}
\nonumber   \Jcb' & = \varphi(\Jab)    & \Jab' & = \varphi(\Jcb) & I_\ab' & = \varphi(f_{\ab}(I_\ab))\\
\nonumber    \Jcf' & = \varphi(\Jaf) & \Jaf' & = \varphi(\Jaf)   & I_\af' & = \varphi(f_{\af}(I_\af))\\
\nonumber   f_\ab' & = \varphi\circ f_\ab^{-1} \circ \varphi^{-1} & \nonumber f_\af' & = \varphi\circ f_\af^{-1} \circ \varphi^{-1} &  &\\
   \scrA' & = \varphi(\scrA) - 1   & F'_i & = F_{\varphi(i)}. & & 
\end{align}

Let now $i\in \scrA\subset \llbracket 1,n-1\rrbracket$ and set $\nu = \varphi(i)-1\in\scrA'$. Recall the form of $C_i$ from \eqref{Ccomponents}. We compute 
 {\allowdisplaybreaks
\begin{align}\label{CiAdjoint}
\nonumber C_i & = \sum_{\underset{j\in  J_\cb }{j > i}} \wb(q_j) + \sum_{\underset{j\in  \Jcf }{j>i}} \wf(k_j)+\sum_{\underset{j\in \Jab }{j \leq i}} \wb(q_j) +\sum_{\underset{j\in \Jaf }{j \leq i}}\wf(k_j)\\
\nonumber & = \sum_{\underset{j\in J_\cb  }{\varphi(j) < \varphi(i)}} \wb(q_{j}) + \sum_{\underset{j\in  J_\cf }{\varphi(j) < \varphi(i)}} \wf(k_{j})+\sum_{\underset{j\in J_\ab }{\varphi(j) \geq \varphi(i)}} \wb(q_{j}) +\sum_{\underset{j\in J_\af }{\varphi(j)\geq \varphi(i)}}\wf(k_{j})\\
\nonumber & = \sum_{\underset{j\in J_\cb  }{\varphi(j)\leq \nu}} \wb(q_{j}) + \sum_{\underset{j\in  J_\cf }{\varphi(j) \leq \nu}} \wf(k_{j})+\sum_{\underset{j\in J_\ab }{\varphi(j) > \nu}} \wb(q_{j}) +\sum_{\underset{j\in J_\af }{\varphi(j) > \nu}}\wf(k_{j})\\
\nonumber & = \sum_{\underset{\ell\in J'_\ab }{\ell \leq \nu}} \wb(q_{\varphi^{-1}(\ell)}) + \sum_{\underset{\ell\in  J'_\af }{\ell\leq \nu}} \wf(k_{\varphi^{-1}(\ell)})\\
\nonumber & \quad +\sum_{\underset{\ell\in J'_\cb }{\ell > \nu}} \wb(q_{\varphi^{-1}(\ell)}) +\sum_{\underset{\ell\in J'_\cf }{\ell > k}}\wf(k_{\varphi^{-1}(\ell)})\\
& =: \tC'_\nu.
\end{align}
}
 Note that the momentum variables entering into $\tC'_\nu$ has been inverted by $\varphi$ as compared to $C'_\nu$ (defined with the new sets $(J'_\ab,J'_\cb,J'_\af,J'_\cf)$ instead of the original sets).

 Similarly, we may for $i\in\scrA$ compute, up the identification $q_j = q_{f_\ab(j)}$, for $i\in I_\ab$, and 
$k_j = f_\af(k_j)$, for $j\in I_\af$, coming from the delta functions. Recalling the expression of $R_i$ from \eqref{Rcomponents}, we have:
\begin{align*}
R_i & =\sum_{\underset{ j \leq i <  f_{\ab}(j)}{j\in I_\ab}}\wb(q_{f_\ab(j)}) + \sum_{\underset{ j\leq i <  f_{\af}(j)}{j\in I_\af}}\wf(k_{f_\af(j)})\\
    & = \sum_{\underset{ f_\ab^{-1}(\ell) \leq i <  \ell}{\ell \in f_\ab(I_\ab)}}\wb(q_\ell) + \sum_{\underset{ f_\af^{-1}(\ell)\leq i <  \ell}{\ell\in f_\af(I_\af)}}\wf(k_\ell)\\
    & = \sum_{\underset{ f_\ab^{-1}(\varphi^{-1}(u)) \leq  i <  \varphi^{-1}(u)}{u \in I'_\ab}}\wb(q_{\varphi^{-1}(u)}) + \sum_{\underset{ f_\af^{-1}(\varphi^{-1}(u)) \leq  i < \varphi^{-1}(u)}{u\in I'_\af}}\wf(k_{\varphi^{-1}(u)})\\
     & = \sum_{\underset{ u < \varphi(i) \leq  f'_\ab(u)}{u \in I'_\ab}}\wb(q_{\varphi^{-1}(u)}) + \sum_{\underset{ u < \varphi(i) \leq f'_\af(u) }{u\in I'_\af}}\wf(k_{\varphi^{-1}(u)})\\
   & = \sum_{\underset{  u \leq \nu < f'_\af(u) }{u \in I'_\ab}}\wb(q_{\varphi^{-1}(u)}) + \sum_{\underset{ u \leq \nu < f_{\af}(u) }{u\in I'_\af}}\wf(k_{\varphi^{-1}(u)})\\
    & =: \tR'_\nu
\end{align*}
where again $\nu = \varphi(i)-1\in\scrA'$. Note that the momentum variables entering into $\tR'_\nu$ has been inverted by $\varphi$ as compared to $R'_\nu$.

Therefore, for any $i\in \scrA'$, we have 
\[
R_0(\overline{z} - C_i-R_i) = R_0(\overline{z} - \tC'_\nu-\tR'_\nu),
\]
where $\nu= \varphi(i)-1$,
from which it can be deduced by relabeling the integration variables $\nu = \varphi(i)-1$ and $\mu = \varphi(j)$ (comparing with the indices in \eqref{FormOfAdjoint}) that 
  \begin{align*}
 \nonumber &(T(z;\overline{F_1},\dots \overline{F_n}))^*  =   \int \prod^n_{\mu = 1} F_\mu(k_\mu, q_\mu) \ \overline{ \scrL\bigl(\{k_\mu\}_{\mu \in I'_\af}, \{ q_\mu\}_{\mu \in I'_\ab}\bigr)} \\
\nonumber  & \qquad \prod_{\mu \in I_\af'}\kdelta(k_\mu-k_{f'_{\af}(\mu)})  \prod_{\mu \in I_\ab'}\kdelta(q_\mu-q_{f'_{\ab}(\mu)}) 
   \prod_{\mu \in  \Jcf'} \cf(k_\mu)\prod_{\mu \in \Jcb' } \cb(q_\mu) \\
   & \qquad  \prod_{\nu \in \scrA'} R_0(\overline{z}-C'_{\nu}-R'_\nu)  \prod_{\mu \in \Jab' } \ab(q_\mu) \prod_{\mu \in  \Jaf'} \af(k_\mu)\prod_{\mu=1}^n dq_\mu dk_\mu,
\end{align*}
which is of the form of a Wick monomial, as formulated in 
Definition~\ref{WickMonomial}.

It remains to prove that $T^*$ is regular. Let $(P'_\ab, P'_\cb, P'_\af,P'_\cf)$ be a cover for $(J_\ab',J_\ab',J_\af',J_\cf')$, $\sigma'\colon J_\af'\cup J_\cf'\to \NN_{0,\infty}$  be an admissible function and suppose $\{\alpha'_i\}_{i=1}^n$ fulfills the conditions \eqref{Caraalpha}.  Define 
\begin{align*}
     P_\ab & = \varphi^{-1}(P'_\cb) & P_\cb & = \varphi^{-1}(P'_\ab)\\
   P_\af & = \varphi^{-1}(P'_\cf) & P_\cf & = \varphi^{-1}(P'_\af),
\end{align*}
which is cover for $(J_\ab,J_\cb,J_\af,J_\cf)$, and
\begin{equation}\label{alphasigmaprime}
     \alpha_i   = \alpha'_{\varphi(i)},
     \quad  \sigma = \varphi^{-1} \circ \sigma' \circ \varphi.
\end{equation}
Note that $\sigma\colon J_\af\cup J_\cf\to \NN_{0,\infty}$ is admissible.
Furthermore, we observe that $\psigma = \varphi^{-1}(\psigma')$.

Since, $T$ is a regular Wick monomial, there exist $\{\beta_{i}\}_{i=1}^n$ with $0\leq \beta_i\leq \alpha_i$ and $\beta_i=0$ for $i\in\psigma$, and a collection of admissible exponents, cf. Definition~\ref{def-admexp}, 
$\{\gamma_{i;j}\}_{i\in\scrA,j\in \scrJ}$, 
 such that for any $z\in \CC_-^*$, the estimate \eqref{regularityproperty} from Definition~\ref{RegularOperator} holds.
 
Before continuing, we pause to relate $C_i(\sigma)$, $i\in\scrA$,  to $C'_\nu(\sigma')$, where
$\nu = \varphi(i)-1\in\scrA'$. Recall that $C_i(\sigma)$ is defined in Definition~\ref{SumsinRes}. For the two summands involving boson momenta in the slightly simpler computation \eqref{CiAdjoint}, there is no change. As for the two sums over fermion momenta, we start with $\sigma(j)$ in place of simply $j$ in \eqref{CiAdjoint}. Through the steps in the computation \eqref{CiAdjoint}, the index $j$ undergoes the changes
\[
j \to \varphi(j) \to \varphi(\varphi^{-1}(\ell)) = \ell.
\]
Repeating the computation for $C_i(\sigma)$, we would instead have 
\[
\sigma(j) \to \varphi(\sigma(j)) \to \varphi(\sigma(\varphi^{-1}(\ell))) = \sigma'(j).
\]
With this in mind, one may follow the same computation as in \eqref{CiAdjoint} to conclude that
\begin{align*}
C_i(\sigma)  = \tC'_\nu(\sigma') & = \sum_{\underset{\ell > \nu}{\ell\in J'_\cb}} \wb(q_{\varphi^{-1}(\ell)}) + \sum_{\underset{\sigma'(\ell)>\nu}{\ell\in  J'_\cf }} \wf(k_{\varphi^{-1}(\ell)})\\
& \quad +\sum_{\underset{\ell \leq \nu}{\ell\in \Jab' }} \wb(q_{\varphi^{-1}(\ell)}) +\sum_{\underset{\sigma'(\ell) \leq \nu}{\ell\in \Jaf' }}\wf(k_{\varphi^{-1}(\ell)}),
\end{align*}
which equals $C'_\nu(\sigma')$ up to relabeling of the boson and fermion momenta.

Define $\{\beta'_{i}\}_{i=1}^n$  by setting $ \beta'_{i} = \beta_{\varphi^{-1}(i)}$, for $i\in\llbracket 1,n\rrbracket$. Clearly, $0\leq \beta'_i\leq \alpha'_i$ and $\beta'_i = 0$ for $i\in\varphi(\psigma) = \psigma'$. 

Define $\{\gamma'_{i;j}\}_{i\in\scrA',j\in \scrJ'}$ as follows
\[
\forall i\in\scrA', \forall j\in \scrJ', \gamma'_{i;j} = \gamma_{\varphi^{-1}(i+1);\varphi^{-1}(j)}.
\]
We should check that the $\gamma'_{i;j}$'s are admissible exponents for the primed objects, cf.~Definition~\ref{def-admexp}. Let $i\in\scrA'$ and $j\in\scrJ'$, such that $\varphi^{-1}(i+1)\in\scrA$ and $\varphi^{-1}(j) \in\scrJ$. Recalling \eqref{adjoint-primeobjects}, we observe that
$\varphi^{-1}(f'_\ab(I'_\ab)) = f_\ab^{-1}(\varphi^{-1}(I'_\ab)) = I_\ab$ and $\varphi^{-1}(I'_\ab) = f_\ab(I_\ab)$. Hence we may check the constraints \eqref{ConstraintOnGamma}: 
\begin{itemize}[left=0pt .. \parindent]
    \item If $j\in P_\cb'\cup \bigl(P'_\cf\cap (J'_\cb \cup f'_\ab(I'_\ab))\bigr)$  and  $i\geq j$, then $\varphi^{-1}(j) \in P_\ab\cup \bigl(P_\af\cap (J_\ab \cup I_\ab)\bigr)$ and  $\varphi^{-1}(i+1)< \varphi^{-1}(i)\leq \varphi^{-1}(j)$. Hence $\gamma'_{i;j} = \gamma_{\varphi^{-1}(i+1);\varphi^{-1}(j)}=0$. 
    \item If $j\in P_\ab'\cup \bigl(P'_\af\cup (J'_\ab\cup I'_\af)\bigr)$ and  $i <j$,
    then $\varphi^{-1}(j) \in P_\cb\cup \bigl(P_\cf\cup (J_\cb\cup f_\ab(I_\af))\bigr)$ and 
    $\varphi^{-1}(i+1)= \varphi^{-1}(i)-1>\varphi^{-1}(j)-1$, implying that $\varphi^{-1}(i+1)\geq \varphi^{-1}(j)$.
  Therefore $\gamma'_{i;j} = \gamma_{\varphi^{-1}(i+1);\varphi^{-1}(j)}=0$.
    \item If $j\in P'_\cf\cap I'_\ab$ and $i \geq f'_\ab(j)$, then $\varphi^{-1}(j) \in P_{\af}\cap f_\ab(I_\ab)$ and $\varphi^{-1}(i+1) < \varphi^{-1}(i) \leq \varphi^{-1}(f'_\ab(j)) = f_\ab^{-1}(\varphi^{-1}(j))$. Therefore $\gamma'_{i;j} = \gamma_{\varphi^{-1}(i+1);\varphi^{-1}(j)}=0$.
    \item If $j\in P'_\af\cap f'_\ab(I'_\ab)$ and $i < {f'}^{-1}_\ab(j)$, then $\varphi^{-1}(j) \in P_\cb\cap I_\ab$ and
    $\varphi^{-1}(i+1)= \varphi^{-1}(i)-1> \varphi^{-1}({f'}^{-1}_\ab(j))-1=f_\ab(\varphi^{-1}(j))-1$, implying that $\varphi^{-1}(i+1)\geq f_\ab(\varphi^{-1}(j))$.
    Therefore $\gamma'_{i;j} = \gamma_{\varphi^{-1}(i+1);\varphi^{-1}(j)}=0$.
\end{itemize}
In addition 
\begin{itemize}[left=0pt .. \parindent]
    \item For any $j \in \scrJ'$, 
    \begin{equation*}
        \alpha'_j  = \alpha_{\varphi^{-1}(j)} 
     = \sum_{i \in \scrA} \gamma_{i,\varphi^{-1}(j)} 
        = \sum_{i \in \scrA'} \gamma_{\varphi^{-1}(i+1),\varphi^{-1}(j)} 
          = \sum_{i \in \scrA'} \gamma'_{i,j}.
    \end{equation*}
    \item For any $i \in \scrA'$, 
    \begin{equation*}
        \overline{\gamma}'_i  = \sum_{j \in \scrJ'} \gamma'_{i , j}
        = \sum_{j \in \scrJ'} \gamma_{\varphi^{-1}(i+1) , \varphi^{-1}(j)}
        = \sum_{j \in \scrJ} \gamma_{\varphi^{-1}(i+1) , j}
        = \overline{\gamma}_{\varphi^{-1}(i+1)}
        \leq 1.
    \end{equation*}
\end{itemize}

Finally, by relabelling the momenta $j = \varphi^{-1}(\nu)$ and the index $i = \varphi^{-1}(\mu+1)$, we find that
\begin{equation*}
\begin{aligned}
& \Bigl|\overline{\scrL\bigl( \{ k_\mu\}_{\mu \in I'_\af}, \{ q_\mu\}_{\mu \in I'_\ab}\bigr)} \Bigr|  \prod_{\mu \in P'_{\ab}\cup P'_{\cb}} \wb(q_\mu)^{-\alpha'_\mu} \prod_{ \mu \in P'_{\af}\cup P'_{\cf}} \wf(k_\mu)^{-\alpha'_\mu} \\ 
& \qquad \prod_{\mu\in I'_\ab}\kdelta\bigl(q_{\mu} - q_{f'_\ab(\mu)}\bigr) \prod_{\mu\in I'_\af}\kdelta\bigl(k_{\mu} - k_{f'_\af(\mu)}\bigr) \biggl\|  \prod_{\nu\in \scrA'} R_0\bigl(\overline{z}-C'_\nu(\sigma') -R'_\nu\bigr)^{1-\overline{\gamma'_\nu}}  \biggr\|
\\
& \quad =  \Bigl|\scrL\bigl( \{ k_j\}_{j \in I_\af}, \{ q_j\}_{j \in I_\ab}\bigr) \Bigr|  \prod_{j \in P_{\cb}\cup P_{\ab}} \wb(q_{j})^{-\alpha'_{\varphi(j)}} \prod_{ j \in P_{\cf}\cup P_{\af}} \wf(k_j)^{-\alpha'_{\varphi(j)}} \\ 
& \qquad \prod_{i\in I_\ab}\kdelta\bigl(q_{i} - q_{f_\ab(i)}\bigr) \prod_{j\in I_\af}\kdelta\bigl(k_{j} - k_{f_\af(j)}\bigr)\\
& \qquad  \biggl\|  \prod_{i\in \scrA} R_0\bigl(\overline{z}-\tC'_{\varphi(i)-1}(\sigma') -\tR'_{\varphi(i)-1}\bigr)^{1-\overline{\gamma}'_{\varphi(i)-1}}  \biggr\|
\\
& \quad = \Bigl|\scrL\bigl( \{ k_j\}_{j \in I_\af}, \{ q_j\}_{j \in I_\ab}\bigr)\Bigr|  \prod_{j \in P_{\ab}\cup P_{\cb}} \wb(q_j)^{-\alpha_j} \prod_{ j \in P_{\af}\cup P_{\cf}} \wf(k_j)^{-\alpha_j} \\ 
& \qquad \prod_{i\in I_\ab}\kdelta\bigl(q_{i} - q_{f_\ab(i)}\bigr) \prod_{j\in I_\af}\kdelta\bigl(k_{j} - k_{f_\af(j)}\bigr) \biggl\|  \prod_{i\in \scrA} R_0\bigl(z-C_i(\sigma) -R_i\bigr)^{1-\overline{\gamma}_i}  \biggr\| \\
&\quad \leq c_T\frac{ \prod_{i\in I_\ab}\kdelta(q_{i} - q_{f_\ab(i)}) \prod_{j\in I_\af}\kdelta(k_{j} - k_{f_\af(j)})}{\prod^n_{i=m}[\wb(q_i)]^{\alpha_i-\beta_{i}}[\wf(k_i)]^{\beta_{i}}},
\end{aligned}
\end{equation*}
where we used \eqref{regularityproperty} in the last step. This concludes the proof. 
\end{proof}

The following result is a straightforward consequence of Lemma~\ref{LemAdjoint}.
\begin{Cor}
\begin{enumerate}[label = \textup{(\arabic*)}]
    \item The adjoint of  a right-handed Wick monomial is left-handed.
    \item The adjoint of a left-handed Wick monomial is right-handed.
    \item The adjoint of a fully contracted Wick monomial is fully contracted.
\end{enumerate}
\end{Cor}

\subsection{A calculus for regular Wick monomials}\label{subsec-WickCalc}

This goal of this section is to establish Lemma~\ref{InductionLemma}. This lemma states that for $\us\in\scrS^{(n)}$, the corresponding renormalized block $T^{(n)}_\us(z,\cdot)$, can be written as a sum of right-handed Wick monomials if $\us\in\scrS^{(n)}_\Right$, as sum of left-handed Wick monomials, if $\us\in\scrS^{(n)}_\Left$ and finally, if $\us\in\scrS^{(n)}_\LR$,  as a sum of Wick monomials that are both left- and right-handed, as well as fully contracted Wick monomials from which counter-terms have been removed.  The bulk of the section is concerned with the stability of the regularity property for Wick monimials under multiplication, which is studied in Lemma~\ref{HNFCT} and Lemma~\ref{HBFCTANFCO}.

\begin{lem}[Products of Wick monomials]
\label{HNFCT} Let $m,n\in \NN$ with $m< n$. There exists a constant $M = M(n)$, such that the following holds. Let $T_1$ be a Wick monomial of length $m$ and $T_2$ a Wick monomial of length $n-m$.
\begin{enumerate}[label = \textup{(\arabic*)}]
\item\label{item-Her1-Basic} There is a collection of Wick monomials $\{T'_{j}\}_{j=1}^{M'}$ with $M'\leq M$, all of length $n$, such that for any $z\in\CC_-^*$ and $F_1,\dotsc, F_m,F_{m+1},\dotsc, F_n\in L^2(\RR^d\times\RR^d)$, we have 
\[
T_1(z;F_1, \dotsc, F_m) R_0(z)T_2(z;F_{m+1}, \dotsc, F_n)
=\sum_{j=1}^{M'}T'_j(z;F_1, \dotsc, F_n).
\]
\item\label{item-Her1-signs} If $\us_1$ is the signature string affiliated with $T_1$ and $\us_2$ is the signature string affiliated with $T_2$ (see Remark \ref{FromSetsToSignature}), then $\us = \us_1\circ\us_2$ is the signature string affiliated with each of the $T'_j$'s.
\end{enumerate}
From now on, we assume that $T_1=\rightT_1$ is right-handed  and that $T_2=\leftT_2$ is left-handed.
\begin{enumerate}[resume*]
 \item\label{item-Her1-RH} If   $n_{\ab}(\rightT_1)+ n_{\ab}(\leftT_2)  > 0$, or
$n_{\ab}(\rightT_1)+ n_{\ab}(\leftT_2) =  0$ and $n_{\af}(\rightT_1) + n_{\af}(\leftT_2) \geq 0$, then the Wick monomials $T'_j$ from \ref{item-Her1-Basic} are   
either right-handed or fully contracted.
\item\label{item-Her1-LH} If  $n_{\ab}(\rightT_1)+ n_{\ab}(\leftT_2)  <  0$, or
$n_{\ab}(\rightT_1)+ n_{\ab}(\leftT_2) =  0$ and $n_{\af}(\rightT_1)+ n_{\af}(\leftT_2) \leq 0$,  then the Wick monomials $T'_j$ from \ref{item-Her1-Basic} are  either
left-handed or fully contracted.
\end{enumerate} 
Moreover, in \ref{item-Her1-RH} and \ref{item-Her1-LH}, the bounding constant for all the $T'_j$'s can be taken to be $c_{\rightT_1}\cdot c_{\leftT_2}$.
\end{lem}

\begin{proof}
We will use superscripts $(1)$ on objects affiliated with the Wick monomial $T_1$ and, likewise, the superscript $(2)$ for objects affiliated with $T_2$. For the sets affiliated with $T_2$, it is convenient to shift the labeling from $\llbracket 1,n-m\rrbracket$ to $\llbracket m+1,n\rrbracket$, such that we get the right labeling, from $1$ to $n$, for the composition.
For the purpose of this proof, as well as the next, we make use of the abbreviations for $\ell=1,2$
\begin{equation}\label{abbrev-LDelta}
\begin{aligned}
    &\scrL^{(\ell)} := \scrL^{(\ell)}\bigl( \{ k_j\}_{j \in I^{(\ell)}_\af}, \{ q_j\}_{j \in I^{(\ell)}_\ab}\bigr)\\
    & \Delta^{(\ell)}_\ab := \prod_{i\in I_\ab}\kdelta\bigl(q_i-q_{f^{(1)}_{\ab}(i)}\bigr), \qquad \Delta^{(\ell)}_\af =  \prod_{i \in I_\af^{(1)}}\kdelta\bigl(k_i-k_{f^{(1)}_{\af}(i)}\bigr).
\end{aligned}
\end{equation} 
 We begin with \ref{item-Her1-Basic}, so let us consider $T_1(z;F_1, \dots, F_m)$ of the form: 
\begin{align*}
\nonumber &   \int \prod^m_{i = 1} F_i(k_i, q_i) \scrL^{(1)}\Delta^{(1)}_\ab\Delta^{(1)}_\af
  \prod_{j \in  \Jcf^{(1)}} \cf(k_j)\prod_{j \in \Jcb^{(1)} } \cb(q_j)\\
  & \quad \prod_{i \in \scrA^{(1)}} R_0\bigl(z-C^{(1)}_{i}-R^{(1)}_i\bigr)  \prod_{j \in \Jab^{(1)} } \ab(q_j) \prod_{j \in  \Jaf^{(1)}} \af(k_j) \prod_{j=1}^m dk_j dq_j,
\end{align*}
keeping \eqref{abbrev-LDelta} in mind, and   $T_2(z;F_{m+1}, \dots, F_n)$ of the form: 
\begin{align*}
\nonumber
&\int \prod^n_{i = m+1} F_i(k_i, q_i) \scrL^{(2)}\Delta^{(2)}_{\ab}\Delta^{(2)}_{\af} 
 \prod_{j \in  \Jcf^{(2)}} \cf(k_j)\prod_{j \in \Jcb^{(2)} } \cb(q_j) \\
 & \quad \prod_{i \in \scrA^{(2)}} R_0\bigl(z-C^{(2)}_{i}-R^{(2)}_i\bigr)  \prod_{j \in \Jab^{(2)} } \ab(q_j) \prod_{j \in  \Jaf^{(2)}} \af(k_j)\prod_{j=m+1}^n dk_j dq_j,
\end{align*}
where we made use of the relabelling from $\llbracket 1,n-m\rrbracket$ to $\llbracket m+1,n\rrbracket$.

\emph{Step I:} Normal ordering the product. We may now compute: 
\begin{align} \label{Eq-HerComp}
& \nonumber T_1(z;F_1, \dotsc, F_m) R_0(z) T_2(z;F_{m+1}, \dots, F_n)\\
\nonumber
& \quad = \int \prod^n_{i = 1} F_i(k_i, q_i) \scrL^{(1)} \scrL^{(2)}\Delta^{(1)}_\ab\Delta^{(1)}_\af\Delta^{(2)}_{\ab}\Delta^{(2)}_{\af} \\
\nonumber & \qquad
 \prod_{j \in  \Jcf^{(1)}} \cf(k_j)\prod_{j \in \Jcb^{(1)} } \cb(q_j) \prod_{i \in \scrA^{(1)}} R_0\bigl(z-C^{(1)}_{i}-R^{(1)}_i\bigr) \\
\nonumber &\qquad \biggl[ \prod_{j \in \Jab^{(1)} } \ab(q_j) \prod_{j \in  \Jaf^{(1)}} \af(k_j) R_0(z)\prod_{j \in  \Jcf^{(2)}} \cf(k_j)\prod_{j \in \Jcb^{(2)} } \cb(q_j)\biggr]\\
 &  \qquad  \prod_{i \in \scrA^{(2)}} R_0\bigl(z-C^{(2)}_{i}-R^{(2)}_i\bigr)  \prod_{j \in \Jab^{(2)} } \ab(q_j) \prod_{j \in  \Jaf^{(2)}} \af(k_j)\prod_{j=1}^n dk_j dq_j,
\end{align}
as an identity in $\scrL_\fin$.
In this first step we normal order the term in the square brackets $[\cdots]$, 
using the pull through formula. We obtain: 
\begin{align*}
&\prod_{j \in \Jab^{(1)} } \ab(q_j) \prod_{j \in  \Jaf^{(1)}} \af(k_j) \ R_0(z)\prod_{j \in  \Jcf^{(2)}} \cf(k_j)\prod_{j \in \Jcb^{(2)} } \cb(q_j)\\ 
& \qquad  =  \biggl\{\prod_{j \in \Jab^{(1)} } \ab(q_j) \prod_{j \in  \Jaf^{(1)}} \af(k_j) \prod_{j \in  \Jcf^{(2)}} \cf(k_j)\prod_{j \in \Jcb^{(2)} } \cb(q_j) \biggr\}\\
& \qquad \quad R_0\bigl(z - \sum_{j \in  \Jcf^{(2)}}\wf(k_j)-\sum_{j \in  \Jcb^{(2)}}\wb(q_j)\bigr).
\end{align*}
We can then normal order the expression in the curly brackets $\{\cdots\}$, 
leading to a representation of the contents in the square brackets $[\cdots]$ of \eqref{Eq-HerComp} as a sum of terms of the form: 
\begin{equation}\label{Eq-HerNO}
\begin{aligned}
& \prod_{i \in I'_\af} \kdelta\bigl(k_i-k_{f'_{\af}(i)}\bigr)\prod_{i \in I'_\ab} \kdelta\bigl(q_i-q_{f'_{\ab}(i)}\bigr)\prod_{j \in \Jcf'}\cf(k_j) \prod_{j \in \Jcb'}\cb(q_j)\\
&\qquad \prod_{j \in \Jaf'}\af(k_j)\prod_{j \in \Jab'}\ab(k_j),
\end{aligned}
\end{equation}
up to a sign arising from the anti-commutation relations for fermionic annihilation and creation operators. (Note that if $T$ is a regular Wick monomial, then $-T$ is also a regular Wick monomial with $\scrL$ replaced by $-\scrL$.)  In \eqref{Eq-HerNO}, the primed objects satisfy
\begin{equation}\label{J-I-primes}
\begin{aligned}
&\Jcb' \subset \Jcb^{(2)}, & & \Jab' \subset \Jab^{(1)},\\
&\Jcf' \subset \Jcf^{(2)}, & & \Jaf' \subset \Jaf^{(1)},\\
&I_\af' = \Jaf^{(1)}\setminus  \Jaf', & &   f'_{\af} \colon I_\af' \to\Jcf^{(2)}\setminus \Jcf'\textup{ is a bijection}, \\
&I_\ab' = \Jab^{(1)}\setminus \Jab', & &  f'_{\ab}\colon I_\ab'\to\Jcb^{(2)}\setminus \Jcb' \textup{ is a bijection}.
\end{aligned}
\end{equation}  
\emph{Step II:} 
Before we reinsert the expression \eqref{Eq-HerNO} into the square brackets $[\cdots]$ of \eqref{Eq-HerComp}, we need to 
define some objects that will enable us to recognize the resulting expression as a Wick monomial. We write
\begin{equation}\label{New-J-I}
\begin{aligned}
\Jab & =  \Jab' \cup \Jab^{(2)}& 
\Jaf & =  \Jaf' \cup \Jaf^{(2)} \\
\Jcb & =  \Jcb^{(1)} \cup \Jcb' & 
\Jcf & =  \Jcf^{(1)} \cup \Jcf' \\
    I_\ab & =  I_\ab^{(1)} \cup I_\ab^{(2)}\cup I'_\ab  &
I_\af  & =  I_\af^{(1)}\cup I_\af^{(2)}\cup I'_\af.
\end{aligned}
\end{equation}
Observe that
\[
\scrJ = J_\ab\cup J_\cb\cup J_\af\cup J_\cf = \bigl(\scrJ^{(1)}\cup \scrJ^{(2)}\bigr)\setminus \bigl(I_\ab'\cup f_\ab'(I_\ab')\cup I_\af'\cup f_\af'(I_\af')\bigr).
\]
We finally introduce functions $f_\ab\colon I_\ab\to \llbracket 1,n\rrbracket$ and $f_\af\colon I_\af\to \llbracket 1,n\rrbracket$ by setting
\begin{equation}\label{New-f}
    f_\ab(i) = \begin{cases} f_\ab^{(1)}(i),& i\in I_\ab^{(1)} \\f_\ab^{(2)}(i), & i\in I_\ab^{(2)} \\ f_\ab'(i), & i \in I_\ab' \end{cases}  \qquad \textup{and} \qquad
     f_\af(i) = \begin{cases} f_\af^{(1)}(i), & i\in I_\af^{(1)} \\f_\af^{(2)}(i), & i\in I_\af^{(2)} \\ f_\af'(i), & i \in I_\af' \end{cases}.
\end{equation}
Observe the identities

\begin{equation}\label{Prep-For-gs}
    J_\ab \cup I_\ab =   J_\ab^{(1)} \cup I_\ab^{(1)} \cup J_\ab^{(2)} \cup I_\ab^{(2)}. \quad J_\cb\cup f_\ab(I_\ab) = J_\cb^{(1)}\cup f_\ab^{(1)}(I_\ab^{(1)})\cup J_\cb^{(2)}\cup f_\ab^{(2)}(I_\ab^{(2)}),
\end{equation}
where we used \eqref{J-I-primes} and \eqref{New-J-I}.
Note also that $I_\ab \cap (\Jab\cup \Jcb) =I_\af \cap (\Jaf\cup \Jcf) = \emptyset$, $f_\ab\colon I_\ab \to \llbracket 1,n\rrbracket \setminus (\Jab\cup \Jcb\cup I_\ab)$ and $f_\af\colon I_\af\to \llbracket 1,n\rrbracket\setminus (\Jaf\cup \Jcf\cup I_\af)$ are bijections with $f_\ab(i)>i$, for all $i\in I_\ab$, and $f_\af(i)>i$, for all $i\in I_\af$, such that the subsets $\Jab$, $\Jcb$, $\Jaf$, $\Jcf$, $I_\ab$, $I_\af$ of $\llbracket 1, n\rrbracket$,  together with the functions $f_\ab$ and $f_\af$, satisfy Definition~\ref{Notation}.
Furthermore, we abbreviate
\[
\scrL\bigl(\{ k_j\}_{j \in I_\af}, \{ q_j\}_{j \in I_\ab}\bigr) = \scrL^{(1)} \scrL^{(2)},
\]
which does not depend on the extra variables indexed by $I_\af'$ and $I_\ab'$.

With the above notation, we can express \eqref{Eq-HerComp} as a sum of terms of the form
{\allowdisplaybreaks
 \begin{align}\label{Stepone}
\nonumber  &  \int \prod_{i=1}^n F_i(k_i,q_i) \scrL^{(1)} \scrL^{(2)}\Delta^{(1)}_\ab\Delta^{(1)}_\af\Delta^{(2)}_{\ab}\Delta^{(2)}_{\af}\\ 
  \nonumber &\qquad
 \prod_{j \in  \Jcf^{(1)}} \cf(k_j)\prod_{j \in \Jcb^{(1)} } \cb(q_j) \prod_{j \in \Jcf'}\cf(k_j) \prod_{j \in \Jcb'}\cb(q_j)\\
 \nonumber &\qquad
  \prod_{i \in \scrA^{(1)}} R_0\Bigl(z-C^{(1)}_{i}-R^{(1)}_i- \sum_{j \in \Jcf'} \wf(k_j)-\sum_{j \in \Jcb'} \wb(q_j)\Bigr)\\
   \nonumber &\qquad
  R_0\Bigl(z - \sum_{j \in  \Jcf^{(2)}}\wf(k_j)-\sum_{j \in  \Jcb^{(2)}}\wb(q_j) - \sum_{j \in  \Jaf'}\wf(k_j) - \sum_{j \in  \Jab'}\wb(q_j) \Bigr)  \\
  \nonumber & \qquad  
  \prod_{i \in \scrA^{(2)}} R_0\Bigl(z-C^{(2)}_{i}-R^{(2)}_i - \sum_{j \in \Jaf'}\wf(k_j)-\sum_{j \in \Jab'} \wb(q_j) \Bigr) \\
  \nonumber & \qquad \prod_{j \in \Jaf'}\af(k_j)\prod_{j \in \Jab'}\ab(k_j) \prod_{j \in \Jab^{(2)} } \ab(q_j) \prod_{j \in  \Jaf^{(2)}} \af(k_j) \prod_{j=1}^n dq_j dk_j\\
  \nonumber  & \quad =  \int \prod_{i=1}^n F_i(k_i,q_i) \scrL\bigl(\{k_j\}_{j\in I_\af}, \{q_j\}_{j\in I_\ab}\bigr)\\
 \nonumber & \qquad \prod_{i \in I_\af} \kdelta\bigl(k_i-k_{f_{\af}(i)}\bigr)\prod_{i \in I_\ab }\kdelta\bigl(q_i-q_{f_{\ab}(i)}\bigr) \prod_{j \in  \Jcf} \cf(k_j)\prod_{j \in \Jcb } \cb(q_j)  \\
 \nonumber &\qquad
  \prod_{i \in \scrA^{(1)}} R_0\Bigl(z-C^{(1)}_{i}-R^{(1)}_i- \sum_{j \in \Jcf'} \wf(k_j)-\sum_{j \in \Jcb'} \wb(q_j)\Bigr) \\
  \nonumber &\qquad
  R_0\Bigl(z - \sum_{j \in  \Jcf^{(2)}}\wf(k_j)-\sum_{j \in  \Jcb^{(2)}}\wb(q_j)
   - \sum_{j \in  \Jaf'}\wf(k_j) - \sum_{j \in  \Jab'}\wb(q_j) \Bigr)  \\
    \nonumber & \qquad
   \prod_{i \in \scrA^{(2)}} R_0\Bigl(z-C^{(2)}_{i}-R^{(2)}_i - \sum_{j \in \Jaf'}\wf(k_j)-\sum_{j \in \Jab'} \wb(q_j) \Bigr) \\
  & \qquad \prod_{j \in \Jaf}\af(k_j)\prod_{j \in \Jab}\ab(k_j) \prod_{j=1}^n dq_j dk_j.
 \end{align}
 }
 The right-hand side of \eqref{Stepone} brings us closer to being able to recognize the form a Wick monomial \eqref{GeneralFormRegularOperator}.  

 \emph{Step III:} The spectral shifts.
Let us define:
\begin{equation}
 \scrA  =  \scrA^{(1)}\cup\{m\}\cup \scrA^{(2)} \subseteq \llbracket 1, n -1\rrbracket
\end{equation}
and abbreviate for $i\in\scrA$
\begin{equation}\label{Ci-inproof}
   C_i = \sum_{\underset{j \leq i}{j \in \Jab }} \wb(q_j) +  \sum_{\underset{j \leq i}{j \in  \Jaf}} \wf(k_j)  + \sum_{\underset{j >i}{j \in \Jcb}} \wb(q_j) + \sum_{\underset{j>i}{j \in  \Jcf}} \wf(k_j).
  \end{equation}
  For $i\in\scrA$, the spectral shift coming from pull-through, in the corresponding resolvents on the right hand side of \eqref{Stepone}, should be written on the form $C_i+S_i$, where we proceed to compute the remainders $S_i$. 
  
  Recall the definition of $R_i^{(1)}$ and $R_i^{(2)}$ from \eqref{Rcomponents}.
  We compute for any $i\in\scrA^{(1)}$, using \eqref{J-I-primes}, \eqref{New-J-I} and \eqref{New-f}:
 \begin{equation}\label{AfterPullThrough1}
  \begin{aligned}
  S_i &:= R^{(1)}_i + C^{(1)}_i+\sum_{j \in \Jcb'} \wb(q_j) + \sum_{j \in \Jcf'} \wf(k_j) - C_i\\
    & =  R^{(1)}_i + \sum_{\underset{j\leq i}{j\in \Jab^{(1)}\setminus \Jab'}}\wb(q_j)
    + \sum_{\underset{j\leq i}{j\in \Jaf^{(1)}\setminus \Jaf'}}\wf(k_j)\\
    &=  \sum_{\underset{ j \leq i <  f^{(1)}_{\ab}(j)}{j\in I_\ab^{(1)}}}\wb(q_j) + \sum_{\underset{ j\leq i < f^{(1)}_{\af}(j)}{j\in I_\af^{(1)}}}\wf(k_j)  + \sum_{\underset{j\leq i}{j\in I_\ab'}}\wb(q_j)
    + \sum_{\underset{j\leq i}{j\in I_\af'}}\wf(k_j)\\
    & =  \sum_{\underset{ j \leq i <  f_{\ab}(j)}{j\in I_\ab}}\wb(q_j) + \sum_{\underset{ j\leq i <  f_{\af}(j)}{j\in I_\af}}\wf(k_j).
     \end{aligned}
\end{equation}
For $i\in\llbracket m,n \rrbracket$, we compute first 
\begin{align*}
   \sum_{\underset{j > i}{j\in \Jcb^{(2)}\setminus \Jcb'}}\wb(q_j) + \sum_{\underset{j > i}{j\in \Jcf^{(2)}\setminus \Jcf'}}\wf(k_j)
   &= \sum_{\underset{j > i}{j\in f'_\ab(I_\ab')}}\wb(q_j)
    + \sum_{\underset{j > i}{j\in f_\af'(I_\af')}}\wf(k_j)\\
    &= \sum_{\underset{f_\ab'(j) > i}{j\in I_\ab'}}\wb(q_{f'_\ab(j)})
    + \sum_{\underset{f_\af'(j) > i}{j\in I_\af'}}\wf(k_{f_\af'(j)}).
\end{align*}
Using this we can now compute for any $i \in  \scrA^{(2)}$:
   \begin{align}\label{AfterPullThrough2}
 \nonumber   S_i &:= R^{(2)}_i + C^{(2)}_i +\sum_{j \in \Jab'} \wb(q_j)+ \sum_{j \in \Jaf'} \wf(k_j) - C_i\\
\nonumber    & =  R^{(2)}_i + \sum_{\underset{j > i}{j\in \Jcb^{(2)}\setminus \Jcb'}}\wb(q_j)
    + \sum_{\underset{j > i}{j\in \Jcf^{(2)}\setminus \Jcf'}}\wf(k_j)\\
\nonumber &    = \sum_{\underset{ j \leq i <  f^{(2)}_{\ab}(j)}{j\in I_\ab^{(2)}}}\wb(q_j) + \sum_{\underset{ j\leq i <  f^{(2)}_{\af}(j)}{j\in I_\af^{(2)}}}\wf(k_j)\\
& \quad +  \sum_{\underset{f_\ab'(j) > i}{j\in I_\ab'}}\wb(q_{f'_\ab(j)})
    + \sum_{\underset{f_\af'(j) > i}{j\in I_\af'}}\wf(k_{f_\af'(j)}).
      \end{align}
and
   \begin{align}\label{AfterPullThroughm}    
 \nonumber   S_m & := \sum_{j \in  \Jcf^{(2)}}\wf(k_j)+\sum_{j \in  \Jcb^{(2)}}\wb(q_j)  + \sum_{j \in  \Jaf'}\wf(k_j) + \sum_{j \in  \Jab'}\wb(q_j) - C_m\\
 \nonumber  & = \sum_{j\in \Jcb^{(2)}\setminus \Jcb'}\wb(q_j)+ \sum_{j\in \Jcf^{(2)}\setminus \Jcf'}\wf(k_j)\\
 \nonumber   & = \sum_{j\in I_\ab'}\wb(q_{f_\ab'(j)})+ \sum_{j\in I_\af'}\wf(k_{f_\af'(j)})\\
    & = \sum_{\underset{j\leq m < f_\ab(j)}{j\in I_\ab}}\wb(q_{f_\ab'(j)})+ \sum_{\underset{j\leq m< f_\af(j)}{j\in I_\af}}\wf(k_{f_\af'(j)}).
  \end{align}
 Note that the $S_i$'s are not all of the form $R_i$ required by a Wick monomial, cf.~\eqref{Rcomponents} in Definition~\ref{WickMonomial}. But recalling that we have delta functions in play, we observe that
 if we set $q_j = q_{f_\ab(j)}$, for $j\in I_\ab$, and $k_j = k_{f_\af(j)}$, for $j\in I_\af$, 
then $S_i=R_i$ for all $i\in\scrA$.

The right-hand side of \eqref{Stepone} can then be written as: 
 \begin{align} \label{SteponeCorrectForm}
 & \prod_{i \in I_\af} \kdelta\bigl(k_i-k_{f_{\af}(i)}\bigr)\prod_{i \in I_\ab} \kdelta\bigl(q_i-q_{f_{\ab}(i)}\bigr)\\
 \nonumber & \qquad  \prod_{j \in  \Jcf} \cf(k_j)\prod_{j \in \Jcb } \cb(q_j)
  \prod_{i \in \scrA} R_0\bigl(z-C_{i}-R_i\bigr)  \prod_{j \in \Jaf}\af(k_j)\prod_{j \in \Jab}\ab(k_j),
 \end{align}
and therefore,  $T_1(z;F_1, \dots, F_m) R_0(z) T_2(z;F_{m+1}, \dots, F_n)$ is a sum of Wick monomials of the form \eqref{GeneralFormRegularOperator}. 

\emph{Step IV:} Signature strings. 
Let $\us_1\in\scrS^{(m)}_0$ and $\us_2\in\scrS^{(n-m)}_0$ be the signature strings pertaining to $T_1$ and $T_2$, respectively. Then $\us = \us_1\circ\us_2\in\scrS^{(n)}_0$ is the signature string
affiliated with all the summands of the form \eqref{SteponeCorrectForm}. This proves \ref{item-Her1-signs}.

Calling a summand of the form \eqref{SteponeCorrectForm} $T'$, we may compute
$n_\ab(T') = n_\ab(T_1) + n_\ab(T_2)$ and
$n_\af(T') = n_\af(T_1) + n_\af(T_2)$, which is independent of the choice of summand, since the signature string $\us$ is the same for all the $T'$'s. The indices $n_\ab$ and $n_\af$ were defined, in terms of the underlying $\us$, in Definition~\ref{def-Tnumbers}.

It now follows from the constraints in the formulation of 
\ref{item-Her1-RH} and \ref{item-Her1-LH}, together with Remark~\ref{rem-naToHanded}, that what remains of the proof is to establish that Wick monomials associated with \eqref{SteponeCorrectForm} are regular, cf.~Definition~\ref{RegularOperator}. From now on $T_1 = \rightT_1$ is right-handed and $T_2 = \leftT_2$ is left-handed.

\emph{Step V:} Setting up for regularity.
Let $(P_{\ab},P_{\af}, P_{\cb},P_{\cf})$ be a cover of $J_{\ab}\cup J_{\cb} \cup  J_{\af}\cup J_{\cf}$ (cf.~Definition~\ref{DefCover}). In order to use the hypothesis that $\rightT_1$ and $\leftT_2$ are regular, we must estimate terms in \eqref{regularityproperty} for $i\in \llbracket 1,m\rrbracket$ and for $i\in \llbracket m+1,n\rrbracket$, separately. For this we first need to go from the cover $(P_{\ab},P_{\af}, P_{\cb},P_{\cf})$ to covers
$(P_{\ab}^{(j)},P_{\af}^{(j)}, P_{\cb}^{(j)},P_{\cf}^{(j)})$ of $\scrJ^{(j)}= J_{\ab}^{(j)}\cup J_{\cb}^{(j)} \cup  J_{\af}^{(j)}\cup J_{\cf}^{(j)}$, for $j=1,2$. Let
\begin{equation}\label{Simple-Ps}
\begin{aligned}
  P_\cb^{(1)} & := P_\cb\cap \llbracket 1,m\rrbracket,  &  P_\cf^{(1)} &:= P_\cf\cap \llbracket 1,m\rrbracket, \\
  P_\ab^{(2)}&:= P_\ab\cap \llbracket m+1,n\rrbracket, &  P_\af^{(2)} &:= P_\af\cap \llbracket m+1,n\rrbracket.
 \end{aligned}
\end{equation}
For the remaining four sets, we need to take contractions into account. Choose a partition $P_\ab'\subset J_\ab^{(1)}\setminus J_\ab'$ and $P_\af'\subset J_\af^{(1)}\setminus J_\af'$, with $P_\ab'\cap P_\af' = \emptyset$ and
$P_\ab'\cup P_\af' = (J_\ab^{(1)}\setminus J_\ab')\cup (J_\af^{(1)}\setminus J_\af')$, and a partition $P_\cb'\subset J_\cb^{(2)}\setminus J_\cb'$ and $P_\cf'\subset J_\cf^{(2)}\setminus J_\cf'$, with $P_\cb'\cap P_\cf' = \emptyset$ and
$P_\cb'\cup P_\cf' = (J_\cb^{(2)}\setminus J_\cb')\cup (J_\cf^{(2)}\setminus J_\cf')$. Then we may define the last four sets to be
\begin{equation}\label{NotSimple-Ps}
\begin{aligned}
P_\ab^{(1)} &:= (P_\ab\cap \llbracket 1,m\rrbracket) \cup P_\ab',  & P_\af^{(1)}&:= (P_\af\cap \llbracket 1,m\rrbracket) \cup P_\af',\\
P_\cb^{(2)} &:= (P_\cb\cap \llbracket m+1,n\rrbracket)\cup P_\cb', & P_\cf^{(2)}&:= (P_\cf\cap \llbracket m+1,n\rrbracket)\cup P_\cf'.
\end{aligned}
\end{equation}

Next step is to introduce $\sigma^{(1)}$ and $\sigma^{(2)}$, starting from $\sigma$. Recalling \eqref{New-J-I}, we observe that $(J_{\cf}\cup J_{\af})\cap \llbracket 1,m\rrbracket= J_{\cf}^{(1)}\cup J_\af'$ and $(J_\cf\cup J_{\af})\cap \llbracket m+1,n\rrbracket= J_\cf'\cup J_{\af}^{(2)}$. With this in mind we define $\sigma^{(i)}\colon J_\cf^{(j)}\cup J_\af^{(j)}\to\NN_\infty$, for $j=1,2$, by:
\begin{equation}
\begin{aligned}
\sigma^{(1)}(i) &= \begin{cases} {\sigma}(i) & \textup{if }  i\in J_\cf^{(1)}\cup J_{\af}' \textup{ and } \sigma(i) \in \{0\} \cup \llbracket 1, m \rrbracket\\
\infty & \textup{if }  i\in J_\cf^{(1)}\cup J_{\af}' \textup{ and } \sigma(i) > m\\
i & \textup{if }  i\in J_\af^{(1)} \setminus J_{\af}'\end{cases}  \\
\sigma^{(2)}(i) & = \begin{cases} {\sigma}(i) & \textup{if }  i\in J_{\cf}'\cup J_\af^{(2)} \textup{ and } \sigma(i) \in \llbracket m+1, n \rrbracket\cup\{\infty\} \\
 0 & \textup{if }  i\in J_{\cf}'\cup J_\af^{(2)} \textup{ and } \sigma(i) < m+1 \\
i & \textup{if }  i\in J_\cf^{(2)} \setminus J_{\cf}'\end{cases}.
\end{aligned}
\end{equation}
Note that, $\sigma^{(1)}$, respectively $\sigma^{(2)}$, has its range included in $\{0\} \cup \llbracket 1, m \rrbracket \cup \{\infty \}$, respectively in $\{0\} \cup \llbracket m+1, n \rrbracket \cup \{\infty \}$. Note that the relabelling of the $\leftT_2$-indices from $\llbracket 1,n-m\rrbracket$ to $\llbracket m+1,n\rrbracket$ is reflected in the shifted range of $\sigma^{(2)}$.  Both maps $\sigma^{(j)}$, $j=1,2$, are moreover admissible with respect to the choice of partition $P_{\cf}^{(j)}$ and $P_{\af}^{(j)}$ made above, in the sense of Definition~\ref{Admissiblemaps} (and taking the relabelling into account for $\sigma^{(2)}$). Finally, we observe that
\begin{equation}\label{Her-psigma}
\psigma\cap\llbracket 1,m\rrbracket \subset\psigma^{(1)} \quad\textup{and}\quad
\psigma\cap\llbracket m+1,n\rrbracket \subset\psigma^{(2)}. 
\end{equation}

We will extract the improved momentum decay from the extra resolvent $R_0(z - C_m(\sigma)-R_m)$. We now fix two momentum indices, where we know that annihilation/creation operators have been pulled through this central resolvent.
Because $\rightT_1$ is a right-handed Wick monomial, and $\leftT_2$ is left-handed, we know by Definition~\ref{def-handedWick}~\ref{item-RHWick} and~\ref{item-LHWick}, that it is possible to choose
\begin{equation}\label{Her1-Choiceofj1j2}
j_1\in J_\ab^{(1)}\cup \bigl(J_\af^{(1)}\setminus J_\cb^{(1)}\bigr)\neq \emptyset \quad \textup{and} \quad j_2\in J_\cb^{(2)}\cup \bigl( J_\cf^{(2)}\setminus J_\ab^{(2)} \bigr)\neq \emptyset. 
\end{equation}

Consider an arbitrary choice of real numbers $\{\alpha_i\}_{i=1}^n$ fulfilling:
\begin{equation}\label{sum0}
\forall i \in \llbracket 1, n \rrbracket: \ 0 \leq \alpha_i \leq 1 \qquad \textup{and}\qquad
 \sum^n_{i =1} \alpha_i = n-1,
\end{equation}
We define
\begin{equation}\label{Her1-alpha1}
\forall i\in\llbracket 1,m\rrbracket:\qquad     \alpha_i^{(1)} = \begin{cases} \alpha_i, & i \neq j_1 \\
 \displaystyle m-1 - \sum_{j\in\llbracket 1,m\rrbracket \setminus \{j_1\}}\alpha_i,& i = j_1\end{cases}
\end{equation}
and 
\begin{equation}\label{Her1-alpha2}
\forall i\in\llbracket m+1,n\rrbracket:\qquad     \alpha_i^{(2)} = \begin{cases} \alpha_i, & i \neq j_2 \\
\displaystyle n -m -1- \sum_{j\in\llbracket m+1,n\rrbracket \setminus \{j_2\}}\alpha_i,& i = j_2.\end{cases}
\end{equation}
The constraint on the $\alpha_i$'s in \eqref{sum0} ensures that $0\leq \alpha_{i}^{(1)}\leq 1$, for $i\in\llbracket 1,m\rrbracket$, and $\sum_{i=1}^{m} \alpha_i^{(1)} = m-1$. Similarly, $0\leq \alpha_{i}^{(2)}\leq 1$, for $i\in\llbracket m+1,n\rrbracket$, and $\sum_{i=m+1}^{n} \alpha_i^{(2)} = n-m-1$.

Since $\rightT_1$  and $\leftT_2$ are regular, there exists two collections of admissible exponents $\{\gamma_{i;j}^{(1)}\}_{i\in \scrA^{(1)}, j \in  \scrJ^{(1)}}$ and $\{\gamma_{i;j}^{(2)}\}_{i\in \scrA^{(2)}, j \in  \scrJ^{(2)}}$, cf.~Definition~\ref{def-admexp}.

Recalling that $\scrJ\subset \scrJ^{(1)}\cup\scrJ^{(2)}$, we define $\gamma_{i;j}$, for $i\in\scrA$ and $j\in \scrJ$, in the following way: 
\begin{equation}\label{Her1-gamma}
\gamma_{i;j} = \begin{cases}
\gamma^{(1)}_{i;j}, & \textup{if } j\in \scrJ\cap \llbracket 1, m \rrbracket, i \in \scrA^{(1)} \\
\gamma^{(2)}_{i;j}, & \textup{if } j\in \scrJ\cap \llbracket m+1, n \rrbracket, i \in \scrA^{(2)}\\
\kdelta_{j,j_1}\bigl(\alpha_{j_1}-\alpha^{(1)}_{j_1}\bigr) + \kdelta_{j,j_2}\bigl(\alpha_{j_2}-\alpha_{j_2}^{(2)}\bigr)& \textup{if } i=m \\
0 & \text{otherwise}
\end{cases}
\end{equation}
We claim that $\{\gamma_{i;j}\}_{i\in\scrA,j\in\scrJ}$ are admissible exponents for the composed objects. To see this, observe that it follows from \eqref{Prep-For-gs}, \eqref{Simple-Ps} and \eqref{NotSimple-Ps} that
\[
\begin{aligned}
\llbracket 1,m\rrbracket \cap \Bigl(P_\cb\cup \bigl( P_\cf\cap (J_\cb\cup f_\ab(I_\ab)) \bigr)\Bigr) & = P_\cb^{(1)} \cup \bigl( P_\cf^{(1)} \cap (J_\cb^{(1)}\cup f_\ab^{(1)}(I_\ab^{(1)}))\bigr),\\
\llbracket 1,m\rrbracket \cap \Bigl(P_\ab\cup \bigl( P_\af\cap (J_\ab\cup I_\ab) \bigr)\Bigr) & \subset P_\ab^{(1)} \cup \bigl( P_\af^{(1)} \cap (J_\ab^{(1)}\cup I_\ab^{(1)})\bigr),\\
\rrbracket m,n\rrbracket \cap \Bigl(P_\cb\cup \bigl( P_\cf\cap (J_\cb\cup f_\ab(I_\ab)) \bigr)\Bigr) & \subset P_\cb^{(2)} \cup \bigl( P_\cf^{(2)} \cap (J_\cb^{(2)}\cup f_\ab^{(2)}(I_\ab^{(2)}))\bigr),\\
\rrbracket m,n\rrbracket \cap \Bigl(P_\ab\cup \bigl( P_\af\cap (J_\ab\cup I_\ab) \bigr)\Bigr) & = P_\ab^{(2)} \cup \bigl( P_\af^{(2)} \cap (J_\ab^{(2)}\cup I_\ab^{(2)})\bigr),
\end{aligned}
\]
where $\rrbracket m,n\rrbracket := \llbracket m+1,n\rrbracket$.
From this, the first two constraints on the $\gamma_{i;j}$'s in \eqref{ConstraintOnGamma} follow. The last two constraints from \eqref{ConstraintOnGamma}, follows from the observations that
\[
\begin{aligned}
\llbracket 1,m\rrbracket \cap P_\cf \cap I_\ab &=P_\cf^{(1)} \cap (I_\ab^{(1)}\cup I_\ab')= P_\cf^{(1)}\cap I_\ab^{(1)},\\
\llbracket 1,m\rrbracket \cap P_\af \cap f_\ab(I_\ab) &
=\llbracket 1,m\rrbracket \cap P_\af \cap f_\ab^{(1)}(I_\ab^{(1)})\subset P_\af^{(1)}\cap f_\ab^{(1)}(I_\ab^{(1)}),
\end{aligned}
\]
where in the first identity we used that $I'_\ab \subset I'_\ab\cup I'_\af = P'_\ab\cup P'_\af \subset P^{(1)}_\ab\cup P^{(1)}_\af$, which has empty intersection with $P_\cf^{(1)}$. Finally,
\[
\begin{aligned}
\llbracket m+1,n\rrbracket \cap P_\cf \cap I_\ab &= \llbracket m+1,n\rrbracket \cap P_\cf \cap I_\ab^{(2)}\subset P_\cf^{(2)}\cap I_\ab^{(2)},\\
\llbracket m+1,n\rrbracket \cap P_\af \cap f_\ab(I_\ab) &
= P_\af^{(2)} \cap (f_\ab^{(2)}(I_\ab^{(2)})\cup f'_\ab(I'_\ab))= P_\af^{(2)}\cap f_\ab^{(2)}(I_\ab^{(2)}),
\end{aligned}
\]
where, in the last identity, we used that $f'_\ab(I'_\ab) \subset f'_\ab(I'_\ab)\cup f'_\af(I'_\af) = P'_\cb\cup P'_\cf \subset P^{(2)}_\cb\cup P^{(2)}_\cf$, which has empty intersection with $P_\af^{(2)}$. 

 Note that the choice \eqref{Her1-Choiceofj1j2} of $j_1$ and $j_2$ ensures that $\gamma_{m;j_1}$ and $\gamma_{m;j_2}$ are not constrained to be zero. This can be inferred from the above set considerations and concludes the verification of \eqref{ConstraintOnGamma}. 

As for \eqref{SumOfGammas}, we abbreviate
$\overline{\gamma}_i = \sum_{j\in \scrJ} \gamma_{i;j}$, for $i\in\scrA$,
and observe that 
\[
\forall i\in\scrA^{(1)}: \quad \overline{\gamma}_i = \overline{\gamma}_i^{(1)}\leq 1, \qquad \forall i\in\scrA^{(2)}: \quad \overline{\gamma}_i = \overline{\gamma}_i^{(2)}\leq 1,
\]
since the $\gamma^{(1)}_{i;j}$'s and the $\gamma^{(2)}_{i;j}$'s both satisfy \eqref{SumOfGammas}.
Furthermore,  we observe that 
\begin{equation}\label{gamma-bar-m}
\overline{\gamma}_m \leq \bigl(\alpha_{j_1}-\alpha_{j_1}^{(1)}\bigr) + \bigl(\alpha_{j_2}-\alpha_{j_2}^{(2)}\bigr) =  1. 
\end{equation}
 Finally, recalling \eqref{Her1-alpha1}, \eqref{Her1-alpha2} and~\eqref{Her1-gamma}, it is clear by construction that $\sum_{i\in\scrA} \gamma_{i;j} = \alpha_j$ for all $j\in\scrJ$. Hence, we have established that $\{\gamma_{i;j}\}_{i\in\scrA,j\in\scrJ}$ satisfies \eqref{SumOfGammas}.

\emph{Step VI:} Concluding the regularity estimate \eqref{regularityproperty}. Recall the abbreviations \eqref{abbrev-LDelta}.
Finally, recalling \eqref{New-J-I}, \eqref{New-f}, \eqref{Her1-alpha1} and \eqref{Her1-alpha2}, we obtain 
\begin{align}
\nonumber &  \Bigl|\scrL\bigl( \{ k_j\}_{j \in I_\af}, \{ q_j\}_{j \in I_\ab}\bigr)\Bigr|\prod_{i\in I_\ab}\kdelta\bigl(q_{i} - q_{f_\ab(i)}\bigr) \prod_{j\in I_\af}\kdelta\bigl(k_{j} - k_{f_\af(j)}\bigr) \\ 
\nonumber &  \qquad  \prod_{i \in P_{\ab}\cup P_{\cb}} \wb(q_i)^{-\alpha_i} \prod_{ i \in P_{\af}\cup P_{\cf}} \wf(k_j)^{-\alpha_i}  \biggl\|  \prod_{i\in \scrA} R_0\bigl(z_i-C_i(\sigma) -R_i\bigr)^{1-\overline{\gamma}_i}  \biggr\|\\
\nonumber & \quad \leq  \Biggl[\  \bigl|\scrL^{(1)}\bigr| \Delta^{(1)}_\ab\Delta^{(1)}_\af
  \prod_{i \in P^{(1)}_{\ab}\cup P^{(1)}_{\cb}} \wb(q_i)^{-\alpha^{(1)}_i} \prod_{ i \in P^{(1)}_{\af}\cup P^{(1)}_{\cf}} \wf(k_j)^{-\alpha^{(1)}_i} \\
\nonumber & \qquad \qquad \biggl\|  \prod_{i\in \scrA^{(1)}}  R_0\bigl(z-C_i(\sigma) -R_i\bigr)^{1-\overline{\gamma}_i^{(1)}}  \biggr\| \ \Biggr]
\\
\nonumber &  \qquad \Biggl[\ 
  \bigl|\scrL^{(2)}\bigr| \Delta^{(2)}_\ab \Delta^{(2)}_\af 
  \prod_{i \in P^{(2)}_{\ab}\cup P^{(2)}_{\cb}} \wb(q_i)^{-\alpha^{(2)}_i} \prod_{ i \in P^{(2)}_{\af}\cup P^{(2)}_{\cf}} \wf(k_j)^{-\alpha^{(2)}_i}  \\
\nonumber   & \qquad \qquad   \biggl\|\prod_{i\in \scrA^{(2)}} R_0\bigl(z-C_i(\sigma) -R_i\bigr)^{1-\overline{\gamma}_i^{(2)}} \biggr\|\ \Biggr] \\
\nonumber   & \qquad \Biggl[ \  \prod_{i\in I'_\ab}\kdelta\bigl(q_{i} - q_{f'_\ab(i)}\bigr) \prod_{j\in I'_\af}\kdelta\bigl(k_{j} - k_{f'_\af(j)}\bigr) \prod_{\ell=1}^2 \wb(q_{j_\ell})^{-\kdelta_{j_\ell\in P_\ab\cup P_\cb} (\alpha_{j_\ell}-\alpha_{j_\ell}^{(\ell)})}   \\
  & \quad\qquad 
  \wf(k_{j_\ell})^{-\kdelta_{j_\ell\in P_\af\cup P_\cf} (\alpha_{j_\ell}-\alpha_{j_\ell}^{(\ell)})}
  \biggl\| R_0\bigl(z-C_m(\sigma)-R_m\bigr)^{1-\overline{\gamma}_m}\biggr\| \  \Biggr].
\end{align}

Next, up to an identification of $q_j$ with $q_{f_\ab'(j)}$, for $j\in I_\ab'$, and of $k_j$ with $k_{f_\af'(j)}$, for $j\in I_\af'$, we have: 
\[
\begin{aligned}
&\forall i \in \scrA^{(1)}:\quad & \Bigl\|R_0\bigl(z-C_i(\sigma) -R_i\bigr)\Bigr\| \leq\Bigl\|R_0\bigl(z-C^{(1)}_{i}(\sigma^{(1)})-R^{(1)}_i\bigr)\Bigr\|,\\
&\forall i \in \scrA^{(2)}:\quad & \Bigl\|R_0\bigl(z-C_i(\sigma) -R_i\bigr)\Bigr\| \leq\Bigl\|R_0\bigl(z-C^{(2)}_{i}(\sigma^{(2)})-R^{(2)}_i\bigr)\Bigr\|.
\end{aligned}
\]
As a consequence, there exist $\{\beta^{(1)}_{i}\}_{ i=1}^m$ with $0\leq \beta_i^{(1)}\leq \alpha_i^{(1)}$ and
$\beta^{(1)}_i =0$ for $i\in\psigma^{(1)}$,  together with  $\{\beta^{(2)}_{i}\}_{ i=m+1}^n$   with $0\leq \beta_i^{(2)}\leq \alpha_i^{(2)}$ and
$\beta^{(2)}_i =0$ for $i\in\psigma^{(2)}$, such that
\begin{align}\label{Her1AlmostThere}
\nonumber &  \prod_{i\in I_\ab}\kdelta\bigl(q_{i} - q_{f_\ab(i)}\bigr) \prod_{j\in I_\af}\kdelta\bigl(k_{j} - k_{f_\af(j)}\bigr)\Bigl|\scrL\bigl( \{ k_j\}_{j \in I_\af}, \{ q_j\}_{j \in I_\ab}\bigr)\Bigr| \\ 
\nonumber& \qquad  \prod_{i \in P_{\ab}\cup P_{\cb}} \wb(q_i)^{-\delta_i} \prod_{ i \in P_{\af}\cup P_{\cf}} \wf(k_j)^{-\delta_i} \biggl\|  \prod_{i\in \scrA} R_0\bigl(z-C_i(\sigma) -R_i\bigr)^{1-\overline{\gamma}_i}  \biggr\|
\\
\nonumber& \quad \leq \prod_{i\in I_\ab}\kdelta\bigl(q_{i} - q_{f_\ab(i)}\bigr) \prod_{j\in I_\af}\kdelta\bigl(k_{j} - k_{f_\af(j)}\bigr) \\
\nonumber& \qquad \frac{ c_{\rightT_1}}{\prod^m_{i=1}[\wb(q_i)]^{\alpha^{(1)}_i - \beta^{(1)}_{i}}[\wf(k_i)]^{\beta^{(1)}_{i}}}
\cdot  \frac{c_{\leftT_2}}{\prod^n_{i=m+1}[\wb(q_i)]^{\alpha^{(2)}_i - \beta^{(2)}_{i}}[\wf(k_i)]^{\beta^{(2)}_{i}}} \\
\nonumber&\qquad \prod_{\ell=1}^2 \Bigl[\wb(q_{j_\ell})^{-\kdelta_{j_\ell\in P_\ab\cup P_\cb} (\alpha_{j_\ell}-\alpha_{j_\ell}^{(\ell)})} \wf(k_{j_\ell})^{-\kdelta_{j_\ell\in P_\af\cup P_\cf} (\alpha_{j_\ell}-\alpha_{j_\ell}^{(\ell)})} \Bigr]\\
& \qquad 
\Bigl\| R_0\bigl(z-C_m(\sigma)-R_m\bigr)^{1-\overline{\gamma}_m}\Bigr\|.
\end{align}

Note that by \eqref{gamma-bar-m}, we have 
\[
1-\overline{\gamma}_m = 1 -  \kdelta_{j_1\in\scrJ}\bigl(\alpha_{j_1}-\alpha^{(1)}_{j_1}\bigr) - \kdelta_{j_2 \in\scrJ}\bigl(\alpha_{j_2}-\alpha_{j_2}^{(2)}\bigr).
\]
We claim that the following choice of exponents $\{\beta_{i}\}_{i=1}^n$ will work. For $i\in\llbracket 1,n\rrbracket$, we set
\[
 \beta_{i} = \begin{cases}
    \beta^{(1)}_{i} & \textup{if }  i \in \llbracket 1 , m \rrbracket\setminus \{j_1\}  \\
    \beta^{(1)}_{j_1} + \kdelta_{j_1 \in P_{\af}\cup I'_\af} \bigl(\alpha_{j_1} - \alpha^{(1)}_{j_1}\bigr)& \textup{if }  i=j_1\\
     \beta^{(2)}_{i} & \textup{if }  i \in \llbracket m+1 , n \rrbracket\setminus\{j_2\}\\
     \beta^{(2)}_{j_2} + \kdelta_{j_2 \in P_{\cf}\cup f'_\af(I'_\af)} \bigl(\alpha_{j_2} - \alpha^{(2)}_{j_2}\bigr)& \textup{if }  i=j_2.
\end{cases} 
\]
Note that $0\leq \beta_i\leq \alpha_i$ for all $i\in\llbracket 1,n\rrbracket$ and it follows from \eqref{Her-psigma} that $\beta_i = 0$ if $i\in\psigma$.

In order to go from \eqref{Her1AlmostThere} to the desired estimate \eqref{regularityproperty} with bounding constant $c_{\rightT_1}\cdot c_{\leftT_2}$ and the claimed exponents just introduced, we need to establish the following: 
\begin{equation}\label{ReqOnCmRm}
\begin{aligned}
&\textup{if } j_1\in P_\ab\cup P_\cf\cup I'_\ab: & & C_m(\sigma) + R_m\geq \wb(q_{j_1})\\
&\textup{if } j_1\in P_\af\cup I'_\af: & & C_m(\sigma) + R_m\geq \wf(k_{j_1})\\
&\textup{if } j_2 \in P_\cb\cup P_\af: & & C_m(\sigma) + R_m\geq\wb(q_{j_2})\\
&\textup{if } j_2 \in  f'_\ab(I'_\ab): & & C_m(\sigma) + R_m\geq\wb(q_{{f'}_\ab^{-1}(j_2)})\\
&\textup{if } j_2 \in P_\cf: & & C_m(\sigma) + R_m\geq \wf(k_{j_2})\\
&\textup{if } j_2 \in f'_\af(I'_\af): & & C_m(\sigma) + R_m\geq\wf(k_{{f'}_\af^{-1}(j_2)})\\
\end{aligned}
\end{equation}
This can be seen on a case-by-case basis by inspection from the formula
\[
C_m(\sigma)+R_m = \sum_{j \in  \Jcb'}\wb(q_j) +\sum_{\underset{j \in  \Jcf'}{\sigma(j)>m}}\wf(k_j)+ \sum_{j \in  \Jab^{(1)}}\wb(q_j) + \sum_{\underset{j \in  \Jaf^{(1)}}{\sigma(j)\leq m}}\wf(k_j).
\]
\end{proof}

\begin{lem}[Products of regular Wick monomials]
\label{HBFCTANFCO}
Let $m,n\in\NN$ with $m< n$. There exists a constant $M = M(n)$, such that the following holds. Let $T_1$ be a regular Wick monomial of length $m$ that is not fully contracted, with bounding constant $c_{T_1}$ and signature string $\us_1$. Let $T_2$ be a fully contracted Wick monomial of length $n-m$ with bounding constant $c_{T_2}$ and signature string $\us_2$. Finally, let $E_2$ be defined as:
\begin{equation*}
E_2 \colon (F_{m+1},\dotsc, F_n) \to  \bigl\langle \Omega \,\big\vert\, T_2(0;F_{m+1},\dotsc, F_n) \Omega\bigr\rangle.
\end{equation*}
 The following properties hold:
\begin{enumerate}[label = \textup{(\arabic*)}]
\item\label{itemHeredityContracted1} Suppose $T_1=\rightT_1$ is  right-handed. Then the composition $\rightT_1 R_0(z) (T_2 - E_2)$ is a sum of at most $M$ right-handed Wick monomials of length $n$, all with bounding constant $c_{T_1}\cdot c_{T_2}$ and signature string $\us_1\circ\us_2$. 
\item\label{itemHeredityContracted2}  Let $T_1 = \leftT_1$ be left-handed. Then the composition  $(T_2 - E_2 ) R_0(z)\leftT_1$ is a sum of at most $M$ left-handed Wick monomials of length $n$, all with bounding constant $c_{T_1}\cdot c_{T_2}$ and signature string $\us_2\circ\us_1$. 
\end{enumerate}
\end{lem}

\begin{proof} During the proof, we will employ again the abbreviations \eqref{abbrev-LDelta} from the proof of Lemma~\ref{HNFCT}. In addition, we will again employ the convention from the proof of Lemma~\ref{HNFCT} that the sets and indices that label $T_2$-related objects are shifted from $\llbracket 1,n-m\rrbracket$ to $\llbracket m+1,n\rrbracket$.

We will prove only item \ref{itemHeredityContracted1}, item \ref{itemHeredityContracted2} being a consequence of the first case. Indeed, if $\leftT$ is a left-handed Wick monomial then $\bigl((T_2 - E_2 ) R_0(z)\leftT\bigr)^*$ is equal to $\leftT^* R_0(\bar{z})(T^*_2 - E^*_2 )$ which is of the form stated in item \ref{itemHeredityContracted1}. 

Before we begin, let us recall that since $T_2$ is a fully contracted Wick monomial (cf.~Definitions~\ref{Notation} and~\ref{def-handedWick}~\ref{item-FCWick}), we have 
\begin{equation}\label{T2-FC}
\begin{aligned}
& J_\ab^{(2)}=J_\cb^{(2)} = J_\af^{(2)} = J_\cf^{(2)} = \emptyset, \\
& I_\ab^{(2)}\cup f_\ab^{(2)}(I_\ab^{(2)}) = I_\af^{(2)}\cup f_\af^{(2)}(I_\af^{(2)}) = \llbracket m+1,n\rrbracket.
\end{aligned}
\end{equation}
Let us consider $\rightT_1(z;F_1, \dotsc, F_m)$ of the form: 
\begin{align*}
\nonumber &   \int \prod^m_{i = 1} F_i(k_i, q_i) \scrL^{(1)}\Delta^{(1)}_\ab \Delta^{(1)}_\af  \prod_{j \in  \Jcf^{(1)}} \cf(k_j)\prod_{j \in \Jcb^{(1)} } \cb(q_j)\\
 & \qquad   \prod_{i \in \scrA^{(1)}} R_0\bigl(z-C^{(1)}_{i}-R^{(1)}_i\bigr)  \prod_{j \in \Jab^{(1)} } \ab(q_j) \prod_{j \in  \Jaf^{(1)}} \af(k_j)\prod_{j=1}^m dq_j dk_j,
\end{align*}
together with $T_2(z;F_{m+1}, \dotsc, F_n)$ of the form 
\begin{equation*}
   \int \prod^n_{i = m+1} F_i(k_i, q_i) \scrL^{(2)}\Delta^{(2)}_\ab \Delta^{(2)}_\af  \prod_{i \in \scrA^{(2)}} R_0\bigl(z-R^{(2)}_i\bigr)\prod_{j=m+1}^n dq_j dk_j.
\end{equation*}
Recall from Remark~\ref{rem-Wick} that $R^{(2)}_k\neq 0$ for all $k\in\scrA^{(2)}$.
 Consequently, $E_2$ is well-defined
\begin{align*}
    & E_2(F_{m+1}, \dotsc, F_n)  = \bigl\langle \Omega \big| T_2(0;F_{m+1}, \dotsc, F_n) \Omega \bigr\rangle\\
    & \quad = \int \prod^n_{i = m+1} F_i(k_i, q_i) \scrL^{(2)}\Delta^{(2)}_\ab \Delta^{(2)}_\af   \prod_{i \in \scrA^{(2)}} \frac{1}{R^{(2)}_i}\prod_{j=m+1}^n dq_j dk_j.
\end{align*}
Therefore 
\begin{align*}
    &T_2(z;F_{m+1}, \dotsc, F_n) - E_2(F_{m+1}, \dotsc, F_n)\\
     & \qquad  = \int \prod^n_{i = m+1} F_i(k_i, q_i) \scrL^{(2)} \Delta^{(2)}_\ab \Delta^{(2)}_\af\\
     & \qquad \quad  \Biggl[\prod_{i \in \scrA^{(2)}} R_0\bigl(z-R^{(2)}_i\bigr) -  \prod_{i \in \scrA^{(2)}} \frac{1}{R^{(2)}_i}  \Biggr] \prod_{j=m+1}^n dq_j dk_j.
\end{align*}
Moreover,
\begin{align*}
    &\prod_{i \in \scrA^{(2)}} R_0\bigl(z-R^{(2)}_i\bigr) -  \prod_{i \in \scrA^{(2)}} \frac{1}{R^{(2)}_i} \\
    & \  = \sum_{k\in \scrA^{(2)}}  \biggl\{\prod_{i\in \scrA^{(2)}, i<k}\frac{1}{R^{(2)}_i}\biggr\} \biggl( R_0\bigl(z-R^{(2)}_k\bigr) - \frac{1}{R^{(2)}_k}\biggr) \biggl\{\prod_{i \in \scrA^{(2)}, i>k} R_0\bigl(z-R^{(2)}_i\bigr)\biggr\}\\
    & \ = - \sum_{k \in \scrA^{(2)}} \biggl\{\prod_{i\in \scrA^{(2)}, i<k}\frac{1}{R^{(2)}_i} \biggr\} \frac{R_0\bigl(z-R^{(2)}_k\bigr) }{R^{(2)}_k}\bigl( H_0-z\bigr)\biggl\{\prod_{i \in \scrA^{(2)}, i>k} R_0\bigl(z-R^{(2)}_i\bigr)\biggr\}\\
     & \ = - \sum_{k \in \scrA^{(2)}} \biggl\{\prod_{i\in \scrA^{(2)}, i\leq k}\frac{1}{R^{(2)}_i}\biggr\} \bigl( H_0-z\bigr)\biggl\{ \prod_{i \in  \scrA^{(2)}, i\geq k} R_0\bigl(z-R^{(2)}_i\bigr)\biggr\}\\
\end{align*}
and therefore, we may compute
{\allowdisplaybreaks 
\begin{align}\label{NormalOrderingAfterCancellation}
\nonumber &\rightT_1(z;F_1, \dotsc, F_m) R_0(z) \bigl(T_2(z;F_{m+1}, \dotsc, F_n) - E_2(F_{m+1}, \dotsc, F_n)\bigr) \\
 \nonumber & \quad =\sum_{k \in \scrA^{(2)}} \Biggl\{ \Biggl[ \int \prod^m_{i = 1} F_i(k_i, q_i) \scrL^{(1)} \Delta^{(1)}_\ab\Delta^{(1)}_\af\prod_{j \in  \Jcf^{(1)}} \cf(k_j) \prod_{j \in \Jcb^{(1)} } \cb(q_j)  \\
\nonumber &  \qquad  \prod_{i \in \scrA^{(1)}} R_0\bigl(z-C^{(1)}_{i}-R^{(1)}_i\bigr)   \prod_{j \in \Jab^{(1)} } \ab(q_j) \prod_{j \in  \Jaf^{(1)}} \af(k_j) \prod_{j=1}^m dq_j dk_j \Biggr] R_0(z) \\
\nonumber  &  \qquad \Biggl[ -\int \prod^n_{i = m+1} F_i(k_i, q_i) 
  \scrL^{(2)}\Delta^{(2)}_\ab\Delta^{(2)}_\af  \biggl\{
 \prod_{i\in \scrA^{(2)},i\leq k}\frac{1}{R^{(2)}_i}\biggr\} \bigl( H_0-z\bigr) \\
 \nonumber &\qquad\quad
     \biggl\{\prod_{i \in  \scrA^{(2)}, i\geq k} R_0\bigl(z-R^{(2)}_i\bigr) \biggr\} \prod_{j= m+1}^n dq_j dk_j \Biggr] \Biggr\}\\
\nonumber & \quad = - \sum_{k \in \scrA^{(2)}}  \Biggl[\int \prod^n_{i = 1} F_i(k_i, q_i) \frac{\scrL^{(1)} \scrL^{(2)}}{\prod_{i\in \scrA^{(2)}, i\leq k}R^{(2)}_i}\Delta^{(1)}_\ab\Delta^{(1)}_\af \Delta^{(2)}_\ab\Delta^{(2)}_\af \\
\nonumber &  \qquad   \prod_{j \in  \Jcf^{(1)}} \cf(k_j) \prod_{j \in \Jcb^{(1)} } \cb(q_j)   \biggl\{\prod_{i \in \scrA^{(1)}} R_0\bigl(z-C^{(1)}_{i}-R^{(1)}_i\bigr) \biggr\} \\
\nonumber & \qquad \biggl\{\prod_{i \in\scrA^{(2)}, i\geq k} R_0\Bigl(z-R^{(2)}_i-\sum_{j\in \Jab^{(1)}}\wb(q_j) -\sum_{j\in \Jaf^{(1)}}\wf(k_j) \Bigr) \biggr\}   \\
 &  \qquad  \prod_{j \in \Jab^{(1)} } \ab(q_j) \prod_{j \in  \Jaf^{(1)}} \af(k_j) \prod_{j=1}^n dq_j dk_j\Biggr].
\end{align}
}

In light of the above computation, we fix from now on a $k\in\scrA^{(2)}$ and define 
\begin{equation}\label{Her2-sets}
\begin{aligned}
& I_\af  = I_\af^{(2)}\cup I_\af^{(1)}, & &  I_\ab  = I_\ab^{(2)}\cup I_\ab^{(1)},\\
&\Jab  = \Jab^{(1)}, & &\Jcb = \Jcb^{(1)},\\
& \Jaf = \Jaf^{(1)}, & &\Jcf = \Jcf^{(1)},\\
& 
 \scrA = \scrA^{(1)}\cup\bigl(\scrA^{(2)}\cap \llbracket k, n \rrbracket \bigr), & & 
 \end{aligned}
 \end{equation}
 together with
 \begin{equation*}
 \begin{aligned}
& \forall i \in I_\ab: & &  f_{\ab}(i) = \begin{cases}
   f^{(1)}_{\ab}(i), & \textup{if } i \in I_\ab^{(1)} \\ 
    f^{(2)}_{\ab}(i), & \textup{if } i \in I_\ab^{(2)},
 \end{cases} \\
 &\forall i \in I_\af:  & & f_{\af}(i) = \begin{cases}
   f^{(1)}_{\af}(i), & \textup{if } i \in I_\af^{(1)} \\ 
    f^{(2)}_{\af}(i), & \textup{if } i \in I_\af^{(2)},
 \end{cases}
\end{aligned}
\end{equation*}
and
\begin{equation*}
\scrL\bigl( \{ k_j\}_{j \in I_\af}, \{ q_j\}_{j \in I_\ab}\bigr) = 
 - \frac{ \scrL^{(1)} \scrL^{(2)}}{ \prod_{i\in \scrA^{(2)}, i\leq k}R^{(2)}_i }. 
\end{equation*}

With the above definitions, we may now compute $C_i$ and $R_i$ for $i\in\scrA$  as follows:
\begin{equation}\label{FC-terms}
\begin{aligned}
  C_i & =  \begin{cases} C^{(1)}_{i},   &  \textup{if } i \in \scrA^{(1)}\\
\sum_{j\in \Jab^{(1)}}\wb(q_j) + \sum_{j\in \Jaf^{(1)}}\wf(k_j), &
\textup{if }i \in \llbracket k, n\rrbracket \cap \scrA^{(2)},
\end{cases}
\\
R_i &= \begin{cases} R^{(1)}_i,  & \textup{if } i \in \scrA^{(1)}\\ 
R^{(2)}_i, & \textup{if } i \in \llbracket k, n\rrbracket \cap \scrA^{(2)}.
\end{cases}
\end{aligned}
\end{equation}
Then each term involved in the sum on the right-hand side of \eqref{NormalOrderingAfterCancellation} is of the form \eqref{GeneralFormRegularOperator}. 

Let us focus on one specific term (with the chosen $k$) from the right-hand side of \eqref{NormalOrderingAfterCancellation}:
\begin{equation}\label{ToEstimateFCandNFC}
\begin{aligned}
  &\int \prod^n_{i = 1} F_i(k_i, q_i)  \scrL\bigl( \{ k_j\}_{j \in I_\af}, \{ q_j\}_{j \in I_\ab}\bigr) \prod_{i \in I_\af}\kdelta\bigl(k_i-k_{f_{\af}(i)}\bigr)  \\
  & \qquad \prod_{i \in I_\ab}\kdelta\bigl(q_i-q_{f_{\ab}(i)}\bigr)   \prod_{j \in  \Jcf^{(1)}} \cf(k_j) \prod_{j \in \Jcb^{(1)} } \cb(q_j)\\
 & \qquad 
  \prod_{i \in \scrA} R_0\bigl(z-C_{i}-R_i\bigr) \prod_{j \in \Jab^{(1)} } \ab(q_j) \prod_{j \in  \Jaf^{(1)}} \af(k_j)\prod^n_{j = 1} dq_j dk_j 
\end{aligned}
\end{equation}
and prove that it is a regular Wick monomial.

Any cover $(P_\ab,P_\cb,P_\af,P_\cf)$ for \eqref{ToEstimateFCandNFC} is a cover for $\rightT_1(z;F_1, \dotsc, F_m)$. This follows directly from \eqref{T2-FC}. 

Let $\sigma\colon J_\af\cap J_\cf \to \NN_{0,\infty}$ be an admissible map for \eqref{ToEstimateFCandNFC}. Recalling from \eqref{Her2-sets} that $J_\af^{(1)}= J_\af$ and $J_\cf^{(1)} = J_\cf$, we can define another map 
$\sigma^{(1)}\colon J_\af\cap J_\cf \to \NN_{0,\infty}$, by setting $\sigma^{(1)}(i)=\sigma(i)$ if $\sigma(i)\leq m$ and $\sigma^{(1)}(i)=+\infty$ if $\sigma(i)>m$. Then $\sigma^{(1)}$ is admissible 
for $\rightT_1$. Finally, observe that
\begin{equation}\label{Her2-psigma}
    \psigma\subset \psigma^{(1)}.
\end{equation}

Let now $\{\alpha_i\}_{i=1}^n$ be such that $0\leq \alpha_i\leq 1$ and $\alpha_1+\cdots+\alpha_n = n-1$.  Note that $\scrJ\subset \llbracket 1,m\rrbracket$. As in the proof of Lemma~\ref{HNFCT}, we observe that $J^{(1)}_\ab\cup (J^{(1)}_\af\setminus J^{(1)}_\cb)\neq\emptyset$, since $\rightT_1$ is right-handed (cf. Definition~\ref{def-handedWick}~\ref{item-RHWick}), and choose
\begin{equation}\label{Her2-choiceofj1j2}
\begin{aligned}
& j_1 \in J_{\ab}^{(1)} \cup \bigl(J_{\af}^{(1)}\setminus J_\cb^{(1)}\bigr), 
& & j_2 \in I'_\ab\cup I'_\af,\\
&I_\ab' = \bigset{i \in I_\ab^{(2)}}{ i \leq k < f^{(2)}_{\ab}(i)}, & & I'_\af = \bigset{i \in I_\af^{(2)}}{ i \leq k < f_{\af}^{(2)}(i)}.
\end{aligned}
\end{equation}
Note that the set $I_\ab'\cup I_\af'$ that $j_2$ is chosen from is not empty, because we have $R_k\neq 0$. See \eqref{Rcomponents}.

We again define 
\begin{equation}
\forall i\in\llbracket 1,m\rrbracket:\quad     \alpha_i^{(1)} =  \begin{cases} \alpha_i, & i \neq j_1 \\
\displaystyle m -1- \sum_{j\in\llbracket 1,m\rrbracket \setminus \{j_1\}}\alpha_i,& i = j_1\end{cases} 
\end{equation}
and 
\begin{equation}
\forall i\in\llbracket m+1,n\rrbracket:\quad     \alpha_i^{(2)} = \begin{cases} \alpha_i, & i \neq j_2 \\
\displaystyle n -m -1- \sum_{j\in\llbracket m+1,n\rrbracket \setminus \{j_2\}}\alpha_i,& i = j_2.\end{cases}
\end{equation}

Applying Definition~\ref{RegularOperator} to $\rightT_1$ and $T_2$, yields now admissible exponents $\{\gamma^{(1)}_{i;j}\}_{i\in\scrA^{(1)}, j\in \scrJ^{(1)}}$ and exponents
$\{\beta^{(1)}_{i}\}_{i=1}^m$ and $\{ \beta^{(2)}_{i}\}_{i=m+1}^n$ with $0\leq \beta_i^{(1)}\leq \alpha_i^{(1)}$, $0\leq \beta_i^{(2)}\leq \alpha_i^{(2)}$ and $\beta^{(1)}_i=0$ if $i\in\psigma^{(1)}$, such that the estimate \eqref{regularityproperty} holds for each of the two operators.

 Recalling that $\scrJ = \scrJ^{(1)}$, We introduce for $i\in\scrA$ and $j\in\scrJ$
\begin{equation}
\gamma_{i;j} = \begin{cases}
\gamma^{(1)}_{i;j}, & \textup{if } i \in \scrA^{(1)} \\
\kdelta_{j,j_1}\bigl(\alpha_{j_1}-\alpha^{(1)}_{j_1}\bigr)& \textup{if } i=k \\
0 & \text{otherwise.}
\end{cases}
\end{equation}
Note that this time it is easy to see that the $\gamma_{i;j}$'s satisfy the constraints in \eqref{ConstraintOnGamma}, since the $\gamma_{i;j}^{(1)}$'s do and there are no new contractions. We only have to remark that the choice of $j_1$ ensures that no constraint is imposed on $\gamma_{k,j_1}$. 
Then $\sum_{i\in\scrA}\gamma_{i;j} = \alpha_j$ for $j\in\scrJ$
and $\overline{\gamma}_i = \sum_{j\in \scrJ}\gamma_{i;j} = \overline{\gamma}^{(1)}_i\leq 1$ for $i\in\scrA^{(1)}$
and for $i\in \scrA^{(2)}\cap \llbracket m+1,n\rrbracket$, we have
$\overline{\gamma}_i = \kdelta_{i,k}\bigl(\alpha_{j_1}-\alpha^{(1)}_{j_1}\bigr)\leq 1$. In fact, as with \eqref{gamma-bar-m}, we have
\begin{equation}\label{gamma-bar-k}
1-\overline{\gamma}_k=   \alpha_{j_2}-\alpha^{(2)}_{j_2}.
\end{equation}
We have thus also verified \eqref{SumOfGammas} and the collection $\{\gamma_{i;j}\}_{i\in\scrA,j\in \scrJ}$ are admissible exponents, cf. Definition~\ref{def-admexp}.

We compute using \eqref{Her2-sets}, \eqref{FC-terms} and the definition of $\scrL$:
{\allowdisplaybreaks 
\begin{align}\label{Her2-Laststep}
\nonumber &    \Bigl|\scrL\Bigl(\{q_j\}_{j\in I_\ab},\{k_j\}_{j\in I_\af}\Bigr)\Bigr| \prod_{i\in I_\ab}\kdelta\bigl(q_i-q_{f_\ab(i)}\bigr)\prod_{i\in I_\af}\kdelta\bigl(k_i-k_{f_\af(i)}\bigr)\\
\nonumber&\qquad  \prod_{i\in P_\ab\cup P_\cb} \wb(q_i)^{-\alpha_i}
    \prod_{i\in P_\af\cup P_\cf} \wf(q_i)^{-\alpha_i} \prod_{i\in\scrA} \Bigl\| R_0\bigl(z-C_i(\sigma)-R_i\bigr)^{1-\overline{\gamma}_i}\Bigr\| \\
\nonumber    & \quad  = \biggl\{ \  \bigl|\scrL^{(1)}\bigr| \Delta^{(1)}_\ab \Delta^{(1)}_\af\prod_{i\in P^{(1)}_\ab\cup P^{(1)}_\cb} \wb(q_i)^{-\alpha^{(1)}_i}
    \prod_{i\in P^{(1)}_\af\cup P^{(1)}_\cf} \wf(q_i)^{-\alpha^{(1)}_i}  \\
\nonumber    & \qquad \quad \prod_{i\in\scrA^{(1)}}
    \Bigl\| R_0\bigl(z-C_i(\sigma)-R_i^{(1)}\bigr)^{1-\overline{\gamma}^{(1)}_i}\Bigr\| \
    \biggr\}\\
\nonumber  &\qquad  \biggl\{\
\bigl|\scrL^{(2)}\bigr| \Delta^{(2)}_\ab \Delta^{(2)}_\af   \frac{\prod_{i\in\scrA^{(2)}, i>k} \bigl\| R_0\bigl(z-C_i(\sigma)-R_i^{(2)}\bigr)\bigr\|}{\prod_{i\in \scrA^{(2)},i\leq k} R_i^{(2)}} 
 \   \biggr\} \\
\nonumber  &\qquad   \wb(q_{j_1})^{- \kdelta_{j_1\in P_\ab\cup P_\cb}(\alpha_1-\alpha_1^{(1)})}
    \wf(k_{j_1})^{- \kdelta_{j_1\in P_\af\cup P_\cf}(\alpha_1-\alpha_1^{(1)})} \\
    & \qquad\Bigl\| R_0\bigl(z-C_k(\sigma)-R_k^{(2)}\bigr)^{1-\overline{\gamma}_k}\Bigr\|.
\end{align}
}
Note also that
\[
\begin{aligned}
&\forall i\in\scrA^{(1)}:& & \bigl\| R_0\bigl(z-C_i(\sigma)-R_i^{(1)}\bigr)\bigr\|\leq \bigl\| R_0\bigl(z-C^{(1)}_i(\sigma^{(1)})-R_i^{(1)}\bigr)\bigr\|,\\
&\forall i\in\scrA^{(2)},i>k:& & \bigl\| R_0\bigl(z-C_i(\sigma)-R_i^{(2)}\bigr)\bigr\|\leq \bigl\| R_0\bigl(z-R_i^{(2)}\bigr)\bigr\|.
\end{aligned}
\]
With this in mind, one may now use
 \eqref{regularityproperty} on the two brackets $\{\cdots\}$ on the right-hand side of \eqref{Her2-Laststep}, with $z=0$ for the second bracket (cf.~Remark~\ref{rem-Wick}). The proof now concludes as for Lemma~\ref{HNFCT}, using \eqref{gamma-bar-k}
and with exponents $\{\beta_{i}\}_{i=1}^n$ given as at the end of the proof of Lemma~\ref{HNFCT}, by
\[
 \beta_{i} = \begin{cases}
    \beta^{(1)}_{i} & \textup{if }  i \in \llbracket 1 , m \rrbracket\setminus \{j_1\}  \\
    \beta^{(1)}_{j_1} + \kdelta_{j_1 \in P_{\af}} \bigl(\alpha_{j_1} - \alpha^{(1)}_{j_1}\bigr)& \textup{if }  i=j_1\\
     \beta^{(2)}_{i} & \textup{if }  i \in \llbracket m+1 , n \rrbracket\setminus\{j_2\}\\
     \beta^{(2)}_{j_2} + \kdelta_{j_2 \in I'_\af} \bigl(\alpha_{j_2} - \alpha^{(2)}_{j_2}\bigr)& \textup{if }  i=j_2.
\end{cases} 
\]
We have $0\leq \beta_i\leq\alpha_i$ and from \eqref{Her2-psigma}, we see that $\beta_i=0$ if $i\in\psigma$. In order to verify that these exponents work, it remains to observe the estimates:
\begin{equation*}
\begin{aligned}
&\textup{if } j_1\in P_\ab\cup P_\cf: & & C_k(\sigma) + R_k^{(2)}\geq \wb(q_{j_1})\\
&\textup{if } j_1\in P_\af: & & C_k(\sigma) + R_k^{(2)}\geq \wf(k_{j_1})\\
&\textup{if }j_2\in I'_\ab: & & C_k(\sigma)+R_k\geq \wb(q_{j_2}),\\
&\textup{if }j_2\in I'_\af: & & C_k(\sigma)+R_k\geq \wf(k_{j_2}).
\end{aligned}
\end{equation*}
These estimates follow from the choice of $j_1$ and $j_2$, cf.~\eqref{Her2-choiceofj1j2}, and the formulas for $C_k(\sigma)$, cf. Definition~\ref{SumsinRes}, and for $R_k^{(2)}$, cf.~\eqref{Rcomponents}. 
\end{proof}

\subsection{Normal ordering renormalized handed blocks}\label{subsec-NO-blocks}

In this section, we show how any renormalized handed block of operators can be expressed as a finite sum of Handed Wick Monomials.

\begin{lem}
\label{InductionLemma}
 For any $\ell \in \NN$, there exists $M=M(\ell)$ such that the following holds:
 \begin{enumerate}[label = \textup{(\arabic*)}]
 \label{lem1}
 \item\label{itemtotcont1}  for any  $\us \in \scrS^{(\ell)}_\Right$, there exists $N\in \NN$ with $N\leq M$, and a collection of $N$ right-handed Wick monomials of length $\ell$, $\{\rightT_i \}_{i\in \llbracket 1, N \rrbracket}$, such that:
  \begin{equation}\label{lem2}  
 T_\us^{(\ell)}  = \sum^{N}_{i=1}  \rightT_i.
\end{equation} 
Moreover, the signature associated with each operator $\rightT_i $ is $\us$ and the bounding constants are $c_{\rightT_i} = 1$.
 \item\label{itemtotcont2}  for any  $\us \in \scrS^{(\ell)}_\Left$, there exists $N\in\NN$ with $N\leq M$ and a collection of left-handed  Wick monomials of length $\ell$, $\{ \leftT_i \}_{i\in \llbracket 1, N\rrbracket}$ , such that:
  \begin{equation}\label{lem3} 
 T^{(\ell)}_\us  = \sum^{N}_{i=1} \leftT_i. 
\end{equation} 
Moreover, the signature associated with each operator $\leftT_i $ is $\us$ and the bounding constants are $c_{\leftT_i} = 1$.
 \item\label{itemtotcont3}  for any  $\us \in \scrS^{(\ell)}_\LR$, there exist $N_1, N_2 \in \NN$ with $N_1+N_2\leq M$, a collection of $N_1$ Wick monomials of length $\ell$, $\{ \lrT_i \}_{i\in \llbracket 1, N_1 \rrbracket}$ that are both left- and right-handed, and a collection of $N_2$ fully contracted Wick monomials of length $\ell$, $\{ T_i \}_{i\in \llbracket 1, N_2 \rrbracket}$, such that:
  \begin{equation}
  \begin{aligned}
 T_\us^{(\ell)} & = \sum^{N_1}_{i=1}  \lrT_i +  \sum^{N_2}_{i=1}  (T_i - E_i),\\
   E_i \colon &  (F_1,\dots, F_\ell)\to \bigl\langle\Omega \, \big\vert\, T_i(0; F_1, \dots, F_\ell) \Omega \bigr\rangle.
 \end{aligned}
\end{equation}
Moreover, the signature associated with each operator $\lrT_i  $ or $T_i$ is $\us$ and the bounding constants are $c_{\lrT_i} = c_{T_i} =1$.
\end{enumerate}
\end{lem}

\begin{proof} 
The proof is done by induction and we begin with $\ell=1$.
Here there are four elements of $\scrS^{(1)}$, cf. Remark~\ref{rem-handedsigns}~\ref{rem-signslength1}, two left-handed and two right-handed so only \ref{itemtotcont1} and \ref{itemtotcont2} are involved. We treat only the case $\us = (\ab\af)$, the other three cases being similar.  Recall from Definition~\ref{def-remblocks} that
\[
T^{(1)}_{(\ab\af)}(z;F) = - \int F(k,q) \ab(q) \af(k) dq dk
\]
does not depend on $z$. We argue that this is actually a regular Wick monomial by itself, such that we may choose $N_1=1$. Indeed, we recognize $T^{(1)}_{(\ab\af)}$ as a Wick monomial of length $1$ with $\scrA = I_\af = I_\ab = \Jcf = \Jcb  =   \emptyset$, $\Jaf= \Jab = \{1\}$ and $\scrL=-1$. 

To see that $T^{(1)}_{(\ab\af)}$ is also regular, we observe that there are only two possible covers, one with $P_\ab = \{1\}$ and one with $P_\af = \{1\}$ (the remaining three sets being empty in both cases).
Given one of these two covers and any of the two possible choices of admissible map $\sigma\colon \{1\}\to \{0,1,+\infty\}$, the following holds. There exists only one collection $\{\alpha_i\}$ fulfilling \eqref{Caraalpha}, which is $\alpha_1=0$. Therefore, we are forced to take $\beta_{1} = 0$ and hence $\overline{\gamma}_1 = 0$. As a conclusion, the estimate \eqref{regularityproperty} holds with $c_T = 1$ and $T^{(1)}_{(\ab\af)}$ is a regular Wick monomial.

Now let $\ell \in\NN$ with $\ell\geq 2$. Assume that \ref{itemtotcont1}, \ref{itemtotcont2} and \ref{itemtotcont3} all hold for any $\ell' \leq \ell-1 $. We only have to prove \ref{itemtotcont1} and \ref{itemtotcont3}, since \ref{itemtotcont2} follows from \ref{itemtotcont1} by passing to adjoints. 

If $\us \in \scrS^{(\ell)}_\Right$, then there exists $k\in\Split(\us)$, such that $\us = \us'\circ \us''$,
$\us'\in \scrS^{(k)}_\Right$  and $\us''\in \scrS^{(\ell-k)}_\Left\cup\scrS^{(\ell-k)}_\LR$, such that, for any $\{F_i\}_{i \in \llbracket 1 , \ell \rrbracket} \in L^{2}(\RR^d\times \RR^d)^\ell$:
\begin{equation}
\begin{aligned}
& T^{(\ell)}_\us\bigl(z;F_1, \dotsc, F_\ell\bigr) = T^{(k)}_{\us'}\bigl(z;F_1, \dotsc, F_k\bigr) R_0(z) T^{(\ell-k)}_{\us''}\bigl(z;F_{k+1}, \dotsc, F_\ell\bigr) \\
& \qquad - \bigl\langle\Omega \, \big\vert\, T^{(k)}_{\us'}\bigl(0;F_1, \dotsc, F_k\bigr) R_0(0) T^{(\ell-k)}_{\us''}\bigl(0;F_{k+1}, \dotsc, F_\ell\bigr)\Omega\bigr\rangle
\end{aligned}
\end{equation}
and either 
\begin{equation}
n_\ab(1,\ell;\us) =n_{\ab}(1,k;\us')+ n_{\ab}(1,\ell-k;\us'')  >  0
\end{equation}
or
\begin{equation}\label{condition2}
\begin{cases}  n_\ab(1,\ell;\us) =n_{\ab}(1,k;\us')+ n_{\ab}(1,\ell-k;\us'')  = 0, \\
n_\af(1,\ell;\us)=  n_{\af}(1,k;\us')+n_\af(1,\ell-k;\us'') \geq 0.
\end{cases}
\end{equation}
The induction hypothesis, together with and Lemma~\ref{HNFCT} and Lemma~\ref{HBFCTANFCO} can then be used to conclude that there exist $M$, depending only on $\ell$ and not on the choice of $\us$, as well as $N_1\in\NN, N_2\in\NN_0 $, fulfilling $N_1+N_2\leq M$, $\bigl\{\rightT_i \bigr\}_{i\in \llbracket 1, N_1 \rrbracket}$  a collection of right-handed Wick monomials and $\{ T_i \}_{i\in \llbracket 1, N_2 \rrbracket}$ a collection of fully contracted Wick monomials, all of length $\ell$, such that:
 \begin{equation}
T_\us^{(\ell)} = \sum^{N_1}_{i=1} \rightT_i +  \sum^{N_2}_{i=1}  T_i
 \end{equation}
 and 
 \[
     c_{T_i} =c_{\rightT_i} = c_{T^{(k)}_{\us'}} \cdot c_{T^{(\ell-k)}_{\us''}} = 1.
 \]
In particular, writing $\uF^{(\ell)} = (F_1,\dotsc,F_\ell)$, 
 \begin{align*}
 \bigl\langle\Omega \,\big\vert\, T^{(\ell)}_\us\bigl(0;\uF^{(\ell)}\bigr) \Omega \bigr\rangle  & =   \bigl\langle\Omega \,\big\vert\, \sum^{N_1}_{i=1} \rightT_i\bigl(0;\uF^{(\ell)}\bigr) +  \sum^{N_2}_{i=1}  T_i\bigl(0;\uF^{(\ell)} \bigr)\Omega \bigr\rangle  \\ 
  & =  \sum^{N_2}_{i=1} \bigl\langle\Omega \,\big\vert\,  T_i\bigl(0;\uF^{(\ell)} \bigr)\Omega \bigr\rangle=  \sum^{N_2}_{i=1} E_i\bigl(\uF^{(\ell)} \bigr),
 \end{align*}
which concludes the proof of \ref{itemtotcont1} and \ref{itemtotcont3}, since $N_2=0$ if $\us$ is right-handed. (Recall that the claim \ref{itemtotcont2} is a consequence of claim \ref{itemtotcont1}, since it is the adjoint case.)
\end{proof}

\section{Ordered Wick Monomials}\label{Sec-OrdOp}

As explained earlier, regular Wick monomials are difficult to estimate directly because there may be up to twice as many creation and annihilation operators as there are resolvents. We need to free up as many resolvents as possible, to exploit the momentum decay they provide due to pull-through. The solution is to use boundedness of the smeared fermion annihilation and creation operators $\af(f)$ and $\cf(f)$, for square-integrable $f$, cf.~\eqref{SmearedFermions}. However, in order to arrive at such smeared fermionic operators, we first need to systematically eliminate the relevant explicit fermion momentum dependence introduced into resolvents through normal ordering and pull-through. The procedure we follow is motivated and exemplified in Appendix~\ref{app-oo}.

We are in principle only interested in the starting point, the regular Wick monomials, and the objects we end up with that we estimate in Subsection~\ref{subsec-estregWM}. However, in order to keep track of the terms that appear when we systematically eliminate the relevant fermionic momenta from the resolvents, we introduce a larger class of operators that we call Ordered Wick Monomials that encompass regular Wick monomials and the objects we end up actually estimating, along with everything that appears along the way, when we start to extract the smeared fermionic operators. Figure~\ref{fig-OWM} illustrates the situation with the arrows indicating the momentum elimination procedure that is performed in Subsection~\ref{subsec-RWMtoOWM}.


\begin{figure}
\begin{center}
\begin{tabular}{|c|}
\hline
Wick Monomials, $(P_\af,P_\cf)= (J_\af,J_\cf)$ \\
\hdashline
$\downarrow$
\\
$\downarrow$
\\
$\downarrow$
\\
\hdashline
\ Estimates, $(P_\af,P_\cf)= (J_\af\setminus(J_\ab\cup J_\cb),J_\cf\setminus(J_\ab\cup J_\cb))$ \ \\
\hline
\end{tabular}
\end{center}
\caption{The set of Ordered Operators}\label{fig-OWM}
\end{figure}

\subsection{Definition of ordered Wick monomials}

Before starting, we ask the reader to recall the definitions of a cover, cf.~Definition~\ref{DefCover}, admissible map $\sigma$, cf. Definition~\ref{Admissiblemaps}, as well as the definition $C_i(\sigma)$ from Definition~\ref{SumsinRes}. Note that these objects were introduced with a view towards the analysis of this section.
The following notation will be convenient: 
\begin{equation}\label{tildesigma}
\tilde{\sigma}(i) = \begin{cases}
    \sigma(i), & \textup{ if } i\in \Jcf\\
     \sigma(i)-1, & \textup{ if } i\in \Jaf.
\end{cases}
\end{equation}
Below, $\tilde{\sigma}$ will only be applied to $i\not\in\psigma$.


\begin{Def}[Ordered Wick Monomial]
\label{OrderOperator} 
Let $n\in\NN$ and $\scrA\subset \llbracket 1,n-1\rrbracket$.
Suppose $J_\ab,J_\cb,J_\af,J_\cf,I_\ab,I_\cb$ and $f_\ab,f_\af$ are as in Definition~\ref{Notation}. Let $(P_\ab,P_\cb,P_\af,P_\cf)$ be a cover and $\sigma$ an admissible map.
Abbreviate $\scrB =  (\Jcf\cup \Jaf)\setminus ( P_{\cf}\cup P_{\af})$. Let $z\in\CC_-^*$. An operator-valued function $T_{(P_{\cf},P_{\af})}(z;\cdot)\in\scrM^{(n)}$ is an ordered Wick monomial if there exist two total orders $\preceq$ and $\preceq_*$ on the set $(\Jcf \cup\Jaf) \setminus \psigma$,    such that: 
 \begin{equation}\label{orderprop}
     \forall i,j \in \bigl(\Jcf \cup J_\af\bigr)\setminus \psigma: \qquad \tilde{\sigma}(i) < \tilde{\sigma}(j) \quad  \Rightarrow \quad  i\prec j \textup{ and } i\prec_{*} j
 \end{equation}
 and for $z\in\CC_-^*$ and $\uF^{(n)}\in (L^2(\RR^d\times\RR^d))^n$:
\begin{align}
\label{OrderedOperatorDef}
\nonumber T_{(P_{\cf},P_{\af})}(z;\uF^{(n)})  = &  \int \prod_{i \in \llbracket 1, n \rrbracket \setminus \psigma} F_i(k_i, q_i) \scrL\bigl(\{ k_j\}_{j \in I_\af}, \{ q_j\}_{j \in I_\ab}\bigr) \prod_{i \in \scrB\setminus \psigma}\wf(k_i) \\
\nonumber & \prod_{i \in I_\af}\kdelta\bigl(k_i-k_{f_{\af}(i)}\bigr)  \prod_{i \in I_\ab}\kdelta\bigl(q_i-q_{f_{\ab}(i)}\bigr)  \prod_{i \in \Jcf \cap \psigma} \cf\bigl(F_i(., q_i)\bigr)  \\
\nonumber&  \prod_{j \in \Jcb} \cb(q_j) \prod_{j \in \Jcf\setminus \psigma} \cf(k_j)  \prod_{i \in \scrA} R_0\bigl(z-C_{i}(\sigma) - R_i\bigr)   \\
\nonumber  &  \prod_{i \in \scrB\setminus \psigma} R_0\bigl(z-D_i^{(\preceq, \preceq_*)}(\sigma)\bigr)   \prod_{j \in \Jaf \setminus \psigma} \af(k_j) \prod_{j \in \Jab}\ab(q_j)  \\
& \prod_{i \in \Jaf \cap \psigma}\af\bigl(\overline{F_i(., q_i)}\bigr)\prod_{i \in \llbracket 1, n \rrbracket } dq_i\prod_{i \in \llbracket 1, n \rrbracket \setminus \psigma} dk_i.
\end{align}
Here $\{D_i\}_{i \in \scrB\setminus \psigma}$ is a collection of sums of dispersion relations given by:
\begin{equation}\label{Dis}
\begin{aligned}
    \forall j \in  \scrB\setminus \psigma:\quad    D_j^{(\preceq, \preceq_*)}(\sigma ) & = \sum_{\underset{\ell\in  \Jcb }{\ell > \tilde{\sigma}(j)}} \wb(q_\ell) + \sum_{\underset{\ell\in \Jcf\setminus \psigma}{{j\preceq_* \ell}}}\wf(k_{\ell})\\
    & \quad +  \sum_{\underset{\ell\in \Jab  }{\ell \leq \tilde{\sigma}(j)}} \wb(q_\ell) +\sum_{\underset{\ell\in \Jaf\setminus \psigma }{{\ell \preceq j}}}\wf(k_\ell) +R_{\tilde{\sigma}(j)}.
\end{aligned}
\end{equation}
Moreover, $n$ is called the length of $T_{(P_{\cf},P_{\af})}$. 
\end{Def}

We remind the reader that examples motivating the above definition can be found in Appendix~\ref{app-oo}. The rest of this subsection is devoted to establishing the following lemma.

\begin{Def}[Adjoint of Ordered Operator] Let $T$ be an ordered operator of length $n$. We define the adjoint of $T$, labeled by $T^*$, by setting 
\[
T^*(z; F_1, \dotsc, F_n)= T(\overline{z};\overline{F_n},\dotsc, \overline{F_1})^*_{\vert\Hfin},
\]
with $z\in\CC_-^*$ and $(F_1,\dotsc, F_n)\in (L^2(\RR^d\times \RR^d))^{n}$.
\end{Def}

\begin{lem}\label{lem-OrderedOpAdjoint}
    Let $T_{(P_{\cf},P_{\af})}$ be an ordered Wick monomial. Then its adjoint, $T_{(P_{\cf},P_{\af})}^*$, is also an ordered Wick monomial.
\end{lem}

\begin{proof}
    Let $T_{(P_{\cf},P_{\af})}$ be an ordered Wick monomial, $z\in\CC_-^*$ and $\uF^{(n)}\in (L^2(\RR^d\times\RR^d))^n$. Then, following the notation of Definition~\ref{OrderOperator}, we have 
   \begin{equation*}
\begin{aligned}
 (T_{(P_{\cf},P_{\af})}(z;\uF^{(n)}))^*  = &  \int \prod_{i \in \llbracket 1, n \rrbracket \setminus \psigma} \overline{F_i(k_i, q_i)} \overline{\scrL\bigl(\{ k_j\}_{j \in I_\af}, \{ q_j\}_{j \in I_\ab}\bigr) }\prod_{i \in \scrB\setminus \psigma}\wf(k_i)\\
 &  \prod_{i \in I_\af}\kdelta\bigl(k_i-k_{f_{\af}(i)}\bigr)  \prod_{i \in I_\ab}\kdelta\bigl(q_i-q_{f_{\ab}(i)}\bigr)  \prod_{i \in \Jaf \cap \psigma} \cf\bigl(\overline{F_i(., q_i)}\bigr)  \\
&  \prod_{j \in \Jab} \cb(q_j) \prod_{j \in \Jaf\setminus \psigma} \cf(k_j)  \prod_{i \in \mathscr{A}} R_0\bigl(\overline{z}-C_{i}(\sigma) - R_i\bigr)  \\
  &  \prod_{i \in \scrB\setminus \psigma} R_0\bigl(\overline{z}-D_i^{(\preceq, \preceq_*)}(\sigma)\bigr)   \prod_{j \in \Jcf \setminus \psigma} \af(k_j) \prod_{j \in \Jcb}\ab(q_j) \\
   & \prod_{i \in \Jcf \cap \psigma} \af\bigl(F_i(., q_i)\bigr)\prod_{i \in \llbracket 1, n \rrbracket } dq_i\prod_{i \in \llbracket 1, n \rrbracket \setminus \psigma} dk_i,
\end{aligned}
\end{equation*} 
where $\sigma\colon J_\af\cup J_\cf\to \NN_{0,\infty}$ is an admissible map. Recall the definition \eqref{tildesigma} of $\tilde{\sigma}$.
We will use the same notation for primed objects, cf.~\eqref{adjoint-primeobjects}, that was introduced in the proof of Lemma~\ref{LemAdjoint}. Recall also the bijection $\varphi$ from \eqref{varphi}. Moreover, let us introduce $\sigma' = \varphi \circ \sigma \circ \varphi^{-1}$ and $\preceq'$ and $\preceq'_*$, two total orders on the set $(\Jcf' \cup\Jaf') \setminus \psigma'$, defined as follows.
\[
i \preceq' j \Leftrightarrow \varphi^{-1}(j) \preceq_* \varphi^{-1}(i)  
\]
and 
\[
i \preceq'_* j \Leftrightarrow \varphi^{-1}(j) \preceq \varphi^{-1}(i).
\] 
We now need to check that \eqref{orderprop} is fulfilled. 
First of all, 
note that $\varphi^{-1}\circ \tilde{\sigma}'\circ \varphi = \tilde{\sigma} +1$. Indeed, let $j\in \Jcf$, then $\varphi(j) \in \Jaf'$ and therefore  
\begin{align*}
(\varphi^{-1}\circ \tilde{\sigma}'\circ \varphi)(j) & = \varphi^{-1}\bigl(\sigma'(\varphi(j))-1 \bigr) = \varphi^{-1}\bigl(\sigma'(\varphi(j))\bigr)+1\\
& = \sigma(j)+1 = \tilde{\sigma}(j)+1.
\end{align*}
Analogously, if $j\in \Jaf$, then $\varphi(j) \in \Jcf'$ and
\begin{equation*}
    (\varphi^{-1}\circ \tilde{\sigma}'\circ \varphi)(j)  =  (\varphi^{-1}\circ \sigma'\circ \varphi)(j)
     = \sigma(j)
     = \tilde{\sigma}(j) +1.
\end{equation*}
We are now ready to prove \eqref{orderprop}. Let $i,j \in (\Jcf' \cup J_\af')\setminus \psigma' $. Assume $\tilde{\sigma}'(i) < \tilde{\sigma}'(j)$ then 
\[
\varphi^{-1}(\tilde{\sigma}'(j)) < \varphi^{-1}(\tilde{\sigma}'(i)) 
\]
implying that
\[
\tilde{\sigma}(\varphi^{-1}(j)) < \tilde{\sigma}(\varphi^{-1}(i)).
\]
Consequently $\varphi^{-1}(j) \prec_{\sharp} \varphi^{-1}(i)$, which implies -- by definition -- that $i\prec_{\sharp}' j$.

We now need to prove that $(T_{(P_{\cf},P_{\af})}(z;F_1,\dots F_n))^*$ is of the form described by \eqref{OrderedOperatorDef} and \eqref{Dis}. In the proof of Lemma~\ref{LemAdjoint}, it was established that for $i\in \scrA'$, we have
\[
R_0(\overline{z} - C'_i(\sigma')-R'_i) = R_0(\overline{z} - C_k(\sigma)-R_k),
\]
where $k=\varphi^{-1}(i+1)\in\scrA$.
Let us then consider, for $j\in (\Jcf' \cup\Jaf') \setminus \psigma'$:
\begin{align*}
{D'}^{(\preceq', \preceq'_*)}_j(\sigma')  =&  \sum_{\underset{\ell\in  \Jcb' }{\ell > \tilde{\sigma}'(j)}} \wb(q_\ell) + \sum_{\underset{\ell\in \Jcf'\setminus \psigma'}{{j\preceq'_* \ell}}}\wf(k_{\ell}) \\
& +  \sum_{\underset{\ell\in \Jab'  }{\ell \leq \tilde{\sigma}'(j)}} \wb(q_\ell)+\sum_{\underset{\ell\in \Jaf'\setminus \psigma' }{{\ell \preceq' j}}}\wf(k_\ell) +R_{\tilde{\sigma}'(j)}.
\end{align*}
First, set $k = \varphi^{-1}(j) \in (\Jaf\cup \Jcf) \setminus \psigma$, so that $j = \varphi(k)$. Let us consider the first term of ${D'}^{(\preceq', \preceq'_*)}_j(\sigma')$. Clearly,
\begin{align*}
    \sum_{\underset{\ell\in  \Jcb' }{\ell > \tilde{\sigma}'(j)}} \wb(q_\ell) & = \sum_{\underset{\ell\in  \Jab }{\varphi(\ell) > \tilde{\sigma}'(\varphi(k))}} \wb(q_{\ell})
     = \sum_{\underset{\ell\in  \Jab }{\ell < \varphi^{-1}(\tilde{\sigma}'(\varphi(k)))}} \wb(q_{\ell})\\
    & = \sum_{\underset{\ell\in  \Jab }{\ell < \tilde{\sigma}(k)+1}} \wb(q_{\ell})
      = \sum_{\underset{\ell\in  \Jab }{\ell \leq \tilde{\sigma}(k)}} \wb(q_{\ell}).
\end{align*}
In the same way 
\begin{align*}
    \sum_{\underset{\ell\in \Jab'  }{\ell \leq \tilde{\sigma}'(j)}} \wb(q_\ell) & = \sum_{\underset{\ell\in \Jcb  }{\varphi(\ell) \leq \sigma'(\varphi(k))}} \wb(q_{\ell})
      = \sum_{\underset{\ell\in \Jcb  }{\ell \geq \varphi^{-1}(\sigma'(\varphi(k)))}} \wb(q_{\ell})\\
      & = \sum_{\underset{\ell\in \Jcb  }{\ell \geq \tilde{\sigma}(k)+1}} \wb(q_{\ell})
       = \sum_{\underset{\ell\in \Jcb  }{\ell > \tilde{\sigma}(k)}} \wb(q_{\ell}).
\end{align*}
Moreover,
\begin{equation*}
 \sum_{\underset{\ell\in \Jcf'\setminus \psigma'}{{j\preceq'_* \ell}}}\wf(k_{\ell})  =   \sum_{\underset{\ell\in \Jaf\setminus \psigma}{{\varphi(k) \preceq'_* \varphi(\ell)}}}\wf(k_{\ell})
  = \sum_{\underset{\ell\in \Jaf\setminus \psigma}{{\ell \preceq k}}}\wf(k_{\ell})
\end{equation*}
and
\[
\sum_{\underset{\ell\in \Jaf'\setminus \psigma' }{{\ell \preceq' j}}}\wf(k_\ell) = \sum_{\underset{\ell\in \Jcf\setminus \psigma }{{k \preceq_* \ell}}}\wf(k_{\ell}).
\]
It remains now to consider $R_{\tilde{\sigma}'(j)}$. 
\begin{align*}
    R_{\tilde{\sigma}'(j)} & = \sum_{\underset{ \ell \leq \tilde{\sigma}'(j) <  f'_{\ab}(\ell)}{\ell \in I_\ab'}}\wb(q_\ell) + \sum_{\underset{ \ell \leq \tilde{\sigma}'(j) <  f'_{\af}(\ell)}{\ell \in I_\af'}}\wf(k_\ell)\\
    & = \sum_{\underset{ \varphi^{-1}(\ell) \geq \varphi^{-1}(\tilde{\sigma}'(\varphi(k))) >  \varphi^{-1}(f'_{\ab}(\ell))}{\ell \in I_\ab'}}\wb(q_\ell) + \sum_{\underset{ \varphi^{-1}(\ell) \geq \varphi^{-1}(\tilde{\sigma}'(j)) >  \varphi^{-1}(f'_{\af}(\ell))}{\ell \in I_\af'}}\wf(k_\ell)\\
    & = \sum_{\underset{ f_{\ab}(\ell) \geq \tilde{\sigma}(k)+1 >  \ell}{\ell \in I_\ab}}\wb(q_\ell) + \sum_{\underset{ f_{\af}(\ell) \geq \tilde{\sigma}(k) +1 >  \ell}{\ell \in I_\af}}\wf(k_\ell)\\
     & = \sum_{\underset{  \ell \leq \tilde{\sigma}(k) <  f_{\ab}(\ell)}{\ell \in I_\ab}}\wb(q_\ell) + \sum_{\underset{ \ell \leq \tilde{\sigma}(k) < f_{\af}(\ell)}{\ell \in I_\af}}\wf(k_\ell).
\end{align*}
Abbreviating $\scrB' =  (\Jcf'\cup \Jaf')\setminus ( P_{\cf}'\cup P_{\af}')$, we conclude  
\[
\prod_{i \in \scrB'\setminus \psigma'} R_0\bigl(\overline{z}-{D'}^{(\preceq', \preceq'_*)}_i(\sigma')\bigr) = \prod_{i \in \scrB\setminus \psigma} R_0\bigl(\overline{z}-D^{(\preceq, \preceq_*)}_i(\sigma)\bigr)  
\]
implying that the adjoint of an ordered Wick monomial is ordered.
 \end{proof}

\subsection{From regular to ordered Wick monomials}\label{subsec-RWMtoOWM}

In this section, we establish the link between regular Wick monomials and ordered Wick monomials. The main result of this section is Lemma \ref{ReorderingFermionLemma} which states that any regular Wick monomial is a finite sum of  ordered Wick monomials of a specific type, which, as will be seen later, are easy to estimate (see Proposition \ref{RNFCT}).

\begin{lem}
\label{InitialisationReorderingFermionLemma}
For any ordered Wick monomial $T_{(P_{\cf},P_{\af})}$ and any $j_0\in (P_\cf\cup P_\af)\cap (\Jcb\cup\Jab) $, there exists $L\leq 2n$, a collection of ordered Wick monomials $\left\{T^{(i)}_{(P_{\cf}\backslash\{j_0\},P_{\af}\backslash\{j_0\})}\right\}_{i=1}^L$ of the same length as $T_{(P_{\cf},P_{\af})}$, such that:
\begin{equation}
T_{(P_{\cf},P_{\af})} = \sum_{i=1}^L T^{(i)}_{(P_{\cf}\setminus\{j_0\},P_{\af}\setminus\{j_0\})}.
\end{equation}
Moreover, the sets $\scrA, \Jcb, \Jab, \Jaf, \Jcf, I_\ab, I_\af$ and the functions $f_\ab,f_\af$ and $\scrL$ are identical for $T_{(P_{\cf},P_{\af})}$ and any $ T^{(i)}_{(P_{\cf}\setminus\{j_0\},P_{\af}\setminus\{j_0\})}$.
 \end{lem}

 \begin{proof} 
 As we did with \eqref{abbrev-LDelta}, we will in this proof be using the abbreviations:
  \begin{equation}\label{abbrev-LDelta2}
    \begin{aligned}
        &\scrL := \scrL\bigl( \{ k_j\}_{j \in I_\af}, \{ q_j\}_{j \in I_\ab}\bigr),\\
        &\Delta_\ab := \prod_{i \in I_\ab}\kdelta\bigl(q_i-q_{f_{\ab}(i)}\bigr),\quad \Delta_\af :=  \prod_{i \in I_\af}\kdelta\bigl(k_i-k_{f_{\af}(i)}\bigr). 
    \end{aligned}
  \end{equation}
These objects are mostly bystanders during the proof.
 
In the following $z\in\CC_-^*$ and $\uF^{(n)} \in (L^2(\RR^d\times\RR^d))^n$. Let $\sigma\colon J_\af\cup J_\cf\to\NN_{0,\infty}$ be an admissibe map, cf. Definition~\ref{Admissiblemaps}, and recall the notation $\psigma = \set{j\in J_\cf\cup J_\af}{\sigma(j)=0 \textup{ or }\sigma(j)=+\infty}$. Recall also the abbreviation $\scrB = (J_\cf\cup J_\af)\setminus (P_\cf\cup P_\af)$.
We consider an ordered Wick monomial $T_{(P_{\cf},P_{\af})}(z;\uF^{(n)})$ of the form \eqref{OrderedOperatorDef}.

 Let $j_0\in (P_{\cf}\cup P_{\af})\cap (\Jab\cup \Jcb)$ and observe that $j_0\not\in\scrB$.  We will treat the case where $j_0 \in P_{\cf}$ as it implies that the result holds in the other case, $j_0\in P_\af$. Indeed, if one consider $j_0\in P_\af$, one can consider the adjoint Wick monomial which is an ordered Wick monomial too for which $j_0$ becomes a fermionic creation index. See Lemma~\ref{lem-OrderedOpAdjoint}.  Let us define new covers 
 \[
 \begin{aligned}
 P_{\cf}^{(0)}& =P_{\cf}\setminus \{j_0\}, & P_{\af}^{(0)}&=P_{\af},\\
  P_{\cb}^{(0)} &= P_{\cb} \cup ( \Jcb\cap \{j_0\}), &   P_{\ab}^{(0)} &= P_{\ab} \cup ( \Jab\cap \{j_0\}),
 \end{aligned}
 \]
 and a new admissible map $\sigma^{(0)}\colon J_\cf\cup J_\af\to \NN_{0,\infty}$ by setting: 
\begin{equation}
    \sigma^{(0)}(i) = 
    \begin{cases}
       \sigma(i),  & \text{if } i\in (J_\cf\cup J_\af)\setminus \{j_0\} \\
        0, & \text{if } i=j_0. 
    \end{cases}
\end{equation}

 Let us moreover introduce two total orders on  $(\Jcf  \cup \Jaf) \setminus \psigma^{(0)} = [(\Jcf  \cup \Jaf) \setminus \psigma]\setminus\{j_0\}$ as the restriction of  $\preceq$ and $\preceq_*$ onto this set. The restricted orders will be denoted by $\preceq^{(0)}$ and $\preceq^{(0)}_*$.
 Observe that the required properties hold: 
 \begin{align*}
    & \forall i,\ell \in \bigl(\Jcf \cup \Jaf\bigr) \setminus \psigma^{(0)} : \qquad \tilde{\sigma}^{(0)}(i) < \tilde{\sigma}^{(0)}(\ell) \quad  \Rightarrow \quad  i \prec^{(0)} \ell
 \end{align*}
 and 
  \begin{align*}
    & \forall i,\ell \in \bigl(\Jcf \cup \Jaf\bigr) \setminus \psigma^{(0)} : \qquad \tilde{\sigma}^{(0)}(i) < \tilde{\sigma}^{(0)}(\ell) \quad  \Rightarrow \quad  i \prec^{(0)}_* \ell.
 \end{align*}
Set $\scrB^{(0)} = (J_\cf\cup J_\af)\setminus (P_\cf^{(0)}\cup P_\af^{(0)})$ and observe as well that
\begin{equation}
\scrB^{(0)} = \scrB\cup \{j_0\} \qquad \textup{and}\qquad 
\psigma^{(0)} = \psigma\cup \{j_0\}.
\end{equation}
 In particular, we conclude that $\scrB^{(0)} \setminus \psigma^{(0)} = \scrB \setminus \psigma$. Let us introduce another ordered Wick monomial
  \begin{equation*}
\begin{aligned}
 & T^{(0)}_{(P^{(0)}_{\cf},P^{(0)}_{\af})}\bigl(z;\uF^{(n)}\bigr) \\
& \quad =  \int \prod_{i \in \llbracket 1, n \rrbracket \setminus \psigma^{(0)}} F_i(k_i, q_i) \scrL\, \Delta_\ab\Delta_\af \prod_{i \in \scrB^{(0)}\setminus \psigma^{(0)}}\wf(k_i)  \prod_{i \in \Jcf \cap \psigma^{(0)}} \cf\bigl(F_i(., q_i)\bigr)  \\
 & \qquad\prod_{j \in \Jcb} \cb(q_j) \prod_{j \in \Jcf\setminus \psigma^{(0)}} \cf(k_j)  \prod_{i \in \scrA} R_0\bigl(z-C_{i}(\sigma^{(0)}) - R_i\bigr)  \\
  & \qquad \prod_{i \in \scrB^{(0)}\setminus\psigma^{(0)}} R_0\bigl(z-D_i^{(\preceq^{(0)},\preceq_*^{(0)})}(\sigma^{(0)})\bigr)   \prod_{j \in \Jaf \setminus \psigma^{(0)}} \af(k_j) \prod_{j \in \Jab}\ab(q_j)\\
  & \qquad \prod_{i \in \Jaf \cap \psigma^{(0)}} \af\bigl(\overline{F_i(., q_i)}\bigr)\prod_{i \in \llbracket 1, n \rrbracket \setminus \psigma^{(0)}}  dk_i\prod_{i \in \llbracket 1, n \rrbracket } dq_i
  \end{aligned}
\end{equation*}
that can be expressed as follows
  \begin{equation}\label{T0P0}
\begin{aligned}
& T^{(0)}_{(P^{(0)}_{\cf},P^{(0)}_{\af})}\bigl(z;\uF^{(n)}\bigr) \\
&\quad  =   \int \prod_{i \in \llbracket 1, n \rrbracket \setminus\psigma} F_i(k_i, q_i) \scrL\, \Delta_\ab\Delta_\af \prod_{i \in \scrB\setminus \psigma}\wf(k_i)    \prod_{i \in \Jcf \cap \psigma} \cf\bigl(F_i(., q_i)\bigr)  \\
 & \qquad\prod_{j \in \Jcb} \cb(q_j) \prod_{j \in \Jcf\setminus \psigma} \cf(k_j)  \prod_{i \in \scrA} R_0\bigl(z-C_{i}(\sigma^{(0)}) - R_i\bigr)  \\
  & \qquad \prod_{i \in \scrB\setminus \psigma} R_0\bigl(z-D_i^{(\preceq^{(0)},\preceq_*^{(0)})}(\sigma^{(0)})\bigr)   \prod_{j \in \Jaf \setminus \psigma} \af(k_j) \prod_{j \in \Jab}\ab(q_j)\\
  & \qquad  \prod_{i \in \Jaf \cap \psigma} \af\bigl(\overline{F_i(., q_i)}\bigr)\prod_{i \in \llbracket 1, n \rrbracket \setminus \psigma}  dk_i\prod_{i \in \llbracket 1, n \rrbracket } dq_i. 
\end{aligned}
\end{equation}
Note that as part of the computation above, we have -- for the chosen $j_0\in P_\cf$ -- written
$\displaystyle \cf(F_{j_0}(\cdot,q_{j_0})) = \int F_{j_0}(k_{j_0},q_{j_0}) \cf(k_{j_0}) dk_{j_0}$. If $n=1$, we have $T^{(0)}_{(\emptyset,\emptyset)}(F_1) = T_{(\{j_0\},\emptyset)}(F_1)$ and we are done. This follows easily from \eqref{T0P0}, noting that $\scrB\setminus\psigma=\scrA =\emptyset$ in the case $n=1$. In the following, we may therefore assume that $n>1$. 

We are interested in studying $T_{(P_{\cf},P_{\af})}(z;\uF^{(n)}) - T^{(0)}_{(P^{(0)}_{\cf},P^{(0)}_{\af})}(z;\uF^{(n)})$. Our goal is to prove that it is a sum of ordered Wick monomials, which is enough to conclude the proof. First, we have: 
\begin{align}\label{pre-telescopicsum}
\nonumber &    T_{(P_{\cf},P_{\af})}\bigl(z;\uF^{(n)}\bigr) -  T^{(0)}_{(P^{(0)}_{\cf},P^{(0)}_{\af})}\bigl(z;\uF^{(n)}\bigr) \\
 \nonumber     & \quad =  \int \prod_{i \in \llbracket 1, n \rrbracket \setminus \psigma} F_i(k_i, q_i) \scrL\, \Delta_\ab\Delta_\af  \prod_{i \in \scrB\setminus\psigma}\wf(k_i) \\
\nonumber      &\qquad   \prod_{i \in \Jcf \cap \psigma} \cf\bigl(F_i(., q_i)\bigr)\prod_{j \in \Jcb} \cb(q_j) \prod_{j \in \Jcf\setminus \psigma} \cf(k_j) \\
\nonumber  & \qquad \Biggl[  \prod_{i \in \scrA} R_0\bigl(z-C_{i}(\sigma) - R_i\bigr)\prod_{i \in \scrB\setminus \psigma} R_0\bigl(z-D_i^{(\preceq, \preceq_*)}(\sigma)\bigr)  \\
\nonumber  &\qquad \quad -  \prod_{i \in \scrA} R_0\bigl(z-C_{i}(\sigma^{(0)}) - R_i\bigr)\prod_{i \in \scrB\setminus \psigma} R_0\bigl(z-D_i^{(\preceq^{(0)},\preceq_*^{(0)})}(\sigma^{(0)})\bigr)\Biggr] \\
  &\qquad \prod_{j \in \Jaf \setminus \psigma} \af(k_j) \prod_{j \in \Jab}\ab(q_j) \prod_{i \in \Jaf \cap \psigma} \af\bigl(\overline{F_i(., q_i)}\bigr)\prod_{i \in \llbracket 1, n \rrbracket \setminus \psigma}  dk_i\prod_{i \in \llbracket 1, n \rrbracket } dq_i.
\end{align}

 In the rest of this proof, the following convention will be used. For $S\subset\NN_0$,  we  define the characteristic function $\myTheta_S\colon \NN_0 \to \{0,1\}$ by
 \[
\myTheta_S(k) = \begin{cases}
    1, & \text{if }  k \in S  \\
    0, & \text{otherwise}
\end{cases}
 \]
and use the convention that  $\myTheta_{S^\comp} = \myTheta_{\NN_0\setminus S}$.
 We compute the term in the square brackets $[\dots]$ in \eqref{pre-telescopicsum}: 
  {\allowdisplaybreaks
  \begin{align}\label{telescopicsum}
    \nonumber &  \prod_{i \in \scrA} R_0\bigl(z-C_{i}(\sigma) - R_i\bigr)\prod_{i \in \scrB\setminus \psigma} R_0\bigl(z-D_i^{(\preceq, \preceq_*)}(\sigma)\bigr)   \\
 \nonumber& \quad -  \prod_{i \in \scrA} R_0\bigl(z-C_{i}(\sigma^{(0)}) - R_i\bigr)\prod_{i \in \scrB\setminus \psigma} R_0\bigl(z-D_i^{(\preceq^{(0)},\preceq_*^{(0)})}(\sigma^{(0)})\bigr)\\
 \nonumber& \ = \sum^{n-1}_{k=0} \Biggl[ \prod_{\underset{i<k}{i \in \scrA}} R_0\bigl(z-C_{i}(\sigma) - R_i\bigr)\prod_{\underset{\tilde{\sigma}(i) < k+1}{i \in \scrB\setminus \psigma}} R_0\bigl(z-D_i^{(\preceq, \preceq_*)}(\sigma)\bigr)   \\
 \nonumber & \quad \Biggl\{ \Bigl(\myTheta_\scrA(k) R_0\bigl(z-C_{k}(\sigma) - R_{k}\bigr)+\myTheta_{\scrA^\comp}(k)\Bigr)\prod_{\underset{\tilde{\sigma}(i) = k+1}{i \in \scrB\setminus \psigma}} R_0\bigl(z-D_i^{(\preceq, \preceq_*)}(\sigma)\bigr)\\
\nonumber & \quad  -  \Bigl( \myTheta_\scrA(k) R_0\bigl(z-C_{k}(\sigma^{(0)}) - R_{k}\bigr)+\myTheta_{\scrA^\comp}(k)\Bigr) \\
 \nonumber & \quad \prod_{\underset{\tilde{\sigma}(i) = k+1}{i \in \scrB\setminus \psigma}} R_0\bigl(z-D_i^{(\preceq^{(0)},\preceq_*^{(0)})}(\sigma^{(0)})\bigr)\Biggr\}\\
    &\quad \prod_{\underset{i>k}{i \in \scrA}} R_0\bigl(z-C_{i}(\sigma^{(0)}) - R_i\bigr)\prod_{\underset{\tilde{\sigma}(i) > k+1}{i \in \scrB\setminus \psigma}} R_0\bigl(z-D_i^{(\preceq^{(0)},\preceq_*^{(0)})}(\sigma^{(0)})\bigr)\Biggr].
 \end{align}
 }
 
 Let us note that for any $k \in \scrB \setminus\psigma$, we have
 \begin{align*}
     & D_{k}^{(\preceq,\preceq_*)}(\sigma) - D_{k}^{(\preceq^{(0)},\preceq_*^{(0)})}(\sigma^{(0)}) \\
     & \ = \sum_{\underset{l\in  \Jcb }{\ell > \tilde{\sigma}(k)}} \wb(q_\ell) + \sum_{\underset{\ell\in \Jcf\setminus \psigma}{{ k \preceq_* \ell}}}\wf(k_{\ell}) +  \sum_{\underset{\ell\in \Jab  }{\ell \leq \tilde{\sigma}(k)}} \wb(q_\ell)\
  +\sum_{\underset{\ell\in \Jaf\setminus \psigma }{{\ell \preceq k}}}\wf(k_\ell) \\
  &  \quad
   -\sum_{\underset{\ell\in  \Jcb }{\ell > \tilde{\sigma}^{(0)}(k)}} \wb(q_\ell)
  - \sum_{\underset{\ell\in \Jcf\setminus \psigma^{(0)}}{{k \preceq_*^{(0)} \ell}}}\wf(k_{\ell}) - \sum_{\underset{\ell\in \Jab  }{\ell \leq \tilde{\sigma}^{(0)}(k)}} \wb(q_\ell) - \sum_{\underset{\ell\in \Jaf\setminus \psigma^{(0)} }{{\ell \preceq^{(0)} k}}}\wf(k_\ell) \\
  & \quad
  +R_{\tilde{\sigma}(k)}
  -R_{\tilde{\sigma}^{(0)}(k)}.
 \end{align*}
 Most of the terms above cancel. Indeed;
 since $j_0\not\in \scrB$, we must have $k\neq j_0$, which implies that $\sigma(k) = \sigma^{(0)}(k)$. Therefore $R_{\tilde{\sigma}(k)} =  R_{\tilde{\sigma}^{(0)}(k)}$, 
 \begin{equation*}
     \sum_{\underset{\ell\in  \Jcb }{\ell > \tilde{\sigma}(k)}} \wb(q_\ell) = \sum_{\underset{\ell\in  \Jcb }{\ell > \tilde{\sigma}^{(0)}(k)}} \wb(q_\ell) \qquad  \textup{and}\qquad
     \sum_{\underset{\ell\in \Jab }{\ell \leq \tilde{\sigma}(k)}} \wb(q_\ell) = \sum_{\underset{\ell\in \Jab  }{\ell \leq \tilde{\sigma}^{(0)}(k)}} \wb(q_\ell).
      \end{equation*}
Furthermore, by definition of $\preceq^{(0)}$ and because $j_0 \notin \Jaf$ and $k\neq j_0$, we have : 
\[
\sum_{\underset{\ell\in \Jaf\setminus \psigma }{{\ell \preceq k}}}\wf(k_\ell) = \sum_{\underset{\ell\in \Jaf\setminus \psigma^{(0)} }{{  \ell \preceq^{(0)}  k}}}\wf(k_\ell).
\]
Finally, $j_0 \notin \Jcf \setminus \psigma^{(0)}$ but $j_0\in J_\cf \setminus\psigma$, since $\sigma(j_0)=j_0$ (due to the chosen case $j_0\in P_\cf$), and therefore :  
 \begin{align*}
     D_{k}^{(\preceq,\preceq_*)}(\sigma) - D_{k}^{(\preceq^{(0)},\preceq_*^{(0)})}(\sigma^{(0)}) & = \begin{cases}
       \wf(k_{j_0}),   & \text{if }  k \preceq_* j_0 \\
         0, & \text{otherwise}. 
     \end{cases}  
 \end{align*}
 In the same way, cf. Definition~\ref{SumsinRes}, we have
  \begin{align*}
     C_{i}(\sigma)  -C_{i}(\sigma^{(0)}) & = \begin{cases}
       \wf(k_{j_0}),   & \text{if }  i < j_0\\
         0, & \text{otherwise}. 
     \end{cases} 
 \end{align*}
 We can now compute the difference in the curly brackets $\{\dots\}$ in Eq.~\eqref{telescopicsum}: 
 \begin{align}\label{ResDiff}
\nonumber & \Bigl(\myTheta_\scrA(k)
R_0\bigl(z-C_{k}(\sigma) - R_k\bigr)+\myTheta_{\scrA^\comp}(k)\Bigr)\prod_{\underset{ \tilde{\sigma}(i) = k+1}{i \in \scrB\setminus \psigma}} R_0\bigl(z-D_i^{(\preceq,\preceq_*)}(\sigma)\bigr)\\
\nonumber & -  \Bigl(\myTheta_\scrA(k) R_0\bigl(z-C_{k}(\sigma^{(0)}) - R_k\bigr)+\myTheta_{\scrA^\comp}(k)\Bigr)\prod_{\underset{\tilde{\sigma}(i) = k+1}{i \in \scrB\setminus \psigma}} R_0\bigl(z-D_i^{(\preceq^{(0)},\preceq_*^{(0)})}(\sigma^{(0)})\bigr)\\
 \nonumber  & \quad = - \myTheta_\scrA(k) \myTheta_{\llbracket 0, j_0-1\rrbracket}(k) \wf(k_{j_0}) R_0\bigl(z-C_{k}(\sigma^{(0)}) - R_{k}\bigr)R_0\bigl(z-C_{k}(\sigma) - R_{k}\bigr) \\
\nonumber & \qquad  \prod_{\underset{\tilde{\sigma}(i) = k+1}{i \in \scrB\setminus \psigma}} R_0\bigl(z-D_i^{(\preceq^{(0)},\preceq_*^{(0)})}(\sigma^{(0)})\bigr)\\
\nonumber &\qquad  - \Bigl(\myTheta_\scrA(k) R_0\bigl(z-C_{k}(\sigma) - R_k\bigr)
 + \myTheta_{\scrA^\comp}(k)\Bigr) \wf(k_{j_0})\\
 \nonumber & \qquad \sum_{\underset{\tilde{\sigma}(k') = k+1}{k'  \in \scrB\setminus \psigma}} \prod_{\underset{i \prec k'}{\underset{ \tilde{\sigma}(i) = k+1}{i \in \scrB\setminus \psigma}}} R_0\bigl(z-D_i^{(\preceq, \preceq_*)}(\sigma)\bigr) \myTheta_{\{ \preceq j_0\}}(k')   R_0\bigl(z-D_{k'}^{(\preceq,\preceq_*)}(\sigma)\bigr) \\
 &\qquad \quad  R_0\bigl(z-D_{k'}^{(\preceq^{(0)},\preceq_*^{(0)})}(\sigma^{(0)})\bigr)  \prod_{\underset{k' \prec i }{\underset{\tilde{\sigma}(i) = k+1}{i \in \scrB\setminus \psigma}}} R_0\bigl(z-D_i^{(\preceq^{(0)},\preceq_*^{(0)})}(\sigma^{(0)})\bigr),
\end{align}
where  $\{\preceq j_0\}:=\set{k'\in (J_\af\cup J_\cf)\setminus \psigma}{k'\preceq j_0}$.
Inserting back into \eqref{telescopicsum} gives a sum of operators that we now proved to identify as ordered Wick monomials.

First, let $k\in\scrA$ with $k<j_0$ and note that $k\geq 1$. Recall the abbreviations \eqref{abbrev-LDelta2}. The first term, $T^{(k)}$, in the difference $T_{(P_{\cf},P_{\af})} - T^{(0)}_{(P^{(0)}_{\cf},P^{(0)}_{\af})}$ that we focus on corresponds to the first term on the right hand side of  \eqref{ResDiff}:
 {\allowdisplaybreaks
  \begin{align} \label{Changingorder1}
  \nonumber & T^{(k)}\bigl(z;\uF^{(n)}\bigr) \\
\nonumber  &\quad = - \int \prod_{i \in \llbracket 1, n \rrbracket \setminus \psigma} F_i(k_i, q_i)\scrL\, \Delta_\ab\Delta_\af \prod_{i \in \scrB\setminus \psigma}\wf(k_i)   \prod_{i \in \Jcf \cap \psigma} \cf\bigl(F_i(., q_i)\bigr)\\
 \nonumber& \qquad \prod_{j \in \Jcb} \cb(q_j) \prod_{j \in \Jcf\setminus \psigma} \cf(k_j)  \prod_{\underset{i<k}{i \in \scrA}} R_0\bigl(z-C_{i}(\sigma) - R_i\bigr) \\
\nonumber & \qquad\prod_{\underset{\tilde{\sigma}(i) < k+1}{i \in \scrB\setminus \psigma}} R_0\bigl(z-D_i^{(\preceq, \preceq_*)}(\sigma)\bigr)  \Biggl\{  \wf(k_{j_0}) R_0\bigl(z-C_{k}(\sigma^{(0)}) - R_{k}\bigr) \\
\nonumber  &  \qquad R_0\bigl(z-C_{k}(\sigma) - R_{k}\bigr) \prod_{\underset{\tilde{\sigma}(i) = k+1}{i \in \scrB\setminus \psigma}} R_0\bigl(z-D_i^{(\preceq^{(0)},\preceq_*^{(0)})}(\sigma^{(0)})\bigr)\Biggr\} \\
\nonumber & \qquad 
\prod_{\underset{i>k}{i \in \scrA}} R_0\bigl(z-C_{i}(\sigma^{(0)}) - R_i\bigr) \prod_{\underset{\tilde{\sigma}(i) > k+1}{i \in \scrB\setminus \psigma}} R_0\bigl(z-D_i^{(\preceq^{(0)},\preceq_*^{(0)})}(\sigma^{(0)})\bigr) \\
& \qquad \prod_{j \in \Jaf \setminus \psigma} \af(k_j) \prod_{j \in \Jab}\ab(q_j) \prod_{i \in \Jaf \cap \psigma} \af\bigl(\overline{F_i(., q_i)}\bigr)\prod_{i \in \llbracket 1, n \rrbracket \setminus \psigma}  dk_i\prod_{i \in \llbracket 1, n \rrbracket } dq_i.
\end{align}
}
For the $k$ under consideration, 
we define $\sigma^{(k)}\colon J_\cf\cup J_\af\to\mathbb N_\infty$ by setting: 
\[
\forall i\in J_\cf\cup J_\af:\qquad \sigma^{(k)}(i) = 
\begin{cases}
    \sigma(i), & \textup{if } i \neq j_0   \\
   k,  & \textup{if }  i = j_0.
\end{cases}
\]
Note that $\psigma^{(k)}= \psigma$, such that we may define
 new total orders $\preceq_*^{(k)}$ by 
\[
\forall a,b \in \bigl(\Jcf  \cup\Jaf\bigr) \setminus\psigma^{(k)}:\qquad
\begin{cases}
    a \preceq_*^{(k)} b,  & \textup{ if } a,b\neq j_0 \textup{ and } a \preceq_* b \\
     a \preceq_*^{(k)} j_0,   & \textup{ if }  \tilde{\sigma}^{(k)}(a) \leq k \\
      j_0 \preceq_*^{(k)} b,  & \textup{ if }   \tilde{\sigma}^{(k)}(b)>k.
\end{cases}
\]
and $\preceq^{(k)}$ by
\[
\forall a,b \in \bigl(\Jcf  \cup\Jaf\bigr) \setminus\psigma^{(k)}:\qquad
\begin{cases}
    a \preceq^{(k)} b,  & \textup{ if } a,b\neq j_0 \textup{ and } a \preceq b \\
     a \preceq^{(k)} j_0,   & \textup{ if }  \tilde{\sigma}^{(k)}(a) < k \\
      j_0 \preceq^{(k)} b,  & \textup{ if }   \tilde{\sigma}^{(k)}(b)\geq k.
\end{cases}
\]
It is easy to check that 
 \begin{equation*}
     \forall i,\ell \in \bigl(\Jcf \cup \Jaf\bigr) \setminus \psigma^{(k)}: \qquad \tilde{\sigma}^{(k)}(i) < \tilde{\sigma}^{(k)}(\ell) \quad  \Rightarrow \quad  i\prec^{(k)}_* \ell
 \end{equation*}
 and 
  \begin{equation*}
     \forall i,\ell \in \bigl(\Jcf \cup \Jaf\bigr) \setminus \psigma^{(k)}: \qquad \tilde{\sigma}^{(k)}(i) < \tilde{\sigma}^{(k)}(\ell) \quad  \Rightarrow \quad  i\prec^{(k)} \ell.
 \end{equation*}
The strategy is now to rewrite \eqref{Changingorder1} as an ordered Wick monomial using $\sigma^{(k)}$ and $\preceq^{(k)}$. 
One may readily verify, using that $k<j_0$, the following identities: 
\begin{align*}
    \forall i \in \scrA, i < k: \qquad & C_i(\sigma) = C_i(\sigma^{(k)}), \\
    \forall i \in \scrA, i\geq k:\qquad & C_i(\sigma^{(0)}) = C_i(\sigma^{(k)}).
\end{align*}

Let us now turn to the term involving a $D_i$. Let us recall that 
\begin{align*}
D_i^{(\preceq, \preceq_*)}(\sigma) & = \sum_{\underset{\ell > \tilde{\sigma}(i)}{\ell\in  \Jcb }} \wb(q_\ell) + \sum_{\underset{i\preceq_* \ell}{\ell\in \Jcf\setminus \psigma}}\wf(k_{\ell}) +  \sum_{\underset{\ell \leq \tilde{\sigma}(i)}{\ell\in \Jab  }} \wb(q_\ell)\\
& \qquad +\sum_{\underset{\ell \preceq i}{\ell\in \Jaf\setminus \psigma }}\wf(k_\ell) +R_{\tilde{\sigma}(i)}. 
\end{align*}
First, $j_0 \notin \scrB$ and therefore $i \neq j_0$. As a consequence, we have that $\tilde{\sigma}^{(k)}(i) = \tilde{\sigma}(i) $ and therefore $R_{\tilde{\sigma}(i)}   = R_{\tilde{\sigma}^{(k)}(i)}$, 
\begin{equation*}
    \sum_{\underset{\ell > \tilde{\sigma}(i)}{\ell\in  \Jcb }} \wb(q_\ell)  = \sum_{\underset{\ell > \tilde{\sigma}^{(k)}(i)}{\ell\in  \Jcb }} \wb(q_\ell)\quad \textup{and} \quad 
   \sum_{\underset{\ell \leq \tilde{\sigma}(i)}{\ell\in \Jab  }} \wb(q_\ell)  = \sum_{\underset{\ell \leq \tilde{\sigma}^{(k)}(i)}{\ell\in \Jab  }} \wb(q_\ell) .
\end{equation*}
Moreover, using the fact that $j_0 \notin \scrB$ and that $j_0 \notin \Jaf$, we have that: 
\[
\sum_{\underset{\ell \preceq i}{\ell\in \Jaf\setminus \psigma }}\wf(k_\ell) = \sum_{\underset{\ell \preceq^{(k)} i}{\ell\in \Jaf\setminus \psigma }}\wf(k_\ell).
\]
To treat the last term, let us distinguish two cases. First, if $\tilde{\sigma}(i)< k+1$ then for any $\ell \in \Jcf \setminus(\psigma\cup \{j_0\}) $, $i \preceq_* \ell$ is equivalent to $i \preceq^{(k)}_* \ell$. It remains to consider the case where $\ell=j_0$. First, $j_0>k\geq \tilde{\sigma}(i)$. As a consequence, $i\preceq_* j_0$. On the other hand, $ k \geq \tilde{\sigma}^{(k)}(i)$ implying that $i \preceq^{(k)}_* j_0$. Therefore, if $\tilde{\sigma}(i)< k+1$, we have 
\[
\sum_{\underset{i\preceq_* \ell}{\ell\in \Jcf\setminus \psigma}}\wf(k_{\ell}) = \sum_{\underset{i\preceq^{(k)}_* \ell}{\ell\in \Jcf\setminus \psigma}}\wf(k_{\ell}).
\]
If now $\tilde{\sigma}(i)\geq k+1$ then again it is enough to consider the case where $l=j_0$. Noticing that $\tilde{\sigma}^{(k)}(i) > k$, we have $j_0 \preceq^{(k)}_* i$ and since $j_0\neq i$, we then have $j_0 \prec^{(k)}_* i$. Consequently, 
\[
\sum_{\underset{i\preceq^{(0)}_* \ell}{\ell\in \Jcf\setminus \psigma}}\wf(k_{\ell}) = \sum_{\underset{i\preceq^{(k)}_* \ell}{\ell\in \Jcf\setminus \psigma}}\wf(k_{\ell}).
\]
The last term that remains to treat is $R_0\bigl(z-C_{k}(\sigma) - R_k\bigr)$. Let us recall that: 
\begin{align*}
C_{k}(\sigma) &= \sum_{\underset{\ell > k}{\ell\in  \Jcb }} \wb(q_\ell) + \sum_{\underset{{\sigma}(\ell)>k}{\ell\in  \Jcf }} \wf(k_\ell)+\sum_{\underset{\ell \leq k}{\ell\in \Jab }} \wb(q_\ell) +\sum_{\underset{{\sigma}(\ell) \leq k}{\ell\in \Jaf }}\wf(k_\ell).
\end{align*}
We need to prove that $C_{k}(\sigma) +R_k= D_{j_0}^{(\preceq^{(k)},\preceq_*^{(k)})}(\sigma^{(k)})$. Since $\tilde{\sigma}^{(k)}(j_0) = k$, we have $R_k = R_{\tilde{\sigma}^{(k)}(j_0)}$ and 
\begin{equation*}
   \sum_{\underset{\ell > k}{\ell\in  \Jcb }} \wb(q_\ell)  = \sum_{\underset{\ell > \tilde{\sigma}^{(k)}(j_0)}{\ell\in  \Jcb }} \wb(q_\ell) \quad \textup{and}\quad 
   \sum_{\underset{\ell \leq k}{\ell\in \Jab }} \wb(q_\ell)  = \sum_{\underset{\ell \leq \tilde{\sigma}^{(k)}(j_0) }{\ell\in \Jab }} \wb(q_\ell).
\end{equation*}
Moreover, for any $\ell\in \Jcf\setminus \psigma$ such that ${\sigma}(\ell)>k$ and $\ell\neq j_0$,  we have $\tilde{\sigma}(\ell)=\sigma(\ell)>k$ and therefore $j_0 \prec^{(k)}_* \ell$. Consequently, because $\sigma(\ell) = \tilde{\sigma}(\ell)$ for any $\ell\in\Jcf$, it is easy to deduce that:
\[
\sum_{\underset{{\sigma}(\ell)>k}{\ell\in  \Jcf }} \wf(k_\ell) = \sum_{\underset{j_0 \preceq^{(k)}_* \ell}{\ell\in  \Jcf\setminus \psigma}} \wf(k_\ell).
\]
Finally, for any $l\in \Jaf\setminus \psigma$ such that ${\sigma}(\ell) \leq k $ we have $\tilde{\sigma}(\ell) < k $ implying that $\ell \prec^{(k)} j_0$. Moreover, for any $\ell\neq j_0$ such that $\tilde{\sigma}(\ell)\geq k$, we know, by definition, that  $j_0 \prec^{(k)} \ell $ implying that: 
\[
\sum_{\underset{\ell\in \Jaf }{{\sigma}(\ell) \leq k}}\wf(k_\ell) = \sum_{\underset{\ell\in \Jaf }{\tilde{\sigma}(\ell) < k}}\wf(k_\ell)= \sum_{\underset{\ell\in \Jaf\setminus \psigma }{\ell \preceq^{(k)} j_0}}\wf(k_\ell).
\]
As a conclusion
\[
C_{k}(\sigma) + R_k= D_{j_0}^{(\preceq^{(k)},\preceq_*^{(k)})}(\sigma^{(k)}).
\]
All together, these remarks prove that \eqref{Changingorder1} can be written as an ordered Wick monomial $T^{(k)} = T^{(k)}_{(P_\cf^{(0)},P_\af^{(0)})}(z;\uF^{(n)})$. 

Finally, let $k\in \llbracket 1, n-1\rrbracket$ and $k'\in\scrB\setminus \psigma$ with $\tilde{\sigma}(k') = k+1$ and $k'\preceq j_0$. Recall the abbreviations \eqref{abbrev-LDelta2}. The $k,k'$-contribution, $T^{(k,k')}$,  to \eqref{telescopicsum} coming from the second term on the right-hand side of \eqref{ResDiff} is: 
 {\allowdisplaybreaks
\begin{align}\label{ChanginOrder2}
\nonumber & T^{(k,k')}(z;\uF^{(n)})\\
\nonumber   &\quad =-\int \prod_{i \in \llbracket 1, n \rrbracket \setminus \psigma} F_i(k_i, q_i) \scrL\, \Delta_\ab\Delta_\af \prod_{i \in \scrB\setminus \psigma}\wf(k_i)     \prod_{i \in \Jcf \cap \psigma} \cf\bigl(F(., q_i)\bigr) \\
 \nonumber &  \qquad \prod_{j \in \Jcb} \cb(q_j) \prod_{j \in \Jcf\setminus \psigma} \cf(k_j)  \prod_{\underset{i<k}{i \in \scrA}} R_0\bigl(z-C_{i}(\sigma) - R_i\bigr) \\
\nonumber  & \qquad \prod_{\underset{\tilde{\sigma}(i) < k+1}{i \in \scrB\setminus \psigma}} R_0\bigl(z-D_i^{(\preceq, \preceq_*)}(\sigma)\bigr)    \Bigl(\myTheta_\scrA(k) R_0\bigl(z-C_{k}(\sigma) - R_k\bigr)+ \myTheta_{\scrA^\comp}(k)\Bigr)\\
\nonumber  &  \qquad  \prod_{\underset{i \prec k'}{\underset{ \tilde{\sigma}(i) = k+1}{i \in \scrB\setminus \psigma}}} R_0\bigl(z-D_i^{(\preceq, \preceq_*)}(\sigma)\bigr)\Biggl[\wf(k_{j_0})  R_0\bigl(z-D_{k'}^{(\preceq,\preceq_*)}(\sigma)\bigr)
 \\
\nonumber  & \qquad
R_0\bigl(z-D_{k'}^{(\preceq^{(0)},\preceq_*^{(0)})}(\sigma^{(0)})\bigr) \Biggr] \prod_{\underset{k' \prec i }{\underset{\tilde{\sigma}(i) = k+1}{i \in \scrB\setminus \psigma}}} R_0\bigl(z-D_i^{(\preceq^{(0)},\preceq_*^{(0)})}(\sigma^{(0)})\bigr) 
\\
\nonumber  & \qquad
\prod_{\underset{i>k}{i \in \scrA}} R_0\bigl(z-C_{i}(\sigma^{(0)}) - R_i\bigr) \prod_{\underset{\tilde{\sigma}(i) > k+1}{i \in \scrB\setminus \psigma}} R_0\bigl(z-D_i^{(\preceq^{(0)},\preceq_*^{(0)})}(\sigma^{(0)})\bigr)\\
 &\qquad \prod_{j \in \Jaf \setminus \psigma} \af(k_j) \prod_{j \in \Jab}\ab(q_j) \prod_{i \in \Jaf \cap \psigma} \af\bigl(\overline{F(., q_i)}\bigr)\prod_{i \in \llbracket 1, n \rrbracket \setminus \psigma}  dk_i\prod_{i \in \llbracket 1, n \rrbracket } dq_i.
\end{align}
}
Let us define the relations $\preceq^{(k,k')}$ and $\preceq^{(k,k')}_*$ as follows, recalling that we have $\psigma^{(k+1)}=\psigma$. If $k'\in \Jcf$, we insert $j_0$ just before $k'$: 
\begin{equation}
\label{Case1}
\forall a,b \in \bigl(\Jcf \cup \Jaf\bigr) \setminus \psigma^{(k+1)}:\qquad 
\begin{cases}
    a \preceq^{(k,k')}_{\sharp} b,  & \textup{ if } a,b \neq j_0 \textup{ and } a \preceq_{\sharp} b \\
     a \preceq^{(k,k')}_{\sharp} j_0,  & \textup{ if }  a \prec_{\sharp} k' \\
      j_0 \preceq^{(k,k')}_{\sharp} b,  & \textup{ if }   k' \preceq_{\sharp} b.
\end{cases}
\end{equation}
If on the other hand $k'\in J_\af$, then we insert $j_0$ just after $k'$:
\begin{equation}
\label{Case2}
\forall a,b \in \bigl(\Jcf \cup \Jaf\bigr) \setminus \psigma^{(k+1)}:\qquad 
\begin{cases}
    a \preceq^{(k,k')}_{\sharp} b,  & \textup{ if } a,b \neq j_0 \textup{ and } a \preceq_{\sharp} b \\
     a \preceq^{(k,k')}_{\sharp} j_0,  & \textup{ if }  a \preceq_{\sharp} k' \\
      j_0 \preceq^{(k,k')}_{\sharp} b,  & \textup{ if }   k' \prec_{\sharp} b.
\end{cases}
\end{equation}

In either case
 \begin{equation*}
     \forall i,l \in \bigl(\Jcf\cup \Jaf\bigr) \setminus \psigma^{(k+1)}: \qquad \tilde{\sigma}^{(k+1)}(i) < \tilde{\sigma}^{(k+1)}(l) \quad  \Rightarrow \quad  i\prec^{(k,k')}_{\sharp} l.
 \end{equation*}

The strategy is again to rewrite \eqref{ChanginOrder2} as an ordered Wick monomial using $\sigma^{(k+1)}$ and $\preceq^{(k,k')}$. First, 
\begin{align}
    \forall i \in \scrA, i \leq k:\qquad & C_i(\sigma) = C_i(\sigma^{(k+1)}) \\
    \forall i \in \scrA, i> k:\qquad & C_i(\sigma^{(0)}) = C_i(\sigma^{(k+1)}).
\end{align}
Using the same type of argument as before, one can also check that 
\begin{align*}
    \forall i \in \scrB\setminus \psigma, i \prec k':\qquad & D_i^{(\preceq, \preceq_*)}(\sigma)  = D_i^{(\preceq^{(k,k')},\preceq_*^{(k,k')})}(\sigma^{(k+1)})\\
    \forall i \in \scrB\setminus \psigma, k' \preceq i:\qquad & D_i^{(\preceq^{(0)},\preceq_*^{(0)})}(\sigma^{(0)}) = D_i^{(\preceq^{(k,k')},\preceq_*^{(k,k')})}(\sigma^{(k+1)}).
\end{align*}
Hence, it remains to prove that 
$ D_{k'}^{(\preceq,\preceq_*)}(\sigma) = D_{j_0}^{(\preceq^{(k,k')},\preceq_*^{(k,k')})}(\sigma^{(k+1)})$. Let us note that 
\[
\begin{aligned}
D_{k'}^{(\preceq,\preceq_*)}(\sigma) =&  \sum_{\underset{\ell\in  \Jcb }{\ell > \tilde{\sigma}(k')}} \wb(q_\ell) + \sum_{\underset{\ell\in \Jcf\setminus \psigma}{{k'\preceq_* \ell}}}\wf(k_{\ell})\\
& +  \sum_{\underset{\ell\in \Jab  }{\ell \leq \tilde{\sigma}(k')}} \wb(q_\ell) +\sum_{\underset{\ell\in \Jaf\setminus \psigma }{{\ell \preceq k'}}}\wf(k_\ell) +R_{\tilde{\sigma}(k')}.
\end{aligned}
\]
From the fact that $\tilde{\sigma}(k') = k+1=\tilde{\sigma}^{(k+1)}(j_0)$, 
we have $R_{\tilde{\sigma}(k')} = R_{\tilde{\sigma}^{(k+1)}(j_0)}$,
\begin{equation*}
    \sum_{\underset{\ell\in  \Jcb }{\ell > \tilde{\sigma}(k')}} \wb(q_\ell)  = \sum_{\underset{\ell\in  \Jcb }{\ell > \tilde{\sigma}^{(k+1)}(j_0)}} \wb(q_\ell)
\end{equation*}
and 
\begin{equation*}
    \sum_{\underset{\ell\in \Jab  }{\ell \leq \tilde{\sigma}(k')}} \wb(q_\ell) =  \sum_{\underset{\ell\in \Jab  }{\ell \leq \tilde{\sigma}^{(k+1)}(j_0)}} \wb(q_\ell).
\end{equation*}
If $k'\in \Jcf$, the order $\preceq^{(k,k')}_*$ is defined by \eqref{Case1}. Then $k' \preceq_* l$ is equivalent to the condition $j_0 \preceq^{(k,k')}_* l$ and therefore 
\[
\sum_{\underset{l\in \Jcf\setminus \psigma}{{k'\preceq_* l}}}\wf(k_{l}) = \sum_{\underset{l\in \Jcf\setminus \psigma^{(k+1)}}{{j_0\preceq^{(k,k')}_* l}}}\wf(k_{l}).
\]
Moreover, if $k' \in \Jcf$,
\[
\sum_{\underset{l\in \Jaf\setminus \psigma }{{l \preceq k'}}}\wf(k_l) = \sum_{\underset{l\in \Jaf\setminus \psigma }{{l \prec k'}}}\wf(k_l) 
\]
and $l\prec k'$ implies that $l \preceq^{(k,k')} j_0$. Similarly if $l \preceq^{(k,k')} j_0$ then we claim that $l\preceq k'$. Indeed, if $k' \prec l$ then $j_0 \preceq^{(k,k')} l$ implying that $l =j_0$ contradicting the fact that $l\in \Jaf$. Note that $j_0\prec^{(k,k')} k'$, however, since $k' \notin \Jaf$ we can conclude that
\[
\sum_{\underset{l\in \Jaf\setminus \psigma }{{l \preceq k'}}}\wf(k_l) = \sum_{\underset{l\in \Jaf\setminus \psigma^{(k+1)} }{{l \preceq^{(k,k')} j_0}}}\wf(k_l).
\]
The strategy is similar if $k'\in \Jaf$, defining the orders $\prec^{(k,k')}$ and $\prec^{(k,k')}_*$ by \eqref{Case2}. In any case, we can conclude that $ D_{k'}^{(\preceq,\preceq_*)}(\sigma) = D_{j_0}^{(\preceq^{(k,k')},\preceq_*^{(k,k')})}(\sigma^{(k+1)})$.

In conclusion, the contribution $T^{(k,k')}$ in \eqref{ChanginOrder2}, is an ordered Wick monomial $T^{(k,k')}=T^{(k,k')}_{(P_\cf^{(0)},P_\af^{(0)})}(z;\uF^{(n)})$.

Summing up,
\begin{align*}
T_{(P_\cf,P_\af)}\bigl(z;\uF^{(n)}\bigr)  = & T^{(0)}_{(P_\cf^{(0)},P_\af^{(0)})}\bigl(z;\uF^{(n)}\bigr)+
\sum_{k\in\scrA, k<j_0} T^{(k)}_{(P_\cf^{(0)},P_\af^{(0)})}\bigl(z;\uF^{(n)}\bigr)\\
& +\sum_{k=0}^{n-1} \sum_{\underset{\underset{ \tilde{\sigma}(k')=k+1}{k'\preceq j_0}}{k'\in \scrB\setminus\psigma} }T^{(k,k')}_{(P_\cf^{(0)},P_\af^{(0)})}\bigl(z;\uF^{(n)}\bigr).
\end{align*}
We are done, since there are at most $2n$ terms on the right hand side above. 
\end{proof}

 The following lemma is a straightforward application of 
Lemma~\ref{InitialisationReorderingFermionLemma}. 
 
\begin{lem}
\label{ReorderingFermionLemma}
Let $n\in\NN$. There exists a constant $M=M(n)$, such that the following holds. For any ordered Wick monomial $T_{(P_{\cf},P_{\af})}$ of length $n$, There exist $N \in \NN$ with $N\leq M$ and a collection of ordered Wick monomials 
\[
\{T^{(i)}_{(\Jcf\setminus (\Jcb\cup\Jab),\Jaf\setminus (\Jcb\cup\Jab))}\}_{i=1}^N 
\]
all of length $n$, such that:
\begin{equation}
T_{(P_{\cf},P_{\af})} = \sum^N_{i = 1 } T^{(i)}_{(\Jcf\setminus (\Jcb\cup\Jab),\Jaf\setminus (\Jcb\cup\Jab))}.
\end{equation}
 Finally, the sets $\scrA, \Jcb, \Jab, \Jaf, \Jcf, I_\af, I_\ab$ and the functions $f_\ab,f_\af$ and $\scrL$ are identical for $T_{(P_{\cf},P_{\af})}$ and any $ T^{(i)}_{(\Jcf\setminus (\Jcb\cup\Jab),\Jaf\setminus (\Jcb\cup\Jab))}$.
\end{lem}

\section{Estimating the Renormalized Blocks}\label{Sec-GSetsEst}

The goal of this section is to prove operator norm estimates on the renormalized blocks $T^{(n)}_\us$ from Subsect.~\ref{subsec-RenBlocks} that are uniform in the ultraviolet cutoff.
We do this by establishing operator norm estimates on regular Wick monomials, which suffices by Lemma~\ref{InductionLemma}. However, this we cannot do directly. We first rewrite regular Wick monomials as a sum of ordered Wick monomials, which was done in Lemma~\ref{ReorderingFermionLemma} in the previous section. We will then finally establish the desired estimates, first for regular Wick monomials in Propositions~\ref{RNFCT} and~\ref{RFCT}, and for the renormalized blocks in Propositions~\ref{RegularityofthegeneralisedGsets}, \ref{FinalEstimate1} and \ref{FinalEstimate2}.

\subsection{Preliminary}

Let us first introduce some convenient notations.

\begin{Def}[The functions $\theta$ and $\theta_*$]
\label{FunctionTheta}
Let $n\in\NN$ and consider subsets of $\llbracket 1,n\rrbracket$:  $\Jaf, \Jcf, \Jab, \Jcb, I_\af, I_\ab$, as well as functions $f_{\ab}, f_{\cb}$ as in Definition~\ref{Notation}, together with a cover $(P_{\cb},P_{\ab},P_{\cf},P_{\af})$, cf. Definition \ref{DefCover}, and an admissible map $\sigma$, cf. Definition \ref{Admissiblemaps}. We define two functions  $\theta_* \colon \Jcf\setminus\psigma \to  \llbracket 1, n \rrbracket$ and $\theta \colon  \Jaf\setminus\psigma \to  \llbracket 1, n \rrbracket$ by setting
\begin{equation}
\label{thetestar}
    \theta_*(i) = \begin{cases}
      f_{\ab}(i), & \textup{ if } i \in P_{\cf} \cap I_\ab\\
      n, &  \textup{ if } i \in P_{\cf} \cap \Jab\\
      \sigma(i), & \textup{ otherwise}
    \end{cases}
\end{equation}
and
\begin{equation}
    \theta(i) = \begin{cases}
      f^{-1}_{\ab}(i), & \textup{ if } i \in P_{\af} \cap f_{\ab}(I_\ab)\\
            1, &  \textup{ if } i \in P_{\af} \cap \Jcb\\
      \sigma(i), & \textup{ otherwise.}
    \end{cases}
\end{equation}
\end{Def}
When we estimate ordered operators in the proof of Proposition~\ref{RNFCT} below, we order the pointwise fermionic creation and annihilation operators using the functions $\theta_*$ and $\theta$ and the total orders $\preceq_*$ and $\preceq$, respectively. We then use kinetic energy bounds to control the pointwise annihilation and creation operators, both the bosonic and the fermionic ones. For this we insert the identity $(H_0-z)^{-1/2}(H_0-z)^{1/2}$ to the right of the creation operators and to the left of the annihilation operators. The unbounded factors $(H_0-z)^{1/2}$ are then pulled through onto the product of resolvents in the middle and controlled by those. The functions $\theta_*$ and $\theta$ in addition helps keep track of the factors of $\wf(k_i)$ and $\wb(q_i)$ that are picked up via pull-through in the process.

\begin{Def}
\label{SumsinRes2}
Let us consider $\Jaf, \Jcf, \Jab, \Jcb$  defined as in Definition~\ref{Notation}, a cover $(P_{\ab}, P_{\af}, P_{\cb}, P_{\cf})$ defined as in Definition \ref{DefCover}, together with an admissible map $\sigma$. Moreover, let $\theta_*\colon \Jcf \setminus \psigma\to \llbracket 1, n \rrbracket$ and $\theta\colon \Jaf \setminus \psigma \to  \llbracket 1, n \rrbracket$ be two functions. We abbreviate 
\begin{equation}\label{ABis}
\begin{aligned}
&\forall i\in \llbracket 1,n\rrbracket:&      A_i(\theta_*) & =  \sum_{\underset{j \geq i}{j\in \Jcb }} \wb(q_j)+\sum_{\underset{\theta_*(j) \geq i}{j \in 
    \Jcf\setminus \psigma }} \wf(k_{j}),\\
&\forall i\in \llbracket 1,n\rrbracket:&   B_i(\theta) & =   \sum_{\underset{j \leq i}{j\in \Jab }} \wb(q_j) +\sum_{\underset{\theta(j) \leq i  }{j \in \Jaf\setminus \psigma }} \wf(k_j).
\end{aligned}
\end{equation}
\end{Def}

The following result is an important first step towards estimating Wick monomials. 

\begin{lem}
\label{RegOpEstimates}
Let $T$ be a regular Wick monomial with associated index sets $\scrA,J_\ab, J_\cb, J_\af, J_\cf,I_\ab,I_\af\subset \llbracket 1,n\rrbracket$.  Define a cover (cf. Definition~\ref{DefCover}) by setting  $P_{\ab} = \Jab$, $P_{\cb} = \Jcb$, $P_{\af} = J_{\af}\setminus (J_{\ab}\cup J_{\cb})$ and $P_{\cf} = J_{\cf}\setminus (J_{\ab}\cup J_{\cb})$. Let $\sigma$ be an admissible map (cf. Definition~\ref{Admissiblemaps})
and $\{\alpha_i\}_{i=1}^n$ a family of non-negative numbers satisfying \eqref{Caraalpha}.
Then there exists a family of exponents $\{\beta_{i}\}_{i=1}^n$ with $0\leq \beta_i\leq \alpha_i$ and $\beta_i=0$ for $i\in\psigma$, such that for all $z\in\CC_-$:
\begin{equation}
\begin{aligned}
\label{regularitypropertyV2}
& \Bigl|\scrL\bigl( \{ k_j\}_{j \in I_\af}, \{ q_j\}_{j \in I_\ab}\bigr)\Bigr|  \prod_{i \in P_{\ab}\cup P_{\cb}} \wb(q_i)^{-\alpha_i} \prod_{ i \in P_{\af}\cup P_{\cf}} \wf(k_j)^{-\alpha_i} \\
& \qquad \prod_{i\in I_\ab}\kdelta\bigl(q_{i} - q_{f_\ab(i)}\bigr) \prod_{j\in I_\af}\kdelta\bigl(k_{j} - k_{f_\af(j)}\bigr)\biggl\| \prod_{i \in P_{\cb}}\bigl(H_0-z+ A_i(\theta_*)  \bigr)^{\alpha_i}\\
& \qquad \quad  \prod_{i \in  P_{\cf}}\bigl(H_0-z+ A_{\theta_*(i)}(\theta_*)  \bigr)^{\alpha_i} \prod_{i\in \scrA} R_0\bigl(z-C_i(\sigma) -R_i\bigr)
\\
& \qquad  \quad  \prod_{i \in P_{\ab} }\bigl(H_0-z+B_i(\theta) \bigr)^{\alpha_i} \prod_{i \in P_{\af} }\bigl(H_0-z+B_{\theta(i)}(\theta)\bigr)^{\alpha_i}   \biggr\| \\  
&\quad \leq  c_T\prod_{i\in (P_\cf\cap I_\ab)\cup (P_\af\cap f_\ab(I_\ab))}\biggl(1+\frac{\wf(k_i)}{ \wb(q_i) }\biggr)^{n-1} \\
& \qquad \quad \frac{ \prod_{i\in I_\ab}\kdelta\bigl(q_{i} - q_{f_\ab(i)}\bigr) \prod_{j\in I_\af}\kdelta\bigl(k_{j} - k_{f_\af(j)}\bigr)}{\prod^n_{i=1}[\wb(q_i)]^{\alpha_i-\beta_{i}}[\wf(k_i)]^{\beta_{i}}}.
\end{aligned}
\end{equation}
\end{lem}

\begin{proof} Let $z\in\CC_-$ and
recall the notation $\scrJ= J_\ab\cup J_\cb\cup J_\af\cup J_\cf$. From Definition~\ref{RegularOperator}, we get exponents 
 $\{\beta_{i}\}_{i=1}^n$, as in the formulation of the lemma, as well as admissible exponents  
 $\{\gamma_{i;j}\}_{i\in \scrA, j\in \scrJ}$ obeying Definition~\ref{def-admexp}.

Let $i\in \scrA$. We aim to estimate 
\begin{align}\label{ReducToRegEstim}
&  \biggl\| \prod_{j \in P_{\cb}}\bigl(H_0-z+ A_j(\theta_*) \bigr)^{\gamma_{i;j}} \prod_{j \in  P_{\cf}}\bigl(H_0-z+ A_{\theta_*(j)}(\theta_*) \bigr)^{\gamma_{i;j}} 
\\
\nonumber & \quad R_0\bigl(z-C_i(\sigma) -R_i\bigr)  \prod_{j \in P_{\ab} }\bigl(H_0-z+B_j(\theta)\bigr)^{\gamma_{i;j}} \prod_{j \in P_{\af} }\bigl(H_0-z+B_{\theta(j)}(\theta)\bigr)^{\gamma_{i;j}}   \biggr\|.
\end{align}

Let now $j \in P_{\cb}$ and observe, recalling the first constraint in \eqref{ConstraintOnGamma}, that if $j\leq i$ then $\gamma_{i;j} = 0$. Therefore, for the purpose of estimating \eqref{ReducToRegEstim}, one can assume that $j>i$. To sum up, we have $i\in\scrA$ and $j\in P_\cb$ with $j>i$. Split the second sum in the definition \eqref{ABis} of $A_j(\theta_*)$ as follows
\begin{equation}\label{firstineqtoreg}
 A_j(\theta_*)  =  \sum_{\underset{\ell \geq j}{\ell\in \Jcb }} \wb(q_\ell)+\sum_{\underset{\sigma(\ell)> i}{\underset{\theta_*(\ell) \geq j}{\ell \in 
    \Jcf\setminus \psigma }}} \wf(k_{\ell}) + \sum_{\underset{\sigma(\ell)\leq i}{\underset{\theta_*(\ell) \geq j}{\ell \in 
    \Jcf\setminus\psigma }}} \wf(k_{\ell}) .
\end{equation}
As for the first two sums on the right-hand side of \eqref{firstineqtoreg}, let $\ell\in J_\cb$ with $\ell\geq j$. Since  $\ell\geq j  > i$, then the first two sums on the right-hand side of \eqref{firstineqtoreg} are controlled by $C_i(\sigma)$:
\begin{equation}\label{estbyCi}
\sum_{\underset{\ell \geq j}{\ell\in \Jcb }} \wb(q_l)+\sum_{\underset{\sigma(\ell)> i}{\underset{\theta_*(\ell) \geq j}{\ell \in 
    \Jcf\setminus \psigma }}} \wf(k_{\ell}) \leq C_i(\sigma).
\end{equation}
See Definition~\ref{SumsinRes} for the form of $C_i(\sigma)$. 

As far as the third and final sum on the right-hand side of \eqref{firstineqtoreg} is concerned, let $\ell\in \Jcf \setminus \psigma$ with $\sigma(\ell)\leq i$ and $\theta_*(\ell)\geq j$. Therefore $\sigma(\ell) \leq i < j\leq \theta_*(\ell)$ and in particular $\theta_*(\ell) > \sigma(\ell)$. From the fact that $P_{\cf}\cap \Jab = \emptyset$,  together with Definition~\ref{FunctionTheta}, we can consequently conclude that $\ell \in P_{\cf} \cap I_\ab$. Since this implies that $\sigma(\ell)=\ell$, we deduce $\ell \leq i < \theta_*(\ell) = f_{\ab}(\ell)$. Hence
\[
\sum_{\underset{\sigma(\ell)\leq i}{\underset{\theta_*(\ell) \geq j}{\ell \in 
    \Jcf \setminus \psigma }}} \wf(k_{\ell})
    \leq \biggl(\sum_{\underset{\ell\leq i}{\underset{f_\ab(\ell) \geq j}{\ell \in 
    P_\cf\cap I_\ab }}} \frac{\wf(k_{\ell})}{\wb(q_\ell)}\biggr)  R_i, 
\]
where $R_i$ is defined in \eqref{Rcomponents}.
Inserting this estimate together with \eqref{estbyCi}, into \eqref{firstineqtoreg} we obtain
\[
A_j(\theta_*) \leq \prod_{\underset{\ell\leq i}{\underset{f_\ab(\ell) \geq j}{\ell \in 
    P_\cf\cap I_\ab }}}\biggl( 1+ \frac{\wf(k_{\ell})}{\wb(q_\ell)}\biggr)\bigl(C_i(\sigma) +R_i\bigr).
\]
This implies that
\begin{equation*}
 \Bigl\| \bigl(H_0-z+ A_j(\theta_*) \bigr)^{\gamma_{i;j}} R_0\bigl(z-C_i(\sigma) -R_i\bigr)^{\gamma_{i;j}}   \Bigr\|
     \leq \prod_{\underset{\ell\leq i}{\underset{f_\ab(\ell) \geq j}{\ell \in  P_\cf\cap I_\ab}}}
     \biggl(1+\frac{\wf(k_{\ell})}{\wb(q_\ell)}\biggr)^{\gamma_{i;j}}.
\end{equation*}

Let us now consider the case where $j\in P_{\ab}$, recalling this time the third constraint in \eqref{ConstraintOnGamma}. Since $\gamma_{i;j}=0$ if $i<j$, we may assume that
$j \leq i$.
We may now conclude similarly to above that 
\begin{align*}
B_j(\sigma) &=  \sum_{\underset{\ell \leq j}{\ell\in \Jab }} \wb(q_\ell) +\sum_{\underset{\sigma(\ell)\leq i}{\underset{\theta(\ell) \leq j  }{\ell \in \Jaf\setminus \psigma }}} \wf(k_\ell)  + \sum_{\underset{\sigma(\ell) > i}{\underset{\theta(\ell) \leq j  }{\ell \in \Jaf\setminus \psigma }}} \wf(k_\ell) . \\
& \leq C_i(\sigma) + \sum_{\underset{\ell>i}{\underset{f_\ab^{-1}(\ell) \leq j}{\ell \in 
    P_\af\cap f_\ab(I_\ab) }}} \wf(k_{\ell})\\
    & 
    \leq \prod_{\underset{\ell>i}{\underset{f_\ab^{-1}(\ell) \leq j}{\ell \in 
    P_\af\cap f_\ab(I_\ab)}}} \biggl( 1+\frac{\wf(k_{\ell})}{\wb(q_\ell)}\biggr)
    \biggl(C_i(\sigma)+  \sum_{\underset{f_\ab^{-1}(\ell) \leq j < \ell}{\ell \in 
    f_\ab(I_\ab) }} \wb(q_{\ell})  \biggr),
\end{align*}
where we used that $P_\af\cap J_\cb = \emptyset$ and that for $\ell\in J_\af$ with  $\theta(\ell)\leq j \leq i < \sigma(\ell)$, we have $\ell\in P_\af \cap f_\ab(I_\ab)$ and hence, $\sigma(\ell)=\ell$ and $\theta(\ell) = f_\ab^{-1}(\ell)$. For the properties of $\theta$, see Definition~\ref{FunctionTheta}, and
for the form of $R_i$, see \eqref{Rcomponents}. Keeping in mind that boson momenta $q_\ell$ and $q_{f_\ab^{-1}(\ell)}$, for $\ell\in f_\ab(I_a)$, are identified by delta-functions in the estimate \eqref{regularitypropertyV2} from the lemma, we conclude - with some abuse of notation - that
\begin{equation*}
 \Bigl\| R_0\bigl(z-C_i(\sigma) -R_i\bigr)^{\gamma_{i;j}}  \bigl(H_0-z+ B_j(\theta) \bigr)^{\gamma_{i;j}}  \Bigr\|
     \leq \prod_{\underset{\ell>i}{\underset{f_\ab^{-1}(\ell) \leq j}{\ell \in 
    P_\af\cap f_\ab(I_\ab) }}} \biggl(1+\frac{\wf(k_{\ell})}{\wb(q_\ell)}\biggr)^{\gamma_{i;j}},
\end{equation*}
under the constraint that $q_\ell = q_{f_\ab^{-1}(\ell)}$, for $\ell\in f_\ab(I_\ab)$.

Now let $j\in P_\cf$. We may assume $\gamma_{i,j}\neq 0$, and therefore $\theta_*(j)>i$. Too see this, compare Definitions~\ref{def-admexp} and~\ref{FunctionTheta} and recall that $\sigma(j)=j$ for $j\in P_\cf$.
Write
\[
\begin{aligned}
A_{\theta_*(j)}(\theta_*)&  = \sum_{\underset{\ell \geq \theta_*(j)}{\ell\in \Jcb }} \wb(q_\ell)+\sum_{\underset{\sigma(\ell)> i}{\underset{\theta_*(\ell) \geq \theta_*(j)}{\ell \in 
    \Jcf\setminus \psigma }}} \wf(k_{\ell}) + \sum_{\underset{\sigma(\ell)\leq i}{\underset{\theta_*(\ell) \geq \theta_*(j)}{\ell \in 
    \Jcf\setminus \psigma }}} \wf(k_{\ell})\\
    & \leq C_i(\sigma)+ \sum_{\underset{\sigma(\ell)\leq i}{\underset{\theta_*(\ell) \geq \theta_*(j)}{\ell \in 
    \Jcf\setminus \psigma }}} \wf(k_{\ell}).
\end{aligned}
\]
 For indices $\ell\in J_\cf\setminus \psigma$ in the last summand, we have $\sigma(\ell) \leq i < \theta_*(j) \leq \theta_*(\ell)$. Recalling Definitions~\ref{FunctionTheta} and that $P_{\cf}\cap \Jab = \emptyset$, we have $\ell\in P_\cf\cap I_\af$, $\sigma(\ell)=\ell$ and $\theta_*(\ell) = f_\ab(\ell)$. We thus get
\[
A_{\theta_*(j)}(\sigma) \leq C_i(\sigma) + \sum_{\underset{\ell\leq i}{\underset{f_\ab(\ell) \geq \theta_*(j)}{\ell \in 
    P_\cf\cap I_\ab }}} \wf(k_{\ell})\leq \prod_{\underset{\ell\leq i}{\underset{f_\ab(\ell) \geq \theta_*(j)}{\ell \in 
    P_\cf\cap I_\ab }}} \biggl( 1+\frac{\wf(k_{\ell})}{\wb(q_\ell)}\biggr)\bigl(C_i(\sigma)+R_i\bigr),
\]
where we again used the form \eqref{Rcomponents} of $R_i$.

Finally, for $j\in P_\af$ with $\gamma_{i,j}\neq 0$, we similarly have $\theta(j) \leq i$ and
\[
B_{\theta(j)}(\sigma) \leq  C_i(\sigma) + \sum_{\underset{\ell>i}{\underset{f_\ab^{-1}(\ell) \leq \theta(j)}{\ell \in 
    P_\af\cap f_\ab(I_\ab) }}} \wf(k_\ell)
\leq 
\prod_{\underset{\ell>i}{\underset{f_\ab^{-1}(\ell) \leq \theta(j)}{\ell \in 
    P_\af\cap f_\ab(I_\ab) }}} \biggl( 1+\frac{\wf(k_{\ell})}{\wb(q_\ell)}\biggr) \bigl(C_i(\sigma)+ R_i\bigr),
\]
under identification of $q_\ell$ with $q_{f_\ab^{-1}(\ell)}$, for $\ell\in f_\ab(I_\ab)$, when estimating by $R_i$ in the last inequality. See also the estimate of $B_j(\sigma)$ above.

Taken together, and recalling that $\bar{\gamma}_i\leq 1$, we may estimate \eqref{ReducToRegEstim} for $z\in\CC_-$
\begin{align*}
&  \biggl\| \prod_{j \in P_{\cb}}\bigl(H_0-z+ A_j(\theta_*) \bigr)^{\gamma_{i;j}} \prod_{j \in  P_{\cf}}\bigl(H_0-z+ A_{\theta_*(j)}(\theta_*) \bigr)^{\gamma_{i;j}} 
 R_0\bigl(z-C_i(\sigma) -R_i\bigr)\\
& \qquad  \prod_{j \in P_{\ab} }\bigl(H_0-z+B_j(\theta)\bigr)^{\gamma_{i;j}} \prod_{j \in P_{\af} }\bigl(H_0-z+B_{\theta(j)}(\theta)\bigr)^{\gamma_{i;j}}   \biggr\|\\
&\quad \leq 
\prod_{\underset{f_\ab(\ell) > i \geq \ell}{\ell \in 
    P_\cf\cap I_\ab }} \biggl(1+\frac{\wf(k_{\ell})}{\wb(q_\ell)}\biggr)^{\sum_{j\in P_\cb\cup P_\cf}\gamma_{i;j}} \\
  & \qquad \prod_{\underset{f_\ab^{-1}(\ell) \leq i < \ell}{\ell \in 
    P_\af\cap f_\ab(I_\ab) }} \biggl(1+\frac{\wf(k_{\ell})}{\wb(q_\ell)}\biggr)^{\sum_{j\in P_\ab\cup P_\af}\gamma_{i;j}}  \Bigl\| R_0\bigl(z-C_i(\sigma) -R_i\bigr)^{1-\bar{\gamma}_i}\Bigr\|\\
     &\quad \leq \bigg\{ \prod_{\ell\in (P_\cf\cap I_\ab)\cup (P_\af\cap f_\ab(I_\ab))}\biggl(1+\frac{\wf(k_{\ell})}{\wb(q_\ell)}\biggr)^{\bar{\gamma}_{i}}\bigg\}
     \Bigl\| R_0\bigl(z-C_i(\sigma) -R_i\bigr)^{1-\bar{\gamma}_i}\Bigr\|,
\end{align*}
under the identification of $q_l$ with $q_{f_\ab^{-1}(l)}$, for $l\in f_\ab(I_\ab)$, coming from the delta-functions in \eqref{regularitypropertyV2}. Note that
$\sum_{i\in \scrA} \bar{\gamma}_i = \sum_{i\in\scrA,j\in\scrJ} \gamma_{i;j} =  \sum_{j\in\scrJ} \alpha_j \leq n-1$.
Therefore, taking the product over $i\in \scrA$, the lemma follows from \eqref{regularityproperty}.  
\end{proof}

\subsection{Estimates for regular Wick monomials}\label{subsec-estregWM}

Let $\uF^{(n)}=(F_1, \dotsc, F_n) \in (L^2(\RR^d\times \RR^d))^n$. We introduce the notation: 
\begin{align}\label{tildeK}
\nonumber& \tK_{\alpha}\bigl(\uF^{(n)}\bigr) = \\
 &\quad  \biggl( \prod^n_{i=1} \int  \biggl(1+\frac{\wf(k_i)}{ \wb(q_i) }\biggr)^{2 n } \biggl(1+\frac{\wb(q_i)}{ \wf(k_i) }\biggr)^{2\alpha } \frac{|F_i(k_i,q_i)|^2}{\wb(q_i)^{2\alpha}  }  dq_i dk_i\biggr)^\frac12,
\end{align}
where $\alpha$ and 
$\{\beta_{i}\}_{i=1}^n$ are real numbers
with $0\leq \beta_i\leq \alpha$, for all $i\in\llbracket 1,n\rrbracket$.

\begin{Prop}[Estimates of non fully contracted Wick monomials]
\label{RNFCT} Let $n\in \NN$. There exists a constant $c_n$, such that the following holds true:
\begin{enumerate}[label = \textup{(\arabic*)}]
\item\label{item-regbound1} Assume that  $\leftT$ is a left-handed Wick monomial of length $n$ and bounding constant $c_T$. If $n \geq 2$, respectively if $n=1$, then for any $0\leq \delta\leq 1$, respectively  $\frac{1}{2}\leq \delta\leq 1$, for any $z\in\CC_-$ and $\uF^{(n)}\in (L^2(\RR^d\times\RR^d))^n$, we have
\begin{equation}\Bigl\|R_0(z)^{\delta}\leftT\bigl(z,\uF^{(n)}\bigr) \Bigr\| \leq
c_n c_T \tK_{\alpha}\bigl(\uF^{(n)}\bigr),
\end{equation}
with $\alpha = 1-\frac{1-\delta}{n}$.
\item\label{item-regbound2} Assume that,  $\rightT$ is a right-handed Wick monomial of length $n$ and bounding constant $c_T$. If $n \geq 2$, respectively if $n=1$, then for any $0\leq \gamma\leq 1$, respectively  $\frac{1}{2}\leq \gamma\leq 1$, for any $z\in\CC_-$ and $\uF^{(n)}\in (L^2(\RR^d\times\RR^d))^n$, we have
\begin{equation}
\Bigl\|\rightT\bigl(z;\uF^{(n)}\bigr) R_0(z)^{\gamma} \Bigr\| \leq c_n c_T
\tK_{\alpha}\bigl(\uF^{(n)}\bigr),
\end{equation}
with $\alpha = 1-\frac{1-\gamma}{n}$.
\item\label{item-regbound3} Assume that,  $\lrT$ is both a left- and a right-handed Wick monomial of length $n$ and bounding constant $c_T$. Then $n \geq 2$ and for any $\gamma,\delta\geq 0$ with $\gamma+\delta\leq 1$, 
for any $z\in\CC_-$ and $\uF^{(n)}\in (L^2(\RR^d\times\RR^d))^n$, we have
\begin{equation}
\Bigl\|R_0(z)^{\delta}\lrT\bigl(z;\uF^{(n)}\bigr) R_0(z)^{\gamma} \Bigr\| \leq c_n c_T 
K_{ \alpha}\bigl(\uF^{(n)}\bigr),
\end{equation}
with $\alpha=1-\frac{1-\gamma-\delta}{n}$.
\end{enumerate}
\end{Prop}

\begin{proof} It suffices to prove the estimates for $z\in\CC_-^*$. The estimates will then extend to $z=0$, cf.~Remark~\ref{rem-Wick}. Hence, in the following we consider only $z\in\CC_-^*$.
Let us start with item \ref{item-regbound2}. A right-handed Wick monomial $\rightT$ is a regular Wick monomial, and therefore an ordered Wick monomial of the form $T_{(\Jcf,\Jaf)}$. Here we use $\sigma=\mathrm{id}$ and the cover $P_\cf = J_\cf$, $P_\af=J_\af$ and hence $P_\cb = J_\cb\setminus (J_\cf\cup J_\af)$ and $P_\ab = J_\ab\setminus (J_\cf\cup J_\af)$. Lemma~\ref{ReorderingFermionLemma} implies that there exist $M(n) \in \NN$ and a collection of ordered Wick monomials $\{T^{(i)}_{(\Jcf\setminus ( \Jcb \cup \Jab),\Jaf\setminus ( \Jcb \cup \Jab ))}\}_{i=1}^N$ with $N\leq M(n)$, such that
\begin{equation}\label{eq-ooT-expansion}
\rightT = \sum^N_{i = 1}T^{(i)}_{(\Jcf\setminus ( \Jcb \cup \Jab ),\Jaf\setminus ( \Jcb \cup \Jab ))}.
\end{equation}
The sets $\scrA, \Jcb, \Jab, \Jaf, \Jcf, I_\af, I_\ab$ and the functions $f_\ab,f_\af$ and $\scrL$ are identical for all the $T^{(i)}_{(\Jcf\setminus ( \Jcb \cup \Jab ),\Jaf\setminus ( \Jcb \cup \Jab ))}$'s. What depends on $i$ are the admissible maps $\sigma^{(i)}$ and the total orders $\preceq^{(i)}$ and $\preceq_*^{(i)}$.

The formula \eqref{eq-ooT-expansion} implies the following cover that we will use in the rest of the proof:
\begin{equation}\label{FinalCover5}
P_\cf = J_\cf\setminus (J_\cb\cup J_\ab), \quad P_\af = J_\af \setminus (J_\cb\cup J_\ab), \quad P_\cb = J_\cb \quad \textup{and}\quad P_\ab = J_\ab.
\end{equation}
With this cover, we have 
\[
\scrB = (J_\cf\cup J_\af)\setminus (P_\cf\cup P_\af) = (J_\cf\cup J_\af)\cap (J_\cb\cup J_\ab).
\]
We will now fix an $i\in\{1,2,\dotsc,N\}$ and estimate the contribution from $T^{(i)}_{(\Jcf\setminus ( \Jcb \cup \Jab ),\Jaf\setminus ( \Jcb \cup \Jab ))}$.
To lighten the notation, we drop the superscript $(i)$ from the notation for the admissible map $\sigma^{(i)}$ and the total orders $\preceq^{(i)}$ and $\preceq_*^{(i)}$. The operator $T^{(i)}_{(P_\cf,P_\af)}(z;F_1,\dots F_n)$, is of the form:
\begin{equation}
\label{ToEstimate}
\begin{aligned}
  &   \int \prod_{i \in \llbracket 1, n \rrbracket \setminus \psigma} F_i(k_i, q_i) \scrL\, \Delta_\ab\Delta_\af \prod_{i \in \scrB\setminus \psigma}\wf(k_i)  \prod_{i \in \Jcf \cap \psigma} \cf\bigl(F_i(., q_i)\bigr)  \\
 & \qquad \prod_{j \in \Jcb} \cb(q_j) \prod_{j \in \Jcf\setminus \psigma} \cf(k_j)  \prod_{i \in \scrA} R_0\bigl(z-C_{i}(\sigma) - R_i\bigr)  \\
  &\qquad  \prod_{i \in \scrB\setminus \psigma} R_0\bigl(z-D_i(\sigma)\bigr)   \prod_{j \in \Jaf \setminus \psigma} \af(k_j) \prod_{j \in \Jab}\ab(q_j)\\
  & \qquad  \prod_{i \in \Jaf \cap \psigma} \af\bigl(\overline{F_i(., q_i)}\bigr)\prod_{i \in \llbracket 1, n \rrbracket } dq_i\prod_{i \in \llbracket 1, n \rrbracket \setminus \psigma} dk_i,
\end{aligned}
\end{equation}
where we have recycled the abbreviations $\scrL$, $\Delta_\ab$ and $\Delta_\af$ from \eqref{abbrev-LDelta2} as well as introduced the abbreviation $D_i(\sigma) = D_i^{(\preceq,\preceq_*)}(\sigma)$ that we will be employing throughout this proof.

The strategy is now to reorder the creation and annihilation terms in \eqref{ToEstimate} so that it can be estimated by using the pull-through formula and the regularity conditions arising from the definitions of ordered and regular Wick monomials.  
 
 First, since $\rightT$ was assumed to be a right-handed Wick monomial (cf. Definition~\ref{def-handedWick}~\ref{item-RHWick}), we know that at least one of the sets $J_\ab$ and $P_\af\subseteq J_\af\setminus\psigma$ is not empty. Let $j_1 = \max(\Jab)$ and $j_2= \max_{\preceq}\bigl(\Jaf\setminus \psigma\bigr) $. We define
 \begin{equation}\label{ChoiceOf-j}
     j_0 = \begin{cases}
         j_1 & \textup{ if } \theta(j_2) \leq j_1\\
         j_2 &  \textup{ if } j_1 < \theta(j_2)
     \end{cases}
  \end{equation}
  with the convention that $j_0 = j_1$ if $J_\af\setminus \psigma=\emptyset$ and $j_0=j_2$ if $J_\ab=\emptyset$.

 Recall that $\scrJ = J_\ab\cup J_\cb\cup J_\af\cup J_\cf$.  Let us now define  $\{\alpha_i\}_{i=1}^n$ to be: 
\begin{equation}\label{alphadeltainbound}
\begin{aligned}
\forall i \in \llbracket 1, n \rrbracket\backslash \{ j_0 \}:\quad &\alpha_i = 1 - \frac{1-\gamma}{n} = \alpha\\
&\alpha_{ j_0 } = 1-\frac{1}{n}-\frac{n-1}{n}\gamma = \alpha-\gamma.
\end{aligned}
\end{equation}
Note that for any $i \in \llbracket 1, n \rrbracket\setminus \{ j_0 \}, \alpha_i \geq \frac12$. 
It is also useful to introduce sequences $\{\delta^{(\af)}_i\}_{i=1}^n$ and $\{\delta^{(\ab)}_i\}_{i=1}^n$ as follows
\[
\delta^{(\af)}_i = 
\begin{cases}
  \alpha_i,  & \text{if~} i \in P_\af\cup P_\cf \\
    1, & \text{if~} i \in \scrB\setminus \psigma\\
    0, & \text{otherwise}
\end{cases}
\qquad
\textup{and} 
\qquad
\delta^{(\ab)}_i = 
\begin{cases}
  \alpha_i,  & \text{if~} i \in P_\cb\cup P_\ab \\
    0, & \text{otherwise.}
\end{cases}
\]
We introduce the following notation for $i\in\llbracket 1,n\rrbracket$:
\[
\ab_{i}= \begin{cases}
    \ab(q_i), & \textup{if~} i \in \Jab  \\
    1,  & \text{otherwise} 
\end{cases}
\qquad
\textup{and} 
\qquad
\ab_i^*= \begin{cases}
    \cb(q_i), & \textup{if~} i \in \Jcb  \\
    1, & \text{otherwise.} 
\end{cases}
\]
Finally, we define variants of $A_i(\theta_*)$ and $B_i(\theta)$ from \eqref{ABis}. In this proof, we drop the argument from $A_i = A_i(\theta_*)$ and $B_i = B_i(\theta)$ for brevity,
and introduce the variants $\tA_i$ and $\tB_i$ with this more compact notation:
\begin{equation}\label{ABtildes}
\begin{aligned}
    &\forall i \in J_\cf\setminus\psigma :& &\tA_i  = \sum_{\underset{j \in\Jcb}{j> \theta_*(i)}}\wb(q_j) + \sum_{\underset{j \in \Jcf\setminus \psigma}{\underset{\theta_*(i)=\theta_*(j)}{i \preceq j}}} \wf(k_{j})  + \sum_{\underset{j \in \Jcf\setminus \psigma}{\theta_*(i)< \theta_*(j)}} \wf(k_{j}),\\
     &\forall i \in J_\af\setminus \psigma:& & \tB_i  = \sum_{\underset{j\in \Jab}{j<\theta(i)}}\wb(q_j) +\sum_{\underset{j\in \Jaf\setminus \psigma}{\underset{\theta(j)= \theta(i)}{j \preceq i}}}\wf(k_{j}) +\sum_{\underset{j\in \Jaf\setminus \psigma}{\theta(j) < \theta(i)}}\wf(k_{j}).
\end{aligned}
\end{equation}
Recalling \eqref{ABis}, we note that
\begin{equation}\label{FromABtilde-to-AB}
  \forall i \in J_\cf\setminus\psigma :\quad  \tA_i\leq A_{\theta_*(i)} \quad \textup{and}\quad  \forall i\in J_\af\setminus \psigma:\quad   \tB_i\leq B_{\theta(i)}.
\end{equation}

Using the abbreviation
\begin{equation}\label{W-function}
    W(k,q) = 1+\frac{\wf(k)}{\wb(q)},
\end{equation}
we may now rewrite \eqref{ToEstimate} in the following way
\begin{align}
\label{ToEstimateStepTwo}
\nonumber  &   \int \prod_{i \in \llbracket 1, n \rrbracket \setminus \psigma} F_i(k_i, q_i)  \prod_{ j\in ((J_\ab\cap J_\af)\cup (J_\cb\cap J_\cf))\cap \psigma} \bigl\|W(\cdot,q_j)^{n} F_j(\cdot,q_j)\bigr\| \\
\nonumber   & \qquad \scrL\,\Delta_\ab\Delta_\af \prod_{i \in P_\ab\cup P_\cb} \wb(q_i)^{-\alpha_i} \prod_{i \in P_\af\cup P_\cf} \wf(k_i)^{-\alpha_i}\prod_{i \in (\Jcf \cap \psigma)\setminus \Jcb}  \cf\bigl(F_i(., q_i)\bigr) \\
\nonumber  & \qquad \prod_{i=1}^n\biggl\{ R_0(z)^{\delta^{(\ab)}_i}\ab_i^* \wb(q_i)^{\delta^{(\ab)}_i} \bB^*_i {\prod_{\underset{\theta_*(j)=i}{j \in \Jcf\setminus \psigma}}}^{\hspace{-0.3cm}(\preceq_*)}  \biggl(R_0(z)^{\delta^{(\af)}_j} \cf(k_{j})\wf(q_{j})^{\delta^{(\af)}_{j}} \biggr)\biggr\} \\
\nonumber  &\qquad \prod_{i \in P_\cf} \bigl( H_0-z+ \tA_{i} \bigr)^{\alpha_i} \prod_{i\in P_\cb} \bigl(H_0-z+ A_i  \bigr)^{\alpha_i}  \\
\nonumber  &  \qquad \prod_{i \in \scrA} R_0\bigl(z-C_{i}(\sigma) - R_i\bigr) \prod_{i \in P_\ab} \bigl(H_0-z + B_i\bigr)^{\alpha_i}\prod_{i \in P_\af}\bigl(H_0-z+\tB_{i}\bigr)^{\alpha_i} \\
\nonumber  & \qquad \prod_{i \in \Jcf\cap (\scrB\setminus \psigma)} \bigl( H_0-z+\tA_{i} \bigr) \prod_{i \in \scrB\setminus \psigma} R_0\bigl(z-D_i(\sigma)\bigr)  \prod_{i \in \Jaf\cap (\scrB\setminus \psigma)}  \bigl(H_0-z+\tB_{i}\bigr)\\
\nonumber   & \qquad \prod_{i=1}^n\biggl\{{\prod_{\underset{\theta(j)=i}{j \in \Jaf \setminus \psigma}}}^{\hspace{-0.2cm}(\preceq)}\biggl( \wf(k_{j})^{\delta^{(\af)}_{j}}\af(k_{j})R_0(z)^{\delta^{(\af)}_{j}}\biggr)\bB_{i} \,\wb(q_{i})^{\delta^{(\ab)}_i} \ab_i R_0(z)^{\delta^{(\ab)}_i}\biggr\}\\
  &\qquad \prod_{i \in (\Jaf \cap \psigma)\setminus J\ab} \af\bigl(\overline{F_i(., q_i)}\bigr)  \prod_{i\in \llbracket 1,n\rrbracket \setminus \psigma} dk_i 
\prod_{j\in \llbracket 1,n\rrbracket} dq_j,
\end{align}
where the following notation for operator-valued functions $\bB_i$ and $\bB_i^*$ have been used: if $i  \in \Jcf \cap \psigma\cap \Jcb$ then
\begin{equation*}
\begin{aligned}
  \bB_i^* & :=  \bB^*_i\bigl(\{q_j\}_{\underset{j\leq i}{j\in J_\cb}},\{k_j\}_{\underset{\theta_*(j)<i}{j\in J_\cf\setminus\psigma}}\bigr) \\
 & =  \frac1{\bigl\|W(\cdot,q_i)^{n} F_i(\cdot,q_i)\bigr\|} \prod_{\underset{j\leq i}{j \in \Jcb}} \bigl(H_0-z+A_{j;i} \bigr)^{\alpha_j}  \prod_{\underset{\theta_*(j) < i}{j\in \Jcf\setminus \psigma}}
 \bigl( H_0-z+ \tA_{j;i}\bigr)^{\delta^{(\af)}_j}
  \\
&\quad  \cf\bigl(F_i(.,q_i)\bigr) \prod_{\underset{\theta_*(j) < i}{j\in \Jcf\setminus \psigma}}
 \bigl( H_0-z+ \tA_{j;i}\bigr)^{-\delta^{(\af)}_l}   \prod_{\underset{j\leq i}{j\in \Jcb}} \bigl(H_0-z+A_{j;i}\bigr)^{-\alpha_j}
\end{aligned}
  \end{equation*}
 and otherwise, for $i\not\in \Jcf \cap \psigma\cap \Jcb$, we just set $\bB_i^* = 1$, the identity operator. Here we make use of the following two abbreviations for $i,j\in\llbracket 1,n\rrbracket$. If $j\leq i$, we set
 \begin{equation*}
  A_{j;i}   =  \sum_{\underset{j\leq \ell\leq i }{\ell\in \Jcb }} \wb(q_\ell)+\sum_{\underset{j\leq \theta_*(q) < i}{q \in 
    \Jcf\setminus \psigma }} \wf(k_{q})
    \end{equation*}
    and, if $j\in J_\cf\setminus\psigma$ with $\theta_*(j)< i$, we set
    \begin{equation*}
\tA_{j;i}  =  \sum_{\underset{\theta_*(j)< \ell\leq i }{\ell\in \Jcb }} \wb(q_\ell)+\sum_{\underset{\theta_*(q) = \theta_*(j), j\preceq_* q}{q \in 
    \Jcf\setminus \psigma }} \wf(k_{q})+\sum_{\underset{\theta_*(j) < \theta_*(q) < i}{q \in 
    \Jcf\setminus \psigma }} \wf(k_{q}).
 \end{equation*}

 In the same way, if $i \in \Jaf \cap \psigma\cap \Jab $ then
\begin{equation*}
\begin{aligned}
\bB_i :=&  \bB_i\bigl(\{q_j\}_{\underset{j\geq i}{j\in J_\ab}},\{k_j\}_{\underset{\theta_*(j)>i}{j\in J_\af\setminus\psigma}}\bigr) \\
  = & \frac1{\bigl\|W(\cdot,q_i)^{n} F_i(\cdot,q_i)\bigr\|} \prod_{\underset{j\geq i}{j \in \Jab}} \bigl(H_0-z + B_{j;i}\bigr)^{-\alpha_j} \prod_{\underset{\theta(j) >i }{j\in \Jaf\setminus \psigma}} \bigl(H_0-z+  \tB_{j;i}\bigr)^{-\delta^{(\af)}_j} \\
& \af\bigl(\overline{F_i(.,q_i)}\bigr)\prod_{\underset{\theta(j) >i }{j\in \Jaf\setminus \psigma}} \bigl(H_0-z  + \tB_{j;i}\bigr)^{\delta^{(\af)}_j} \prod_{\underset{j\geq i}{j \in \Jab}} \bigl(H_0-z  + B_{j;i}\bigr)^{\alpha_j}
\end{aligned}
\end{equation*}
and otherwise, for $i \not\in \Jaf \cap \psigma\cap \Jab $, we use the convention $\bB_i=1$. Similar to above, we make use of two abbreviations for  $i,j\in\llbracket 1,n\rrbracket$. If $j\geq i$, we set
 \begin{equation*}
B_{j;i}  =   \sum_{\underset{i \leq  \ell\leq j}{\ell\in \Jab }} \wb(q_\ell) +\sum_{\underset{i< \theta(q) \leq j  }{q \in \Jaf\setminus \psigma }} \wf(k_q)
\end{equation*}
and, if $j\in J_\af\setminus\psigma$ with $\theta(j)> i$, we set
\begin{equation*}
\tB_{j;i} =   \sum_{\underset{i \leq \ell <  \theta(j)}{\ell\in \Jab }} \wb(q_\ell) 
+\sum_{\underset{\theta(q) = \theta(j), q\preceq j}{q \in 
    \Jaf\setminus \psigma }} \wf(k_{q})
+\sum_{\underset{i< \theta(q) < \theta(j)  }{q \in \Jaf\setminus \psigma }} \wf(k_q).
 \end{equation*}

A reader worried about diving by zero in the definitions of $\bB_i$ and $\bB_i^*$, may add $\varepsilon \cdot\exp(-|q_j|^2)$ to $\|\tF_j(\cdot,q_j)\|$, and at the end of the estimates to follow take the limit $\varepsilon\to 0$ to recover the same conclusion.

Due to \eqref{FromABtilde-to-AB}, we find that
\begin{align*}
   & \Biggl\| \prod_{i\in P_\cb} \bigl(H_0-z+A_i \bigr)^{\alpha_i}  \prod_{i \in P_\cf} \bigl( H_0-z+ \tA_{i}\bigr)^{\alpha_i} \prod_{i \in \scrA} R_0\bigl(z-C_{i}(\sigma) - R_i\bigr)  \\
 & \qquad \prod_{i \in P_\af} \bigl(H_0-z +\tB_{i}\bigr)^{\alpha_i} \prod_{i \in P_\ab}\bigl(H_0-z+B_i\bigr)^{\alpha_i}\Biggr\|\\ 
 & \quad \leq  \biggl\| \prod_{i \in P_\cb }\bigl(H_0-z+A_i\bigr)^{\alpha_i} \prod_{i \in P_\cf }\bigl(H_0-z+A_{\theta_*(i)}\bigr)^{\alpha_i}
\\
& \qquad \prod_{i\in \scrA} R_0\bigl(z-C_{i}(\sigma)-R_i\bigr) \prod_{i \in P_\af }\bigl(H_0-z+B_{\theta(i)}\bigr)^{\alpha_i} \prod_{i \in P_\ab }\bigl(H_0-z+B_{i}\bigr)^{\alpha_i} \biggr\|. 
\end{align*}
We can therefore apply Lemma~\ref{RegOpEstimates} to conclude that there exists a family of real numbers  
$\{\beta_{i}\}_{i=1}^n$ with $0\leq \beta_i\leq \alpha_i$ and $\beta_i=0$ for $i\in\psigma$. Here the $\alpha_i$'s were introduced in \eqref{alphadeltainbound}. The estimate  \eqref{regularitypropertyV2} then yields
\begin{align}\label{EstimFirstStep}
     \nonumber&   \bigl| \scrL\bigr| \Delta_\ab\Delta_\af \prod_{i \in P_\ab\cup P_\cb} \wb(q_i)^{-\alpha_i} \prod_{i \in P_\af\cup P_\cf} \wf(k_i)^{-\alpha_i} \\
 \nonumber &  \qquad  \Biggl\| \prod_{i\in P_\cb} \bigl(H_0-z+A_i \bigr)^{\alpha_i} \prod_{i \in P_\cf} \bigl( H_0-z+\tA_{i}\bigr)^{\alpha_i} \\
 \nonumber &  \qquad\quad    \prod_{i \in \scrA} R_0\bigl(z-C_{i}(\sigma) - R_i\bigr)  \prod_{i \in P_\ab} \bigl(H_0-z + B_i\bigr)^{\alpha_i} \prod_{i \in P_\af} \bigl(H_0-z+\tB_{i}\bigr)^{\alpha_i}\Biggr\|\\
 & \quad \leq c_T \biggl(\prod_{i\in P_{\cf}\cup P_{\af}}W(k_i,q_i)^{n-1}\biggr) \frac{ \Delta_\ab\Delta_\af}{\prod^n_{i=1}[\wb(q_i)]^{\alpha_i - \beta_{i}}[\wf(k_i)]^{\beta_{i}}}.
\end{align}
Here $c_T$ is the bounding constant for the right-handed Wick monomial $\rightT$ that we started out with.

 We now turn to the estimate of 
 \begin{align*}
    & \prod_{i \in \Jcf\cap (\scrB\setminus \psigma)} \bigl( H_0-z+ \tA_{i}\bigr) \prod_{i \in \scrB\setminus \psigma} R_0\bigl(z-D_i(\sigma)\bigr)  \prod_{i \in \Jaf\cap (\scrB\setminus \psigma)}  \bigl(H_0-z+\tB_{i}\bigr).
\end{align*}
Recall from \eqref{Dis} that $D_i(\sigma)=D_i^{(\preceq,\preceq_*)}(\sigma)$ is given by the expression
\begin{equation}\label{Di-in-proof}
      \sum_{\underset{\ell > \tilde{\sigma}(i)}{\ell\in  \Jcb }} \wb(q_\ell) + \sum_{\underset{i\preceq_* \ell}{\ell\in \Jcf\setminus \psigma}}\wf(k_{\ell}) +  \sum_{\underset{\ell \leq \tilde{\sigma}(i)}{\ell\in \Jab  }} \wb(q_\ell)
     +\sum_{\underset{\ell \preceq i}{\ell\in \Jaf\setminus \psigma }}\wf(k_\ell) +R_{\tilde{\sigma}(i)}.
\end{equation}
Consider first terms with $i \in \Jcf\cap (\scrB\setminus \psigma)$. 
Recall from \eqref{ABtildes} that
\begin{equation}\label{tA-in-proof}
\tA_{i} = \sum_{\underset{j> \theta_*(i)}{j \in\Jcb}}\wb(q_j) + \sum_{\underset{\theta_*(i)= \theta_*(j)}{\underset{i \preceq_* j}{j \in \Jcf\setminus \psigma}}} \wf(k_{j})  + \sum_{\underset{\theta_*(i)< \theta_*(j)}{j \in \Jcf\setminus \psigma}} \wf(k_{j}).
\end{equation}
Since, $i\in\scrB$, we have $i\not\in P_\cf$ and therefore, cf.  Definition~\ref{FunctionTheta}, we may conclude that  $\theta_*(i) = \sigma(i)=\tilde{\sigma}(i)$. Hence the first term on the right-hand sides of \eqref{Di-in-proof} and \eqref{tA-in-proof} are identical. 
As for the sum of the last two terms in \eqref{tA-in-proof}, we compute
\begin{equation}\label{Di-compu}
\sum_{\underset{\theta_*(i)= \theta_*(j)}{\underset{i \preceq_* j}{j \in \Jcf\setminus\psigma}}} \wf(k_{j})  + \sum_{\underset{\theta_*(i)< \theta_*(j)}{j \in \Jcf\setminus \psigma}} \wf(k_{j})  = \sum_{\underset{\sigma(i)\leq \theta_*(j)}{\underset{i\preceq_* j}{j\in \Jcf\setminus \psigma}}}\wf(k_{j}) + \sum_{\underset{\sigma(i) < \theta_*(j)}{\underset{j\prec_* i}{j\in \Jcf\setminus \psigma }}}\wf(k_{j}).
\end{equation}
 The first sum on the right-hand side of \eqref{Di-compu} above is bounded by the second sum in the expression \eqref{Di-in-proof} for $D_i^{(\preceq,\preceq_*)}(\sigma)$. As for the second sum on the right hand side of \eqref{Di-compu}, let us note that for $j\in J_\cf\setminus \psigma$ with $j\prec_* i$ we have $\tilde{\sigma}(j) \leq \tilde{\sigma}(i)$ and since $\tilde{\sigma}(j)= \sigma(j)$, $\tilde{\sigma}(i)= \sigma(i)$ and $\theta_*(j)>\sigma(i)$, we therefore have $\sigma(j)\leq\sigma(i) < \theta_*(j)$. Here we used the definition \eqref{tildesigma} of $\tilde{\sigma}$ and the property \eqref{orderprop} of the total order $\preceq_*$. Hence, $\theta_*(j) \neq \sigma(j)$ and since $P_\cf \cap \Jab = \emptyset$, we therefore have $j \in P_\cf\cap I_\ab$, cf.~Definition~\ref{FunctionTheta}. By Definition~\ref{Admissiblemaps} it follows that $\sigma(j)=j$. By definition $\theta_*(j) = f_{\ab}(j)$, therefore  
 \[
 j = \sigma(j) \leq \sigma(i) < \theta_*(j) \leq f_{\ab}(j).
 \]
 Moreover, 
\begin{align*}
    R_{\sigma(i)} & =   \sum_{\underset{ j \leq \sigma(i) <  f_{\ab}(j)}{j\in I_\ab}}\wb(q_j) + \sum_{\underset{ j\leq \sigma(i) <  f_{\af}(j)}{j\in I_\af}}\wf(k_j).
\end{align*}
As a conclusion 
\begin{align*}
\sum_{\underset{\sigma(i) < \theta_*(j)}{\underset{j\prec_* i}{j\in \Jcf\setminus \psigma }}}\wf(k_{j}) & \leq 
\sum_{{\underset{j \leq \sigma(i) < f_a(j)}{j\in P_{\cf}\cap I_\ab}}} W(k_j,q_j)\wb(q_{j})\leq
\biggl(\prod_{{\underset{j \leq \sigma(i) < f_a(j)}{j\in P_{\cf}\cap I_\ab}}}
W(k_j,q_j)\biggr)
R_{\sigma(i)}
\end{align*}
 and therefore: 
 \begin{equation}\label{AibyDi}
     \tA_i\leq \biggl(\prod_{{\underset{j \leq \sigma(i) < f_a(j)}{j\in P_{\cf}\cap I_\ab}}} W(k_j,q_j)\biggr)
     D_i(\sigma).
 \end{equation}
 Similarly, if $i \in \Jaf\cap (\scrB \setminus \psigma)$, then 
 \begin{equation}\label{tB-in-proof}
 \tB_i = \sum_{\underset{j<\theta(i)}{j\in \Jab}}\wb(q_j) +\sum_{\underset{\theta(j)\leq \theta(i)}{\underset{j \preceq i}{j\in \Jaf\setminus \psigma}}}\wf(k_{j}) +\sum_{\underset{\theta(j) < \theta(i)}{\underset{i \prec j}{j\in \Jaf\setminus \psigma}}}\wf(k_{j}).
 \end{equation}
 First, $\theta(i) = \sigma(i)=\tilde{\sigma}(i)+1$ so the first term in \eqref{tB-in-proof} is identical to the third term in \eqref{Di-in-proof}. In addition, the second term in \eqref{tB-in-proof} can be bounded by the fourth term in \eqref{Di-in-proof}.   As for the third and last term in \eqref{tB-in-proof}, $i\prec j$ implies that $\tilde{\sigma}(i) \leq \tilde{\sigma}(j)$ and since $i,j \in \Jaf$ we have $\sigma(i) \leq \sigma(j)$ (recall from \eqref{tildesigma} the definition of $\tilde{\sigma}$). Consequently, $\theta(j) < \sigma(i)\leq \sigma(j)$. Hence, $j\in P_\af\cap f_\ab(I_\ab)$ implying both that $\theta(j)=f_\ab^{-1}(j)$ and $\sigma(j)=j$. Therefore, $f_{\ab}^{-1}(j)< \sigma(i) \leq j$ which implies that $f_{\ab}^{-1}(j)\leq \tilde{\sigma}(i) < j$. Consequently, the last term in \eqref{tB-in-proof} can be estimated as follows 
 \[
\sum_{\underset{\theta(j) < \theta(i)}{\underset{i \prec j}{j\in \Jaf\setminus \psigma}}}\wf(k_{j})\leq
 \sum_{\underset{f^{-1}_{\ab}(j) \leq \tilde{\sigma}(i) < j}{j\in f_\ab(I_\ab)}}\wf(k_j).
 \]
 As a conclusion
 \begin{equation}\label{BibyDi}
     \tB_i\leq  \biggl(\prod_{\underset{ f_\ab^{-1}(\ell) \leq \tilde{\sigma}(i) \leq \ell}{\ell \in  P_{\af}\cap f_\ab(I_\ab)}} W(k_j,q_j)\biggr)
     D_i(\sigma),
 \end{equation}
with the convention that $q_\ell = q_{f_\ab^{-1}(\ell)}$ for $\ell\in f_\ab(I_\ab)$, due to the presence of the delta distributions. 
 
 Combining \eqref{AibyDi} and \eqref{BibyDi}, we arrive at
 \begin{align}\label{EstimOnDis}
  \nonumber & \biggl\|\prod_{i \in \Jcf\cap (\scrB\setminus \psigma)} \bigl( H_0-z+ \tA_{i}\bigr) \prod_{i \in \scrB\setminus \psigma} R_0\bigl(z-D_i(\sigma)\bigr)  \prod_{i \in \Jaf\cap (\scrB\setminus \psigma)}  \bigl(H_0-z+\tB_{i}\bigr)\biggr\|\\
   & \qquad
  \leq \prod_{j\in P_\cf\cup P_\af} W(k_j,q_j).
  \end{align}
Combining \eqref{EstimFirstStep} and \eqref{EstimOnDis},
we obtain:
\begin{align}\label{EstimSecondStep}
  \nonumber   &   \bigl| \scrL\bigr| \Delta_\ab\Delta_\af   \prod_{i   \in P_\ab\cup P_\cb} \wb(q_i)^{-\alpha_i} \prod_{i \in P_\af\cup P_\cf} \wf(k_i)^{-\alpha_i}  \Biggl\| \prod_{i\in P_\cb} \bigl(H_0-z+A_i \bigr)^{\alpha_i}\\
  \nonumber     & \qquad\quad  \prod_{i \in P_\cf} \bigl( H_0-z+ \tA_{i}\bigr)^{\alpha_i}   \prod_{i \in \scrA} R_0\bigl(z-C_{i}(\sigma) - R_i\bigr) \\
  \nonumber &\qquad \quad \prod_{i \in P_\ab} \bigl(H_0-z + B_i\bigr)^{\alpha_i} \prod_{i \in P_\af}\bigl(H_0-z+\tB_{i}\bigr)^{\alpha_i}\Biggr\|\\
  \nonumber &\qquad \Biggl\| \prod_{i \in \Jcf\cap \scrB\setminus\psigma} \bigl( H_0-z+ \tA_{i}\bigr) \prod_{i \in \scrB\setminus\psigma} R_0\bigl(z-D_i(\sigma)\bigr) \prod_{i \in \Jaf\cap \scrB\setminus \psigma}  \bigl(H_0-z+\tB_{i}\bigr)\Biggr\|\\
 & \quad \leq c_T \biggl(\prod_{i \in P_\cf\cup P_\af} W(k_i,q_i)^{n}\biggr) \frac{ \Delta_\ab\Delta_\af}{\prod^n_{i=1}[\wb(q_i)]^{\alpha_i-\beta_{i}}[\wf(k_i)]^{\beta_{i}}}.
\end{align}
In light of the estimate \eqref{EstimSecondStep}, we introduce for all $i\in\llbracket 1,n\rrbracket$ the modified $F_i$:
\[
\tF_i(k_i,q_i) = \frac{W(k_i,q_i)^{n \kdelta_{i\in P_\af\cup P_\cf}+n \kdelta_{i\in ((\Jab\cap \Jaf)\cup (\Jcb\cap \Jcf))\cap \psigma}} \ F_i(k_i,q_i)}{[\wb(q_i)]^{\alpha-\beta_{i}}[\wf(k_i)]^{\beta_{i}}},
\]
where the weight $W$ was defined in \eqref{W-function}.
Recall, from Definition~\ref{RegularOperator} and Lemma~\ref{RegOpEstimates} that $\beta_{i}  = 0$ if $i \in \psigma$. This means that for $i\in\psigma$, the $\tF_i$'s only differ from the $F_i$'s by the function $\wb(q_i)^{-\alpha}$ of the boson momentum $q_i$. Hence, for $i\in\psigma$, we have $\wb(q_i)^{-\alpha}\af(F_i(\cdot,q_i)) = \af(\tF_i(\cdot,q_i))$. 

From now one we will assume that $j_0 \in \Jab$. If $j_0\not\in J_\ab$, then $j_0\in J_\af\setminus\psigma$ and we will comment on how to deal with this case along the way.
Recalling from \eqref{alphadeltainbound} that $\alpha_i = \alpha$ for $i\neq j_0$ and $\alpha_{j_0} = \alpha -\gamma$, we may now use \eqref{EstimSecondStep} to estimate
\begin{equation}
\label{ToEstimateRHS}
\begin{aligned}
&\Bigl|\Bigl\langle \phi \Big| T^{(k)}_{(P_\cf,P_\af)}\bigl(z;\uF^{(n)}\bigr) R_0(z)^{\gamma}\psi \Bigr\rangle\Bigr|\\
&\quad \leq c_T \int \Delta_\ab\Delta_\af \prod_{i \in \llbracket 1, n \rrbracket \setminus\psigma} 
\bigl| \tF_i(k_i, q_i) \bigr| \prod_{j\in ((J_\ab\cap J_\af)\cup (J_\cb\cap J_\cf))\cap \psigma} \bigl\|\tF_j(\cdot,q_j)\bigr\|  \\
& \qquad \Biggl\| \prod^n_{i=1}  \Biggl[  {\prod_{\underset{\theta_*(j)=n+1-i}{j \in \Jcf\setminus \psigma}}}^{\hspace{-0.5cm}(\succeq_*)} \Biggl(\wf(k_{j})^{\delta^{(\af)}_{j}}\af(k_{j})  R_0(\overline{z})^{\delta^{(\af)}_j} \Biggr) \Bigl(\bB^*_{n+1-i}\Bigr)^*  \\
&\qquad \quad
 \wb(q_{n+1-i})^{\delta^{(\ab)}_{n+1-i}} (\ab_{n+1-i}^*)^* R_0(\overline{z})^{\delta^{(\ab)}_{n+1-i}}\Biggr]\prod_{i \in (\Jcf \cap \psigma)\setminus \Jcb} \af\bigl(\tF_i(., q_i)\bigr)\phi \Biggr\|
\\
    & \qquad  \Biggl\| \prod^n_{i=1}\Biggl[ {\prod_{\underset{\theta(j)=i}{j \in \Jaf \setminus \psigma}}}^{\hspace{-0.2cm}(\preceq)} \Biggl(\wf(k_{j})^{\delta^{(\af)}_{j}}\af(k_{j})R_0(z)^{\delta^{(\af)}_{j}}\Biggr)\bB_{i}\, \wb(q_{i})^{\delta^{(\ab)}_i} \ab_{i} R_0(z)^{\delta^{(\ab)}_i}\Biggr]\\
    & \qquad \wb(q_{j_0})^{\gamma} \quad \prod_{i \in (\Jaf \cap \psigma)\setminus J_\ab} \af\bigl(\overline{\tF_i(., q_i)}\bigr) R_0(z)^{\gamma}\psi\Biggr\|\prod_{i\in \llbracket 1,n\rrbracket \setminus \psigma} dk_i 
\prod_{j\in \llbracket 1,n\rrbracket} dq_j  \\
    & \quad = c_T\int \Delta_\ab\Delta_\af  \prod_{i \in \llbracket 1, n \rrbracket \setminus\psigma} 
\bigl| \tF_i(k_i, q_i) \bigr|  \prod_{j\in ((J_\ab\cap J_\af)\cup (J_\cb\cap J_\cf))\cap \psigma} \bigl\|\tF_j(\cdot,q_j)\bigr\| \\
& \qquad \Fl\bigl(\{k_i\}_{i\in J_\cf\setminus\psigma}, \{q_j\}_{j\in J_\cb\cup (J_\cf \cap \psigma\cap J_\ab)}\bigr) \Fr\bigl(\{k_i\}_{J_\af\setminus\psigma}, \{q_j\}_{J_\ab\cup (J_\af \cap \psigma\cap J_\cb)}\bigr) \\
& \qquad \prod_{i\in \llbracket 1,n\rrbracket \setminus \psigma} dk_i 
\prod_{j\in \llbracket 1,n\rrbracket} dq_j ,
\end{aligned}
\end{equation}
where 
\begin{align*}
  & \Fl\bigl(\{k_i\}_{J_\cf\setminus\psigma}, \{q_j\}_{J_\cb\cup (J_\cf \cap \psigma\cap J_\ab)}\bigr) \\
   & \quad = 
  \Biggl\| \prod^n_{i=1}\Biggl[    {\prod_{\underset{\theta_*(j)=n+1-i}{j \in \Jcf\setminus \psigma}}}^{\hspace{-0.5cm}(\succeq_*)} \Biggl(\wf(k_{j})^{\delta^{(\af)}_{j}}\af(k_{j})  R_0(\overline{z})^{\delta^{(\af)}_j} \Biggr) \Bigl(\bB^*_{n+1-i}\Bigr)^* \\
& \qquad \quad
 \wb(q_{n+1-i})^{\delta^{(\ab)}_{n+1-i}} (\ab_{n+1-i}^*)^* R_0(\overline{z})^{\delta^{(\ab)}_{n+1-i}}\Biggr]\prod_{i \in (\Jcf \cap \psigma)\setminus \Jcb} \af\bigl(\tF_i(., q_i)\bigr)\phi \Biggr\|
 \end{align*}
 and
 \begin{align*}
 & \Fr\bigl(\{k_i\}_{i\in J_\af\setminus\psigma}, \{q_j\}_{j\in J_\ab\cup (J_\af \cap \psigma\cap J_\cb)}\bigr) \\
 & \quad= 
 \Biggl\| \prod^n_{i=1}\Biggl[{\prod_{\underset{\theta(j)=i}{j \in \Jaf \setminus \psigma}}}^{\hspace{-0.2cm}(\preceq)} \Biggl(\wf(k_{j})^{\delta^{(\af)}_{j}}\af(k_{j})R_0(z)^{\delta^{(\af)}_{j}}\Biggr)\bB_{i}\\
 &\qquad \quad \wb(q_{i})^{\delta^{(\ab)}_i} \ab_{i} R_0(z)^{\delta^{(\ab)}_i} \Biggr] \wb(q_{j_0})^{\gamma}
 \prod_{i \in (\Jaf \cap \psigma)\setminus J_\ab} \af\bigl(\overline{\tF_i(., q_i)}\bigr) R_0(z)^{\gamma}\psi\Biggr\|.
\end{align*}
Note that $j_0$ is chosen  such that $\bB_{j_0}\, \wb(q_{j_0})^{\delta^{(\ab)}_{j_0}} \ab(q_{j_0}) R_0(z)^{\delta^{(\ab)}_{j_0}}$ is the right-most term in the product in $\Fr$. If $j_0\not\in J_\ab$, then $\wf(k_{j_0})^{\delta^{(\af)}_{j_0}}\af(k_{j_0})R_0(z)^{\delta^{(\af)}_{j_0}}$ is the right-most term, and the term $\wb(q_{j_0})^{\gamma}$ should be replaced by $\wf(k_{j_0})^{\gamma}$.

Abbreviate 
\begin{align*}
  & \tFl\bigl(\{k_i\}_{J_\cf\setminus\psigma}, \{q_j\}_{J_\cb\setminus (J_\af \cap \psigma)}\bigr)  \\
   & \quad  = \biggl(\int \Fl\bigl(\{k_i\}_{J_\cf\setminus\psigma}, \{q_j\}_{J_\cb\cup (J_\cf \cap \psigma\cap J_\ab)}\bigr)^2 \prod_{j\in ((J_\ab\cap J_\cf)\cup (J_\cb\cap J_\af))\cap \psigma} dq_j\biggr)^\frac12,\\
   & \tFr\bigl(\{k_i\}_{J_\af\setminus\psigma}, \{q_j\}_{J_\ab\setminus (J_\cf \cap \psigma)}\bigr) \\
    & \quad = \biggl( \int \Fr\bigl(\{k_i\}_{J_\af\setminus\psigma}, \{q_j\}_{J_\ab\cup (J_\af \cap \psigma\cap J_\cb)}\bigr)^2\prod_{j\in((J_\ab\cap J_\cf)\cup (J_\cb\cap J_\af))\cap \psigma} dq_j \biggr)^\frac12.
\end{align*}
Note that $\| \tFl\| = \| \Fl\|$ and  $\| \tFr\| = \| \Fr\|$. Let us focus on the estimate of $\Fr$. Remark that the following estimates are easy to derive: 
\begin{align*}
    \forall j \in \Jaf\setminus \psigma: & \qquad\int \Bigl\| \wf(k_{j})^{\delta^{(\af)}_{j}}\af(k_{j})R_0(z)^{\delta^{(\af)}_{j}} \Bigr\|^2 d k_{j}  \leq 1,\\
    \forall j \in \Jab\setminus\{j_0\}: & \qquad \int \Bigl\| \wb(q_{j})^{\delta^{(\ab)}_{j}}\ab(k_{j})R_0(z)^{\delta^{(\ab)}_{j}} \Bigr\|^2 d q_{j}  \leq 1\\
    \forall j \in \Jab\cap \Jaf \cap \psigma: & \qquad \Bigl\|\bB_j\bigl(\{q_j\}_{\underset{j\leq i}{j\in J_\ab}},\{k_j\}_{\underset{\theta(j)<i}{j\in J_\af\setminus\psigma}}\bigr) \Bigr\| \leq c_{2n},
\end{align*}
where we used Lemma \ref{RegChainFermionOp} for the last inequality, which is the source of the constant $c_{2n}$. 
Therefore, if one defines 
\[
\psi' = \psi'\bigl(z;q_{j_0}, \{q_i\}_{(J_\af\cap \psigma)\setminus J_\ab}\bigr)=\wb(q_{j_0})^{\gamma} \prod_{i \in (\Jaf \cap \psigma)\setminus J_\ab} \af\bigl(\overline{\tF_i(., q_i)}\bigr) R_0(z)^{\gamma} \psi,
\]
then 
\begin{align*}
\int 
 \Biggl\| \prod^n_{i=1}\Biggl[ & {\prod_{\underset{\theta(j)=i}{j \in \Jaf \setminus \psigma}}}^{\hspace{-0.2cm}(\preceq)} \Biggl(\wf(k_{j})^{\delta^{(\af)}_{j}}\af(k_{j})R_0(z)^{\delta^{(\af)}_{j}}\Biggr)\bB_{i}\bigl(\{q_j\}_{\underset{j\leq i}{j\in J_\ab}},\{k_j\}_{\underset{\theta(j)<i}{j\in J_\af\setminus\psigma}}\bigr) \\
 & \wb(q_{i})^{\delta^{(\ab)}_i} \ab_{i} R_0(z)^{\delta^{(\ab)}_i}\Biggr] \psi' \Biggr\|^2 \prod_{j\in \Jaf\setminus \psigma}d k_j \prod_{j\in \Jab\setminus \{j_0\} }d q_j
\end{align*}
is of the form described in Lemma \ref{lem6}. Consequently, applying Lemma~\ref{lem6} -- leaving the last $j_0$-term in the product -- yields
\begin{align*}
\|F_r\|^2 \leq & c^{2n-2}_{2n}\int \int \Bigl\| \bB_{j_0}\bigl(\{q_j\}_{\underset{j\leq j_0}{j\in J_\ab}},\{k_j\}_{\underset{\theta(j)<j_0}{j\in J_\af\setminus\psigma}}\bigr) \wb(q_{j_0})^{\delta^{(\ab)}_{j_0}} \ab(q_{j_0})\\
& \qquad R_0(z)^{\delta^{(\ab)}_{j_0}}\psi'\bigl(z;q_{j_0}, \{q_i\}_{(J_\af\cap \psigma)\setminus J_\ab}\bigr) \Bigr\|^2  d q_{j_0} \prod_{j \in (\Jaf\cap\psigma )\setminus\Jab} d q_j \\
\leq &  c^{2n}_{2n} \int  \int \biggl\|\wb(q_{j_0})^{\delta^{\ab}_{j_0}+\gamma} \ab(q_{j_0}) R_0(z)^{\delta^{(\ab)}_{j_0}}\\
& \qquad \prod_{i \in (\Jaf \cap \psigma)\setminus J_\ab} \af\bigl(\overline{\tF_i(., q_i)}\bigr) R_0(z)^{\gamma} \psi\biggr\|^2 d q_{j_0} \prod_{j \in (\Jaf\cap\psigma )\setminus\Jab} d q_j.
\end{align*}
Let us now focus on 
\begin{align*}
&  \int  \biggl\|\wb(q_{j_0})^{\delta^{(\ab)}_{j_0}+\gamma} \ab(q_{j_0}) R_0(z)^{\delta^{(\ab)}_{j_0}}\prod_{i \in (\Jaf \cap \psigma)\setminus J_\ab} \af\bigl(\overline{\tF_i(., q_i)}\bigr) R_0(z)^{\gamma}\psi\biggr\|^2  d q_{j_0}\\
&\ = \int  \biggl\|\wb(q_{j_0})^{\delta^{(\ab)}_{j_0}+\gamma-\frac12}\wb(q_{j_0})^{\frac12} \ab(q_{j_0}) R_0(z)^{\delta^{(\ab)}_{j_0}+\gamma}\\
& \ \quad (H_0-z)^{\gamma}\prod_{i \in (\Jaf \cap \psigma)\setminus J_\ab} \af\bigl(\tF_i(., q_i)\bigr) R_0(z)^{\gamma}\psi\biggr\|^2  d q_{j_0}\\
&\ \leq 
\int  \biggl\|\wb(q_{j_0})^{\frac12} \ab(q_{j_0}) R_0(z)^{\frac12} (H_0-z)^{\gamma}\prod_{i \in (\Jaf \cap \psigma)\setminus J_\ab} \af\bigl(\tF_i(., q_i)\bigr) R_0(z)^{\gamma}\psi\biggr\|^2  d q_{j_0}\\
&\ \leq
\biggl\|(H_0-z)^{\gamma}\prod_{i \in (\Jaf \cap \psigma)\setminus J_\ab} \af\bigl(\overline{\tF_i(., q_i)}\bigr) R_0(z)^{\gamma}\psi\biggr\|^2.
\end{align*}
From Lemma \ref{RegFermionProp2} we can therefore conclude that 
\begin{equation}\label{Fr-norm}
  \bigl\| \Fr \bigr\|^2  \leq  (2c_{2n})^{2n} \bigl\|\psi\bigr\|^2\prod_{i \in (\Jaf \cap \psigma)\setminus J_\ab} \bigl\|\tF_i\bigr\|^2 .
\end{equation}
The function $\Fl$ can be treated in the same way to get the estimate
\begin{equation}\label{Fl-norm}
\bigl\|\Fl\bigr\|\leq  (2c_{2n})^{n}\bigl\|\phi\bigr\| \prod_{i \in (\Jcf \cap \psigma)\setminus J_\cb} \bigl\|\tF_i\bigr\| ,
\end{equation}
but without the extra complication coming from treating the $j_0$-term separately. With the above notation, and using Cauchy-Schwarz, we estimate \eqref{ToEstimateRHS} as follows
\begin{align*}
&\Bigl|\Bigl\langle \phi \Big| T^{(k)}_{(P_\cf,P_\af)}\bigl(z;\uF^{(n)}\bigr) R_0(z)^{\gamma}\psi \Bigr\rangle\Bigr|\\
&\quad \leq c_T\int \Delta_\ab\Delta_\af
\prod_{i \in \llbracket 1, n \rrbracket \setminus\psigma} 
\bigl| \tF_i(k_i, q_i) \bigr|  \prod_{j\in ((J_\ab\cap J_\af)\cup (J_\cb\cap J_\cf))\cap \psigma} \bigl\|\tF_j(\cdot,q_j)\bigr\|\\
&\qquad \tFl\bigl(\{k_i\}_{J_\af\setminus\psigma}, \{q_j\}_{J_\ab\setminus (J_\cf \cap \psigma)}\bigr) \tFr\bigl(\{k_i\}_{J_\cf\setminus\psigma}, \{q_j\}_{J_\cb\setminus (J_\af \cap \psigma)}\bigr) \\
& \qquad \prod_{i\in \llbracket 1,n\rrbracket \setminus \psigma} dk_i 
\prod_{j\in \llbracket 1,n\rrbracket\setminus ((J_\ab\cap J_\cf)\cup (J_\cb\cap J_\af))\cap \psigma)} dq_j.
\end{align*}

Using Cauchy-Schwarz again, now with respect to the integration variables appearing in $\tFl$ and $\tFr$, yields
\begin{align}
\label{ToEstimateRHS2}
\nonumber &\Bigl|\Bigl\langle \phi \Big| T^{(k)}_{(P_\cf,P_\af)}\bigl(z;\uF^{(n)}\bigr) R_0(z)^{\gamma}\psi \Bigr\rangle\Bigr|\leq  c_T\bigl\|\tFl\bigr\| \bigl\|\tFr\bigr\| \int \Delta_\ab\Delta_\af \\
\nonumber& \qquad \biggl( \int \prod_{i\in \llbracket 1,n\rrbracket\setminus \psigma} \bigl| \tF_i(k_i,q_i)\bigr|^2
  \prod_{j\in ((J_\ab\cap J_\af)\cup (J_\cb\cap J_\cf))\cap \psigma}\bigl\|\tF_j(\cdot,q_j)\bigr\|^2 \\
\nonumber  & \qquad\quad  \prod_{i\in (J_\af\cup J_\cf)\setminus\psigma} dk_i\prod_{j\in (J_\cb\setminus (J_\af\cap\psigma))\cup (J_\ab\setminus (J_\cf\cap \psigma))} dq_j\biggr)^\frac12 \\
\nonumber&\qquad
 \prod_{i\in I_\af} dk_i dk_{f_\af(i)} \prod_{j\in I_\ab} dq_j dq_{f_\ab(j)}\\
\nonumber & \quad  =  c_T\bigl\|\Fl\bigr\| \bigl\|\Fr\bigr\|\prod_{j\in ((J_\ab\cap J_\af)\cup (J_\cb\cap J_\cf))\cap \psigma}\bigl\|\tF_j\bigr\| 
 \int  \Delta_\ab\Delta_\af
  \biggl( \int \prod_{i\in \llbracket 1,n\rrbracket\setminus \psigma} \bigl| \tF_i(k_i,q_i)\bigr|^2\\
  &\qquad\quad  \prod_{i\in (J_\af\cup J_\cf)\setminus\psigma} dk_i\prod_{j\in (J_\ab\cup J_\cb )\setminus \psigma} dq_j\biggr)^\frac12 \prod_{i\in I_\af} dk_i dk_{f_\af(i)} \prod_{j\in I_\ab} dq_j dq_{f_\ab(j)}.
\end{align}

Recall from \eqref{FinalCover5} that $P_\ab = J_\ab$, $P_\cb = J_\cb$ and that $P_\af,P_\cf,P_\ab,P_\cb$ are pairwise disjoint and form a cover of $\scrJ$. From the computation
\begin{align*}
& \int \prod_{i\in \llbracket 1,n\rrbracket\setminus \psigma} \bigl| \tF_i(k_i,q_i)\bigr|^2 \prod_{i\in (J_\af\cup J_\cf)\setminus\psigma} dk_i\prod_{j\in (J_\ab\cup J_\cb )\setminus \psigma} dq_j \\
& \qquad  = \prod_{i\in \llbracket 1,n\rrbracket \setminus \scrJ} \bigl|\tF_i(k_i,q_i)\bigr|^2
\prod_{j\in P_\af\cup P_\cf} \bigl\| \tF_j(\cdot, q_j) \bigr\|^2 \\
& \qquad \quad\prod_{i\in (J_\ab\cup J_\cb)\setminus (J_\af\cup J_\cf)} \bigl\| \tF_i (k_i,\cdot) \bigr\|^2 \prod_{i\in ((J_\ab\cup J_\cb) \cap ( J_\af\cup J_\cf))\setminus \psigma} \bigl\|\tF_i\bigr\|^2
\end{align*}
and the estimates
\[
\bigl\|\Fl\bigr\|\leq  (2c_{2n})^{n}\bigl\|\phi\bigr\|\!\! \prod_{i \in (\Jcf \cap \psigma)\setminus J_\cb} \bigl\|\tF_i\bigr\| 
\quad \textup{and} \quad
\bigl\|\Fr\bigr\|\leq  (2c_{2n})^{n}\bigl\|\psi\bigr\|\!\! \prod_{i \in (\Jaf \cap \psigma)\setminus J_\ab} \bigl\|\tF_i\bigr\| ,
\]
we conclude from \eqref{ToEstimateRHS2} that
\begin{equation}\label{ToEstimateRHS3}
\begin{aligned}
&\Bigl|\Bigl\langle \phi \Big| T^{(k)}_{(P_\cf,P_\af)}\bigl(z;\uF^{(n)}\bigr) R_0(z)^{\gamma}\psi \Bigr\rangle\Bigr| \\
& \quad
\leq (2c_{2n})^{2n} c_T \prod_{i\in (J_\ab\cup J_\cb) \cap ( J_\af\cup J_\cf)} \bigl\|\tF_i\bigr\|\int \prod_{i \in I_\af}\kdelta\bigl(k_i-k_{f_{\af}(i)}\bigr)   \prod_{j \in I_\ab}\kdelta\bigl(q_j-q_{f_{\ab}(j)}\bigr) \\
& \qquad  \prod_{i\in \llbracket 1,n\rrbracket \setminus \scrJ} \bigl|\tF_i(k_i,q_i)\bigr|
\prod_{j\in P_\af\cup P_\cf} \bigl\| \tF_j(\cdot, q_j) \bigr\| \prod_{i\in (J_\ab\cup J_\cb)\setminus (J_\af\cup J_\cf)} \bigl\| \tF_i (k_i,\cdot) \bigr\|\\
& \qquad \prod_{i\in I_\af} dk_i dk_{f_\af(i)} \prod_{j\in I_\ab} dq_j dq_{f_\ab(j)} \bigl\|\phi \bigr\| \bigl\| \psi \bigr\|.
\end{aligned}
\end{equation}
From the four simple estimates
\begin{equation}    \label{EstimatesForCauchySchwartz}
\begin{aligned}
 \int \bigl|\tF_i(k,q_i)\bigr| \bigl| \tF_{f_\af(i)}(k,q_{f_\af(i)})\bigr| dk & \leq \bigl\| \tF_i(\cdot,q_i)\bigr\| \bigl\| \tF_{f_\af(i)}(\cdot,q_{f_\af(i)})\bigr\|, \\
\int \bigl\|\tF_i(k,\cdot)\bigr\| \bigl\| \tF_{f_\af(i)}(k,\cdot)\bigr\| dk & \leq \bigl\| \tF_i\bigr\| \bigl\| \tF_{f_\af(i)}\bigr\|,\\
   \int \bigl\|\tF_i(k,\cdot)\bigr\| \bigl| \tF_{f_\af(i)}(k,q_{f_\af(i)})\bigr| dk &  \leq \bigl\| \tF_i\bigr\| \bigl\| \tF_{f_\af(i)}(\cdot,q_{f_\af(i)})\bigr\|,\\
\int \bigl|\tF_i(k,q_i)\bigr| \bigl\| \tF_{f_\af(i)}(k,\cdot)\bigr\| dk & \leq \bigl\| \tF_i(\cdot,q_i)\bigr\| \bigl\| \tF_{f_\af(i)}\bigr\|,  
\end{aligned}
\end{equation}
the computation $\llbracket 1,n\rrbracket \cup ( (J_\ab\cup J_\cb)\setminus (J_\af\cup J_\cf) ) = I_\af\cup f_\af(I_\af)$, and Cauchy-Schwarz, we conclude that 
\begin{align*}
& \int \prod_{i \in I_\af}\kdelta\bigl(k_i-k_{f_{\af}(i)}\bigr) \prod_{i\in \llbracket 1,n\rrbracket \setminus \scrJ} \bigl|\tF_i(k_i,q_i)\bigr|\! \prod_{i\in (J_\ab\cup J_\cb)\setminus (J_\af\cup J_\cf)} \bigl\| \tF_i (k_i,\cdot) \bigr\|\prod_{i\in I_\af} dk_i dk_{f_\af(i)}\\
&\qquad \leq \prod_{j\in \llbracket 1,n\rrbracket\setminus \scrJ}  \bigl\| \tF_j(\cdot,q_j)\bigr\|
\prod_{i\in (J_\ab\cup J_\cb)\setminus (J_\af\cup J_\cf)}\bigl\| \tF_i\bigr\|.
\end{align*}
Inserting back into \eqref{ToEstimateRHS3}, recalling that $(\llbracket 1,n\rrbracket \setminus\scrJ)\cup (P_\af\cup P_\cf) = \llbracket 1,n\rrbracket \setminus (J_\ab\cup J_\cb) = I_\ab\cup f_\ab(I_\ab)$ and repeating the argument now with respect to the $q_i$ integration, we arrive at
\begin{equation*}
\begin{aligned}
&\Bigl|\Bigl\langle \phi \Big| T^{(k)}_{(P_\cf,P_\af)}\bigl(z;\uF^{(n)}\bigr) R_0(z)^{\gamma}\psi \Bigr\rangle\Bigr|\leq (2c_{2n})^{2n} c_T\prod_{i\in J_\ab\cup J_\cb} \bigl\|\tF_i\bigr\|\\
& \qquad \int \prod_{j \in I_\ab}\kdelta\bigl(q_j-q_{f_{\ab}(j)}\bigr) 
 \prod_{j\in I_\ab\cup f_\ab(I_\ab)}  \bigl\| \tF_j(\cdot,q_j)\bigr\|
\prod_{j\in I_\ab} dq_j dq_{f_\ab(j)} \bigl\|\phi\bigr\| \bigl\|\psi\bigr\|\\
& \quad \leq  (2c_{2n})^{2n} c_T\Bigl(\prod_{i=1}^n \bigl\| \tF_i\bigr\| \Bigr)\bigl\|\phi\bigr\| \bigl\|\psi\bigr\|.
\end{aligned}
\end{equation*}
Observe now that 
\begin{align*}
\bigl\| \tF_i\bigr\|  & \leq \biggl( \int W(k_i,q_i)^{2n} \frac{|F_i(k_i,q_i)|^2}{[\wb(q_i)]^{2(\alpha-\beta_{i})}[\wf(k_i)]^{2\beta_{i}}} dq_i dk_i\biggr)^{\frac12}\\
& \leq \biggl( \int \Bigl(1+\frac{\wf(k_i)}{ \wb(q_i) }\Bigr)^{2n} \Bigl(1+\frac{\wb(q_i)}{\wf(k_i)}\Bigr)^{2\alpha}\frac{|F_i(k_i,q_i)|^2}{[\wb(q_i)]^{2\alpha}} dq_i dk_i\biggr)^{\frac12},
\end{align*}
where we inserted $W$ from \eqref{W-function} and used that $a^{2\beta_i} \leq (1+a)^{2\alpha}$ with $a = \frac{\wb(q_i)}{\wf(k_i)}$. 
As a conclusion 
\begin{align*}
 \Bigl|\Bigl\langle \phi\, \Big\vert\, \rightT\bigl(z;\uF^{(n)}\bigr) R_0(z)^{\gamma}\psi \Bigr\rangle\Bigr|
&\leq \sum^N_{k =1}\Bigl|\Bigl\langle \phi \Big| T^{(k)}_{(P_\cf,P_\af)}\bigl(z;\uF^{(n)}\bigr) R_0(z)^{\gamma}\psi \Bigr\rangle\Bigr|\\
& \leq \sum^N_{k =1} (2c_{2n})^{2n} c_{T} \Bigl( \prod_{i=1}^n\bigl\| \tF_i\bigr\| \Bigr) \bigl\|\phi\bigr\|\bigl\|\psi\bigr\|\\
& \leq N (2c_{2n})^{2n} c_{T} \tK_{\alpha}\bigl(\uF^{(n)}\bigr) \bigl\|\phi\bigr\|\bigl\|\psi\bigr\|, 
\end{align*}
which completes the proof of \ref{item-regbound2} (note that $N$ depends only on $n$).  The claim \ref{item-regbound1} follows from \ref{item-regbound2} as the adjoint of a left-handed Wick monomial is right-handed. 

Finally, item \ref{item-regbound3} can be derived from \ref{item-regbound1} and \ref{item-regbound2} using Hadamard's Three-line Theorem. Indeed, we already know that there exist exponents $\{\beta_i^{(\ell)}\}_{i=1}^n$ for $\ell=1,2$, such that
\begin{equation}\label{HadamardStart}
\begin{aligned}
    \Bigl\|R_0(z)^{\delta+\gamma}\lrT\bigl(z;\uF^{(n)}\bigr) \Bigr\|  & \leq N(2c_{2n})^{2n} c_{T} 
\tilde{K}_{ \alpha}\bigl(\uF^{(n)}\bigr),\\
\Bigl\|\lrT\bigl(z;\uF^{(n)}\bigr) R_0(z)^{\delta+\gamma} \Bigr\|  & \leq N(2c_{2n})^{2n} c_{T} 
\tilde{K}_{ \alpha}\bigl(\uF^{(n)}\bigr).
\end{aligned}
\end{equation}
Consider for any $\Psi, \Phi \in \scrH$, the continuous function $f\colon \set{z\in\CC}{0\leq \re(z)\leq 1}\to \CC$ defined by 
\[
f(\theta) =\Bigl\langle \Psi\,\Big\vert\, R_0(z)^{(\delta+\gamma)(1-\theta)}\lrT\bigl(z;\uF^{(n)}\bigr) R_0(z)^{(\delta+\gamma)\theta}   \Phi\Bigr\rangle.
\]
Note that $n \geq 2$ and that for $F_1, \dotsc, F_n \in S_{\mathrm{sc}}$,  $\lrT(z;\uF^{(n)})$ is bounded and therefore $f$ is bounded and analytic on $\set{\theta \in \CC}{ 0 < \re(\theta) < 1}$. Moreover, for $\kappa=0,1$, we have
\begin{equation*}
   \sup_{\re(\theta)=\kappa} |f(\theta)| \leq N (2c_{2n})^{2n} c_{T} \tK_{\alpha}\bigl(\uF^{(n)}\bigr)\bigl\|\Phi\bigr\| \bigl\|\Psi\bigr\|
\end{equation*}
and therefore, by Hadamard's Three-line Theorem (cf.~\cite[Theorem~5.2.1]{RS}), 
\[
\sup_{0\leq \re(\theta) \leq 1}|f(\theta)| \leq  N(2c_{2n})^{2n} c_{T} \tK_{\alpha}\bigl(\uF^{(n)}\bigr)\bigl\|\Phi\bigr\| \bigl\|\Psi\bigr\|,
\]
which, together with the fact that $N\leq M(n)$, the universal upper bound on the number of terms in the sum \eqref{eq-ooT-expansion}, concludes the proof of \ref{item-regbound3}, the last item.
\end{proof}

\begin{Prop}[Estimate of fully contracted Wick monomials]
\label{RFCT}
Let $n \in \NN$  and a constant $c_n$ depending only on $n$, $\uF^{(n)} = (F_1,\dotsc, F_n)\in (L^2(\RR^d\times \RR^d))^n$. Let $T$ be a fully contracted Wick monomial of length $n$ and let us introduce: 
\begin{equation}
E \colon \uF^{(n)} \to \bigl\langle \Omega \,\big\vert\, T(0;\uF^{(n)})\Omega \bigr\rangle.
\end{equation}
Then, for any $0\leq \gamma, \delta\leq 1$ we have 
\begin{equation}
\bigl\| R_0(z)^{\delta} \bigl(T(z;\uF^{(n)}) - E(\uF^{(n)}) \bigr) R_0(z)^{\gamma} \bigr\| \leq  c_T\cdot c_n 
\tK_{\alpha}(\uF^{(n)}),
\end{equation}
with $\alpha = 1-\frac{1-\delta-\gamma}{n}$.
\end{Prop}

\begin{proof} We again employ the abbreviations $\scrL$, $\Delta_\ab$ and $\Delta_\af$ from \eqref{abbrev-LDelta2}.
From the proof of Lemma~\ref{HBFCTANFCO}, we know that $T(z;\uF^{(n)}) - E(\uF^{(n)})$ can be written as follows: 
\begin{align*}
    T(z;\uF^{(n)}) - E(\uF^{(n)}) & = - \sum^N_{k=1} \Biggl(\int \prod^n_{i = 1} F_i(k_i, q_i) \scrL\, \Delta_\ab\Delta_\af  \\
 &  \biggl[ \prod_{j\in \llbracket 1, k \rrbracket \cap \scrA}\frac{1}{R_j} ( H_0-z) \prod_{j \in \llbracket k, n \rrbracket \cap \scrA} R_0(z-R_j)  \biggr]\prod_{j=1}^n dq_j dk_j\Biggr).
\end{align*}
Moreover, it follows from Remark~\ref{rem-Wick} that for any $j \in \scrA$, $R_j \neq 0$. Therefore, 
\begin{align*}
  &  R_0(z)^{\delta} \bigl(T(z;\uF^{(n)}) - E(\uF^{(n)}) \bigr) R_0(z)^{\gamma}   = - \sum^N_{k=1}   \Biggl(\int \prod^n_{i = 1} F_i(k_i, q_i) \scrL\, \Delta_\ab\Delta_\af   \\
& \qquad  \biggl[ \prod_{j\in \llbracket 1, k \rrbracket \cap \scrA}\frac{1}{R_j} ( H_0-z)^{1-\gamma-\delta} \prod_{j \in \llbracket k, n \rrbracket \cap \scrA} R_0(z-R_j)  \biggr]\prod_{j=1}^n dq_j dk_j\Biggr),
\end{align*}
which can be estimated as follows
\begin{align*}
&    \Bigl\|R_0(z)^{\delta} \bigl(T(z;\uF^{(n)}) - E(\uF^{(n)}) \bigr) R_0(z)^{\gamma}\Bigr\|   \leq \sum^N_{k=1}  \Biggl(\int \prod^n_{i = 1} \bigl| F_i(k_i, q_i)\bigr|  \bigl|\scrL\bigr| \Delta_\ab\Delta_\af \\
 & \qquad \biggl\| \prod_{j\in \llbracket 1, k \rrbracket \cap \scrA}\frac{1}{R_j} R_0(z-R_k) ^{\gamma+\delta} \prod_{j \in \llbracket k+1, n \rrbracket \cap \scrA} R_0(z-R_j)  \biggr\|\prod_{j=1}^n dq_j dk_j \Biggr).
\end{align*}

From the definition of $R_k$ (see \eqref{Rcomponents}) and the fact that $R_k\neq 0$, there exists an index $j \in I_\ab$ such that $j\leq k < f_{\ab}(j)$ or $j \in I_\af$ such that $j\leq k < f_{\af}(j)$. We may assume, without loss of generality that $j\in I_\ab$. Let us consider
\begin{align*}
\forall i\in \llbracket 1, n \rrbracket \setminus \{j\}:\qquad   \alpha_i &= 1-\frac{1}{n} + \frac{\delta+\gamma}{n} \\
 \alpha_{j} &= 1 - \frac{1}{n}-\frac{n-1}{n}(\delta+\gamma).
\end{align*}
Then, using Definition~\ref{RegularOperator} we obtain $\{\beta_{i}\}_{i=1}^n$ with  $0\leq \beta_{i} \leq \alpha_i $, for any $i\in \llbracket 1 , n \rrbracket $, and $\beta_i=0$ for $i\in\psigma$; and invoking \eqref{regularityproperty} together with Remark \ref{rem-Wick}, we arrive at
\begin{align*}
& \Bigl\|R_0(z)^{\delta} \bigl(T(z;\uF^{(n)}) - \bigl\langle \Omega \,\big\vert\, T(0;\uF^{(n)}) \Omega \bigr\rangle  \bigr) R_0(z)^{\gamma}\Bigr\|\\
& \quad \leq c_T
  \sum^N_{k=1}  \int  \prod^n_{i = 1}   \frac{\bigl| F_i(k_i, q_i)\bigr|   }{[\wb(q_i)]^{\alpha_i-\beta_{i}}[\wf(k_i)]^{\beta_{i}}} \Delta_\ab\Delta_\af   \Bigl\|  R_0(z-R_k) ^{\gamma+\delta}   \Bigr\|\prod_{j=1}^n dq_j dk_j.
\end{align*}
We can conclude using the fact that:
\begin{equation*}
\Bigl\|  R_0(z-R_k) ^{\gamma+\delta}   \Bigr\| \leq \frac{1}{[\wb(q_j)]^{\gamma+\delta}}
\end{equation*}
and the estimates \eqref{EstimatesForCauchySchwartz}, following the same, but slightly simpler, sequence of estimates as at the end of the proof of Proposition \ref{RNFCT}. Note that the weight $\prod_{1=1}^n W(k_i,q_i)^{2n}\geq 1$, cf.~\eqref{tildeK} and \eqref{W-function}, appearing in $\tK$ is not needed here.
\end{proof}

\subsection{Estimates of the renormalized handed blocks of operators}

We first introduce the following set which is convenient to state our  results. 
\begin{equation}
\label{SetOfF}
S_{\mathrm{sc}} = \Bigset{ F \in L^2(\RR^d \times \RR^d ) }{ F(k,q) = h(k,q)g(k\pm q), h\in L^2(\RR^d \times \RR^d ) }.
\end{equation} 
Here $g$ is a spatial cutoff as in Hypothesis~\ref{MainHypothesis}
For any $\uF^{(k)} \in S^{k}_{\mathrm{sc}}$ we define 
\begin{equation}
    K_{\alpha}\bigl(\uF^{(k)}\bigr) = \prod^k_{i=1}\biggl( \int \frac{|F_i(k_i,q_i)|^2}{\wb(q_i)^{2\alpha}} dk_i dq_i\biggr)^{\frac12},
\end{equation}
for $\alpha \in [0,1]$. 

\begin{Prop}[Estimates of the renormalized blocks]
\label{RegularityofthegeneralisedGsets}
Let $N,\ell\in \NN$ with $\ell< N$ and $\uF^{(\ell)} \in S^{\ell}_{\mathrm{sc}}$. There exists a constant $C(\ell) >0$, such that the following holds. 
\begin{enumerate}[label = \textup{(\arabic*)}]
\item\label{item-Gbounds1} if $\us \in \scrS^{(\ell)}_\Right$ then: 
\begin{equation}
\Bigl\| T_\us^{(\ell)}\bigl(z;\uF^{(\ell)}\bigr) R_0(z)^{1-\frac{\ell}{N}}  \Bigr\| \leq C(\ell) K_{ 1-\frac{1}{N}}\bigl(\uF^{(\ell)}\bigr).
\end{equation}
\item\label{item-Gbounds2} if $\us \in \scrS^{(\ell)}_\Left$ then: 
\begin{equation}
\Bigl\| R_0(z)^{1-\frac{\ell}{N}} T_\us^{(\ell)}\bigl(z;\uF^{(\ell)}\bigr)  \Bigr\| \leq C(\ell) K_{ 1-\frac{1}{N}}\bigl(\uF^{(\ell)}\bigr).
\end{equation}
\item\label{item-Gbounds3} if $\us \in \scrS^{(\ell)}_\leftrightarrow$, then for any $\alpha, \beta \in [0,1]$ such that $\alpha+\beta = 1 - \frac{\ell}{N}$ we have: 
\begin{equation}
\Bigl\| R_0(z)^{\alpha} T_\us^{(\ell)}\bigl(z;\uF^{(\ell)}\bigr) R_0(z)^{\beta} \Bigr\| \leq C(\ell) K_{ 1-\frac{1}{N}}\bigl(\uF^{(\ell)}\bigr).
\end{equation}
\end{enumerate}
\end{Prop}

\begin{proof}
Let us consider \ref{item-Gbounds1} and assume that $\us \in \scrS^{(\ell)}_\Right$. From Lemma~\ref{InductionLemma} we know that 
\[
T_\us^{(\ell)}\bigl(z;\uF^{(\ell)}\bigr) = \sum^{N_1}_{i=1}  \rightT_i\bigl(z;\uF^{(\ell)}\bigr),
\]
and that there exists a constant $M=M(\ell)$ such that $N_1 \leq M(\ell)$. As consequence, 
\begin{equation*}
\Bigl\| T_\us^{(\ell)}\bigl(z;\uF^{(\ell)}\bigr)  R_0(z)^{1-\frac{\ell}{N}}  \Bigr\| \leq  \sum^{N_1}_{i=1}  \Bigl\|  \rightT_i\bigl(z;\uF^{(\ell)}\bigr) 
 R_0(z)^{1-\frac{\ell}{N}}  \Bigr\|.
\end{equation*}
 Using Proposition~\ref{RNFCT}, we have: 
\begin{align*}
& \forall i \in \llbracket 1, N_1 \rrbracket: \quad  \Bigl\| \rightT_i\bigl(z; \uF^{(\ell)}\bigr) 
 R_0(z)^{1-\frac{\ell}{N}}  \Bigr\|  \leq  c_\ell 
\tK_{1-\frac{1}{N}}\bigl(\uF^{(\ell)}\bigr).
\end{align*}
Remember that, from Lemma~\ref{InductionLemma}, $c_{\rightT_i}=1$ for any $i\in  \llbracket 1, N_1 \rrbracket$. We then use Lemma~\ref{EstimateWtihRespectToOtherPowers}  to conclude that: 
\begin{equation}
 \tK_{1-\frac{1}{N}}\bigl(\uF^{(\ell)}\bigr)\leq C_{2(1-\frac{1}{N}),2\ell}^\ell K_{ 1-\frac{1}{N}}\bigl(\uF^{(\ell)}\bigr).
\end{equation}
In the rest of the proof, we abbreviate $C:= C_{2(1-\frac{1}{N}),2\ell}$.
Consequently, 
\[
\Bigl\| T_\us^{(\ell)}\bigl(z;\uF^{(\ell)}\bigr)  R_0(z)^{1-\frac{\ell}{N}}  \Bigr\|   \leq M(\ell) C^\ell c_\ell   K_{ 1-\frac{1}{N}}\bigl(\uF^{(\ell)}\bigr).
\]
The statement \ref{item-Gbounds2} follows from \ref{item-Gbounds1} by taking the adjoint. The statement \ref{item-Gbounds3} can be derived in the same way, using Lemma~\ref{InductionLemma} we have
 \begin{equation*}
  \begin{aligned}
 T_\us^{(\ell)} & = \sum^{N_1}_{i=1}  \lrT_i +  \sum^{N_2}_{i=1}  (T_i - E_i),\\
   E_i \colon &  (F_1,\dots, F_\ell)\to \bigl\langle\Omega \, \big\vert\, T_i\bigl(0; F_1, \dots, F_\ell\bigr) \Omega \bigr\rangle.
 \end{aligned}
\end{equation*}
and we conclude using Proposition~\ref{RNFCT} and Proposition~\ref{RFCT}.
\end{proof}

\begin{Prop} \label{FinalEstimate1}Let $N, k \in \NN$. Let $\uut \in \scrT^{(N,k)}$, $z\in\CC_-^*$ and $\uF^{(k)} \in S^{k}_{\mathrm{sc}}$, then there exists a constant $B(N)$, only depending on $N$ and the masses $\mb$ and $\mf$,  such that
\[
\bigl\| S_\uut^{(N,k)}\bigl(z;\uF^{(k)}\bigr) \bigr\|\leq \frac{1}{|z|} \Bigl(\frac{B(N)}{|z|^{\frac{1}{N}}}\Bigr)^{k}  K_{1-\frac{1}{N+1}}\bigl(\uF^{(k)}\bigr).
\]
\end{Prop}

\begin{proof}
    Let $\uut = (\us_1, \dots, \us_\ell) \in \scrT^{(N,k)}$. Here $\us_{\ell'}\in\scrS^{(j_{\ell'})}$ for $\ell'\in\llbracket 1,\ell\rrbracket$ and $j_1+\dotsc+j_\ell = k$. From Definition \ref{def-RenSummands}, we have: 
    \[
      S_\uut^{(N,k)}\bigl(z;\uF^{(k)}\bigr) = R_0(z) \prod^{\ell}_{i=1}\Bigl[ T^{(j_i)}_{\us_i}\bigl(z;F_{b(i;\uut)},\dotsc,F_{e(i;\uut)}\bigr) R_0(z)\Bigr].
    \]
Let $\{\nu_{j_{\ell'}}\}_{\ell' \in \llbracket 1, \ell+1 \rrbracket}$ and $\{\mu_{j_{\ell'}}\}_{\ell' \in \llbracket 1, \ell\rrbracket}$ be two collections of real numbers defined for $\ell'\in\llbracket 1,\ell\rrbracket$ as follows
\begin{itemize}[left=0pt .. \parindent]
    \item  if $s_{\ell'} \in \scrS^{(j_{\ell'})}_\Right$ then $\nu_{j_{\ell'}}=0$ and $\mu_{j_{\ell'}} = 1-\frac{j_{\ell'}}{N+1}$,
    \item  if $s_{\ell'} \in \scrS^{(j_{\ell'})}_\Left$ then $\nu_{j_{\ell'}}= 1-\frac{j_{\ell'}}{N+1}$ and $\mu_{j_{\ell'}} =0$,
    \item  if $s_{\ell'} \in \scrS^{(j_{\ell'})}_\LR$ then $\nu_{j_{\ell'}}= \mu_{j_{\ell'}} =  \frac12-\frac{j_{\ell'}}{2(N+1)}$,
    \item $\nu_{j_{\ell+1}} = 0$.
\end{itemize}
Note that for any $\ell'\in\llbracket 1,\ell\rrbracket$, we have $\mu_{j_{\ell'}}+\nu_{j_{\ell'+1}}\leq 1$, which follows from the definition of $\scrT^{(N,k)}$, cf.~Definition~\ref{def-tuples}. Rewrite
\begin{align*}
     &S_\uut^{(N,k)}\bigl(z;\uF^{(k)}\bigr)\\
     & = R_0(z) \prod^{\ell}_{i=1} T^{(j_i)}_{\us_i}\bigl(z;F_{b(i;\uut)},\dotsc,F_{e(i;\uut)}\bigr) R_0(z)\\
     & = R_0(z)^{1-\nu_{j_1}} \prod^{\ell}_{i=1} \Bigl( R_0(z)^{\nu_{j_{i}}}T^{(j_i)}_{\us_i}\bigl(z;F_{b(i;\uut)},\dotsc,F_{e(i;\uut)}\bigr) R_0(z)^{\mu_{j_{i}}}\Bigr)R_0(z)^{1-\nu_{j_{i+1}}-\mu_{j_{i}}}.
\end{align*}
From Proposition~\ref{RegularityofthegeneralisedGsets}, we have 
\[
\Bigl\|R_0(z)^{\nu_{j_{i}}}T^{(j_i)}_{\us_i}\bigl(z;F_{b(i;\uut)},\dotsc,F_{e(i;\uut)}\bigr) R_0(z)^{\mu_{j_{i}}}\Bigr\|\leq C(j_i) K_{1-\frac{1}{N+1}}\bigl(F_{b(i;\uut)},\dotsc,F_{e(i;\uut)}\bigr) 
\]
and as a consequence, noting that $\|R_0(z)\| = |z|^{-1}$,
\begin{align*}
    & \bigl\|S_\uut^{(N,k)}\bigl(z;\uF^{(k)}\bigr)\bigr\|\\
   &  \quad \leq  \Bigl\|R_0(z)^{1-\nu_{j_1}} \prod^{\ell}_{i=1} \Bigl[\Bigl\{ R_0(z)^{\nu_{j_{i}}}T^{(j_i)}_{\us_i}\bigl(z;F_{b(i;\uut)},\dotsc,F_{e(i;\uut)}\bigr) R_0(z)^{\mu_{j_{i}}}\Bigr\}   \\
   & \qquad \quad  R_0(z)^{1-\nu_{j_{i+1}}-\mu_{j_{i}}}\Bigr]\Bigr\|\\
     & \quad  \leq \Bigl(\frac1{|z|}\Bigr)^{1-\nu_{j_1} + \sum^\ell_{i=1} (1-\nu_{j_{i+1}}-\mu_{j_{i}})}\prod^\ell_{i=1} \Bigl[C(j_i) K_{1-\frac{1}{N+1}}\bigl(F_{b(i;\uut)},\dotsc,F_{e(i;\uut)}\bigr) \Bigr].
\end{align*}
Note that, since $\nu_{\ell+1}=0$,
\begin{align*}
   1-\nu_{j_1} + \sum^\ell_{i=1} (1-\nu_{j_{i+1}}-\mu_{j_{i}}) & =  1+ \sum^\ell_{i=1}(1 - \mu_{j_i} - \nu_{j_i})\\
    =1+ \sum^\ell_{i=1}\bigl(1 - \bigl(1-\frac{j_{i}}{N+1}\bigr)\bigr) & = 1+\frac{k}{N+1},
\end{align*}
leading to 
\begin{align*}
    & \bigl\|S_\uut^{(N,k)}\bigl(z;\uF^{(k)}\bigr)\bigr\| \leq  \frac{\prod^\ell_{i=1} C(j_i)}{|z|^{1+\frac{k}{N+1}}} K_{1-\frac{1}{N+1}}\bigl(\uF^{(k)}\bigr).
\end{align*}
Abbreviating
\[
B(N) = \max\bigl\{1,\max\bigset{ C(j)}{ j \in \llbracket 1, N \rrbracket } \bigr\},
\]
we then have 
\begin{align*}
    \bigl\|S_\uut^{(N,k)}\bigl(z;\uF^{(k)}\bigr)\bigr\| & \leq  \frac{\prod^\ell_{i=1} C(j_i)}{|z|^{1+\frac{k}{N+1}}} K_{1-\frac{1}{N+1}}\bigl(\uF^{(k)}\bigr)\\
    &\leq  \frac{ B(N)^\ell}{|z|^{1+\frac{k}{N+1}}} K_{1-\frac{1}{N+1}}\bigl(\uF^{(k)}\bigr).
\end{align*}
which concludes the proof, since $\ell\leq k$.
\end{proof}

\begin{Prop}
\label{FinalEstimate2}
Let $N, k \in \NN$. Let $\uut \in \scrT^{(N,k)}$, $z\in\CC_-^*$ and $\uF^{(k)}_1, \uF^{(k)}_2 , \in  S^{k}_{\mathrm{sc}}$, then there exists a constant $B(N)$, depending only on $N$ and the masses $\mb$ and $\mf$, such that
\[
\begin{aligned}
&\bigl\| S_\uut^{(N,k)}\bigl(z;\uF^{(k)}_1\bigr)-S_\uut^{(N,k)}\bigl(z;\uF^{(k)}_2\bigr) \bigr\| \leq  \frac{1}{|z|} \Bigl(\frac{B(N)}{|z|^{\frac{1}{N+1}}}\Bigr)^{k}   \sum^k_{i=1}  K_{1-\frac{1}{N+1}}\bigl(\uF^{(i,k)}_{1,2}\bigr),
\end{aligned}
\]
where 
\[
\uF^{(i,k)}_{1,2} = (F_{1;1},\dotsc, F_{1;i-1}, F_{1;i}-F_{2;i},F_{2;i+1},\dotsc,F_{2;k}).
\]
\end{Prop}

\begin{proof}
    Let $\uut = (\us_1, \dots, \us_\ell) \in \scrT^{(N,k)}$. We therefore have 
    \begin{align*}
      S_\uut^{(N,k)}\bigl(z;\uF^{(k)}_1\bigr) & = R_0(z) \prod^{\ell}_{i=1}\Bigl[ T^{(j_i)}_{\us_i}\bigl(z;F_{1;b(i;\uut)},\dotsc,F_{1;e(i;\uut)}\bigr) R_0(z)\Bigr]\\
      S_\uut^{(N,k)}\bigl(z;\uF^{(k)}_2\bigr) & = R_0(z) \prod^{\ell}_{i=1} \Bigl[T^{(j_i)}_{\us_i}\bigl(z;F_{2;b(i;\uut)},\dotsc,F_{2;e(i;\uut)}\bigr) R_0(z)\Bigr].
    \end{align*}
    As a consequence, 
    \begin{align*}
      & S_\uut^{(N,k)}\bigl(z;\uF^{(k)}_1\bigr) -   S_\uut^{(N,k)}\bigl(z;\uF^{(k)}_2\bigr)  \\
      &  = \sum^\ell_{i=1} R_0(z) \prod^{i-1}_{j=1} \Bigl[ T^{(j_i)}_{\us_i}\bigl(z;F_{2;b(i;\uut)},\dotsc,F_{2;e(i;\uut)}\bigr) R_0(z)\Bigr]\\
      & \qquad \Bigl(T^{(j_i)}_{\us_i}\bigl(z;F_{1;b(i;\uut)},\dotsc,F_{1;e(i;\uut)}\bigr)-T^{(j_i)}_{\us_i}\bigl(z;F_{2;b(i;\uut)},\dotsc,F_{2;e(i;\uut)}\bigr)\Bigr) R_0(z) \\
      & \qquad \prod^{\ell}_{j=i+1} \Bigl[T^{(j_i)}_{\us_i}\bigl(z;F_{2;b(i;\uut)},\dotsc,F_{2;e(i;\uut)}\bigr) R_0(z)\Bigr]
    \end{align*}
    and 
    \begin{align*}
      &  T^{(j_i)}_{\us_i}\bigl(z;F_{1;b(i;\uut)},\dotsc,F_{1;e(i;\uut)}\bigr)-T^{(j_i)}_{\us_i}\bigl(z;F_{2;b(i;\uut)},\dotsc,F_{2;e(i;\uut)}\bigr)\\
      & = \sum^{e(i;\uut)}_{j=b(i;\uut)} T^{(j_i)}_{\us_i}\bigl(z;F_{2;b(i;\uut)},\dotsc,F_{2;j-1},F_{1,j}-F_{2;j},F_{1;j+1},\dotsc, F_{1;e(i;\uut)}\bigr).
    \end{align*}
    The same strategy as the one used in the proof of Proposition \ref{FinalEstimate1} can now be used to conclude the proof.
\end{proof}

\section{Proof of Theorem \ref{MainTh} }\label{Sec-MainProof}

The goal of this section is to prove Theorem~\ref{MainTh}. We already know, from \cite{AlvaMoll2021}, that if $p>\frac{d}{2}-\frac{3}{4}$ then $H_{\Lambda} - E^{(2)}_{\Lambda}$ converges in norm resolvent sense to a self-adjoint operator.  We first prove the following theorem which generalise this result: 

\begin{Th}
\label{THinterm}
Let $p>\frac{d}{2}-1$. Pick an $N\in\NN$ with $p> \frac{d}{2}- \frac{N}{N+1}$. Then, in the limit $\Lambda\to+\infty$, the operator $H_{\Lambda}-E^{(N)}_{\Lambda}$ converges in norm resolvent sense to a self-adjoint and operator $H$, which is bounded from below. Moreover, the operator $H$ does not depend on the choice of the cutoff function $\chi$.
\end{Th}

\begin{proof}
  Theorem~\ref{thm-reordering} shows that there exists $C_N(\Lambda)>0$ such that for any $z\in \CC$ fulfilling $\re(z)\leq -C_{N}(\Lambda)$, the Neumann series of the resolvent of $H_\Lambda+ E^{(N)}_\Lambda$ can be reordered as 
  \begin{equation}\label{ReorderedIn7}
R_\Lambda(z):=      R_0(z) + \sum_{k=1}^\infty
    \sum_{[\uut]\in\scrT^{(N,k)}/\!\sim} S^{(N,k)}_\uut\bigl(z;G_{s_1,\Lambda},\dotsc, G_{s_k,\Lambda}\bigr),
  \end{equation}
where for $\uut = (\us_1,\dotsc,\us_\ell)\in\scrT^{(N,k)}$.
  From Proposition \ref{FinalEstimate1}, we can estimate $R_\Lambda(z)$ as follows
  \begin{equation}\label{FinThmStepI}
  \bigl\|R_\Lambda(z)\bigr\|  \leq \frac{1}{|z|} + \sum^{\infty}_{k=1} \frac{1}{|z|}  \sum_{[\uut]\in\scrT^{(N,k)}/\!\sim} \Bigl(\frac{B(N)}{|z|^{\frac{1}{N+1}}}\Bigr)^{k} K_{1-\frac{1}{N+1}}\bigl(\uF^{(k)}_{\us,\Lambda} \bigr),
     \end{equation}
    where 
  \[
\uF^{(k)}_{\us,\Lambda}  = (G_{s_1,\Lambda},\dotsc,G_{s_k,\Lambda}).
\]
From Lemma \ref{RegularityOfK}, we have
\[
 K_{1-\frac{1}{N+1}}\bigl(\uF^{(k)}_{\us,\Lambda} \bigr)\leq \kappa_{\alpha,p}^k,
 \]
 where 
 \[
 \kappa_{\alpha,p} =  \max\{\|h^{(1)}\|_\infty,\|h^{(2)}\|_\infty\}  \|g\|_2  \Bigl(\int \frac{1}{\wb(q)^{\frac{2}{N+1}+2p} } dq\Bigr)^{\frac{1}{2}}.
\]
Exploiting that by Proposition~\ref{prop-tuple-bijection}, $\scrT^{(N,k)}/\!\sim$ has cardinality  $4^k$, we may continue the estimate \eqref{FinThmStepI} and obtain
  \begin{equation*}
     \bigl\|R_\Lambda(z)\bigr\| \leq \frac{1}{|z|}\sum^{\infty}_{k=0}  4^{k} \Bigl(\frac{B(N)\kappa_{\alpha,p}}{|z|^\frac{1}{N+1}}\Bigr)^{k}.
  \end{equation*}
Choose $C_N^{(1)}\in\RR_+$ such that for any $z \in \CC$ with $\re(z)\leq -C_N^{(1)}$, we have 
\[
4 \frac{B(N)\kappa_{\alpha,p}}{|z|^{\frac{1}{N+1}} }\leq \frac12.
\]
As a conclusion, for $z\in\CC$ with $\re(z)\leq -C_N^{(1)}$, the reordered Neumann series \eqref{ReorderedIn7} converges in norm to the holomorphic operator-valued function $R_\Lambda(z)$. Since the limiting operator $R_\Lambda(z)$ agrees with $(H_\Lambda-E_\Lambda^{(N)}-z)^{-1}$ for $z\in\CC$ with $\re(z)\leq - C_N(\Lambda)$, we conclude by unique continuation that $R_\Lambda(z)$ given by the series  \eqref{ReorderedIn7} actually equals $(H_\Lambda-E_\Lambda^{(N)}-z)^{-1}$ for all $z$ with $\re(z)\leq -C_N^{(1)}$. Note in particular that $H_\Lambda -E_\Lambda^{(N)} \geq - C_N^{(1)}\one$.

  We proceed to argue that $\Lambda\to R_\Lambda(z) = (H_\Lambda-E_\Lambda^{(N)}-z)^{-1}$ is Cauchy in norm, possibly after choosing $z$ further left in the complex plane. Let $\varepsilon >0$ and let $\Lambda_1,\Lambda_2\in\RR$ with $\Lambda_1,\Lambda_2>0$ and abbreviate, for $k\in\NN$,  $\us\in\scrS^{(k)}$ and $\ell = 1,2$: $\uF^{(k)}_{\us,\ell} = (G_{s_1,\Lambda_\ell},\dotsc, G_{s_k,\Lambda_\ell})$. 
  
  By what has just been established, for any $z\in \CC$ fulfilling $\re(z) \leq -C^{(1)}_N$, $R_{\Lambda_1}(z) - R_{\Lambda_2}(z)$ can be computed as a difference of two absolutely convergent series: 
  \begin{align*}
      \sum_{k=1}^\infty
    \sum_{[\uut]\in\scrT^{(N,k)}/\!\sim} \Bigl(S^{(N,k)}_\uut\bigl(z;\uF^{(k)}_{\us,1}\bigr)-S^{(N,k)}_\uut\bigl(z;\uF^{(k)}_{\us,2}\bigr)\Bigr)
  \end{align*}
  with the convention that for $\uut=(\us_1,\dotsc,\us_\ell)\in\scrT^{(N,k)}$, $\us= \us_1\circ\cdots\circ\us_\ell = (s_1,\dotsc,s_k)$.
This series can be estimated using Proposition \ref{FinalEstimate2}:
\begin{align*}
  &   \bigl\|R_{\Lambda_1}(z) - R_{\Lambda_2}(z) \bigr\|\\
  & \quad \leq  \sum_{k=1}^\infty
    \sum_{[\uut]\in\scrT^{(N,k)}/\!\sim} \bigl\|S^{(N,k)}_\uut\bigl(z;\uF^{(k)}_{\us,1}\bigr)-S^{(N,k)}_\uut\bigl(z;\uF^{(k)}_{\us,2}\bigr)\bigr\| \\
    & \quad \leq \sum_{k=1}^\infty
    \sum_{[\uut]\in\scrT^{(N,k)}/\!\sim} \frac{1}{|z|} \Bigl(\frac{B(N)}{|z|^{\frac{1}{N+1}}}\Bigr)^{k}  \sum^k_{i=1} K_{1-\frac{1}{N+1}}\bigl(\uF^{(i,k)}_{\us,\Lambda_1,\Lambda_2}\bigr),
\end{align*}
where we used the linearity of $\uF^{(k)}\mapsto S^{(N,k)}_\uut(z;\uF^{(k)})$ and the computation
$\uF^{(k)}_{\us,1} - \uF^{(k)}_{\us,2} = \sum_{i=1}^k \uF^{(i,\ell)}_{\us,\Lambda_1,\Lambda_2}$ with 
\[
\uF^{(i,k)}_{\us,\Lambda_1,\Lambda_2} = \bigl(G_{s_1,\Lambda_1},\dotsc, G_{s_{i-1},\Lambda_1}, G_{s_i,\Lambda_1}-G_{s_i,\Lambda_2},G_{s_{i+1},\Lambda_2},\dotsc,G_{s_\ell,\Lambda_2}\bigr).
\]
From Lemma~\ref{RegularityOfK}~\ref{item-K-Cauchy}, there exists $\Lambda_0>0$, such that for $\Lambda_1, \Lambda_2\geq \Lambda_0$, we have: 
\begin{equation*}
    K_{1-\frac{1}{N+1}}\bigl(\uF^{(i,k)}_{\us,\Lambda_1,\Lambda_2}\bigr) \leq \varepsilon,
\end{equation*}
for all $i\in\llbracket 1,k\rrbracket$ and $\us\in\scrS^{(k)}_0$.
We may now choose $C_N^{(2)}\in \RR$ with $C_N^{(2)}\geq C_N^{(1)}$, such that for any $z \in \CC$ with $\re(z) \leq -C^{(2)}_N$, we have 
\[
4 \frac{B(N)}{|z|^{\frac{1}{N+1}}}\leq \frac12\quad \textup{and}\quad
\frac{1}{|z|} \leq \frac12.
\]
In that case 
\begin{align*}
    \bigl\|R_{\Lambda_1}(z) - R_{\Lambda_2}(z) \bigr\| &\leq \varepsilon \sum_{k=1}^\infty
    k 4^{k} \frac{1}{|z|} \Bigl(\frac{B(N)}{|z|^{\frac{1}{N+1}}}\Bigr)^{k} \\
    & = \varepsilon \frac{1}{|z|}\frac{4 \frac{B(N)}{|z|^{\frac{1}{N+1}}}}{\bigl(1-4\frac{B(N)}{|z|^{\frac{1}{N+1}}}\bigr)^2}\\
    & \leq 2 \varepsilon \frac{1}{|z|}\leq \varepsilon.
\end{align*}
We may therefore conclude that for $z\in\CC$ with $\re(z) \leq -C_N^{(2)}$, the norm limit $R(z):= \lim_{\Lambda\to +\infty}R_{\Lambda}(z)$ exists. Note that for $z, z' \in \CC$ fulfilling $\re(z), \re(z') \leq -C_N^{(2)}$, we have
\begin{align*}
  R(z)^* & = R(\overline{z})\\
        R(z) - R(z') & = (z-z') R(z) R(z').
\end{align*}
This is just standard resolvent identities for finite $\Lambda$ that are carried over to the norm limit as $\Lambda\to+\infty$.
We now aim to prove that for any $\Psi \in \scrH$ we have $z R(z) \Psi \to \Psi$ when $\re(z) \to -\infty$. First, for any $\Phi \in \scrH$
\begin{align*}
    \bigl|\bigl\langle\Phi \,\big|\,   z R(z) \Psi - \Psi \bigr\rangle\bigr| & = \lim_{\Lambda \to \infty} \bigl|\bigl\langle\Phi \,\big|\,   z R_{\Lambda}(z) \Psi - \Psi \bigr\rangle\bigr|\\
    & \leq \|\Phi\| \bigl\| z R_{0}(z) \Psi - \Psi \bigr\|  +  \sum^{\infty}_{k=1}  4^{k} \Bigl(\frac{B(N)\kappa_{\alpha,p}}{|z|^{\frac{1}{N+1}}}\Bigr)^k \|\Psi\| \|\Phi\| \\
    & =  \|\Phi\| \bigl\| z R_{0}(z) \Psi - \Psi \bigr\| +\frac{4B(N)\kappa_{\alpha,p} }{|z|^{\frac{1}{N+1}} -4 B(N)\kappa_{\alpha,p} }  \|\Psi\| \|\Phi\|.
\end{align*}
Since the right-hand side converges to $0$ in the limit $\re(z)\to -\infty$, we conclude that $z R(z) \Psi \to \Psi$ when $\re(z) \to -\infty$. Hence, \cite[Theorem~D.1]{AlvaMoll2021} can then be invoked to obtain the limiting semi-bounded operator $H$.

Let us now prove that this limit does not depend on the choice of cutoff function $\chi$. Consider $\chi_1, \chi_2$ fulfilling  Hypothesis~\ref{Hypothesis-chi}. Define for $j\in\{1,2\}$ 
\begin{align*}
    R_{\chi_j,\Lambda}(z) & = \bigl(H\bigl(G^{(1)}_{\chi_j,\Lambda}, G^{(2)}_{\chi_j,\Lambda}\bigr)-z\bigr)^{-1}\\
    R_{\chi_j}(z) & = \lim_{\Lambda \to +\infty} R_{\chi_j,\Lambda}(z).
\end{align*}
Here we have amended our notation by adding the choice of ultraviolet cutoff function, $\chi_1$ or $\chi_2$, to the subscript of the coupling functions $G^{(1)}_\Lambda$ and $G^{(2)}_\Lambda$. 

First, there exists $C_N >0$ such that for any $\Lambda>0$ and any $z \in \CC$ fulfilling $\Re(z) \leq - C_N$ both $R_{\chi_1,\Lambda}(z)$ and $R_{\chi_2,\Lambda}(z)$ can be expanded as absolutely convergent reordered Neumann series.

 Proceeding as in the first part of the proof, we may estimate the norm of the difference of the two resolvents (with the same $\Lambda$) by subtracting the two reordered Neumann series as follows
\begin{align*}
  &   \bigl\|R_{\chi_1,\Lambda}(z) - R_{\chi_2,\Lambda}(z) \bigr\|\\
    & \quad \leq \sum_{k=1}^\infty
    \sum_{[\uut]\in\scrT^{(N,k)}/\!\sim} \frac{1}{|z|} \Bigl(\frac{B(N)}{|z|^{\frac{1}{N+1}}}\Bigr)^{k}  \sum^k_{i=1} K_{1-\frac{1}{N+1}}\bigl(\uF^{(i,k)}_{\us,\chi_1,\chi_2,\Lambda}\bigr),
\end{align*}
where 
\[
\uF^{(i,k)}_{\us,\chi_1,\chi_2,\Lambda} = \bigl(G_{s_1,\chi_1,\Lambda},\dotsc, G_{s_{i-1},\chi_1,\Lambda}, G_{s_i,\chi_1,\Lambda}-G_{s_i,\chi_2,\Lambda},G_{s_{i+1},\chi_2,\Lambda},\dotsc,G_{s_\ell,\chi_2,\Lambda}\bigr).
\]
Let $\varepsilon>0$. 
Invoking Lemma~\ref{RegularityOfK}~\ref{item-K-Cauchy2}, we obtain  a  $\Lambda_0'\geq \Lambda_0$, such that  $\|R_{\chi_1,\Lambda}(z) - R_{\chi_2,\Lambda}(z) \|\leq \varepsilon$ for $\Lambda\geq \Lambda_0'$. Taking the limit $\Lambda\to +\infty$ in this estimate yields $\|R_{\chi_1}(z) - R_{\chi_2}(z) \|\leq \varepsilon$. Since $\varepsilon>0$ was arbitrary, the proof is complete.
\end{proof}

\begin{proof}[Proof of Theorem \ref{MainTh}]
According to Theorem~\ref{THinterm}, for any $N\geq 1$ and $p>\frac{d}{2}-\frac{N}{N+1}$, there exists $H_{N}$ which is the norm resolvent limit of $H_{\Lambda}-E^{(N)}_{\Lambda}$ as $\Lambda\to +\infty$. The same argument as in \cite[Proof of Theorem~1]{AlvaMoll2021} can be used to prove that the counter-term $E^{(N)}_{\Lambda}$ can be replaced by $E_{\Lambda} = \inf(\sigma(H_{\Lambda}))$. 
\end{proof}

\appendix

\section{Useful Estimates}

\begin{lem}\label{EstimateWtihRespectToOtherPowers}
Let $F \in S_{\mathrm{sc}}$ (see \eqref{SetOfF}). Then  for any  set of exponents $\alpha$, $\beta$ and  $\gamma$ with $\alpha, \beta, \gamma\geq 0$, there exists $C_{\beta,\gamma}>0$, depending only on the two exponents $\beta$ and $\gamma$ as well as the masses $\mb$ and $\mf$, such that
we have 
\begin{equation*}
 \int\Bigl(1 + \frac{\wb(q)}{\wf(k)}\Bigr)^{\beta} \Bigl(1 + \frac{\wf(k)}{\wb(q)}\Bigr)^{\gamma}\frac{|F(k,q)|^2}{\wb(q)^{\alpha}} dk dq \leq C_{\beta,\gamma} \int\frac{|F(k,q)|^2}{\wb(q)^{\alpha}} dk dq.
\end{equation*}
\end{lem}
\begin{proof}
   First, recall from \eqref{SetOfF} that $F$ is of the form
   $F(k,q) = g(k\pm q)h(k,q)$ with $g$ being the spatial cutoff from Hypothesis~\ref{MainHypothesis}. We may assume without loss of generality that $F(k,q) = g(k- q)h(k,q)$, and compute
\begin{align*}
     & \int\Bigl(1 + \frac{\wb(q)}{\wf(k)}\Bigr)^{\beta} \Bigl(1 + \frac{\wf(k)}{\wb(q)}\Bigr)^{\gamma}\frac{|F(k,q)|^2}{\wb(q)^{\alpha}} dk dq\\
     & =  \int\Bigl(1 + \frac{\wb(q)}{\wf(v+q)}\Bigr)^{\beta} \Bigl(1 + \frac{\wf(v+q)}{\wb(q)}\Bigr)^{\gamma}\frac{|g(v)|^2 |h(v+q,q)|^2}{\wb(q)^{\alpha}} dv dq.
\end{align*}
Using the same strategy as the one used in the proof of \cite[Lemma B.5]{AlvaMoll2021} one can prove that there exists a constant $c>0$ depending on $\mb$ and $\mf$ such that
\[
\forall v,q\in\RR^d, \|v\|\leq 1:\quad \frac1{c}\wf(q+v)\leq \wb(q) \leq c \wf(q + v).
\]
As a consequence 
\begin{align*}
     & \int\Bigl(1 + \frac{\wb(q)}{\wf(k)}\Bigr)^{\beta} \Bigl(1 + \frac{\wf(k)}{\wb(q)}\Bigr)^{\gamma}\frac{|F(k,q)|^2}{\wb(q)^{\alpha}} dk dq\\
     & \quad =  \int\Bigl(1 + \frac{\wb(q)}{\wf(q+v)}\Bigr)^{\beta} \Bigl(1 + \frac{\wf(q+v)}{\wb(q)}\Bigr)^{\gamma}\frac{|g(v)|^2|h(q+v,q)|^2}{\wb(q)^{\alpha}} dk dq\\
     &\quad  \leq (1+c)^{\beta+\gamma} \int \frac{|g(v)|^2 |h(v+q,q)|^2}{\wb(q)^{\alpha}} dv dq\\
& \quad = (1+c)^{\beta+\gamma} \int \frac{|g(k-q)|^2 |h(k,q)|^2}{\wb(q)^{\alpha}} dk dq,    
\end{align*}
which completes the proof.
\end{proof}

\begin{lem}
\label{RegFermionProp2}
Let $n\in \NN$, $\{z_i\}_{i=1}^n$ such that $z_i \in  \CC_-^*$  and $F\in L^2(\RR^d)$, then for any family of real numbers $\{\alpha_i\}_{i=1}^n$ with $0\leq\alpha_i\leq 1$, we have $\af(F) \colon \scrD (\prod^n_{i=1}(H_0-z_i)^{\alpha_i}) \to  \scrD(\prod^n_{i=1}(H_0-z_i)^{\alpha_i} ) $ and
\begin{align*}
\biggl\| \prod^n_{i=1}(H_0-z_i)^{\alpha_i} \af(F)  \prod^n_{i=1}R_0(z_i)^{\alpha_i} \biggr\| & \leq  (n+1) \bigl\| F \bigr\|,\\
\biggl\|  \prod^n_{i=1} R_0(z_i)^{\alpha_i} \cf(F)  \prod^n_{i=1}(H_0-z_i)^{\alpha_i} \biggr\| & \leq  (n+1) \bigl\| F \bigr\|.
\end{align*}
\end{lem}

\begin{proof}
    We proceed by induction on $n$. The initialisation with $n=1$ is done in \cite[Lemma~B.1]{AlvaMoll2021}. Assume that the proposition holds true for some $n\in \NN$. Consider $\{z_i\}_{i=1}^{n+1}$ such that $z_i \in  \CC_-^*$  and a family of real numbers $\{\alpha_i\}_{i=1}^{n+1}$ with $0\leq\alpha_i\leq 1$, we aim to estimate 
    \[
    \biggl\| \prod^{n+1}_{i=1}(H_0-z_i)^{\alpha_i} \af(F)  \prod^{n+1}_{i=1}R_0(z_i)^{\alpha_i} \biggr\|.
    \]
    Let $\Psi,\Phi\in\scrH$ with $\Psi\in\scrD(H_0^{n+1})$.
Consider the bounded function $f\colon \set{z\in\CC}{ 0\leq \re(z)\leq 1}\to \CC$
\[
f (\theta) = \Bigl\langle \Psi \,\Big|\, \Bigl(\prod^{n}_{i=1}(H_0-z_i)^{\alpha_i}\Bigr) (H_0-z_{n+1})^{\theta} \af(F) R_0(z_{n+1})^{\theta} \Bigl(\prod^{n}_{i=1}R_0(z_i)^{\alpha_i} \Bigr) \Phi \Bigr\rangle.
\]
Observe, from the induction hypothesis, that for $\lambda\in \RR$, we have
\[
|f(i\lambda)| \leq (n+1) \| F \|\Psi\|\|\Phi\|.
\]
 Moreover, for $\lambda\in \RR$ let 
 $\Psi' = (H_0-\bar{z}_{n+1})^{-i\lambda}\Psi$,
 $\Phi' = R_0(z_{n+1})^{i\lambda}\Phi$,
 \[
 \Psi'' = \Bigl(\prod^{n}_{i=1}(H_0-\bar{z}_i)^{\alpha_i}\Bigr)\Psi'
 \quad \textup{and}\quad
 \Phi'' = \Bigl(\prod^{n}_{i=1}R_0(z_i)^{\alpha_i}\Bigr)\Phi'.
 \]
 Note that $\|\Psi'\|=\|\Psi\|$ and $\|\Phi'\|=\|\Phi\|$. Estimate
 \begin{align*}
     |f(1+i\lambda)| & = \Big| \Bigl\langle \Psi'' \,\Big|\,  (H_0-z_{n+1}) \af(F) R_0(z_{n+1})  \Phi''\Bigr\rangle\Bigr|\\
     & = \Bigl|\Bigl\langle \Psi'' \,\Big|\,  (H_0-z_{n+1}) \int R_0(z_{n+1}-\wf(k)) F(k)\af(k) dk  \Phi'' \Bigr\rangle\Bigr|\\
      & \leq \Bigl|\Bigl\langle \Psi' \,\Big|\, \prod^{n}_{i=1}(H_0-z_i)^{\alpha_i} \af(F) \prod^{n}_{i=1}R_0(z_i)^{\alpha_i}  \Phi'\Bigr\rangle\Bigr|\\
     & \quad +   \Bigl|\Bigl\langle \Psi'' \,\Big|\,  \int \wf(k) R_0(z_{n+1}-\wf(k)) F(k)\af(k) dk   \Phi''\Bigr\rangle\Bigr|.
 \end{align*}
The induction hypothesis applies to the first term on the right-hand side. By Hadamard’s Three-line Theorem \cite[Theorem~5.2.1]{RS}, we are done if we can bound the second term on the right-hand side by $\|F\|\|\Psi\|\Phi\|$. To see this, we estimate
\begin{align*}
    & \Bigl|\Bigl\langle \Psi' \,\Big|\, \prod^{n}_{i=1}(H_0-z_i)^{\alpha_i}  \int \wf(k) R_0(z_{n+1}-\wf(k)) F(k)\af(k) dk \prod^{n}_{i=1}R_0(z_i)^{\alpha_i}  \Phi' \Bigr\rangle\Bigr|\\
    &\qquad   \leq \| \Psi' \| \int \Bigl\| \prod^{n}_{i=1}(H_0-z_i)^{\alpha_i} R_0(z_i-\wf(k))^{\alpha_i} \Bigr\|  \\
    & \qquad\quad \Bigl\| \wf(k) R_0(z_{n+1}-\wf(k))^{\frac12} \af(k) R_0(z_{n+1})^{\frac12} \Phi'\Bigr\| dk\\
    &\qquad  \leq \|F\| \| \Psi \|\| \Phi \|,
\end{align*}
where we used the pull-through formula before estimating.
This completes the proof.
 \end{proof}

\begin{lem}
\label{RegChainFermionOp}
Let  $n\in \NN$ and $F\in L^2(\RR^d\times\RR^d)$ be such that for each $q\in\RR^d$, we have $\|\wf(\cdot)^n F(\cdot,q)\|<\infty$. There exists a constant $c_n>0$, only depending on $n$, such that: For $\{z_i\}_{i=1}^n$ with $z_i \in \CC_-^*$ and any  collection of real numbers $\{\gamma_i\}_{i=1}^n$  with $0\leq \gamma_i \leq 1$, for $i\in\llbracket 1,n\rrbracket$, we have for all  $q\in\RR^d$ that
\begin{align*}
 N(q) &:= \biggl\|\prod^n_{i=1} R_0\bigl(z_i-\wb(q)\bigr)^{\gamma_i} \af\bigl(F(.,q)\bigr) \prod^n_{i=1} \bigl(H_0-z_i+\wb(q)\bigr)^{\gamma_i}\biggr\|,\\
 N_*(q) &:= \biggl\|\prod^n_{i=1} \bigl(H_0-z_i+\wb(q)\bigr)^{\gamma_i} \cf\bigl(F(.,q)\bigr) \prod^n_{i=1} R_0\bigl(z_i-\wb(q)\bigr)^{\gamma_i} \biggr\|
\end{align*}
both satisfy the bound $N(q),N_*(q) \leq c_n\|W(\cdot,q)^n F(\cdot,q)\|$. Recall the notation $W$ from \eqref{W-function}.
\end{lem}

\begin{proof} It suffices to prove the first of the two estimates. The second estimate follows from the first by taking adjoints. For the purpose of this proof, we abbreviate $H_i = H_0-z_i$ and $R_i(\lambda) = R_0(z_i+\lambda)$, for $\lambda\leq 0$ and $i=1,2,\dotsc,n$. We prove only the first estimate, the second one being its adjoint.
First, we compute
\begin{align*}
&\prod^n_{i=1} R_i\bigl(-\wb(q)\bigr)^{\gamma_i} \af\bigl(F(.,q)\bigr) \prod^n_{i=1} \bigl(H_i+\wb(q)\bigr)^{\gamma_i}\\
 &\quad =  \prod^n_{i=1}\bigl(H_i+\wb(q)\bigr)^{1-\gamma_i}\\
 &\qquad 
 \biggl\{\prod_{i=1}^n R_i\bigl(-\wb(q)\bigr) \af\bigl(F(.,q)\bigr) \prod^n_{i=1} \bigl(H_i+\wb(q)\bigr)\biggr\}\prod_{i=1}^n R_i\bigl(-\wb(q)\bigr)^{1-\gamma_i}\\
& \quad =  \prod^n_{i=1}\bigl(H_i+\wb(q)\bigr)^{1-\gamma_i}\\
&\qquad
\biggl\{\sum_{S\subset \llbracket 1,n\rrbracket}  \prod_{j\in S} R_j\bigl(-\wb(q)\bigr) \af\bigl(\wf(.)^{|S|}F(.,q)\bigr) \biggr\} \prod^n_{i=1}R_i\bigl(-\wb(q)\bigr)^{1-\gamma_i}.
\end{align*}
Secondly, using Lemma \ref{RegFermionProp2} we have for $S\subset\llbracket 1,n\rrbracket$
\begin{align*}
  &\biggl\| \prod^n_{i=1}\bigl(H_i+\wb(q)\bigr)^{1-\gamma_i}\af\bigl(\wf(.)^{|S|}F(.,q)\bigr)  \prod^n_{i=1}R_i\bigl(-\wb(q)\bigr)^{1-\gamma_i}  \biggr\|\\
  &\qquad \leq (n+1) \bigl\|\wf(.)^{|S|} F(.,q)\bigr\|.
\end{align*}
We may now conclude by estimating
\begin{align*}
 &\Biggl\| \prod^n_{i=1}\bigl(H_i+\wb(q)\bigr)^{1-\gamma_i}
 \biggl\{\sum_{S\subset \llbracket 1,n\rrbracket} \prod_{j\in S} R_j\bigl(-\wb(q)\bigr) \af\bigl(\wf(.)^{|S|}F(.,q)\bigr) \biggr\} \\
 & \qquad \prod^n_{i=1}R_i(-\wb(q))^{1-\gamma_i}\Biggr\| \\
 & \quad \leq (n+1) \sum_{S\subset \llbracket 1,n\rrbracket} \prod_{j\in S}  \bigl\|R_j\bigl(-\wb(q)\bigr)  \bigr\| \bigl\|\wf(.)^{|S|} F(.,q)\bigr\|\\
 & \quad \leq (n+1) \sum_{S\subset \llbracket 1,n\rrbracket}   \frac{1}{\wb(q)^{|S|}} \bigl\|\wf(.)^{|S|} F(.,q)\bigr\|.
 \end{align*}
 Observe that 
 \[
 \frac{1}{\wb(q)^{|S|}} \bigl\|\wf(.)^{|S|} F(.,q)\bigr\| \leq  \Bigl\|\Bigl(1+\frac{\wf(.)}{\wb(q)} \Bigr)^{|S|} F(.,q)\Bigr\|.
 \]
 Moreover, using the fact that the cardinality of $S$ is at most $n$ and that there are $2^n$ subset of $\llbracket 1, n \rrbracket$, we finally have
 \begin{align*}
& \biggl\| \prod^n_{i=1}\bigl(H_i+\wb(q)\bigr)^{1-\gamma_i}
 \biggl\{\sum_{S\subset \llbracket 1,n\rrbracket} \prod_{j\in S} R_j\bigl(-\wb(q)\bigr) \af\bigl(\wf(.)^{|S|}F(.,q)\bigr) \biggr\} \\
 & \qquad \prod^n_{i=1}R_i(-\wb(q))^{1-\gamma_i}\biggr\|
  \leq (n+1)  2^{n}    \Bigl\|\Bigl(1+\frac{\wf(.)}{\wb(q)}\Bigr)^n F(.,q)\Bigr\|,
\end{align*}
which completes the proof.
\end{proof}

In the following lemma, we use the amended notation from the proof of Theorem~\ref{MainTh} in Section~\ref{Sec-MainProof}, where we add the ultraviolet cutoff function $\chi$ to the subscript of the coupling functions $G^{(1)}_\Lambda$ and $G^{(2)}_\Lambda$, provided there is more than one cutoff function in play.

\begin{lem}
\label{RegularityOfK}
Let $\ell,N\in\NN$ and let us assume that $p>\frac{d}{2}-\frac{N}{N+1}$. For $\us\in\scrS^{(\ell)}_0$, $i\in\llbracket 1,\ell\rrbracket$, $\chi$ and $\chi'$ two cutoff function fulfilling Hypothesis \ref{Hypothesis-chi}, and $\Lambda,\Lambda'\in \RR$ with $\Lambda>0$ and $\Lambda'>0$, we set
\[
\begin{aligned}
\uF^{(\ell)}_{\us,\Lambda}  &= \bigl(G_{s_1,\Lambda},\dotsc,G_{s_\ell,\Lambda}\bigr),\\
\uF^{(i,\ell)}_{\us,\Lambda,\Lambda'} &= \bigl(G_{s_1,\Lambda},\dotsc, G_{s_{i-1},\Lambda}, G_{s_i,\Lambda}-G_{s_i,\Lambda'},G_{s_{i+1},\Lambda'},\dotsc,G_{s_\ell,\Lambda'}\bigr),\\
\uF^{(i,\ell)}_{\us,\chi,\chi',\Lambda} &= \bigl(G_{s_1,\chi,\Lambda},\dotsc, G_{s_{i-1},\chi,\Lambda}, G_{s_i,\chi,\Lambda}-G_{s_i,\chi',\Lambda},G_{s_{i+1},\chi',\Lambda},\dotsc,G_{s_\ell,\chi',\Lambda}\bigr).
\end{aligned}
\]
Define $\alpha = 1 - \frac{1}{N+1}$. Then $\int \wb(q)^{-2\alpha -2p}dq<\infty$ and
\begin{enumerate}[label = \textup{(\arabic*)}]
    \item\label{item-K-bounded} We have the estimates
    \begin{align*}
         K_\alpha\bigl(\uF^{(\ell)}_{\us,\Lambda}\bigr) & \leq \max\{\|h^{(1)}\|_\infty,\|h^{(2)}\|_\infty\}^{\ell}  \|g\|^{\ell}_2 \biggl(\int \frac{1}{\wb(q)^{2\alpha+2p} } dq\biggr)^{\frac{\ell}{2}},\\
         K_\alpha\bigl(\uF^{(i,\ell)}_{\us,\Lambda,\Lambda'}\bigr) & \leq \max\{\|h^{(1)}\|_\infty,\|h^{(2)}\|_\infty\}^{\ell}  \|g\|^{\ell}_2\biggl(\int \frac{1}{\wb(q)^{2\alpha+2p} } dq\biggr)^{\frac{\ell}{2}}, \\
         K_\alpha\bigl(\uF^{(i,\ell)}_{\us,\chi,\chi',\Lambda}\bigr) & \leq \max\{\|h^{(1)}\|_\infty,\|h^{(2)}\|_\infty\}^{\ell}  \|g\|^{\ell}_2  \biggl(\int \frac{1}{\wb(q)^{2\alpha+2p} } dq\biggr)^{\frac{\ell}{2}}.
    \end{align*}
    \item\label{item-K-Cauchy} For any $\varepsilon >0$, there exists $M \in \RR$ such that for any $\us\in\scrS_0^{(\ell)}$, $i\in\llbracket 1,\ell\rrbracket$ and $\Lambda'\geq  \Lambda \geq M$, we have $K_\alpha(\uF^{(i,\ell)}_{\us,\Lambda,\Lambda'})\leq \varepsilon$.
    \item\label{item-K-Cauchy2} For any $\varepsilon >0$, there exists $M \in \RR$ such that for any $\us\in\scrS_0^{(\ell)}$, $i\in\llbracket 1,\ell\rrbracket$ and $  \Lambda \geq M$, we have $K_\alpha(\uF^{(i,\ell)}_{\us,\chi,\chi',\Lambda})\leq \varepsilon$.
\end{enumerate}
\end{lem}

\begin{proof} First note that by the constraint on $p$ and the choice of $\alpha$, we have $2\alpha +2p> d+2-\frac{2N+2}{N+1} = d$ and hence, $\int \wb(q)^{-2\alpha -2p}dq<\infty$.

    We begin with \ref{item-K-bounded}. Note that 
        \[
        K_\alpha\bigl(\uF^{(\ell)}_{\us,\Lambda}\bigr) = \prod^\ell_{i=1} \biggl(\int \frac{\bigl|G_{\us,\Lambda}(k_i,q_i)\bigr|^2}{\wb(q_i)^{2\alpha} } dk dq\biggr)^{\frac12}
        \]
        and for any $j \in \llbracket 1, \ell \rrbracket$,
        \begin{align}\label{EstK1}
          \nonumber \int \frac{\bigl|G_{s_j,\Lambda}(k_j,q_j)\bigr|^2}{\wb(q_j)^{2\alpha} } dk_j dq_j & =   \int \frac{\bigl|h^{\sharp}(k_j,q_j) g(k_j\pm q_j) \chi_\Lambda(k_j) \chi_\Lambda(q_j)\bigr|^2}{\wb(q_j)^{2\alpha+2p} } dk_j dq_j\\
          \nonumber & \leq \|h^{\sharp}\|_{\infty} \int \frac{|g(v) |^2}{\wb(q_j)^{2\alpha+2p} } dv dq_j\\
          & =\|h^{\sharp}\|^2_{\infty}  \|g\|^2_2 \int \frac{1}{\wb(q_j)^{2\alpha+2p} } dq_j.
        \end{align}
        That $K_\alpha(\uF^{(i,\ell)}_{\us,\Lambda,\Lambda'})$ satisfies the same estimate, follows from the observation that $|\chi_\Lambda(q_j)\chi_\Lambda(k_j)-\chi_{\Lambda'}(q_j)\chi_{\Lambda'}(k_j)|\leq 1$.  Moreover, $K_\alpha(\uF^{(i,\ell)}_{\us,\chi,\chi',\Lambda})$ can be estimated following the same strategy and from the observation $|\chi_\Lambda(q_j)\chi_\Lambda(k_j)-\chi'_{\Lambda}(q_j)\chi'_{\Lambda}(k_j)|\leq 1$.
        
         We now turn to \ref{item-K-Cauchy}. Let $\varepsilon>0$.  Due to \eqref{EstK1}, we have already established that
           \[
           \begin{aligned}
                    K_\alpha\bigl(\uF^{(i,\ell)}_{\us,\Lambda,\Lambda'}\bigr)  &\leq 
                   \|h^{\sharp}\|^{\ell-1}_{\infty}  \|g\|^{\ell-1}_2 \left(\int \frac{1}{\wb(q)^{2\alpha+2p} } dq \right)^{\frac{\ell-1}{2}}\\
                   & \quad 
                     \biggl(\int \frac{\bigl|G^{\sharp}_{\Lambda'}(k_i,q_i)-G^{\sharp}_{\Lambda}(k_i,q_i)\bigr|^2}{\wb(q_i)^{2\alpha} } dk_i
                     dq_i\biggr)^{\frac12}.
            \end{aligned}
        \]
Abbreviate $\tchi = \chi-1$ and for $\Lambda>0$, set $\tchi_\Lambda(k) = \tchi(k/\Lambda)$. 
We may without loss of generality assume that $\Lambda'\geq \Lambda$ and estimate, using that $|\tchi|\leq 1$,
\begin{align*}
& \int \frac{\bigl|G^{\sharp}_{\Lambda'}(k_i,q_i)-G^{\sharp}_{\Lambda}(k_i,q_i)\bigr|^2}{\wb(q_i)^{2\alpha} } dk_i dq_i \\
 &  \qquad \leq \|h^{\sharp}\|_{\infty}^2 \int \frac{|g(k\pm q)|^2\bigl|\chi_{\Lambda'}(q)\chi_{\Lambda'}(k)-\chi_\Lambda(q)\chi_\Lambda(k)\bigr|^2}{\wb(q)^{2\alpha+2p} }dq dk\\
&  \qquad  \leq \|h^{\sharp}\|_{\infty}^2 \int \frac{|g(k\pm q)|^2 4\bigl(\tchi_{\Lambda'}(q)+\tchi_\Lambda(q) +\tchi_{\Lambda'}(k)+\tchi_\Lambda(q)\bigr)^2}{\wb(q)^{2\alpha+2p} }dq dk\\
& \qquad\leq 8 \|h^{\sharp}\|_{\infty}^2\bigl(I_\Lambda+I_{\Lambda'}\bigr),
\end{align*}
where, for $\Lambda>0$, 
\[
I_\Lambda = \int \frac{|g(k\pm q)|^2\bigl|\bigl(\tchi_\Lambda(q) +\tchi_\Lambda(k)\bigr)^2}{\wb(q)^{2\alpha+2p} }dq dk
\]
By Lebesgue's dominated convergence theorem, $\lim_{\Lambda\to\infty} I_\Lambda = 0$, where we used that $\tchi$ is continuous at zero and $\tchi(0)=0$. Hence we may pick $\Lambda_0>0$ large enough, such that $K_\alpha\bigl(\uF^{(i,\ell)}_{\us,\Lambda,\Lambda'}\bigr) \leq \varepsilon$, for $\Lambda,\Lambda'\geq \Lambda_0$. This completes the proof of  \ref{item-K-Cauchy}.

Finally, we turn to \ref{item-K-Cauchy2}. Observe that 
\begin{align*}
    K_\alpha\bigl(\uF^{(i,\ell)}_{\us,\chi,\chi',\Lambda}\bigr) & \leq  \max\{\|h^{(1)}\|_\infty,\|h^{(2)}\|_\infty\}^{\ell-1}  \|g\|^{\ell-1}_2 \biggl(\int \frac{1}{\wb(q)^{2\alpha+2p} } dq\biggr)^{\frac{\ell-1}{2}} \\
    &\qquad   \biggl(\int \frac{\bigl|G^{\sharp}_{\chi,\Lambda}(k_i,q_i)-G^{\sharp}_{\chi',\Lambda}(k,q)\bigr|^2}{\wb(q)^{2\alpha} } dk
                     dq\biggr)^{\frac12}.
\end{align*}
Let $\varepsilon>0$. Estimating as we did above, we have
\begin{align*}
& \int \frac{\bigl|G^{\sharp}_{\Lambda'}(k_i,q_i)-G^{\sharp}_{\Lambda}(k_i,q_i)\bigr|^2}{\wb(q_i)^{2\alpha} } dk_i dq_i \\
& \qquad   \leq \|h^{\sharp}\|^2_{\infty}\int \frac{|g(k\pm q)|^2 4\bigl(\tchi'_{\Lambda}(q)+\tchi_\Lambda(q)+\tchi'_{\Lambda}(k)+\tchi_{\Lambda}(k)\bigr)^2}{\wb(q)^{2\alpha+2p} }dq dk,
\end{align*}
where, as for $\tchi$, we abbreviate $\tchi' = \chi'-1$ and $\tchi'_\Lambda(k) = \tchi'(k/\Lambda)$.
By Lebesgue's dominated convergence theorem, applied to the right-hand side above, we find that 
there exists $\Lambda_0>0$ such that for any $\Lambda \geq \Lambda_0$, we have $K_\alpha(\uF^{(i,\ell)}_{\us,\chi,\chi',\Lambda})\leq \varepsilon$.
\end{proof}

The following lemma follows easily by induction.

\begin{lem}
\label{lem6}
    Let $\ell\in \NN$ and consider operator-valued functions $B_i(\{p_k\}_{k=i}^\ell)$ and $A_i(p_{i})$, with $i=1,2,\dotsc,\ell$ and variables $p_1,\dotsc,p_\ell\in\RR^n$. Assume 
    \begin{align*}
       & \forall i \in \llbracket 1, \ell\rrbracket, ~ \exists b_i \in \RR: \quad \bigl\|B_i\bigl(\{p_k\}_{k=i}^\ell\bigr)\bigr\| \leq b_i, \quad \textup{for all } p_i,\dotsc,p_\ell\in\RR^n, \\
       & \forall i \in  \llbracket 1, \ell\rrbracket, ~ \exists a_i \in \RR: \quad  \int \bigl\|A_i(p_{i}) \phi\bigr\|^2  dp_{i}\leq a_i\bigl\| \phi \bigr\|^2, \quad \textup{for all } \phi \in \scrH.
    \end{align*}
    then, for all $\phi \in \scrH$,
    \[
    \int \biggl\| \prod_{i=1}^{\ell} B_i\bigl(\{p_k\}_{k=i}^\ell\bigr) A_i(p_{i}) \phi \biggr\|^2 \prod^\ell_{i=1} d p_{i}\leq \prod_{i=1}^{\ell} \bigl(a_i b^2_i\bigr) \bigl\| \phi \bigr\|^2 . 
    \]
\end{lem}

\section{Ordered Wick Monomials, motivated by examples}\label{app-oo}

   The purpose of this appendix is to discuss some examples that motivate our definition of ordered Wick monomials in Section~\ref{Sec-OrdOp}. The basic challenge is how to exploit that smeared fermionic annihilation and creation operators are bounded, when estimating regular Wick monomials. After normal ordering renormalized handed blocks, expressing them as sums of regular Wick monomials, cf. Subsection~\ref{subsec-NO-blocks},
we do not directly have such smeared objects, since the integration variables, due to the pull-trough formula, appear also in the resolvents. One first attempt of dealing with this is to undo the pull-through in order to arrive at such smeared expressions. A simple example where this works comes from the term without contractions obtained from normal ordering the operator $H^{\ab\af}(F_1) R_0(z) H^{\cb\cf}(F_2)$. The term without contractions has the following kernel: 
\begin{equation}\label{OWM-ex1}
\begin{aligned}
F_1(q_1,k_1)& F_2(q_2,k_2)\cf(k_2)\cb(q_2)  \\
 & R_0\bigl(z-\wb(q_1)-\wb(q_2)-\wf(k_1)-\wf(k_2)\bigr)\ab(q_1) \af(k_1).
\end{aligned}
\end{equation}
One may undo the normal ordering of the $\af,\cf$'s in order to recover 
\begin{equation}\label{OWM-ex1-term}
\af\bigl(\overline{F_1(q_1,\cdot)}\bigr)
\cb(q_2) R_0\bigl(z-\wb(q_1)-\wb(q_2)\bigr)\ab(q_1)  \cf\bigl(F_2(q_2,\cdot)\bigr),
\end{equation}
plus a contraction term that can be estimated separately.
Here we have completed a double integration to arrive at smeared operators.
The above expression, after another double integration, gives rise to an operator that is easily estimated, yielding an estimate with the desired UV behavior.

However, in general, we have not been able to make this strategy work. An obstacle arises for higher-order Wick monomials, when the smeared operators end up in the middle of a product of resolvents, acting as a one-way roadblock.
In fact, Lemma~\ref{RegChainFermionOp} only allows us to pull resolvents through such a smeared object in one direction, whereas our method requires us to be able to distribute fractional resolvents more freely. 

Let us now explain, first on the kernel \eqref{OWM-ex1}, how we elect instead to estimate Wick monomials. In the following computation, we write $z' = z - \wb(q_1) - \wb(q_2)$ and abbreviate $\omega = \wf$. Rewrite by adding and subtracting resolvents 
     \begin{align*}
     & R_0\bigl(z'-\omega(k_1)-\omega(k_2)\bigr)  \\ &\quad = -\omega(k_2)R_0\bigl(z'-\omega(k_1)-\omega(k_2)\bigr)R_0\bigl(z'-\omega(k_1)\bigr) + R_0\bigl(z'-\omega(k_1)\bigr) \\
     & \quad = \bigl\{\omega(k_2)R_0\bigl(z'-\omega(k_2)\bigr)\bigr\}
      \bigl\{\omega(k_1)R_0\bigl(z'-\omega(k_1)-\omega(k_2))\bigr\}
     R_0(z'-\omega(k_1))\\
     &\qquad  +\bigl\{\omega(k_2)R_0\bigl(z'-\omega(k_2)\bigr)\bigr\}\bigl\{\omega(k_1)R_0\bigl(z'-\omega(k_1)\bigr)\bigr\}R_0(z')\\
      &\qquad  -\underline{\bigl\{\omega(k_2)R_0\bigl(z'-\omega(k_2)\bigr)\bigr\}}R_0(z')
      \\  & \qquad 
     - \underline{R(z')\bigl\{\omega(k_1) R_0\bigl(z'-\omega(k_1)\bigr)\bigr\} } + \underline{R(z')}.  
     \end{align*}
     The first two summands (not underlined) can be dealt with using kinetic energy bounds that do not exploit fermi statistics.
Inserting the three underlined terms into \eqref{OWM-ex1} and completing the $k_1$ and $k_2$ integrals, yield
\begin{align*}
&-\int F_2(q_2,k_2) \cf(k_2)\cb(q_2) \omega(k_2)R_0\bigl(z'-\omega(k_2)\bigr)R_0(z')\ab(q_1) \af\bigl(\overline{F_1(\cdot,q_1)}\bigr) dk_2\\
& \quad - \int F_1(q_1,k_1) \cf\bigl(F_2(\cdot,q_2)\bigr)\cb(q_2) R(z')\omega(k_1) R_0\bigl(z'-\omega(k_1)\bigr)\ab(q_1) \cf(k_1) dk_1\\
& \quad +  \cf\bigl(F_2(\cdot,q_2)\bigr)\cb(q_2) R(z') \ab(q_1) \af\bigl(\overline{F_1(\cdot,q_1)}\bigr).
\end{align*}
The third summand above is almost the same as the one we arrived at in \eqref{OWM-ex1-term}, except the smeared $\cf$ and $\af$ factors have switched places. 
In all three terms above, we have effectively gained a full resolvent that is not needed to control the remaining pointwise annihilation and creation operators and therefore supplies extra momentum decay. After integrating out $q_1$ and $q_2$, one may estimate, using boundedness of the smeared fermionic operators and kinetic energy bounds for the rest, gaining an extra momentum decay of the order $(1+|q_1|)^{-\frac12} (1+|q_2|)^{-\frac12}$, apart from what is gained from the kinetic energy bounds alone. 

For Wick monomials of higher order, with multiple resolvents, several of which may include an $\wf(k_j)$ energy shift that should be removed, we have to make a choice of order in which we add and subtract resolvents in a telescopic expansion. We have chosen a particular order that seems to minimize the number of times one has to invoke (almost) momentum conservation to pass between decay in $q_j$ and decay in $k_j$.

\noindent\textbf{Acknowledgement.} The authors are grateful for support by the Independent Research Fund Denmark, via the project grant “Mathematical Aspects of Ultraviolet Renormalization” (8021-00242B) and for the support of CNRS, via the project PEPS JCJC 2023.
Finally, the second author would like to thank Jeremy Faupin and Universit\'e de Lorraine for hospitality.

\bibliographystyle{amsalpha}

\listoffixmes

\end{document}